\documentclass[11pt, letterpaper]{article}

\usepackage{amsfonts}
\usepackage{amssymb}
\usepackage{amsmath}
\usepackage{fancyhdr}
\usepackage{bm}
\usepackage{comment}
\usepackage[margin=1in]{geometry}
\usepackage[
            CJKbookmarks=true,
            bookmarksnumbered=true,
            bookmarksopen=true,
            colorlinks=true,
            citecolor=red,
            linkcolor=blue,
            anchorcolor=red,
            urlcolor=blue,
            ]{hyperref}

\usepackage{hyperref}

\usepackage{subcaption}
\usepackage{graphicx}

\usepackage[ruled,vlined,linesnumbered]{algorithm2e}

\makeatletter

\@addtoreset{algocf}{section}
\makeatother

\hypersetup{
  colorlinks   = true, 
  urlcolor     = blue, 
  linkcolor    = blue, 
  citecolor   = red 
}

\usepackage[square,sort&compress,numbers]{natbib}
\usepackage{amsmath}
\usepackage{amsthm}
\usepackage{dsfont}
\usepackage{graphicx}
\usepackage{enumerate}
\usepackage[hypcap]{caption}
\usepackage{color, soul}
\usepackage{float}
\usepackage{dsfont}
\usepackage{mathtools}
\usepackage{xcolor}

\usepackage{booktabs}

\usepackage{tikz}
\usetikzlibrary{shadings,arrows.meta,positioning,calc,decorations.pathreplacing,patterns,patterns.meta}

\usepackage[font=small]{caption} 
\usepackage{wrapfig}

\newtheorem{theorem}{Theorem}[section]

\newtheorem{remark}{Remark}[section]
\newtheorem{lemma}{Lemma}[section]
\newtheorem{claim}{Claim}[section]
\newtheorem{corollary}{Corollary}[section]

\newtheorem{conjecture}{Conjecture}[section]
\newtheorem{assump}{Believed Hardness}[section]
\newtheorem{proposition}{Proposition}[section]
\newtheorem{fact}{Fact}[section]
\theoremstyle{definition}
\newtheorem{definition}{Definition}[section]

\newtheorem{openproblem}{Open Problem}

\newtheorem{innercustomlem}{Lemma}
\newenvironment{mylem}[1]
  {\renewcommand\theinnercustomlem{#1}\innercustomlem\itshape}
  {\endinnercustomlem}

    \newtheorem{innercustomprop}{Proposition}
\newenvironment{myprop}[1]
  {\renewcommand\theinnercustomprop{#1}\innercustomprop\itshape}
  {\endinnercustomprop}

\newtheorem{innercustomcor}{Corollary}
\newenvironment{mycor}[1]
  {\renewcommand\theinnercustomcor{#1}\innercustomcor\itshape}
  {\endinnercustomcor}

\newtheorem{innercustomclaim}{Claim}
\newenvironment{myclaim}[1]
  {\renewcommand\theinnercustomclaim{#1}\innercustomclaim\itshape}
  {\endinnercustomclaim}

\newtheorem{innercustomdef}{Definition}
\newenvironment{mydef}[1]
  {\renewcommand\theinnercustomdef{#1}\innercustomdef\itshape}
  {\endinnercustomdef}

\renewcommand{\epsilon}{\rho}

\newcommand{\vio}[1]{\textcolor{violet}{#1}}

\newcommand{\gb}[1]{}
\newcommand{\ah}[1]{}

\newcommand{\pr}{{p_R}}
\newcommand{\pl}{{p_L}}

\newcommand{\M}{\mathsf{E}}
\newcommand{\Mr}{{\mathsf{E}_R}}
\newcommand{\Ml}{{\mathsf{E}_L}}

\newcommand{\Mu}{M}
\newcommand{\Mur}{M_R}
\newcommand{\Mul}{M_L}

\newcommand{\Xg}{X^G}
\newcommand{\Xgl}{X^G_L}
\newcommand{\Xgr}{X^G_R}
\newcommand{\Xr}{X_R}
\newcommand{\Xl}{X_L}
\newcommand{\deltaR}{\delta_R}
\newcommand{\deltaL}{\delta_L}

\newcommand{\Y}{\mathsf{Y}}
\newcommand{\Yb}{\mathbf{Y}}
\newcommand{\Zb}{\mathbf{Z}}
\newcommand{\Ss}{\mathsf{S}}
\newcommand{\Zt}{\mathsf{Z}}

\newcommand{\full}{\mathsf{full}}
\newcommand{\norm}{\psi}

\newcommand{\law}{\mathcal{L}}

\newcommand{\Tr}{\mathsf{Tr}}

\newcommand{\mat}[2]{\mathsf{mat}_{#1 \times #2}}
\newcommand{\tenso}{\mathsf{tens}}
\newcommand{\veco}{\mathsf{vec}}
\newcommand{\supp}{\mathsf{supp}}

\newcommand{\tr}{^\top}
\newcommand{\E}{\mathbb{E}}

\renewcommand{\P}{\mathbb{P}}

\newcommand{\unif}{\mathsf{Unif}}
\newcommand{\bin}{\mathsf{Bin}}
\newcommand{\poi}{\mathsf{Pois}}
\newcommand{\ber}{\mathsf{Ber}}
\newcommand{\bern}[1]{\mathsf{Ber}\paren{#1}}
\newcommand{\rad}{\mathsf{Rad}}

\newcommand{\maj}{\mathsf{MAJ}}
\newcommand{\sign}{\mathrm{sign}}

\newcommand{\GOE}{\mathsf{GOE}}
\newcommand{\poly}{\mathrm{poly}}

\newcommand{\alphab}{\overline{\alpha}}

\newcommand{\Pois}{\mathsf{Pois}}
\newcommand{\FULL}{\mathsf{FULL}}
\newcommand{\wor}{\mathsf{wor}}
\newcommand{\asymm}{\mathsf{asymm}}
\newcommand{\nulll}{\mathsf{NULL}}

\newcommand{\kxor}{k\mathsf{\text{-}XOR}}
\newcommand{\ktens}{k\mathsf{\text{-}Gauss}}
\newcommand{\kmodel}{k\mathsf{\text{-}model}}

\newcommand{\kxorvar}[1]{\kxor^{#1}}
\newcommand{\ktensvar}[1]{\ktens^{#1}}

\newcommand{\kxork}[1]{#1\mathsf{\text{-}XOR}}
\newcommand{\ktensk}[1]{#1\mathsf{\text{-}Gauss}}
\newcommand{\kmodelk}[1]{#1\mathsf{\text{-}model}}

\newcommand{\xorvar}[5]{\kxork{#1}^{#5}(#2, #3, #4)}
\newcommand{\gtensvar}[5]{\ktensk{#1}^{#5}(#2, #3, #4)}

\newcommand{\xorpvar}[5]{\kxork{#1}_{#3, #4}^{#5}(#2)}
\newcommand{\gtenspvar}[5]{\ktensk{#1}_{#3, #4}^{#5}(#2)}

\newcommand{\xor}[4]{\xorvar{#1}{#2}{#3}{#4}{}}
\newcommand{\gtens}[4]{\gtensvar{#1}{#2}{#3}{#4}{}}
\newcommand{\xorp}[4]{\xorpvar{#1}{#2}{#3}{#4}{}}
\newcommand{\gtensp}[4]{\gtenspvar{#1}{#2}{#3}{#4}{}}
\newcommand{\xorstar}[4]{\xorvar{#1}{#2}{#3}{#4}{\star}}
\newcommand{\gtensstar}[4]{\gtensvar{#1}{#2}{#3}{#4}{\star}}
\newcommand{\xorpstar}[4]{\xorpvar{#1}{#2}{#3}{#4}{\star}}
\newcommand{\gtenspstar}[4]{\gtenspvar{#1}{#2}{#3}{#4}{\star}}

\newcommand{\xorpnull}[3]{\kxork{#1}_{#3}^{\nulll}(#2)}
\newcommand{\gtenspnull}[4]{\ktensk{#1}_{#3, #4}^{\nulll}(#2)}

\newcommand{\xorpfull}[4]{\xorpvar{#1}{#2}{#3}{#4}{\FULL}}
\newcommand{\gtenspfull}[4]{\gtenspvar{#1}{#2}{#3}{#4}{\FULL}}
\newcommand{\xorpoi}[4]{\xorvar{#1}{#2}{#3}{#4}{\Pois}}
\newcommand{\gtenspoi}[4]{\gtensvar{#1}{#2}{#3}{#4}{\Pois}}
\newcommand{\xorppoi}[4]{\xorpvar{#1}{#2}{#3}{#4}{\Pois}}
\newcommand{\gtensppoi}[4]{\gtenspvar{#1}{#2}{#3}{#4}{\Pois}}
\newcommand{\kxorfull}{\kxorvar{\FULL}}

\newcommand{\kxorpoi}{\kxorvar{\Pois}}
\newcommand{\ktenspoi}{\ktensvar{\Pois}}

\newcommand{\lwe}{\mathsf{LWE}}
\newcommand{\klwe}{k\mathsf{\text{-}LWE}}
\newcommand{\klwek}[1]{#1\mathsf{\text{-}LWE}}

\newcommand{\lweg}[4]{#1\mathsf{\text{-}LWE}{(#2, #3, #4)}}

\newcommand{\lwegpoi}[4]{#1\mathsf{\text{-}LWE}^{\mathsf{Pois}}{(#2, #3, #4)}}
\newcommand{\lwegnull}[3]{#1\mathsf{\text{-}LWE}^{\mathsf{NULL}}{(#2, #3)}}

\newcommand{\conj}[2]{\mathsf{CONJ}\paren{#1, #2}}
\newcommand{\conjeps}[1]{\mathsf{CONJ}_{#1}}
\newcommand{\conjD}[2]{\mathsf{CONJ}^\mathsf{D}\paren{#1, #2}}

\newcommand{\conjR}[2]{\mathsf{CONJ}^\mathsf{R}\paren{#1, #2}}

\newcommand{\noise}{\mathsf{W}}
\newcommand{\noiseG}{\mathsf{G}}

\newcommand{\neff}{n_{\mathsf{eff}}}

\newcommand{\kgtensppoik}[1]{#1\mathsf{\text{-}Gauss}^{\mathsf{Pois}}}

\newcommand{\gtensppoinull}[3]{#1\mathsf{\text{-}Gauss}^{\mathsf{Pois, NULL}}_{#3}(#2)}
\newcommand{\gtenspoinull}[3]{#1\mathsf{\text{-}Gauss}^{\mathsf{Pois, NULL}}(#2, #3)}

\newcommand{\tensorpca}[3]{#1\mathsf{\text{-}TensorPCA}(#2, #3)}
\newcommand{\tensorpcap}[3]{#1\mathsf{\text{-}TensorPCA_{#3}}(#2)}
\newcommand{\ktensorpca}{k\mathsf{\text{-}TensorPCA}}
\newcommand{\ktensorpcak}[1]{#1\mathsf{\text{-}TensorPCA}}

\newcommand{\pcag}[3]{#1\mathsf{\text{-}tPCA}^{\mathsf{Gauss}}(#2, #3)}

\newcommand{\Gaussianize}{\mathsf{Gaussianize}}
\newcommand{\Discr}{\mathsf{Discr}}
\newcommand{\Discretize}{\mathsf{Discretize}}
\newcommand{\GaussClone}{\mathsf{GaussClone}}
\newcommand{\kTensClone}{\mathsf{kGaussClone}}
\newcommand{\kTensSplit}{\mathsf{kSplit}}

\newcommand{\XorClone}{\mathsf{kXORClone}}
\newcommand{\XorSplit}{\mathsf{kSplit}}

\newcommand{\TwoPath}{\mathsf{DiscrRes}}
\newcommand{\GaussResolution}{\mathsf{GaussRes}}
\newcommand{\GaussResolutionTwo}{\mathsf{GaussResTwo}}
\newcommand{\TwoPathTwo}{\mathsf{DiscrResTwo}}
\newcommand{\ReduceK}{\mathsf{ReduceK}}

\newcommand{\DenseRed}{\mathsf{Densify}}
\newcommand{\SparseRed}{\mathsf{Sparsify}}

\newcommand{\SparseToDense}{\mathsf{SparseToDense}}

\newcommand{\Symm}{\mathsf{Symm}}
\newcommand{\DeSymm}{\mathsf{DeSymm}}
\newcommand{\DiscrResolutionLWE}{\mathsf{DiscrResLWE}}
\newcommand{\ResampleIID}{\mathsf{ResamplePois}}
\newcommand{\ResampleM}{\mathsf{ResampleM}}

\newcommand{\CollisionDetection}{\mathsf{CollisionDetection}}
\newcommand{\Decompose}{\mathsf{Decompose}}
\newcommand{\Compose}{\mathsf{Compose}}
\newcommand{\ReduceKDiscrete}{\mathsf{ReduceKDiscrete}}
\newcommand{\ReduceKGauss}{\mathsf{ReduceKGauss}}
\newcommand{\ToTPCA}{\mathsf{ToWORSampling}}
\newcommand{\ToPoiFromTPCA}{\mathsf{ToWRSampling}}

\newcommand{\DiscrEqComb}{\mathsf{DiscrEqRes}}
\newcommand{\GaussEqComb}{\mathsf{GaussEqRes}}
\newcommand{\EqComb}{\mathsf{EqRes}}

\newcommand{\vareps}{\varepsilon}

\newcommand{\la}{\langle}
\newcommand{\ra}{\rangle}
\newcommand{\paren}[1]{\left( #1 \right)}
\newcommand{\sqb}[1]{\left[ #1 \right]}
\newcommand{\set}[1]{\left\{ #1 \right\}}
\newcommand{\B}{\Big}
\renewcommand{\b}{\big}

\newcommand{\binomset}[2]{\binom{\sqb{#1}}{#2}}

\newcommand{\Thetat}{\widetilde{\Theta}}
\newcommand{\Ot}{\widetilde{O}}
\newcommand{\defi}{\eta}
\newcommand{\kstar}{k'}
\newcommand{\snr}{\nu}
\newcommand{\mpoi}{m_{\mathsf{p}}}
\newcommand{\mpoit}{\widetilde{m}_{\mathsf{p}}}

\newcommand{\eps}{\epsilon}

\newcommand{\calN}{\mathcal{N}}
\newcommand{\calS}{\mathcal{S}}
\newcommand{\Ccal}{\mathcal{C}}

\newcommand{\N}{\mathcal{N}}
\newcommand{\Pcal}{\mathcal{P}}
\newcommand{\Qcal}{\mathcal{Q}}

\newcommand{\calA}{\mathcal{A}}
\newcommand{\calY}{\mathcal{Y}}
\newcommand{\calZ}{\mathcal{Z}}
\newcommand{\qq}{\mathcal{Q}}

\newcommand{\R}{\mathbb{R}}
\newcommand{\Z}{\mathbb{Z}}

\newcommand{\otimest}{\widetilde\otimes}
\newcommand{\eqdist}{\stackrel{d}=}

\newcommand{\X}{\mathbf{X}}

\newcommand{\up}[1]{^{(#1)}}

\newcommand{\tv}{d_\mathsf{TV}}
\newcommand{\TV}{d_\mathsf{TV}}
\newcommand{\kl}{\mathsf{KL}}

\title{Average-Case Reductions for $k$-XOR and Tensor PCA}

\author{Guy Bresler\thanks{Supported by NSF Grant CCF-2428619.} \and Alina Harbuzova\thanks{Supported by MathWorks Fellowship and Siebel Scholarship.}}
\date{Massachusetts Institute of Technology}

\begin{document}

\maketitle
\begin{abstract}
    We study the noisy planted $k\mathsf{\text{-}XOR}$ problem, where for a $\{\pm1\}^n$ signal vector $x$, one observes $m$ entries of $x^{\otimes k}$, each independently flipped with probability $1/2-\delta$. We relate the computational properties of $\kxor$ problems with different parameters $(k,m,\delta)$ via poly-time average-case reductions, transferring algorithms for detection and recovery, as well as hardness results. Our central application is a new connection to Tensor PCA -- another canonical planted average-case problem.

    In Tensor PCA, one observes all $n^k$ entries of $x^{\otimes k}$ with additive Gaussian noise and mean $\delta x^{\otimes k}$, at a much smaller signal level $\delta$ than that of $k\mathsf{\text{-}XOR}$. We unify these settings by formalizing a \emph{family of $\kxor$ problems}, parametrized by order $k$, number of samples $m$, and noise level $\delta$. Via reductions, we formally relate Tensor PCA to an extremely noisy dense regime of $\kxor$.
    
    Our main contribution is a broad network of average-case reductions within the family of $\kxor$ problems across all densities $m \in \sqb{1, n^k}$. The denser $n^{k/2} \leq m \leq n^k$ regime interpolates between $k\mathsf{\text{-}XOR}$ with $m\approx n^{k/2}$ entries and Tensor PCA. The conjectured computational threshold $m\delta^2 \asymp n^{k/2}$ (requiring $\delta \ll 1$) here is well understood; in this regime our reductions achieve tight parameter tradeoffs and preserve proximity to this threshold, yielding new hardness results for all conjectured-hard instances. For example, we reduce any $k\mathsf{\text{-}XOR}$ instance at the computational threshold (i.e., $m = n^{k(1+\varepsilon)/2}$ and $\delta \approx n^{-k\varepsilon/4}$) to Tensor PCA at the computational threshold ($m=n^k$ and $\delta \approx n^{-k/4}$). Additionally, we give new order-reducing maps (e.g., $5\to 4$ $k\mathsf{\text{-}XOR}$ and $7\to 4$ Tensor PCA) at fixed entry density. In the sparser $m \leq n^{k/2}$ regime, our reductions yield a new connection between instances of different orders: for example, we reduce $7\mathsf{\text{-}XOR}$ with $m = n^{3.4}$ to the classical setting of $3\mathsf{\text{-}XOR}$ with $m' = \Thetat(n^{1.4})$, with constant noise in both problems.

    Our reductions are built on an equation resolution primitive -- multiplying pairs of entries -- analyzed in both the discrete and Gaussian settings. Using a computational equivalence between these formulations, we freely pass between them, using whichever yields stronger reductions at the given parameter. For the Gaussian setting, the correctness of our reduction rests on a new high-dimensional CLT for Gramians of sparse Gaussian matrices. In the discrete setting, we design a dependence-avoiding resolution procedure. Our framework extends beyond the binary case, yielding analogous reductions for $k$-sparse $\mathsf{LWE}$ over $\mathbb{F}_q$.

    Taken together, these results establish a hardness partial order of discrete and Gaussian problem variants in the space of tensor order $k$, entry density, and noise level, advancing a unified study of planted tensor models and relationships between them. This unified perspective and our results lead us to formulate numerous open problems. 
 
\end{abstract}

\pagenumbering{gobble}

\newpage
\setcounter{tocdepth}{2}
\hypersetup{linktoc=all}

\tableofcontents
\newpage

\newpage
\pagenumbering{arabic}
\setcounter{page}{1}
\section{Introduction}\label{sec:intro}

This paper gives poly-time average-case reductions within a family of planted noisy $\kxor$ problems, trading off order $k$, number of equations $m$, and noise level $\delta$. The central application is a new connection between $\kxor$ and Tensor PCA, two canonical planted average-case problems.

\paragraph{Planted $\kxor$ problem.}
In $\xor k n m \delta$ (also called $k\mathsf{\text{-}LIN}$), there is a secret \emph{signal vector} $x=(x_1,\dots, x_n) \in \set{\pm 1}^n$. We observe $m$ samples $\paren{\alpha, \Y}$, where for each sample $\alpha = \set{i_1,\dots,i_k}$ is a uniformly random\footnote{ Our $\kxor$ model assumes that the $m$ index sets $\alpha$ are sampled independently \emph{with replacement}; other sampling procedures for $\alpha$ are explored in \autoref{subsec:poi_sampling} and \autoref{subsec:w_replacement}.} $k$-subset of $\sqb{n}$ and
\begin{equation}\label{eq:canonical_kxor}
    \Y = \begin{cases} x_{i_1}\cdots x_{i_k} & \text{ w. prob. }1/2+\delta/2\,,\\
    - x_{i_1}\dots x_{i_k} & \text{ otherwise.}\end{cases}
\end{equation}
    Via the natural $\mathbb{F}_2 \leftrightarrow \set{\pm1}$ map, each sample is a noisy linear equation $y=\langle a,s\rangle + e$ over $\mathbb{F}_2$ where $a$ has $k$-sparse support, $s_i\in\{0,1\}$ with $x_i=(-1)^{s_i}$, and $e\sim \mathsf{Bern}(1/2-\delta/2)$. 
    Two fundamental $\kxor$ tasks are \emph{detection} (distinguish these samples from pure noise) and \emph{recovery} (estimate the secret vector $x$). The most extensively studied regime (which we will henceforth call the \emph{canonical regime}) has $\delta = \Theta(1)$ (e.g., $1/3$), in which case poly-time algorithms are known for $m = \widetilde\Omega(n^{k/2})$.
    As discussed in \autoref{subsec:complexity_profiles}, both detection and recovery are believed to be \emph{computationally hard} (i.e., there exists no polynomial-time algorithm) whenever $m \ll n^{k/2}$. Noisy planted $\kxor$ is a 
canonical problem in average-case complexity that is of interest in its own right~\cite{feige2002relations,alekhnovich2003more,feige2006witnesses,barak2016noisy,allen2015refute,raghavendra2017strongly,kothari17soslbanycsp,guruswami2021algorithms,guruswami2023efficient}. Average-case hardness of $\kxor$ tasks has been used to derive hardness results for statistical tasks: agnostically learning halfspaces
\cite{daniely2016complexityHalf}, learning DNFs \cite{daniely2021local,bui2024structured},  and noisy tensor completion~\cite{bangachev2024near}.
It has also found cryptographic applications \cite{alekhnovich2003more,ishai08cryptoconstant,applebaum2010public,dottling12indcca,dao23multipartyhomomorphic,ragavan24io}. 

Eq.~\eqref{eq:canonical_kxor} in fact defines a \emph{family} of planted noisy $\kxor$ problems indexed by the number of equations $m$ and noise level $\delta$. How are the computational properties of these problems related?
We prove \emph{poly-time average-case reductions} within this family, establishing a hardness partial order among its members. Concretely, for many parameter pairs $(k,m,\delta)$ and $(k',m',\delta')$, we show reductions 
\begin{equation*}\label{eq:ex_reduction}
\text{from}\qquad \xor k n m \delta \qquad\text{to}\qquad \xor {k'} n {m'} {\delta'}\,.
\end{equation*}
These reductions transfer poly-time algorithms for detection and recovery, usually up to $\poly\log$ losses in parameters, therefore also transferring hardness. Our reductions are the first to trade off the signal strength $\delta$ against the number of equations $m$. 
Prior $\kxor$ reductions instead focus on search to decision \cite{applebaum2012pseudorandom, bogdanov2019xor}, mapping from dense to sparse coefficient vectors (i.e., $k = n \to O(1)$) \cite{bangachev2024near}, etc. 

One of our main results is a new connection to Tensor PCA (defined next): for this, we study a denser $m\in\sqb{n^{k/2}, n^k}$ regime, where the interesting behavior occurs at higher noise rates (smaller $\delta$). In the $m \geq n^{k/2}$ regime, all our reductions preserve proximity to the computational threshold $m\delta^2 \asymp n^{k/2}$, yielding hardness results for all conjectured-hard instances.

\paragraph{Tensor PCA problem.} In the canonical statistical model for 
Tensor PCA \cite{richard2014statistical}, denoted $\tensorpca k n \delta$, there is a secret \emph{signal vector} $x \in \set{\pm1}^n$, and we observe an order-$k$ tensor
\begin{equation}\label{eq:tensor_pca}
    \Y = \delta x^{\otimes k} + \noiseG\,,
\end{equation}
where $\noiseG$ is a symmetric tensor of $\N(0,1)$. 
Entry $\Y_{i_1,i_2,\dots, i_k}$ of this tensor has the form $\delta x_{i_1}x_{i_2}\cdots x_{i_k}+G_{\set{i_1,i_2,\dots, i_k}}$.
The \emph{detection} and \emph{recovery} problems for Tensor PCA are central in the fields of high-dimensional statistics \cite{lesieur2017statistical,arous2020algorithmic,chen2021spectral,arous2021online,han2022optimal,luo2022tensor,dudeja2024statistical} and statistical physics \cite{arous2019landscape,ros2019complex,wein2019kikuchi,montanari2024friendly,biroli2020iron,el2021optimization}; tensor PCA has also been intensely studied as a prototypical average-case planted problem \cite{hopkins2015tensor,bhattiprolu2016sum,hopkins2017power,kunisky2019notes,brennan2020reducibility,kunisky2024tensor}. As discussed in \autoref{subsec:complexity_profiles}, for $\delta \gg n^{-k/4}$ there are polynomial-time algorithms and below this level, i.e., $\delta \ll n^{-k/4}$, both tasks are believed to be computationally hard.

\bigskip
The models \eqref{eq:canonical_kxor} and \eqref{eq:tensor_pca} are in some respects quite different, and each has its own line of work, even though they exhibit closely related algorithmic phenomena. In particular, both admit spectral, SOS, and Kikuchi-based algorithms, but these have typically been analyzed separately in the two settings. The planted signal $x_{i_1}x_{i_2}\cdots x_{i_k}$ in each entry has the same form, so $\tensorpca k n \delta$ resembles the extremely dense and noisy $\xor k n m \delta$ with $m \asymp n^k$. Can we relate Tensor PCA to the more canonical $\kxor$ regime with only $m \approx n^{k/2}$ equations?
As one possible goal, can one problem be algorithmically reduced to the other, thereby transferring algorithms and hardness conjectures between the two?

It turns out that Tensor PCA can be related, via cosmetic average-case reductions, to the extreme dense point $\xor k n {m=n^k} \delta$; see \autoref{subsec:family_of_kxor}. Our reductions within the $\kxor$ family therefore yield a new connection to Tensor PCA: in Theorem~\ref{thm:example} we reduce $\kxor$ with $m \approx n^{k/2}$ and $ \delta \approx 1$ to $\ktensorpcak {k'}$. In essence, the conjectured $\kxor$ computational threshold at $m\approx n^{k/2}$ implies the believed threshold for Tensor PCA. Beyond hardness implications, our reductions show a fundamental connection between the problems.

\begin{theorem}[Corollary of Thm.~\ref{thm:densifying_general}]\label{thm:example}
Fix constants $k' \geq 2, \vareps \in (0,1)$, and $\defi > 0$. For some sufficiently large $k = k(k',\vareps)$, if there is \textbf{no} polynomial-time algorithm solving detection (resp. recovery) for $\xor k n m {\delta^\mathsf{xor}}$ with
\vspace{-2mm}
$$
\delta^{\mathsf{xor}} = n^{-k\vareps/4 - \defi/2}\text{  and  } m = n^{k(1+\vareps)/2}\,,
$$ 
then there is \textbf{no} polynomial-time algorithm solving detection (resp. recovery) for $\tensorpca {k'} {n} {\delta^{\mathsf{tpca}}}$ and signal level 
$$
\delta^{\mathsf{tpca}} = \Thetat\b(n^{-k'/4 - \defi}\b)\,.
$$
\end{theorem}

Theorem~\ref{thm:example} is part of a broader family of reductions along the $(m,\delta)$ curve (\autoref{subsec:results_intro}).
We also obtain reductions that \emph{change the tensor order $k$}; for example:

\begin{theorem}[Corollary of Thm.~\ref{thm:decrease_k_intro}]\label{thm:example2}
Let $P, Q$ be either of the following pairs:
\begin{enumerate}
    \item $P$ is $\xor 5 n m \delta$ with $m = n^{5/2}$ and $\delta = n^{-\defi/2}$ (for any fixed $\defi > 0$), and $Q$ is $\xor 4 {n} {m'} {\delta'}$ with $m' = \Thetat\b(n^{2}\b)$ and $\delta' = \Thetat\b(n^{-\defi}\b)$;
    \item $P$ is $\tensorpca 7 n \delta$ with $\delta = n^{-7/4-\defi/2}$ (for any fixed $\defi > 0$), and $Q$ is $\tensorpca 4 {n} {\delta'}$ with $\delta' = \Thetat\b(n^{-1 - \defi}\b)$.
\end{enumerate}
If there is \textbf{no} polynomial-time algorithm solving detection (resp. recovery) for $P$, then there is \textbf{no} polynomial-time algorithm solving detection (resp. recovery) for $Q$.
\end{theorem}

\paragraph{Technique: equation resolution.} Our reductions are built around an \emph{equation resolution} primitive, which combines two input equations into a new $\kxork {k'}$ equation on the symmetric-difference support, implemented by multiplying two values in Eq.~\eqref{eq:canonical_kxor} (equiv. adding two equations over $\mathbb{F}_2$), see \autoref{sec:ideas} for the overview of main reduction ideas. Crucially, since we work in the average case, our reductions must map to the correct (approximately in total variation) output distribution. Thus, while forming linear combinations of $\kxor$ constraints is a standard operation (e.g., in Gaussian elimination or in proof systems for refuting $\kxor$; see \autoref{subsec:resolution_lit}), we use this primitive differently by designing resolution patterns that are distributionally correct. The same framework extends beyond binary equations, yielding analogous reductions for $k$-sparse $\lwe$ over larger fields $\mathbb{F}_q$ (\autoref{subsec:resolution_lwe}).

\subsection{Notation and Preliminaries}\label{subsec:family_of_kxor}

We first establish some notation, formalize the detection and recovery tasks and the connection between $\kxor$ and $\ktensorpca$.

\begin{remark}[Polynomial Scale]
We treat $k$ as a constant, consider sequences of problems indexed by $n \to \infty$, and work at polynomial scales: $\log_n m$, and $\log_n \delta$ are viewed as constants. 
As is made precise in \autoref{sec:avg_case_defs} (Prop.~\ref{prop:map_up_to_constant}), two model sequences with parameters $\delta_n, p_n$ and $\delta'_n,p'_n$ such that $\lim_n \log_n\delta_n=\lim_n \log_n\delta_n'$ and $\lim_n \log_n m_n =\lim_n \log_n m_n'$
are equivalent for our purposes.
\end{remark}
\begin{definition}[Density Parameter $\eps$]\label{def:density}

We view $\log_n m$, where $m$ is the number of observed entries in $\kxor, \ktens$, as constant. Accordingly, we parametrize $$m = n^{k(\eps+1)/2} \quad \text{for }\eps \in (-1,1]\,.$$ Canonical $\kxor$ has $m=n^{k/2}$, and hence $\eps=0$, and $\ktensorpca$ has $m = n^k$, and hence $\eps = 1$.\footnote{Models differing by $\poly\log n$ factors in parameters $m,\delta$ are equivalent for our purposes, see \autoref{sec:avg_case_defs}.}
\end{definition}

\paragraph{Tasks.} We formalize the following detection and recovery tasks in \autoref{sec:avg_case_defs}:
\begin{enumerate}
    \item \textbf{Detection:} distinguish an instance of $\kxor$ from a null (pure noise) instance (see Remark~\ref{rmk:null} below) that does not depend on the signal vector $x$;
    \item \textbf{Recovery:} given a $\kxor$ instance, estimate the secret vector $x\in\set{\pm 1}^n$, either exactly or approximately. 
\end{enumerate}
\begin{remark}[Choice of Null Distribution]\label{rmk:null}
    For $\xor k n m \delta$ and $\tensorpca k n \delta$ we define the null distribution by replacing the structured signal $x_{i_1}\dots x_{i_k}$ in each equation/entry with a random sign $w \sim \unif\b(\set{\pm1}\b)$. 
    A common alternative is to set the SNR parameter $\delta=0$. These coincide exactly for the discrete $\xor k n m \delta$; for $\tensorpca k n {\delta\ll n^{-k/4}}$, the two versions are statistically indistinguishable (i.e., their total variation is $o(1)$, see Prop.~\ref{prop:tv_null}). 
\end{remark}

Our results address both the detection and recovery tasks. The separate \emph{refutation} task (certifying an upper bound on the satisfiable fraction under the null) is not our focus; our reductions do not directly apply to it and we leave it as an interesting direction for future research.

\paragraph{Connection to Tensor PCA.} 

To precisely capture $\ktensorpca$, we require several cosmetic reductions between $\xor k n {m=n^k} \delta$ and the standard form in Eq.~\eqref{eq:tensor_pca} (see Cor.~\ref{cor:tensor_pca_wr_equiv}, \autoref{subsubsec:equiv_w_wor_replacement}):

\begin{proposition}[Computational Equivalence of $\kxor$ and Canonical Tensor PCA]\label{prop:intro_tpca_equiv}
    For any $k, n,$ and $ \delta \ll n^{-k/4}$, $\kxor(n,m=n^k,\delta)$ is computationally equivalent to the Tensor PCA model $\Y = \delta x^{\otimes k} + \noiseG$ of Eq.~\eqref{eq:tensor_pca} with only distinct-index entries observed.
    A $\kxor$ model variant $\kxor^\FULL(n,m=n^k,\delta)$ (defined in \autoref{subsec:kxor_variants}) that allows index repetitions in observed equations is computationally equivalent to $\ktensorpca(n, \delta)$ as in Eq.~\eqref{eq:tensor_pca}.

    Moreover, all our results hold verbatim for $\ktensorpca(n, \delta)$ in place of $\kxor(n,m=\widetilde\Theta(n^k),\delta)$.
\end{proposition}

\section{Main Results}\label{subsec:results_intro}

We study the computational complexities of $\xor k n m \delta$ and $\tensorpca k n \delta$. We do so by building \emph{average-case reductions in total variation} between instances with different $(k, n, m ,\delta)$ parameters and consider their algorithmic and hardness implications. We organize the results by regime: in the denser case $m \in \sqb{n^{k/2}, n^k}$ and the sparser case $m \leq n^{k/2}$, we first overview known algorithmic and hardness results, followed by our reductions and their implications.

\subsection{Results for Dense $\kxor$ with $m \in \sqb{n^{k/2}, n^k}$}\label{subsec:complexity_profiles}

\paragraph{Complexity profile overview.}
As stated above, at noise rate $\delta = \Theta(1)$, known polynomial-time algorithms solve both detection and recovery tasks for $\xor k n m \delta$ if $m = \widetilde\Omega(n^{k/2})$. In the regime of $m \in \sqb{n^{k/2}, n^k}$, the interesting behavior thus occurs at high noise rates $\delta \ll 1$. State-of-the-art polynomial-time algorithms and information-theoretic lower bounds for $\xor k n m \delta$ with $m \in \sqb{n^{k/2}, n^k}$ (equiv., $m = n^{k(1+\eps)/2}$ for $\eps \in [0,1]$) are summarized in Fig.~\ref{fig:phase_diag}. 

\begin{figure}[h]
\vspace{-2mm}
    \centering
    \begin{tikzpicture}[
  >=Latex,
  axis/.style={thick,-{Latex[length=3mm]}},
  vline/.style={dashed,gray!90,thick},
  dot/.style={circle,fill,inner sep=1.5pt},
  arc/.style={-{Latex[length=2.2mm]},thick},
  scale = 1
]

\def\k{7} 
\def\ky{4} 
\def\rhoo{0.6}    

\pgfmathsetmacro{\xmax}{1.25*\k}
\pgfmathsetmacro{\ymax}{0.25*\ky}
\pgfmathsetmacro{\ymin}{-0.75*\ky}

\pgfmathsetmacro{\yann}{-\ky*0.7}

\coordinate (O) at (0,0);
\coordinate (A) at (\k*1/10, 0);                    
\coordinate (B) at (\k, -{\ky*(1/3-1/10)});                  
\coordinate (C) at ({\k*(1/10+\rhoo*(1-1/10))}, {- \ky*\rhoo/3});

\coordinate (D) at (\k,\ky*2/10);
\coordinate (E) at (\k*1/10,\ky*2/10);

\coordinate (F) at ({\k*1/10}, - {\ky*(1/3 + 1/10)});
\coordinate (G) at (\k, - \ky*2/3);
\coordinate (FF) at (\k*1/10, - \ky*0.75);
\coordinate (GG) at (\k, - \ky*0.75);

\fill[green!60!black, fill opacity=0.18, draw=none]
  (A) -- (E) -- (D) -- (B) -- cycle;
\path[
  pattern={Lines[angle=60,distance=6pt,line width=3pt]},
  pattern color=gray!70,
  opacity=0.35 
] (A) -- (B) -- (G) -- (F) -- cycle;
\fill[purple!60!black, fill opacity=0.18, draw=none]
  (F) -- (G) -- (GG) -- (FF) -- cycle;

\tikzset{
  axis/.style={very thick, ->},
  tick/.style={very thick},
  dottedline/.style={densely dotted, thick},
  dashedline/.style={dashed, thick},
  bluept/.style={circle, fill=blue!90!black, draw=blue!90!black, inner sep=0pt, minimum size=8pt},
  blueopen/.style={circle, fill=white, draw=blue!90!black, line width=1.2pt, inner sep=0pt, minimum size=8pt},
  purplept/.style={circle, fill=purple!90!black, draw=purple!90!black, inner sep=0pt, minimum size=5pt},
  blackpt/.style={circle, fill=black, draw=black, inner sep=0pt, minimum size=3pt}
}

\draw[axis] (0,\ymin) -- (0,\ymax) node[above left=2pt] {$\log_{n}\delta$};
\draw[axis] (0,0) -- (\xmax,0) node[below right=2pt] {$\log_{n}m$};

\node[left=4pt, font=\small] at (0,0) {$0$};

\draw[tick] (-0.12, -{\ky*(1/3-1/10)}) -- (0.12, -{\ky*(1/3-1/10)});
\node[left=4pt, font=\small] at (0, -{\ky*(1/3-1/10)}) {$-k/4$};

\draw[tick] (-0.12, -{\ky*(1/3 + 1/10)}) -- (0.12, -{\ky*(1/3 + 1/10)});
\node[left=4pt, font=\small] at (0, -{\ky*(1/3 + 1/10)}) {$-k/4+1/2$};

\draw[tick] (-0.12, -\ky*2/3) -- (0.12, -\ky*2/3);
\node[left=4pt, font=\small] at (0, -\ky*2/3) {$-k/2+1/2$};


\draw[tick] (\k, -0.12) -- (\k, 0.12);
\node[above=4pt, font=\small] at (\k, 0) {$k$};

\draw[tick] (\k*1/10, -0.12) -- (\k*1/10, 0.12);
\node[above=4pt, font=\small] at (\k*1/10, 0) {$k/2$};

\draw[thick] (A) -- (B);
\draw[thick] (F) -- (G);

\draw[thick] ($(A)!1/3!(B)$) -- (B) node[midway, above=-1.5pt, sloped, font=\small] {\eqref{eq:tradeoff}};
\draw[thick] (F) -- ($(F)!2/3!(G)$) node[midway, below=-1.5pt, sloped, font=\small] {\eqref{eq:it_threshold}};



\draw[dottedline, color=purple!90!black] (D) -- (GG) node[midway, above=-1.5pt, sloped, font=\small,color=purple!90!black] {$ \ $ tpca ($m=n^k$)};




\pgfmathsetmacro{\legx}{\xmax }   
\pgfmathsetmacro{\legy}{-0.32*\ky}    
\def\legbox{0.28}                     

\node[anchor=north west, font=\small] (legend) at (\legx,\legy) {%
  \begin{tabular}{@{}c@{\hspace{0.6em}}l@{}}

  \tikz{\draw[draw=black!40, line width=0.3pt,
              fill=green!60!black, fill opacity=0.18]
        (0,0) rectangle (\legbox,\legbox);} &
  poly-time algorithms exist \\[0.35em]

  \tikz{\draw[draw=black!40, line width=0.3pt,
              pattern={Lines[angle=60,distance=6pt,line width=3pt]},
              pattern color=gray!70, opacity=0.35]
        (0,0) rectangle (\legbox,\legbox);} &
  conjectured hard \\[0.35em]

  \tikz{\draw[draw=black!40, line width=0.3pt,
              fill=red!70!black, fill opacity=0.18]
        (0,0) rectangle (\legbox,\legbox);} &
  info-theoretically impossible \\

  \end{tabular}
};

\end{tikzpicture}
    \caption{Known algorithms and lower bounds for $\xor k n m \delta$ with $m \in \sqb{n^{k/2}, n^k}$.}
    \vspace{-2mm}
    \label{fig:phase_diag}
\end{figure}

The endpoints $\eps \in \set{0,1}$ correspond to the canonical $\kxor$ and $\ktensorpca$ settings; see \autoref{subsec:lit_algo} for the overview of known algorithms. 
For intermediate densities $\eps\in\paren{0,1}$, known spectral algorithms for detection \cite{wein2019kikuchi} extend naturally and succeed whenever $m \delta^2 = \widetilde\Omega(n^{k/2})$, and our reductions yield recovery algorithms for all $\eps\in \paren{0,1}$ at the same threshold for sufficiently large order $k = k(\eps)$ (\autoref{sec:discussion}). We refer to the boundary of the computationally feasible regime, $m\delta^2 = \widetilde\Omega(n^{k/2})$, as the \emph{computational threshold}; below it, no poly-time algorithms are known:
\begin{equation}\label{eq:tradeoff}
    \qquad\qquad m \delta^2 \asymp n^{k/2}\quad \text{or equivalently},\quad \delta \asymp n^{-k\eps/4}\,,\tag{Comp. Thr.}
\end{equation}
using the parametrization $m = n^{k(1+\eps)/2}$ (see Def.~\ref{def:density}).

The detection task for Tensor PCA (i.e., $\eps = 1$) becomes information-theoretically impossible at $\delta\ll n^{-(k-1)/2}$ \cite{richard2014statistical,montanari2015limitation,perry2016statistical}.
For the remaining densities $\eps\in[0,1)$, the computation of \cite{perry2016statistical} extends and yields the \emph{information-theoretic threshold}
\begin{equation}\label{eq:it_threshold}
    m \delta^2 \asymp n\quad \text{or equivalently},\quad \delta \asymp n^{-k\eps/4 -(k-2)/4}
    \,.\tag{IT Thr.}
\end{equation}

The discrepancy in the signal strength required to solve the problem \emph{efficiently} \eqref{eq:tradeoff} and \emph{information-theoretically} \eqref{eq:it_threshold} leaves a wide regime 
\begin{equation}\label{eq:com_stat_gap}
    n \ll m \delta^2 \ll n^{k/2} \quad {(\text{equiv. } n^{-k\eps/4 - (k-2)/4} \ll \delta \ll n^{-k\eps/4}) } 
\end{equation}
with no known poly-time algorithms and is known as the \emph{computational-statistical gap}. For the two endpoints at $\eps\in \{0,1\}$ that, as discussed earlier, correspond to the extremely well-studied $\kxor$ and $ \ktensorpca$, a line of work described in \autoref{subsec:lit_comp_bounds} 
provides evidence for the computational-statistical gap in \eqref{eq:com_stat_gap}. We expect that the analysis based on restricted classes of algorithms extends to intermediate $\eps\in (0,1)$. We summarize the believed hardness in the following conjectures:

\begin{assump}[Computational Hardness of $\kxor$ at Density $\eps$: $\conj k \eps$]\label{conj:hardness} Fix parameters $k \in \mathbb{Z}^{+}, \eps \in [0,1]$. We denote $\conj k \eps$ for detection (resp. recovery) to be the assumption that:\footnote{By Prop.~\ref{prop:intro_tpca_equiv}, $\conj k 1$ captures the believed hardness of $\ktensorpca$. As stated above, all our results equally apply to standard $\ktensorpca$ and the model where only distinct-index entries are observed.} 
\begin{center}
\begin{minipage}{0.85\textwidth} 
\centering
``There exists no polynomial-time algorithm solving the detection (resp. recovery) task for $\xor k n {m=n^{k(1+\eps)/2}} \delta$ with $\log_n \delta$ any fixed constant $<-k \eps / 4$."
\end{minipage}
\end{center}
\end{assump}

A key consequence of polynomial-time average-case reductions, which are formally defined in \autoref{sec:avg_case_defs}, is the transfer of efficient algorithms (Lem.~\ref{lem:red_implications_gen}). Therefore, our reductions relate the hardness conjectures $\conj k \eps$ and establish a \emph{hardness partial order} on the parameter pairs $(k, \eps)$:

 \begin{myprop}{\ref{prop:hardness_implication}}\textnormal{(Implications of Average-Case Reductions for $\kxor$ Hardness Conjectures).} Fix $k,k'\in \mathbb{Z}^{+}, \eps, \eps' \in [0,1]$. If for every constant $\log_n \delta' < -k\eps'/4$ there exists constant $\log_n \delta < - k \eps/4$ and a poly-time average-case reduction (Def.~\ref{def:avg_case_points}) from $\xor k n {m = n^{k(1+\eps)/2}} \delta$ to $\xor {k'} {n'} {m' = (n')^{k'(1+\eps')/2}} {\delta'}$ with $n' = \poly(n)$, then $\conj k \eps$ implies $\conj {k'} {\eps'}$, for both detection and recovery, which we denote as
 $$\conj k \eps \Rightarrow \conj {k'} {\eps'}
\,.$$
 \end{myprop} 
 \begin{remark}
     Some of our reductions nontrivially alter the signal vector $x$. We formalize the weak/strong/exact recovery tasks (Def.~\ref{def:exact_recovery}, \ref{def:weak_strong_recovery}) in \autoref{subsec:recovery_assumptions} and show that all of them transfer. 
 \end{remark}

\paragraph{Results.}

We prove reductions that (1) \emph{decrease} the tensor order $k$ at fixed density $\eps$ (Thm.~\ref{thm:decrease_k_intro}), and (2) \emph{densify} the instance by increasing $\eps$ and trading off against the tensor order $k$ (Thms.~\ref{thm:densifying_general}, \ref{thm:specific_intro}). While the precise reduction statements are given in \autoref{sec:resolution_red}, \autoref{sec:to_tensor_pca}, and \autoref{sec:decrease_k}, here we present their implications for the $\kxor$ hardness conjectures. 

\begin{theorem}[Reduction Reducing $k$, Cor.~\ref{cor:decreasek}]\label{thm:decrease_k_intro}
    For $\eps \in [0,1]$ the implication 
\begin{equation}\label{eq:decreasek_implication_}
    \conj k \eps \Rightarrow \conj {k'} \eps
    \end{equation} holds for the following choices of $k$ and $k'$: 
    
    \noindent
    \begin{minipage}[t]{0.55\textwidth}
        \begin{enumerate}
            \item \textbf{Canonical $\kxor$, $\eps = 0$:} 
            \begin{enumerate}
                \item \textbf{even $\kstar$:} all $k > \kstar$;
                \item \textbf{odd $\kstar$:} all odd $k > \kstar$ and even $k \geq 2\kstar$.  
            \end{enumerate}
        \end{enumerate}
    \end{minipage}
    \hfill
    \begin{minipage}[t]{0.43\textwidth}
        \begin{enumerate}
            \setcounter{enumi}{1}
            \item \textbf{Tensor PCA, $\eps = 1$:}
            \begin{enumerate}
                \item \textbf{even $\kstar$:} all $k > 3\kstar/2$;
                \item \textbf{odd $\kstar$:} all $k > 3\kstar$. 
            \end{enumerate}
        \end{enumerate}
    \end{minipage}\vspace{3mm}\\
    In particular, for $\eps \in \set{0,1}$ and any $k'$, the relation Eq.~\eqref{eq:decreasek_implication_} holds for \emph{all} sufficiently large $k \geq k\paren{k'}$. 
\end{theorem}

It is convenient to define a unifying conjecture for all $k\geq 3$ at specific density $\eps$. 

\begin{assump}[Computational Hardness of $\kxor$ with density $\eps$: $\conjeps \eps$] Fix $\eps \in \sqb{0,1}$. We consider the following conjectures for detection and recovery: 
    \begin{align*}
        \conjeps \eps: \conj k \eps \text{ holds for all }k\geq 3\,.
    \end{align*}
\end{assump}
\begin{wrapfigure}{r}{0.4\textwidth}
    \centering
    \begin{tikzpicture}[
  >=Latex,
  axis/.style={thick,-{Latex[length=3mm]}},
  vline/.style={dashed,gray!90,thick},
  dot/.style={circle,fill,inner sep=1.5pt},
  arc/.style={-{Latex[length=2.2mm]},thick},
  scale = 0.6
]

\def\kmin{2}
\def\kmax{10}
\def\xscale{5}      
\def\yscale{0.8}   
\def\epsmid{0.6}   

\def\dotsZero{7,3}            
\def\dotsMid{8,6,4,3}         
\def\dotsOne{10,4}            

\def\arrowsZero{7/3}          
\def\arrowsOne{10/4}          

\def\arcoffset{1.0}

\pgfmathsetmacro{\Xzero}{0.3*\xscale}
\pgfmathsetmacro{\Xone}{0.9*\xscale}
\pgfmathsetmacro{\Xmax}{1.4*\xscale}
\pgfmathsetmacro{\Ymax}{(\kmax-\kmin+1)*\yscale}

\draw[axis] (0,0) -- (\Xmax,0) node[below right=2pt] {$\eps$};
\draw[axis] (0,0) -- (0,\Ymax) node[left=2pt] {$k$};

\foreach \X/\label in {\Xzero/0, \Xone/1}{
  \draw[vline] (\X,0) -- (\X,\Ymax);
  \node[below] at (\X,0) {\label};
}

\fill[dot, purple!90!black] (\Xzero,{(4-\kmin)*\yscale}) circle[radius=0.1];
\fill[dot, purple!90!black] (\Xzero,{(5-\kmin)*\yscale}) circle[radius=0.1];
\fill[dot, purple!90!black] (\Xzero,{(6-\kmin)*\yscale}) circle[radius=0.1];
\fill[dot, purple!90!black] (\Xzero,{(7-\kmin)*\yscale}) circle[radius=0.1];
\fill[dot, purple!90!black] (\Xzero,{(3-\kmin)*\yscale}) circle[radius=0.1];
\fill[dot, blue!90!black, opacity=0.3] (\Xzero,{(6-\kmin)*\yscale}) circle[radius=0.1];
\fill[dot, blue!90!black, opacity=0.3] (\Xzero,{(5-\kmin)*\yscale}) circle[radius=0.1];
\fill[dot, blue!90!black, opacity=0.3] (\Xzero,{(3-\kmin)*\yscale}) circle[radius=0.1];


\fill[dot, purple!90!black] (\Xone,{(4-\kmin)*\yscale}) circle[radius=0.1];
\fill[dot, purple!90!black] (\Xone,{(7-\kmin)*\yscale}) circle[radius=0.1];
\fill[dot, blue!90!black, opacity=0.7] (\Xone,{(10-\kmin)*\yscale}) circle[radius=0.1];
\fill[dot, blue!90!black, opacity=0.7] (\Xone,{(3-\kmin)*\yscale}) circle[radius=0.1];

\tikzset{thinheads/.style={-{Latex[length=6pt,width=4pt]}}}

\draw[thinheads, >=Latex, thick, purple!90!black, bend left=60] (\Xzero,{(5-\kmin)*\yscale}) to (\Xzero,{(4-\kmin)*\yscale});
\draw[thinheads, >=Latex, thick, purple!90!black, bend left=60] (\Xzero,{(6-\kmin)*\yscale}) to (\Xzero,{(4-\kmin)*\yscale});
\draw[thinheads, >=Latex, thick, purple!90!black, bend left=60] (\Xzero,{(7-\kmin)*\yscale}) to (\Xzero,{(6-\kmin)*\yscale});
\draw[thinheads, >=Latex, thick, blue!90!black, bend right=60, opacity=0.7] (\Xzero,{(6-\kmin)*\yscale}) to (\Xzero,{(3-\kmin)*\yscale});
\draw[thinheads, >=Latex, thick, blue!90!black, bend right=60, opacity=0.7] (\Xzero,{(5-\kmin)*\yscale}) to (\Xzero,{(3-\kmin)*\yscale});


\draw[thinheads, >=Latex, thick, purple!90!black, bend left=60] (\Xone,{(7-\kmin)*\yscale}) to (\Xone,{(4-\kmin)*\yscale});
\draw[thinheads, >=Latex, thick, blue!90!black, bend right=60, opacity=0.7] (\Xone,{(10-\kmin)*\yscale}) to (\Xone,{(3-\kmin)*\yscale});

\foreach \kk in {\kmin,...,\kmax}{
  \draw[gray!40] (0,{(\kk-\kmin)*\yscale}) -- ++(-0.08,0);
  \node[left,gray!60] at (-0.12,{(\kk-\kmin)*\yscale}) {\scriptsize \kk};
}

\end{tikzpicture}
    \caption{Thm.~\ref{thm:decrease_k_intro}: blue arrows represent reductions to odd $k'$, red -- to even.}
    \label{fig:decrease_k_intro}
    \vspace{-5pt}
\end{wrapfigure}
In light of Thm.~\ref{thm:decrease_k_intro}, for $\eps \in \set{0,1}$, $\conjeps  \eps$ holds if $\conj k \eps$ holds for infinitely many $k$.

Our main result maps sparser to denser instances for a wide parameter range. This includes mapping $\kxor$ arbitrarily close to the canonical version (with any $\eps > 0$) to Tensor PCA.

\begin{theorem}[Densifying Reduction (Growing $\eps$), Corollary of Thm.~\ref{thm:sparse_dense_red} and Cor.~\ref{cor:tensor_pca_summary}]\label{thm:densifying_general}
    Fix a pair of parameters $0 < \eps < \eps' \leq 1$ and any $k' \geq 3$. Then there exists some $k \geq k'$ such that 
    $$\conj k \eps \Rightarrow \conj {k'} {\eps'}\,,$$
    in the following parameter regimes:
    \begin{enumerate}
            \item \textbf{Intermediate $\eps' \in \paren{0,1}$ (Thm.~\ref{thm:sparse_dense_red}):} any $\eps \in \paren{0,1/3}$ and $\eps' \in \paren{2\eps/(1-\eps), 1}$;
        \item \textbf{$\eps' = 1$ (Tensor PCA\footnote{\label{fn:shared}
        This result holds for the standard $\ktensorpca$ in Eq.~\eqref{eq:tensor_pca} ($\delta x^{\otimes k} + \noiseG$) and for $\kxor$ at density $\eps' = 1$, see \autoref{subsubsec:equiv_w_wor_replacement}.} Cor.~\ref{cor:tensor_pca_summary}):} any $\eps \in \paren{0,1}$.
    \end{enumerate}
    In particular, for any $\eps,\eps'$ as above, 
    $
    \conjeps \eps \Rightarrow \conjeps {\eps'}\,.
    $
\end{theorem}

\begin{figure}[h]
    \centering
    \captionsetup[subfigure]{labelformat=simple, labelsep=colon}
    \renewcommand{\thesubfigure}{Fig.~\thefigure(\alph{subfigure})}
    \begin{subfigure}[t]{0.48\textwidth}
        \centering
        \resizebox{\linewidth}{!}{
            \begin{tikzpicture}[
  >=Latex,
  axis/.style={thick,-{Latex[length=3mm]}},
  vline/.style={dashed,gray!90,thick},
  dot/.style={circle,fill,inner sep=1.5pt},
  arc/.style={-{Latex[length=2.2mm]},thick},
  scale = 0.6
]

\def\kmin{2}
\def\kmax{10}
\def\xscale{10}      
\def\yscale{0.55}   

\pgfmathsetmacro{\Xzero}{0.2*\xscale}
\pgfmathsetmacro{\Xone}{1.2*\xscale}
\pgfmathsetmacro{\Xmax}{1.4*\xscale}

\draw[axis] (0,0) -- (\Xmax,0) node[below right=2pt] {$\eps$};
\draw[black, line width=0.6pt, line cap=round]
  ({0.2*\xscale},0) ++(0,-1.5pt) -- ++(0,3pt);
\node[below,gray!60] at ({0.2*\xscale},0) {\scriptsize $0$};

\draw[black, line width=0.6pt, line cap=round]
  ({1.4*\xscale},0) ++(0,-1.5pt) -- ++(0,3pt);

\draw[black, line width=0.6pt, line cap=round]
  ({1.2*\xscale},0) ++(0,-1.5pt) -- ++(0,3pt);
\node[below,gray!60] at ({1.2*\xscale},0) {\scriptsize $1$};

\fill[dot, purple!90!black] ({(0.2+1/10)*\xscale},0) circle[radius=0.1];
\node[below,gray!60] at ({(0.2+1/10)*\xscale},0) {\scriptsize $\frac1{10}$};

\draw[purple!90!black, line width=0.6pt, line cap=round]
  ({(0.2+2/9)*\xscale},0) ++(0,-1.5pt) -- ++(0,3pt);
\node[below,gray!60] at ({(0.2+2/9)*\xscale},0) {\scriptsize $\frac29$};

\draw[purple!90!black, line width=1.5pt] ({(0.2+2/9)*\xscale},0) -- ({1.2*\xscale},0);

\fill[dot, blue!90!black] ({(0.2+1/20)*\xscale},0) circle[radius=0.1];
\fill[dot, blue!90!black] ({(0.2+3/10)*\xscale},0) circle[radius=0.1];
\fill[dot, blue!90!black] ({(0.2+1/2)*\xscale},0) circle[radius=0.1];

\draw[->, >=Latex, thick, purple!90!black, bend left=40] ({(0.2+1/10)*\xscale}, 0) to ({(0.2+11/18)*\xscale}, 0);

\draw[->, >=Latex, thick, blue!90!black, bend right=40] ({(0.2+1/20)*\xscale}, 0) to ({1.2*\xscale}, 0);
\draw[->, >=Latex, thick, blue!90!black, bend right=40] ({(0.2+3/10)*\xscale}, 0) to ({1.2*\xscale}, 0);
\draw[->, >=Latex, thick, blue!90!black, bend right=40] ({(0.2+1/2)*\xscale}, 0) to ({1.2*\xscale}, 0);

\end{tikzpicture}
        }
        \caption{Thm.~\ref{thm:densifying_general}: red arrow exemplifies reductions to sparsities $\eps' \in \paren{2\eps/(1-\eps), 1}$, blue -- to Tensor PCA ($\eps'=1$).}
        \label{fig:red_densify}
    \end{subfigure}
    \hfill
    \begin{subfigure}[t]{0.48\textwidth}
        \centering
        \resizebox{\linewidth}{!}{
            \begin{tikzpicture}[
  >=Latex,
  axis/.style={thick,-{Latex[length=3mm]}},
  vline/.style={dashed,gray!90,thick},
  dot/.style={circle,fill,inner sep=1.5pt},
  arc/.style={-{Latex[length=2.2mm]},thick},
  scale = 0.6
]

\def\kmin{2}
\def\kmax{10}
\def\xscale{10}      
\def\yscale{0.55}   

\pgfmathsetmacro{\Xzero}{0.2*\xscale}
\pgfmathsetmacro{\Xone}{1.2*\xscale}
\pgfmathsetmacro{\Xmax}{1.4*\xscale}

\draw[->, >=Latex, thick, white, bend right=40] ({(0.2+1/20)*\xscale}, 0) to ({1.2*\xscale}, 0);

\draw[axis] (0,0) -- (\Xmax,0) node[below right=2pt] {$\eps$};
\draw[black, line width=0.6pt, line cap=round]
  ({0.2*\xscale},0) ++(0,-1.5pt) -- ++(0,3pt);
\node[below,gray!60] at ({0.2*\xscale},0) {\scriptsize $0$};

\draw[black, line width=0.6pt, line cap=round]
  ({1.4*\xscale},0) ++(0,-1.5pt) -- ++(0,3pt);
\node[below,gray!60] at ({0.25*\xscale},0) {\scriptsize $\eps^{\star}$};

\draw[black, line width=0.6pt, line cap=round]
  ({1.2*\xscale},0) ++(0,-1.5pt) -- ++(0,3pt);
\node[below,gray!60] at ({1.2*\xscale},0) {\scriptsize $1$};


\draw[blue!90!black, line width=1.5pt] ({(0.2)*\xscale},0) -- ({0.25*\xscale},0);
\draw[purple!90!black, line width=1.5pt] ({(0.25)*\xscale},0) -- ({1.2*\xscale},0);


\draw[->, >=Latex, thick, blue!90!black, bend left=40] ({(0.225)*\xscale}, 0) to ({((0.25 + 1.2)/2)*\xscale}, 0);

\end{tikzpicture}
        }
        \caption{Cor.~\ref{cor:densifying_gen}: conjectured hardness in the blue region implies hardness for the red region.}
        \label{fig:red_densify_cor}
    \end{subfigure}
    
\end{figure}

\begin{corollary}[of Theorem~\ref{thm:densifying_general}]\label{cor:densifying_gen} For any $\eps^\star > 0$, if $\conjeps \eps$ holds for all $\eps \in (0, \eps^\star]$, then $\conjeps \eps$ holds for all $\eps \in (0, 1]$: 
$$
\set{\conjeps \eps}_{\eps \in (0, \eps^\star]} \Rightarrow \set{\conjeps \eps}_{\eps \in (0, 1]}\,.
$$
\end{corollary}

Theorem~\ref{thm:densifying_general} follows from Theorem~\ref{thm:specific_intro}, which gives explicit $k \leftrightarrow \eps$ tradeoffs of our densifying reductions.

\begin{figure}[h]
    \centering

    \captionsetup[subfigure]{labelformat=simple, labelsep=colon}
    \renewcommand{\thesubfigure}{Fig.~\thefigure(\alph{subfigure})}
    \begin{subfigure}[t]{0.57\textwidth}
        \centering
        \resizebox{\linewidth}{!}{
            \begin{tikzpicture}[
  >=Latex,
  axis/.style={thick,-{Latex[length=3mm]}},
  vline/.style={dashed,gray!90,thick},
  dot/.style={circle,fill,inner sep=1.5pt},
  arc/.style={-{Latex[length=2.2mm]},thick},
  scale = 0.6
]

\def\kmin{2}
\def\kmax{10}
\def\xscale{10}      
\def\yscale{0.55}   
\def\epsmid{0.6}   

\pgfmathsetmacro{\Xzero}{0.2*\xscale}
\pgfmathsetmacro{\Xone}{1.2*\xscale}
\pgfmathsetmacro{\Xmax}{1.4*\xscale}
\pgfmathsetmacro{\Ymax}{(\kmax-\kmin+1)*\yscale}

\pgfmathsetmacro{\XmidOne}{(0.2+1/20)*\xscale}
\pgfmathsetmacro{\XmidThree}{(0.2+1/2)*\xscale}
\pgfmathsetmacro{\XmidFour}{(0.2+5/40)*\xscale}
\pgfmathsetmacro{\XmidFive}{(0.2+1/3)*\xscale}

\draw[axis] (0,0) -- (\Xmax,0) node[below right=2pt] {$\eps$};
\draw[axis] (0,0) -- (0,\Ymax) node[left=2pt] {$k$};

\foreach \X/\label in {\Xzero/0, \XmidOne/{$\frac 1 {20}$}, \XmidThree/{$\frac 1 {2}$}, \XmidFour/{$\frac 1 {8}$}, \XmidFive/{$\frac 13$}, \Xone/1}{
  \draw[vline] (\X,0) -- (\X,\Ymax);
  \node[below] at (\X,0) {\scriptsize\label};
}

\fill[dot, purple!90!black] (\XmidFour,{(8-\kmin)*\yscale}) circle[radius=0.1];
\fill[dot, purple!90!black] (\XmidThree,{(4-\kmin)*\yscale}) circle[radius=0.1];

\fill[dot, purple!90!black] (\XmidOne,{(5-\kmin)*\yscale}) circle[radius=0.1];
\fill[dot, purple!90!black] (\XmidFour,{(4-\kmin)*\yscale}) circle[radius=0.1];

\fill[dot, blue!90!black] (\XmidFive,{(7-\kmin)*\yscale}) circle[radius=0.1];
\fill[dot, blue!90!black] (\Xone,{(4-\kmin)*\yscale}) circle[radius=0.1];
\fill[dot, blue!90!black] (\XmidThree,{(9-\kmin)*\yscale}) circle[radius=0.1];

\draw[->, >=Latex, thick, purple!90!black, bend left=40] (\XmidFour,{(8-\kmin)*\yscale}) to (\XmidThree,{(4-\kmin)*\yscale});
\draw[->, >=Latex, thick, purple!90!black, bend left=35] (\XmidOne,{(5-\kmin)*\yscale}) to (\XmidFour,{(4-\kmin)*\yscale});

\draw[->, >=Latex, thick, blue!90!black, bend left=35] (\XmidFive,{(7-\kmin)*\yscale}) to (\Xone,{(4-\kmin)*\yscale});
\draw[->, >=Latex, thick, blue!90!black, bend left=30] (\XmidThree,{(9-\kmin)*\yscale}) to (\Xone,{(4-\kmin)*\yscale});

\foreach \kk in {\kmin,...,\kmax}{
  \draw[gray!40] (0,{(\kk-\kmin)*\yscale}) -- ++(-0.08,0);
  \node[left,gray!60] at (-0.12,{(\kk-\kmin)*\yscale}) {\scriptsize \kk};
}

\end{tikzpicture}
        }
        \caption{Thm.~\ref{thm:specific_intro}: red arrow exemplifies reductions to densities $\eps' \in \paren{2\eps/(1-\eps), 1}$, blue -- to Tensor PCA ($\eps'=1$).}
        \label{fig:red_densify_explicit}
    \end{subfigure}
    \hfill
    \begin{subfigure}[t]{0.39\textwidth}
        \centering
        \resizebox{\linewidth}{!}{
            \begin{tikzpicture}[
  >=Latex,
  axis/.style={thick,-{Latex[length=3mm]}},
  vline/.style={dashed,gray!90,thick},
  dot/.style={circle,fill,inner sep=1.5pt},
  arc/.style={-{Latex[length=2.2mm]},thick},
  scale = 0.6
]

\def\kmin{2}
\def\kmax{10}
\def\xscale{6}      
\def\yscale{0.55}   
\def\epsmid{0.6}

\pgfmathsetmacro{\Xzero}{0.2*\xscale}
\pgfmathsetmacro{\Xone}{1.2*\xscale}
\pgfmathsetmacro{\Xmax}{1.4*\xscale}
\pgfmathsetmacro{\Ymax}{(\kmax-\kmin+1)*\yscale}

\pgfmathsetmacro{\XmidOne}{(0.2+1/20)*\xscale}
\pgfmathsetmacro{\XmidTwo}{(0.2+1/10)*\xscale}
\pgfmathsetmacro{\XmidThree}{(0.2+2/5)*\xscale}
\pgfmathsetmacro{\XmidFour}{(0.2+7/40)*\xscale}
\pgfmathsetmacro{\XmidFive}{(0.2+1/3)*\xscale}

\def\ShadeWeak{gray!5}         
\def\ShadeStrong{teal!65!blue}  
\def\ShadeOpacity{0.5}

\draw[axis] (0,{-\yscale}) -- (\Xmax,{-\yscale}) node[below right=2pt] {$\eps$};
\draw[axis] (0,{-\yscale}) -- (0,\Ymax) node[left=2pt] {$k$};

\foreach \X/\label in {\Xzero/0, \Xone/1}{
  \draw[vline] (\X,{-\yscale}) -- (\X,\Ymax);
  \node[below] at (\X,{-\yscale}) {\scriptsize\label};
}

\def\smshift{0.03}

\tikzset{thinheads/.style={-{Latex[length=3pt,width=2pt]}}}
\pgfmathtruncatemacro{\kstop}{\kmax-2}
\foreach \kk in {\kmin,...,\kstop}{
    \foreach \ee in {0,...,9}{
    \pgfmathsetmacro{\xx}{\Xzero + 0.1*\ee*\xscale}
    \pgfmathsetmacro{\yy}{\kk*\yscale}

    \ifodd\kk\relax
    
    \else
      \ifodd\ee\relax

      \else
        \draw[thinheads, purple!90!black] ({\xx+\smshift},{\yy-\smshift}) -- ({\xx+\smshift},{\yy-2*\yscale+\smshift});
        \draw[thinheads, purple!90!black] ({\xx+\smshift},{\yy-\smshift}) -- ({\xx + 2*0.1*\xscale-\smshift},{\yy- \smshift});
        \draw[thinheads, purple!90!black] ({\xx+\smshift},{\yy-\smshift}) -- ({\xx + 2*0.1*\xscale-\smshift},{\yy-2*\yscale+\smshift});
        \draw[thinheads, purple!90!black] ({\xx+\smshift},{\yy-\smshift}) -- ({\xx + 2*0.05*\xscale},{\yy-2*\yscale+\smshift});
      \fi
    \fi
}
}

\foreach \kk in {\kmin,...,\kmax}{
  \draw[gray!40] (0,{(\kk-\kmin)*\yscale}) -- ++(-0.08,0);
  \ifodd\kk\relax
    
\else
    \pgfmathtruncatemacro{\kkk}{(\kk+2)/2 + 1}
  \node[left,gray!60] at (-0.12,{(\kk-\kmin)*\yscale}) {\scriptsize {\kkk}};
  \fi
}

\end{tikzpicture}
        }
        \caption{Conj.~\ref{conj:hardness}: conjectured hardness implication order ``$\Rightarrow$" in the space of $\paren{k,\eps}$.}
        \label{fig:conj}
    \end{subfigure}
    
\end{figure}

\begin{theorem}\label{thm:specific_intro}
    For parameters $0\leq \eps \leq \eps' \leq 1$ and $k \geq k'$ we show 
    $$
    \conj k \eps \Rightarrow \conj {k'} {\eps'}
    $$
    for the following parameter pairs $\paren{k, \eps}$ and $\paren{k', \eps'}$:
    \begin{enumerate}
        \item \textbf{(Discrete Resolution, Lem.~\ref{lem:2path_red}+Lem.~\ref{lem:prelim_decrease_k})} any $k, k'$, $\eps\in[0, 1/3), \eps' \in [0,1)$, s.t.
        \begin{align*}
            \text{even }k': \quad 2\eps< \frac{k'}k <  1-\eps,\ \eps' = \frac{2k}{k'}\eps\qquad\text{and}\qquad\text{odd }k': \quad \eps < \frac{k'}k < \frac{1-\eps}{2},\ \eps' = \frac{k}{k'}\eps\,.
        \end{align*}
        \item \textbf{(Reduction to Tensor PCA\footref{fn:shared}, Cor.~\ref{cor:tensor_pca_summary})} any $k, k'$, $\eps \in (0,1)$, $\eps' = 1$, s.t. 
        $$
        k' < \begin{cases}
            \tau_{\eps} \cdot k, &k' \text{ even}\\
            \tau_{\eps}/2 \cdot k, &k' \text{ odd}\,,
        \end{cases} \qquad \text{where}\qquad \tau_{\eps} = \begin{cases}
    2\eps, &\eps \in \paren{0,1/3}\\
    1-\eps, &\eps \in [1/3,3/5)\\
    (1+\eps)/4 , &\eps \in \sqb{3/5, 1}\,.
\end{cases}
        $$
    \end{enumerate}
\end{theorem}

Observe that the hardness implication of all our results is in the direction of decreasing $k$ and growing $\eps$. We formulate the following conjecture; additional open problems are discussed in \autoref{sec:discussion}.
\begin{conjecture}\label{conj:one_way_hardness}
    For all $0 \leq \eps \leq \eps'\leq 1$ and $k \geq k'$, 
    $
    \conj k \eps \Rightarrow \conj {k'} {\eps'}\,.
    $
\end{conjecture}

\subsection{Results for Sparse $\kxor$ with $m \leq n^{k/2}$}\label{subsec:sparse_res}

Our resolution reduction in Lemma~\ref{lem:2path_red} applies to $\xor k n m \delta$ for all equation counts $m \in (1, n^k]$ (equiv., $m = n^{k(1+\eps)/2}$ for $\eps \in (-1,1]$). It maps $\xor k n m \delta$ to $\xor {k'} {n'} {m'} {\delta^2}$ for every even $k' \leq k(1-\eps)$, where
\begin{equation}\label{eq:resolution_param_tradeoff_sparse}
m' = \widetilde \Theta((n')^{k'(1+\eps')/2})\quad\text{with}\quad\eps' = 2k\eps/ k'\,.
\end{equation}
For example, combined with our order-reducing reduction (\autoref{subsubsec:reducing_k_by_a_factor}), this yields a reduction from $\kxork 7$ with $m = n^{3.4}$ to the classical setting of $\kxork 3$ with $m' = \Thetat(n^{1.4})$, with constant noise in both problems.
In contrast to the denser regime (\autoref{subsec:complexity_profiles}), where the computational threshold is well understood, the complexity profile for $m < n^{k/2}$ is less complete. Consequently, while the $(k,m) \to (k',m')$ tradeoff in \eqref{eq:resolution_param_tradeoff_sparse} does not match that of the best known polynomial-time algorithm in this regime -- which requires extremely low noise $\delta = 1 - n^{-\nu}$ \cite{chen2024algorithms} -- it still yields hardness implications and a new connection across different $\kxor$ orders. See \autoref{sec:discussion} for further discussion.

\begin{theorem}[Corollary of Lem.~\ref{lem:2path_red}]\label{thm:intro_resolution_basic}
    Let $k, k' \geq 2, \eps \in (-1, 1]$ be constants satisfying $k' \leq k(1-\eps)$ and $k'$ even, and let $\delta = \delta(n)$ be a parameter. If there is \textbf{no} polynomial-time algorithm solving detection (resp. recovery) for $\xor k n {m} \delta$ with $m = n^{k(1+\eps)/2}\,,$ then there is \textbf{no} polynomial-time algorithm solving detection (resp. recovery) for $\xor {k'} n {m'} {\delta'}$ with 
    $$
    m'= \Thetat(n^{k'\paren{1+ \frac{2k\eps}{k'}}/2})\quad\text{and}\quad \delta' = \delta^2\,.
    $$
\end{theorem}

\subsection{Results for $k$-sparse LWE}\label{subsubsec:lwe_results_intro}

In \autoref{subsec:resolution_lwe} we extend our reductions to the $k$-sparse $\lwe$ model ($\klwe$), which was first formalized in \cite{jain2024systematic}. While $\kxor$ can be viewed as noisy linear equations over $\mathbb{F}_2$, $\klwe$ consists of noisy linear equations over $\Z_q$ with $k$-sparse supports (for constant $k$). Let $\chi$ be a probability distribution over $\Z_q$.
In $\lweg k n m \chi$ (Def.~\ref{def:k_sparse_lwe}) one observes $m$ i.i.d. samples
\begin{equation}\label{eq:entry_Fq_intro}
    \Y = a_{1}x_{i_{1}}+a_{2}x_{i_{2}}+\dots + a_{k}x_{i_{k}} + e \,,
\end{equation}
where $x\in \Z_q^n$ is the secret, $\alpha = \set{i_{1},\dots,i_{k}}$ is a uniformly random index set, $e \sim \chi \in \Z_q$ is i.i.d. noise, and the coefficients $a_{s} \sim_{\mathrm{i.i.d.}} \unif\paren{\Z_q\setminus 0 }$. 

\begin{theorem}[Reductions for $\klwe$; Informal]\label{thm:lwe_informal}
    Let $\set{\chi_\theta}_{\theta\in \Theta}$ be a \emph{resolution-stable} noise family.\footnote{Informally, the distribution of a difference of random variables remains in the family, see Def.~\ref{def:resolution_stable}.} For a wide range of parameters, there is an average-case reduction for both detection and recovery 
    $$
    \text{from}\quad \lweg k n m {\chi_\theta}\quad \text{to}\quad \lweg {k'} {n'} {m'} {\chi_{\theta'}}
    $$
    for $n' = n(1-o(1))$ and $(k', m', \theta') = f(k, m, \theta)$ for explicit parameter maps $f$ stated in \autoref{subsec:resolution_lwe}. The secret in the output instance is the restriction $x' = x_{\sqb{1:n'}}$.

    In particular, the following noise families are resolution-stable (Def.~\ref{def:resolution_stable}):
    \begin{itemize}
    \item \textbf{Uniform $\mathbb{Z}_q$ Noise:} $e \sim \chi_{\delta}$ if $e = \begin{cases}
        0, &\text{w. prob. }\delta\\
        \unif\paren{\Z_q }, &\text{w. prob. }1-\delta\,.
    \end{cases}$
    \item \textbf{Discrete Gaussian Noise:} 
    $
    \chi_s(e) \propto \sum_{t\in \Z} e^{-\pi (e + tq)^2/s^2}\,.
    $
    \item \textbf{Bounded $\mathbb{Z}_q$ Noise:} $e \sim \chi_{l}$,
    where $\chi_l$ is a given measure on $\set{-l, l}$ for a parameter $l \leq q$.
\end{itemize}
\end{theorem}

Our reductions prove a new connection between different parameter values and, combined with a simple collision-based algorithm for detection (\autoref{subsec:collision_detection}), yield new algorithms that outperform it.
However, the complexity profile of $\klwe$ remains largely underexplored regarding its dependence on $k$, number of samples, and noise. See \autoref{sec:discussion} for $\klwe$ open problems and \autoref{subsec:resolution_lwe} for our results.

\section{Main Reduction Ideas}\label{sec:ideas}

We obtain our hardness implications (Thms~\ref{thm:intro_resolution_basic}, \ref{thm:decrease_k_intro}, \ref{thm:densifying_general}, \ref{thm:specific_intro}) via \emph{average-case reductions}: efficient algorithms transforming $\xor k n m \delta$ to $\xor {k'} {n'} {m'} {\delta'}$ that preserve the planted structure and distributions, up to $o(1)$ error in total variation. Since we study computational phenomena at polynomial scale, instances with equation numbers $m_1, m_2$ satisfying 
$
\lim_{n\to \infty} \log_n m_1 = \lim_{n\to \infty} \log_n m_2
$
are interchangeable. Accordingly, we parametrize the number of equations $m = n^{k(1+\eps)/2}$, with density $\eps \in (-1, 1]$ viewed as constant (Def.~\ref{def:density}). We say that a reduction maps $\eps$ to $\eps'$ if the output has 
$$
m' = \Thetat\b( (n')^{k'(1+\eps')/2} \b) \quad\text{equations}\,.
$$

Recall that in $\xor k n m \delta$ with a secret vector $x \in \set{\pm 1}^n$, there are $m$ samples $\set{\paren{\alpha_j, \Y_j}}_{j\in\sqb{m}}$ with independent index sets $\alpha_j \sim \unif\b[\binomset n k\b]$ sampled with replacement and
\begin{equation*}
    \Y_j = 
    x_{\alpha_j}\cdot \noise_{j}\,, \qquad \noise_j \sim_{i.i.d.} \rad(\delta)\,.
\end{equation*}
Here we write $x_{\alpha} = x_{i_1}\cdots x_{i_k}$ for an index set $\alpha = \set{i_1,\dots,i_k} \in \binomset n k$ and define $\rad(\delta)$ as the $\pm 1$ random variable with $\E\sqb{\rad\paren{\delta}} = \delta$. In the dense regime $\eps \geq 0$, the computational threshold is 
\begin{equation}
\qquad \qquad\quad m\delta^2 \asymp n^{k/2} \quad \text{or equivalently,}\quad \delta \asymp n^{-k\eps/4}\,.\tag{\eqref{eq:tradeoff}}
\end{equation}
Again working at polynomial scale, when $\eps \geq 0$, we encode distance to this threshold via the \emph{SNR deficiency} parameter $\defi$:
\begin{equation}\label{eq:defi_def}
  \qquad\qquad\qquad\defi \coloneq \log_n \frac{n^{-k\eps/2}}{\delta^2} = \log_n \frac{n^{k/2}}{m \delta^2} \,,\tag{\text{SNR deficiency}}
\end{equation}  
which we take to be constant. Efficiently solvable instances have $\defi < 0$ and conversely no efficient algorithms are known for $\defi > 0$. Accordingly, for $\eps \geq 0$, we write $\xorp k n \eps \defi$ to denote $\xor k n m \delta$ with 
$$
m = n^{k(1+\eps)/2}\quad\text{and}\quad \delta = n^{-k \eps/4 - \defi/2}\,.
$$
With this parametrization, a reduction (formally defined in \autoref{sec:avg_case_defs}) from $\xorp k n \eps \defi$ to $\xorp {k'} {n'} {\eps'} {\defi'}$ produces a $\xor {k'} {n'} {m'} {\delta'}$ instance with 
$$
m' = \Thetat\b((n')^{k'(\eps'+1)/2}\b)\quad\text{and}\quad \delta' = \Thetat\b((n')^{-k'\eps'/4 - \defi'/2}\b)\,.
$$

\subsection{Resolution Identity, Na\"ive Reduction, and Challenges}\label{subsec:resolution_id}

Observe that for any two distinct samples $\paren{\alpha_s, \Y_s}$ and $\paren{\alpha_t, \Y_t}$ of $\xor k n m \delta$, the product
\begin{equation}\label{eq:resolution_identity}
    \Y_{s} \Y_{t} = 
    x_{\alpha_s} x_{\alpha_t} \cdot \noise_{s}\noise_{t} = x_{\alpha_s \triangle \alpha_t} \cdot {\paren{\noise_{s}\noise_{t}}} \sim x_{\alpha_s \triangle \alpha_t} \cdot \rad(\delta^2)\,,
\end{equation}
where $\alpha_s\triangle\alpha_t$ denotes the symmetric difference. Let $k'\coloneq |\alpha_s \triangle \alpha_t|$; then $\paren{\alpha_s \triangle \alpha_t, \Y_{s} \Y_{t}}$ is distributed as a $\kxork {k'}$ sample with signal level $\delta' = \delta^2$. We call the operation of multiplying two $\kxor$ entries to obtain another $\kxork {k'}$ entry \emph{resolution} ($\EqComb$).\footnote{In logic and SAT solvers similar resolution procedures have a long history  \cite{robinson1965machine,baumgartner2000taming}; see \autoref{subsec:resolution_lit}.}

Eq.~\eqref{eq:resolution_identity} suggests building a new $\kxork {k'}$ instance by taking products $\Y_s\Y_t$ over all available pairs $(s,t)$ with overlap size $|\alpha_s\triangle\alpha_t| = \kstar$. There are two qualitatively different regimes, depending on how many pairs $(s,t)$ map to the same output index set $\gamma = \alpha_s\triangle\alpha_t \in\binomset n \kstar$. It turns out that if $2k\eps \leq k'$, then there are typically $O(1)$ equations created for each index set $\gamma$: a natural reduction is to simply return this collection of equations.
\vspace{-1mm}
\begin{enumerate}
    \item 
    \textbf{Sparse case -- few pairs $(s,t)$:} for all $s,t \in \sqb{m}$ satisfying $|\alpha_s\triangle\alpha_t| = \kstar$, include entry $\paren{\alpha_s\triangle\alpha_t, \Y_s\Y_t}$ into the output instance $\calZ$: 
    \begin{equation}\label{eq:sparse_naiive_ideas}
        \calZ = \set{ \paren{\alpha_s\triangle\alpha_t, \Y_s\Y_t}:\, s,t \in \sqb{m} \text{ and } |\alpha_s\triangle\alpha_t| = \kstar}\tag{Sparse $\EqComb$}\,.
    \end{equation}
\end{enumerate}
\vspace{-1mm}
However, if $2k\eps > k'$ (which only happens in the dense $\eps\geq 0$ regime), the individual equations that are produced are too noisy: each one has bias only $\delta' = \delta^2 \approx n^{-k\eps / 2} \ll n^{-k'/4}$ (for $\defi\approx 0$). On the other hand, in this regime we expect $\approx n^{k-k'/2} \cdot (m/n^k)^2 = n^{k\eps-k'/2}$ created equations\footnote{Since there are $\approx n^{k-k'/2}$ candidate pairs for each $\gamma \in \binomset n {k'}$ and probability to observe one is $\approx (m/n^k)^2$.} for each index set $\gamma$, so we can \emph{aggregate} these to boost the signal strength.
\vspace{-1mm}
\begin{enumerate}
    
    \item[2.] \textbf{Dense case -- many pairs $(s,t)$:} for each $\gamma \in \binomset n \kstar$ collect $\Y_s\Y_t$ with $\alpha_s\triangle\alpha_t = \gamma$ into $\Zt_{\gamma}^\mathrm{obs} = \set{\Y_{s} \Y_{t}: s,t \in [m]\text{ and }\alpha_s \triangle \alpha_t = \gamma}$ and aggregate them via a majority vote, outputting a collection of equations 
    \begin{equation}\label{eq:dense_naiive_ideas}
    \calZ = \B\{\b(\gamma, \Zt_\gamma \coloneq \maj (\Zt_{\gamma}^\mathrm{obs} )\b):\, \gamma \in \binomset n \kstar \text{ and }|\Zt_{\gamma}^{obs}| \neq 0\B\}\,.\tag{Dense $\EqComb$}
    \end{equation}
    The aggregation in \eqref{eq:dense_naiive_ideas} boosts the signal level in $\Zt_\gamma$ equations to $\delta' \approx \delta^2\cdot\sqrt{|\Zt_{\gamma}^\mathrm{obs}|}$, which turns out to be precisely the signal level needed to trace the computational threshold. 
\end{enumerate}
\vspace{-1mm}

In both cases, the \emph{marginal distribution} of individual equations in $\calZ$ as well as their total number matches that of a $\xor {k'} n {m'} {\delta'}$ instance with the same secret vector $x \in\set{\pm1}^n$, where 
$$
m' = \Thetat(n^{k'(1+\eps')/2})\quad\text{and}\quad \eps' = \min\big\{1, 2k\eps/k'\big\}\,.
$$
For $\eps < 0$, $\delta' = \delta^2$, which is the target parameter map of Thm.~\ref{thm:intro_resolution_basic}. For $\eps \geq 0$, with the $\delta = n^{-k\eps/4 - \defi/2}$ parametrization, 
$$
\delta' = \Thetat(n^{-k'\eps'/4 - \defi'/2})\quad\text{and}\quad \defi' = 2\defi\,.
$$
This change in parameters $(\eps,\defi) \to (\eps',\defi')$ traces the computational threshold \eqref{eq:tradeoff}, mapping to an arbitrarily small $\defi'> 0$ from an appropriate input deficiency $\defi=\defi'/2$. 

However, several obstacles prevent \eqref{eq:sparse_naiive_ideas} and \eqref{eq:dense_naiive_ideas} from being valid reductions: reusing the same input equations introduces \emph{noise dependence} across output equations, and (in both regimes) the resulting equation index sets do not match the target support distribution. Moreover, for any fixed input density $\eps$, the na\"ive resolution map only reaches a finite set of attainable target densities $\eps'$, which is far weaker than the continuum of implications in Thm.~\ref{thm:densifying_general}.

We make resolution into a valid reduction using two techniques: \emph{discrete resolution} ($\DiscrEqComb$) and \emph{Gaussian resolution} ($\GaussEqComb$). In $\DiscrEqComb$ (\autoref{subsec:sparse_resolution}), we avoid entry reuse by design, which ensures noise and index set independence across output samples. 
In $\GaussEqComb$ (\autoref{subsec:dense_resolution}), we introduce a \emph{Gaussian $\kxor$ analog}, $\ktens$, which enables a new dense resolution analysis. Working in  $\ktens$, we replace the majority vote in \eqref{eq:dense_naiive_ideas} by a normalized average of products $\Y_s\Y_t$. Although this introduces noise dependence, the entire aggregation step admits a clean matrix/tensor representation (a Gram-type object built from a sparse Gaussian matrix). This viewpoint lets us control the dependence via a new high-dimensional CLT, yielding an output instance that is close (in total variation) to one with i.i.d. entries.

A crucial ingredient is the intermediate model $\kxor^\FULL$ (see 
\autoref{subsec:full_ideas} below) to which we map from any regular instance. It is needed for the order-reduction subroutine used in the integer-approximation step that removes the discretization in attainable $\eps'$ 
(\autoref{subsubsec:reducing_k_by_a_factor}), and at density $\eps=1$ it provides a connection to standard Tensor PCA.

\subsection{Discrete Resolution: Avoiding Dependence by Design}\label{subsec:sparse_resolution}

In $\DiscrEqComb$ (full details in \autoref{sec:discrete_resolution}) we instantiate both \eqref{eq:sparse_naiive_ideas} and \eqref{eq:dense_naiive_ideas} starting from discrete $\kxor$ samples. Here we sketch some of the challenges and their solutions.

\subsubsection{Distributional Challenges in Discrete Resolution}\label{subsubsec:discrete_res_dist_challenge}

Two issues arise immediately: dependence created by reusing noisy input entries, and a mismatch between the output support distribution and the i.i.d. uniform support of the target model.
    
\textbf{Challenge 1: Noise dependence from entry reuse.}
A single noisy $\Y_s$ can participate in many products $\Y_s\Y_t$ (across different $\gamma = \alpha_s\triangle\alpha_t$), which introduces entry noise dependence in $\calZ$.

\textbf{Challenge 2: Index-set clustering in the sparse regime and support restriction.}
    \begin{align*}
\b\{\overbrace{i_1,\dots,i_{k'/2},\underbrace{\vio{i_{k'/2+1}\dots,i_k}}_{\vio{\kappa}}}^{\alpha_s}\b\} \triangle \b\{\overbrace{\underbrace{\vio{i_{k'/2+1}\dots,i_k}}_{\vio{\kappa}},i_{k+1},\dots,i_{k+k'/2}}^{\alpha_t}\b\} = \b\{\overbrace{i_1,\dots,i_{k'/2},i_{k+1},\dots,i_{k+k'/2}}^{\gamma}\b\}
    \end{align*}
For a pair $\alpha_s, \alpha_t \in \binom{\sqb{n}}{k}$ with $\alpha_s\triangle\alpha_t = \gamma \in \binom{\sqb{n}}{k'}$, let $\kappa = \alpha_s \cap \alpha_t \in \binom{\sqb{n}}{k-k'/2}$.
    For fixed $\kappa$, denote 
    $$ \mathcal{I}_{\kappa} = \set{j \in \sqb{m}: \kappa \subseteq \alpha_j}\,.$$ 
    Then, for most $s, t \in \mathcal{I}_{\kappa}$, we add $\paren{\alpha_s\triangle\alpha_t, \Y_s \Y_t}$ to the output $\calZ$.
    In \eqref{eq:sparse_naiive_ideas} the locations of the samples in $\calZ$ thus cluster within each $\kappa$-block, unlike the i.i.d. uniform index sets of the target model. Moreover, each produced index set $\gamma = \alpha_s\triangle \alpha_t$ is constrained to be disjoint from its originating $\kappa$.
    
\textbf{Idea: Avoiding entry reuse by design.}
For $\kappa \in \binom{\sqb{n}}{k-k'/2}$ and $\mathcal{I}_{\kappa} = \set{j \in \sqb{m}: \kappa \subseteq \alpha_j}$, all pairs $s, t \in \mathcal{I}_{\kappa}$ with $|\alpha_s \triangle \alpha_t|=k'$ yield candidate samples $\paren{\alpha_s\triangle\alpha_t, \Y_s \Y_t}$. 
If $k' \leq k (1-\eps)$, then
\begin{equation}\label{eq:sparse_resolution_regime}
\E\b| \mathcal{I}_{\kappa} \b| \approx \paren{\text{num. of }\alpha:\, \kappa\subseteq \alpha} \cdot \Pr\sqb{\alpha\text{ is observed}} \asymp n^{k'/2}\cdot \frac m {n^k} = n^{k'/2}\cdot n^{k(\eps-1)/2} = O\paren{1}\,,
\end{equation}
so for any fixed $\kappa$ w.h.p. there are at most $O(1)$ candidate products $\Y_{s}\Y_{t}$ ($s, t, \in \mathcal{I}_\kappa$) to consider.

In the regime of \eqref{eq:sparse_resolution_regime}, our \emph{dependence-avoiding} $\DiscrEqComb$ reduction (\autoref{sec:discrete_resolution}) simply selects \emph{one} disjoint pair of elements $s, t \in \mathcal{I}_{\kappa}$ at random and discards the rest.\footnote{The issue of the same observation $\Y_s$ belonging to $\mathcal I_\kappa$ for multiple $\kappa$ is fixed in \autoref{sec:discrete_resolution}.} This removes noise dependence (Challenge 1) by construction, since no observation is reused. Additionally, this eliminates the index-set dependence from the $\kappa$-block clustering (Challenge 2). What remains is the structural support restriction $\gamma\cap\kappa=\emptyset$: full index-set independence is obtained by splitting $[n]$ into dedicated “cancellation indices” (used only to form the $\kappa$ sets) and the remaining indices, see \autoref{sec:discrete_resolution}. Achieving a continuum of target $\eps'$ is a more challenging problem, discussed next.

\subsubsection{Discrete Resolution: Overcoming the Discretization in $\eps'$}\label{subsec:sparse_resolution_discrete}

\textbf{Challenge: Discrete parameter map.}
The algorithm in \autoref{subsubsec:discrete_res_dist_challenge} gives an average-case reduction 
\begin{equation}\label{eq:discrete_set}
\text{from}\qquad\kxor\text{ at density }\eps\qquad\text{to}\qquad \kxork {k'}\text{ at density }\eps' = \min\set{1, \frac{2k\eps}{k'}}\,,
\end{equation}
valid for all even $k' \leq k(1-\eps)$ (Lem.~\ref{lem:2path_red}). Thus, starting at a fixed $\eps$, the reduction cannot reach an arbitrary target density $\eps'\in \b( \frac{2\eps}{1-\eps}, 1\b]$ and order $k'$: hardness implications are obtained only for the pairs $(k', \eps')$ as above. In contrast, Thm.~\ref{thm:densifying_general} fills these gaps and shows for all $\eps > 0$, \footnote{The same argument applies verbatim for $\eps < 0$. However, it incurs an $n^{\vareps}$ loss, for arbitrarily small $\vareps$, in the signal level $\delta$. This is negligible in the dense regime, where the relevant signal levels are already inverse polynomial in $n$, but may be non-negligible in the sparse regime.}
$$
\conjeps \eps \Rightarrow \conjeps {\eps'}\,,\quad\text{for all }\eps' \in \B( \frac{2\eps}{1-\eps}, 1\B]\,.
$$

\textbf{Idea 1: Na\"ive reduction to adjust $\eps'$.} A trivial post-processing reduction (Lem.~\ref{lem:m_adjust}) can change parameter $\eps'$ by either (1) inserting pure noise $\rad(0)$ entries into the instance or (2) removing some of the observed entries at random. This achieves the transformation 
\begin{equation}\label{eq:density_adjust}
\xorp {k'} n {\eps_1} {2\defi} \to \xorp {k'} n {\eps_2} {\defi'}, \quad\text{where}\quad \defi' = 2\defi + {k'|\eps_1-\eps_2|/2}\,,
\end{equation}
where recall that $\xorp k n \eps \defi$ has $m = n^{k(1+\eps)/2}$ equations and $\delta = n^{-k\eps/4-\defi/2}$.
Na\"ively composing this with the resolution step gives deficiency $\defi' = 2\defi + k' \cdot \Delta\eps/2$, where $\Delta\eps$ is the necessary density adjustment. The discrete set of attainable $\eps'=\min\set{1, \frac{2k\eps}{k'}}$ has spacing at least $\frac {4\eps} {k-2} = \frac{2k\eps}{k-2} - \frac{2k\eps}{k}$, yielding $\defi'-2\defi \geq 2k\eps \cdot \frac \eps {k-2} > 2\eps^2$. 
Thus, for a fixed input density $\eps>0$, this fails to cover the entire regime of interest $\defi' > 0$.

Crucially, this lower bound $\defi' > 2\eps^2$ is independent of $k$. Hence, with this analysis, even the stronger $\conjeps \eps$ (i.e., $\conj k \eps$ for all $k\geq 3$) does \textbf{not} imply $\conjeps {\eps'}$ for all $\eps' \in \b( \frac{2\eps}{1-\eps}, 1\b]$. 

\textbf{Idea 2: Choosing an appropriate input order $k$.} In Thm.~\ref{thm:densifying_general}, for any fixed $\eps \in (0,1/3)$, $k'$, and $\eps' \in \b( \frac{2\eps}{1-\eps}, 1\b]$, we construct a reduction $$\text{from}\quad\xorp k n \eps \defi\quad\text{to}\quad\xorp {k'} {n'} {\eps'} {3\defi}\,,\quad n' = n^{\Theta(1)}\,,$$ (with the discrete resolution of Lem.~\ref{lem:2path_red} as subroutine) for a carefully chosen $k = k\paren{k',\eps,\eps',\defi}$.\footnote{One can achieve output deficiency $\defi' = C\defi$ for any $C > 2$.} 

Concretely, we show in Claim~\ref{claim:conv} a Diophantine approximation (i.e., rational approximation with constraints) statement, proved via the theory of continued fractions: for any (arbitrarily small) $c \in (0, 2\defi)$, and for any sufficiently small $\xi=\xi(c)$, there exists 
a $k$ and 
$r\in\mathbb{Z}$ satisfying 
$$
\frac{2k\eps}{rk'} \in \sqb{\eps' \pm \xi} \quad \text{and} \quad k \leq c \xi^{-1}\,.
$$
With this choice of $k$, the resolution reduction to instance of order $K' \coloneq rk'$ gets\footnote{$K' \leq k(1-\eps)$ follows from the $\eps' > \frac{2\eps}{1-\eps}$ condition and sufficiently large (constant) choice of $k$.} to within $\xi$ of $\eps'$ and the post-processing in \eqref{eq:density_adjust} allows us to move to $\eps'$ with a cost in deficiency of 
$$
K' \Delta \eps / 2 \leq (rk')\cdot \xi/2 \leq k\cdot  \xi/2 \leq c/2 \leq \defi\,,
$$ 
where we used $c \leq 2 \defi$. As can be seen from the last display, a weaker result bounding $k\leq c \xi^{-1.01}$ would not be sufficient; conversely, it turns out that a stronger bound $k\leq \xi^{-0.99}$ is impossible.
Note that mapping to a multiple $K' = rk'$ of the desired target $k'$ is necessary; directly mapping to $k'$
results in the discrete set of possible target densities $2kt/k'$ in \eqref{eq:discrete_set}.

The required final step is to reduce the order by a factor of $r$ (Lem.~\ref{lem:prelim_decrease_k}), yielding an instance of order $k'$ with the desired parameters, establishing hardness implications for the entire $\defi' > 0$ range. This step, however, is not immediate -- we describe the ideas and challenges in the following \autoref{subsubsec:reducing_k_by_a_factor} and develop the necessary machinery in \autoref{subsec:full_ideas}.

\subsubsection{Decreasing Order by a Factor}\label{subsubsec:reducing_k_by_a_factor}

As discussed in \autoref{subsec:sparse_resolution_discrete}, the final ingredient of discrete resolution is an order-reduction step: for an integer factor $r$, we want to map a $K'$-order instance $\xorp {K'} n {\tilde\eps}{\tilde\defi}$ with $K' = rk'$ to an instance of the target order $k'$ at the same density $\tilde\eps$. We first describe a natural change-of-variables idea and then explain the technical machinery needed to make it distributionally correct. 

\textbf{Idea: Change-of-variables.} This idea is related to coefficient-sparsifying reductions for noisy linear equations \cite{bangachev2024near} and the $k'=2$ case already appears as a subroutine in the refutation and recovery algorithms of \cite{feige2006witnesses,allen2015refute,barak2016noisy,guruswami2023efficient}, which did not require mapping to the precise target distribution as we do. 

Given an input $\xorp {K'} n {\tilde\eps}{\tilde\defi}$ for $K' = rk'$, the idea is to group the $K'$ indices of every equation into $k'$ blocks of size $r$. For each $r$-tuple $\beta = \paren{i_1,\dots,i_r} \in\sqb{n}^r$, define a new variable $y_{\beta} = x_{i_1}\dots x_{i_r}$, so $y \in \set{\pm1}^{n^r}$. Then an input sample $\paren{\alpha, \Y}$ with a set $\alpha = \set{i_1,\dots,i_{K'}} \in \binomset {n}{K}$ (in any fixed ordering) satisfies
$$
\Y = x_{i_1}\dots x_{i_{K'}} \cdot \noise = \B(\prod_{j=1}^{k'} y_{\paren{i_{(j-1)r+1},\dots,i_{jr}}}  \B)\cdot \noise\,,
$$
i.e., is a noisy observation of a $y^{\otimes k'}$ entry. To restore the uniform prior of the planted signal, we sample an independent random sign vector $w \sim \unif\sqb{\set{\pm 1}^{n^r}}$ and change the entry signal to $y\odot w$ by returning a collection of $\b(\tilde \alpha, \tilde \Y\b)$ for
$$\tilde\alpha = \set{ \paren{i_1,\dots,i_r}, \dots, \paren{i_{r(k'-1)+1,\dots,rk'}}}\quad\text{and}\quad \tilde\Y = \Y \cdot w_{\tilde\alpha}\,.$$ 

\textbf{Challenge: Wrong index-set distribution.} Marginally, each output pair $\b(\tilde \alpha, \tilde \Y\b)$ has the correct label/noise structure for $\xorp {k'} {n^r} {\tilde\eps} {\tilde\defi / r}$. The problem is the distribution of the supports $\tilde \alpha$: the input index sets $\alpha$ do not allow index repetitions, so $r$-tuples in $\tilde \alpha$ are forced to be disjoint. This mismatch is inherent: allowing overlaps between the $r$-tuples would correspond to an input equation with repeated indices, which the regular $\kxor$ model forbids.

\textbf{Solution: Work in a model that allows repetitions.}
We therefore pass to the $\FULL$ variant of $\kxor$, denoted $\kxor^\FULL$, which samples index \emph{multisets} with repetitions allowed, independently from the natural measure (\autoref{subsec:kxor_variants}, \autoref{subsec:repeated_indexes}). In $\kxor^\FULL$ , the change-of-variables mapping becomes distributionally correct, yielding the order-reduction Lemma~\ref{lem:prelim_decrease_k}:
$$
\text{from}\quad \xorpfull {K'} n {\tilde\eps} {\tilde\defi} \quad\text{to}\quad \xorpfull {k'} {n^r} {\tilde\eps} {\tilde\defi/r}\quad\text{for } K' = rk'\,. 
$$
To use this inside $\DiscrEqComb$, we must therefore construct an instance of $\xorpfull {K'} n {\tilde\eps} {\tilde\defi}$ (which includes samples with repeated indices) starting from the regular $\kxor$. This is achieved by a separate “composition” step: in \autoref{subsubsec:full_asymm_kxor_red}, Prop.~\ref{prop:comp_equiv_full_collection} shows how to assemble a $\kxor^\FULL$ instance of order $K'$ from a collection of independent low-order instances: 
$$
    \B\{ \kxork {(K'-2s)}_{\eps_s, \tilde\defi} (n) \B\}_{0\leq s \leq \frac{K'-1}{2}} \quad \to \quad \kxork{K'}_{\tilde\eps, \tilde\defi}^{ \mathsf{FULL}}(n)\,, \qquad \text{where } \eps_s = \min\set{ \frac{K'\tilde\eps}{K'-2s} , 1}\,.
$$
Finally, this collection of low-order instances can be constructed in our setting. Recall that discrete resolution maps an input $\xorp k n \eps \defi$ to order $k^\star$ with density 
$$
\eps^\star = \min\set{1, \frac{2k\eps} {k^\star}}
$$
for all even $k^\star \leq k(1-\eps)$. Setting $k^\star = K'-2s$ we match the $\eps_s$ above. Using initial ``instance splitting"(\autoref{subsubsec:cloning_splitting_overview}), we run resolution independently for each $s$ to obtain the required family $\B\{ \kxork {(K'-2s)}_{\eps_s, \tilde\defi} (n) \B\}_{0\leq s \leq \frac{K'-1}{2}}$. Prop.~\ref{prop:comp_equiv_full_collection} then assembles these into a single $\kxork{K}_{\tilde\eps, \tilde\defi}^{ \mathsf{FULL}}(n)$ instance, after which Lemma~\ref{lem:prelim_decrease_k} decreases the order by the factor $r$. 

We overview the $\kxor^\FULL$ model in the next section, both to explain Prop.~\ref{prop:comp_equiv_full_collection} and to formalize its connection to standard Tensor PCA.

\subsection{$\kxor^\FULL$: Repeated Indices and Reductions to Tensor PCA}\label{subsec:full_ideas}

We introduced the $\kxor^\FULL$ model in \autoref{subsubsec:reducing_k_by_a_factor} as an intermediate instance needed for the order-reduction subroutine of $\DiscrEqComb$. Another reason to introduce $\kxor^\FULL$ is that, at density $\eps=1$, it is computationally equivalent to the standard Tensor PCA model in Eq.~\eqref{eq:tensor_pca}. This equivalence is what enables our reductions to Tensor PCA starting from a regular $\kxor$ instance (Thm.~\ref{thm:densifying_general}, \ref{thm:specific_intro}).

\textbf{$\kxor^\FULL$ and connection to Tensor PCA.} $\xorpfull k n \eps \defi$ (see \autoref{subsec:kxor_variants}, \autoref{subsec:repeated_indexes}) is a variant of $\kxor$ that allows index repetitions -- it consists of $m$ samples $(\alpha, \Y)$ for independent index $k$-multisets $\alpha \in \set{n}^k$ sampled from a natural measure in Def.~\ref{def:multiset} and $\Y \sim x_{\alpha} \cdot \rad(\delta)$. We show in Corollary~\ref{cor:tensor_pca_wr_equiv} the computational equivalence of the two models:
\begin{equation}\label{eq:equiv_tensor_pca_full}
\xorpfull k n {\eps=1} \defi \quad\Leftrightarrow\quad \ktensorpca_\defi (n)\,,
\end{equation}
where recall that a $\ktensorpca_\defi (n)$ instance consists of $\delta x^{ \otimes k} + \noiseG$ with $\noiseG$ a symmetric tensor of i.i.d. $\N(0,1)$ entries. Computational equivalence means there exist average-case reductions both ways preserving signal parameter $\delta$ up to $\poly\log n$ factors: they use our Gaussianization/Discretization subroutines (\autoref{subsubsec:equiv_noise_model}) and reductions changing sampling of $\alpha$ to with or without replacement (\autoref{subsubsec:equiv_w_wor_replacement}). 

As we will see next, all our reductions that output $\xorp{k}{n}{\eps=1}{\defi}$ can be modified to output $\xorpfull{k}{n}{\eps=1}{\defi}$; composing with \eqref{eq:equiv_tensor_pca_full} then yields reductions to $\ktensorpca_\defi(n)$.

\textbf{$\kxor^\FULL$ as a collection of low order instances.} In Prop.~\ref{prop:comp_equiv_full_collection} we show that $\kxor^\FULL$ is computationally equivalent to a collection of independent low-order instances:
\begin{equation}\label{eq:ideas_equiv_decompose}
\kxor_{\eps, \defi}^{ \mathsf{FULL}}(n) \quad \Leftrightarrow \quad \B\{ \kxork {(k-2s)}_{\eps_s, \defi} (n) \B\}_{0\leq s \leq \frac{k-1}{2}}\,, \qquad \text{where } \eps_s = \min\set{ \frac{k\eps}{k-2s} , 1}\,.
\end{equation}
Indeed, consider those samples $(\alpha, \Y)$ of $\kxor_{\eps, \defi}^{ \mathsf{FULL}}(n)$ for which the index multiset $\alpha$ can be expressed as 
$$
\alpha = \beta \cup \set{i_1,i_1,\dots,i_s,i_s} \text{ for } \beta \in \binomset n {k-2s} \text{ and } i_1,\dots,i_s \in \sqb{n}\,,
$$
i.e. $\beta$ collects the indices of \emph{odd} multiplicity in $\alpha$. Then, 
$$
\Y = x_{\alpha} \cdot \rad(\delta)\eqdist x_{\beta} \cdot \rad(\delta)
$$
is distributed as a lower order $\kxork {(k-2s)}$ entry. Such classification (with additional aggregation step) decomposes the $\kxor^\FULL$ instance into the desired collection and the reverse reduction (with a cloning step, \autoref{subsubsec:cloning}) composes the instances into one $\kxor^\FULL$.

\textbf{Reductions to Tensor PCA.} Similarly to discrete resolution in \autoref{subsubsec:reducing_k_by_a_factor}, our reductions from $\xorp k n \eps \defi$ to $\xorp {k'} n {\eps=1} {\defi'}$  (Thm.~\ref{thm:specific_intro}) work for all even $k' \leq \tau k$, where $\tau = \tau(\eps)$ is explicit. To leverage Eq.~\eqref{eq:equiv_tensor_pca_full} and \eqref{eq:ideas_equiv_decompose}, we first apply our splitting reduction (\autoref{subsubsec:cloning_splitting_overview}) to obtain $k'/2$ independent copies of $\xorp k n \eps \defi$. We then run the densifying reduction on these copies to obtain the collection of lower-order instances at orders $k',k'-2,\dots,2$ and densities $\eps_s = \min\set{\frac{k'}{k'-2s}, 1} = 1$, assemble them into a single $\xorpfull{k'}{n}{\eps=1}{\defi'}$ instance using \eqref{eq:ideas_equiv_decompose}, and finally map to $\ktensorpca_{\defi'}(n)$ via~\eqref{eq:equiv_tensor_pca_full}.

Finally, passing through $\kxor^\FULL$ is what lets us reach odd target orders $k'$: the order-reduction step by a factor of $2$ applies in the $\FULL$ model (as in \autoref{subsubsec:reducing_k_by_a_factor}).

\subsection{Gaussian Resolution: Gram Matrix}\label{subsec:dense_resolution}

In Gaussian resolution ($\GaussEqComb$) (see \autoref{sec:gauss_resolution}) we permit entry reuse in \eqref{eq:dense_naiive_ideas}, which lets us map to $\eps'=1$ for a wider parameter range than $\DiscrEqComb$ (\autoref{subsec:sparse_resolution}).

\textbf{Challenge: Strong noise dependence.}  When the same noisy $\Y_s$ is reused across many output $\gamma$, this introduces noise dependence in the output. In the regime of $\GaussEqComb$ each $\Y_s$ is expected to contribute to $\gg 1$ output entries and this reuse is crucial: it preserves the \emph{signal strength} in the output instance -- without it, the reduction maps to a much noisier $\kxor$ instance, yielding weak (or vacuous) complexity implications. 

\textbf{Idea: Gaussian noise and tensor/matrix viewpoint.} To carry out this reduction, we introduce a Gaussian analog of $\kxor$, called $\ktens$, where we replace the multiplicative Rademacher noise $\noise_{j}$ in $\kxor$ equations with additive Gaussian $\noiseG_{j} \sim \N(0,1)$. The key in the Gaussian world is to switch to a \emph{tensor/matrix viewpoint} where the dependence created by reuse can be written explicitly and then controlled via a new high-dimensional CLT. 

Formally, in $\gtens k n m \delta$ (equiv., $\gtensp k n \eps \defi$ with $m = n^{k(1+\eps)/2}, \delta = n^{-k\eps/4-\defi/2}$) with a secret vector $x\in\set{\pm1}^n$ we observe $\calY = \set{\paren{\alpha_j, \Y_{j}}}_{j\in \sqb{m}}$ for $\set{\alpha_j}_{j\in\sqb{m}}$ sampled independently with replacement, and
\begin{equation}\label{eq:def_ktens_ideas}
\Y_{j} = \delta x_{\alpha_j} + \noiseG_{j}\quad\in\R\,,\tag{$\ktens$}
\end{equation}  
where $\noiseG_{j} \sim_\mathrm{i.i.d.} \N(0,1)$. It turns out (see \autoref{subsubsec:equiv_noise_model}) that the discrete $\kxor$ and Gaussian $\ktens$ models are \emph{computationally equivalent} and hence \emph{interchangeable for our purposes}:

\begin{proposition}[Computational Equivalence of $\kxor$ and $\ktens$]\label{prop:gauss_discr_equiv_gen}

    For any $k, n, m, \delta$, models $\xor k n m \delta$ (Eq.~\eqref{eq:def_kxor}) and $\gtens k n m {\delta}$ (Eq.~\eqref{eq:def_ktens_ideas}) are \emph{computationally equivalent}. In particular, there is a poly-time reduction algorithm that preserves signal $x\in\set{\pm1}^n$ from $\xor k n m \delta$ to $\gtens k n m {\delta'}$ for some $\delta' = \Thetat(\delta)$ and vice versa.\footnote{$\Thetat$ hides $\poly\log n$ factors.}

    As a consequence, poly-time detection/recovery algorithms are transferred between the problems.
\end{proposition}

The product of two $\ktens$ entries $\Y_{s}, \Y_{t}$ for distinct $s, t\in \sqb{m}$ with $\kstar = |\alpha_s \triangle \alpha_t|$ has a mean matching the $\ktensk \kstar$ entry with signal level $\delta' = \delta^2$ and index set $\alpha_s \triangle \alpha_t$:
\begin{equation*}
    \E \sqb{\Y_{s} \Y_{t}} = 
    \delta^2 \cdot x_{\alpha_s} x_{\alpha_t} = \delta^2 \cdot x_{\alpha_s \triangle \alpha_t}\,.
\end{equation*}
However, unlike in the discrete model (\autoref{subsec:sparse_resolution}), the product $\Y_s\Y_t$ is not (even approximately) a Gaussian-noise observation of $x_{\alpha_s \triangle \alpha_t}$:
\begin{equation}\label{eq:gauss_entry_product}
\Y_{s} \Y_{t} = 
    \delta^2 \cdot x_{\alpha_s \triangle \alpha_t} + {\delta x_{\alpha_s} \noiseG_t + \delta x_{\alpha_t} \noiseG_s} + {\noiseG_s \noiseG_t}\,.
\end{equation}
The key is the aggregation in \eqref{eq:dense_naiive_ideas}: in the Gaussian case we take a renormalized (variance-1) average of the products $\Y_s\Y_t$, which marginally yields the correct Gaussian law by CLT. 

For this, we adopt a \emph{tensor view} of the input instance $\set{\paren{\alpha_j, \Y_j}}_{j \in \sqb{m}}$, arranging the values $\Y_j$ into corresponding tensor entries and ignoring duplicates. Let $\Yb \in \paren{\R^n}^{\otimes k}$ be a symmetric (up to index permutations) tensor, where entries with repeated indices are set to $0$. Symmetry allows us to index the remaining distinct-index entries by $k$-subsets of $\sqb{n}$; we set for each $\alpha\in\binomset n k$:
\begin{equation*}
    \Yb_{\alpha} = \begin{cases}
        0, &\text{ if }\set{j \in \sqb{m}: \alpha_j = \alpha} = \emptyset\\
        \Y_{j^\star}, &\text{ otherwise, where } j^\star=\min\set{j \in \sqb{m}: \alpha_j = \alpha}\,.
    \end{cases}
\end{equation*}
Consider the matrix flattening $$Y = \mat {k-k'/2} {k'/2} \Yb \in \R^{\binom n {k-k'/2}\times \binom n {k'/2}}\,,$$which is a matrix whose rows are indexed by elements of $\binomset n {k-k'/2}$ and columns by $\binomset n {k'/2}$: $Y_{\kappa,\gamma} = \Yb_{\kappa\cup\gamma}$ for disjoint $\kappa\times\gamma \in \binomset n {k-k'/2} \times \binomset n {k'/2}$ (and $0$ otherwise). 
Each entry $\b(Y^\top Y\b)_{\gamma_1, \gamma_2}$ for disjoint $\gamma_1,\gamma_2 \in \binomset n {k'/2}$ can then be expressed as  
\begin{equation}\label{eq:gauss_resolution_ideas}
    \b(Y^\top Y\b)_{\gamma_1, \gamma_2} = \sum_{\kappa \in \binomset n {k-k'/2}} \Yb_{\gamma_1\cup \kappa} \Yb_{\gamma_2 \cup \kappa} \ \ \approx \sum_{ \substack{s, t\in \sqb{m}:\\ \alpha_s\triangle\alpha_t = \gamma_1\cup \gamma_2}} \Y_s \Y_t\,,
\end{equation}
where the last approximation ignores the potential index set duplicates in the input. Then the tensorization (i.e., the inverse procedure to $\mat {k'/2} {k'/2} \paren{\cdot}$) 
$$
\Zb = \tenso{\b(Y^\top Y\b)}\quad \in \paren{\R^{n}}^{\otimes k'}\,.
$$
has
$
\E\sqb{\Zb} \propto x^{\otimest k'}\,,
$\footnote{$x^{\otimest k'}$ denotes entries of $x^{\otimes k'}$ with distinct indices.} and a natural Gaussian analog of \eqref{eq:dense_naiive_ideas} is to output renormalized $\Zb$. We first use Gaussian Cloning, Poissonization, and index subsampling procedures (see \autoref{sec:gauss_resolution}) to reduce $Y^\top Y$ to a product of two independent non-centered, Bernoulli-masked Gaussian matrices. 

The remaining challenge is the Wishart-like noise dependence and cross-terms (as is evident in Eq.~\eqref{eq:gauss_entry_product}), which we address with our main technical Lemma~\ref{lem:wishart}. Lemma~\ref{lem:wishart} extends prior results of \cite{brennan2021finetti} for Wishart matrices ($X^\top X$ where $X$ is a matrix of i.i.d. Gaussians) to our setting of non-centered, masked Gaussian matrices $X$. This shows that our reduction based on \eqref{eq:gauss_resolution_ideas} outputs $\Zb$ close in total variation distance to a tensor with planted signal $\delta' x^{\otimest k'}$ and additive i.i.d. Gaussian noise as required by the Tensor PCA model. Additional postprocessing (\autoref{subsec:w_replacement}) can then map $\Zb$ back to the equation-based view $\gtensp {k'} {n'} {\eps'=1} {\defi'}$.

In the special case when the input density is $\eps=1$, through our tensor/matrix viewpoint we identify a connection to Spiked Covariance and Spiked Wigner Models and use recent reductions between them \cite{bresler2025computational} to obtain a wider range of output $k'$.

\subsection{Reductions that Decrease Tensor Order $k$}\label{subsec:ideas_decrease_k}

We now overview the main ideas behind the order-reducing reductions of Thm.~\ref{thm:decrease_k_intro}. For the two endpoints $\eps \in\set{0,1}$, the discrete \autoref{subsec:sparse_resolution} (for $\eps=0$) and Gaussian \autoref{subsec:dense_resolution} (for $\eps=1$) resolution techniques give a reduction from 
$$
\xorp k n \eps \defi \quad\text{to}\quad \xorp {k'} {n'} \eps \defi
$$
for any \emph{even} $k' < \begin{cases} k, & \eps = 0\\
2k/3, & \eps = 1\end{cases}$. While the result for $t=0$ is a special case of $\DiscrEqComb$ described in \autoref{subsec:sparse_resolution} (proved in \autoref{sec:discrete_resolution}), for $\eps = 1$, we additionally improve upon the general $\GaussEqComb$ results described in \autoref{subsec:dense_resolution} (proved in \autoref{sec:gauss_resolution}) to get a wider range of attainable $k'$ (see \autoref{subsec:gauss_res_eps_1}). We achieve this by generalizing a recent reduction from Spiked Covariance to Spiked Wigner \cite{bresler2025computational} and apply it to the tensor/matrix $\ktens$ viewpoint described in \autoref{subsec:dense_resolution}.

To reach \emph{odd} $k'$, in the case of even $k$ or $\eps = 1$ we stack this with the $k' \to k'/2$ for $\kxor^\FULL$ instance (see \autoref{subsubsec:reducing_k_by_a_factor}, \autoref{subsec:full_ideas}), obtaining the $k' < \begin{cases} k/2, & \eps = 0\\
k/3, & \eps = 1\end{cases}$ range claimed in Thm.~\ref{thm:decrease_k_intro}.

For $k$ odd, $\eps=0$, we obtain a wider range by doing a two-step reduction\footnote{With additional cloning subroutines that ensure instance independence, \autoref{subsec:cloning_splitting}.}: first map $k$ to $\tilde k$ as above and then use the two instances of orders $k$ and $\tilde k$ to get one of order $k'$. For this, we generalize the resolution reduction to combine two $\kxor$ inputs, which may be of different order $k$, but share the same secret $x \in \set{\pm1}^n$ (Lem.~\ref{lem:2path_red_ext_discr}): 
$$
\xorp {k_1} n \eps {\defi_1} \text{ ``+" } \xorp {k_2} n \eps {\defi_2} \rightarrow \xorp {k'} n {\eps} {\defi_1+\defi_2}\,,
$$
where $k' = k_1+k_2-2a$ and $a = |\alpha \cap \beta|$ is the number of canceled indices in the resolution identity. 

\section{Discussion and Open Problems}\label{sec:discussion}

\paragraph{Missing hardness implications in the dense $\eps \geq 0$ regime.}
As noted in \autoref{subsec:complexity_profiles}, in the regime where the number of $\kxor$ equations $m = n^{k(1+\eps)/2}$ is at least $n^{k/2}$, the believed $\kxor$ hardness at density $\eps \geq 0$ is summarized in the $\conj k \eps$ conjecture (\ref{conj:hardness}). Our reductions establish
\begin{equation}\label{eq:conj_implication}
\qquad\qquad\qquad\qquad\quad\conj k \eps \Rightarrow \conj {k'} {\eps'}\tag{Hardness Implication} 
\end{equation} 
for many pairs with $\eps \leq \eps'$ and $k \geq k'$. We conjecture (Conj.~\ref{conj:one_way_hardness}) that the current parameter restrictions reflect limitations of our techniques, rather than a fundamental barrier. 
\begin{openproblem}[Conj.~\ref{conj:one_way_hardness}]
    Establish \eqref{eq:conj_implication} for all $0 \leq \eps \leq \eps'\leq 1$ and $k \geq k'$.
\end{openproblem}
A particularly important special case is $\eps = 0$ to any $\eps' > 0$. Our densifying reductions map any fixed $\eps > 0$ to larger $\eps' $ (including $\eps'=1$, i.e., Tensor PCA), but do not currently succeed for $\eps = 0$ (canonical $\kxor$) to denser regimes.
\begin{openproblem}[Reductions from Canonical $\kxor$] Establish \eqref{eq:conj_implication} for some $k, k'\geq 3$ and $\eps = 0, \eps' \in (0,1]$.
\end{openproblem}
All our reductions increase density $\eps$: they map instances with fewer entries to denser ones (with appropriately smaller entry-wise SNR). This leaves open the natural opposite direction -- e.g., whether hardness of canonical Tensor PCA implies hardness of canonical $\kxor$.
\begin{openproblem}[Sparsifying Reduction] Establish \eqref{eq:conj_implication} for some $k, k'\geq 3$ and $\eps > \eps' \in \sqb{0,1}$.
\end{openproblem}
Such a result from $\eps=1$ to $\eps'=0$ would have an additional consequence: it would imply hardness of canonical $\kxor$ from the secret-leakage planted-clique conjecture. This is because there exists an average-case reduction from secret-leakage planted clique to Tensor PCA \cite{brennan2020reducibility}.

\paragraph{Hardness against subexponential-time algorithms.} Prior work gives \emph{subexponential-time} algorithms and matching lower bounds in restricted models in the regime where polynomial-time algorithms are believed not to exist \cite{basu2025solving,raghavendra2017strongly,kothari2017sum,wein2019kikuchi,kunisky2019notes,bhattiprolu2017sum}. In particular, it is conjectured\footnote{We believe the results for Tensor PCA, i.e., $\eps = 1$, extend to intermediate densities $\eps \in (0,1)$.} that there is no $n^{ O\paren{\ell}}$-time algorithm solving $\xor k n m \delta$ below
\begin{equation}\label{eq:subexp_tradeoff}
    m \delta^2 \asymp n \paren{n/\ell}^{k/2-1}\,,
\end{equation}
while $n^{ O\paren{\ell}}$-time algorithms exist above this threshold. Poly-time average-case reductions transfer sub-exponential runtime algorithms and, consequently, hardness; more generally, for this purpose the reduction itself may run in a smaller subexponential time.

In the sparse case, where $m = n^{k(1+\eps)/2}$ with $\eps < 0$ and $\delta$ is a fixed constant, Eq.~\eqref{eq:subexp_tradeoff} becomes
$$
\eps = - (1 - 2/k) \log_n \ell\,.
$$
Our reductions in Theorem~\ref{thm:intro_resolution_basic} map $\kxor$ at density $\eps < 0$ to $\kxork {k'}$ at density $\eps' = \frac{2k\eps}{k'} < 0$ for any even $k' \leq k(1-\eps)$. However, this transformation does not preserve the tradeoff in Eq.~\eqref{eq:subexp_tradeoff}, except in the regime $k(1+\eps) \leq 2$, which is at or below the information-theoretic threshold (\autoref{subsec:complexity_profiles}).

\begin{openproblem}[Subexponential Trade-off for $\eps < 0$] 
    Construct an average-case reduction from $\kxor$ at density $\eps < 0$ to $\kxork {k'}$ at density $\eps' < 0$ that preserves the conjectured optimal subexponential runtime curve, namely
    $$\frac \eps {1 - 2/k} = \frac {\eps'} {1 - 2/k'}\,.$$
\end{openproblem}

In the dense case, where $m = n^{k(1+\eps)/2}$ has $\eps > 0$ and $\delta = n^{-k\eps/4 - \defi/2}$, Eq.~\eqref{eq:subexp_tradeoff} becomes
$$
\defi = \paren{k/2-1}\log_n \ell\,.
$$
Most of our reductions are based on resolution, i.e., multiplying pairs of entries, which \emph{doubles the deficiency}. Since these reductions also decrease the order $k$, they do not trace the conjectured subexponential runtime curve. 

\begin{openproblem}[Subexponential Trade-off for $\eps > 0$] 
    Construct an average-case reduction from 
    $
    \xorp k n \eps \defi$ to $\xorp {k'} {n'} {\eps'} {\defi'}\,,
    $
    that preserves the conjectured optimal subexponential runtime curve, namely
    $$\frac \defi{k/2-1} = \frac {\defi'} {k'/2-1}\,.$$
\end{openproblem}

\paragraph{Sparse LWE.} As mentioned in \autoref{subsubsec:lwe_results_intro}, our $\DiscrEqComb$ average-case reductions extend to the $k$-sparse $\lwe$ model (for constant $k$), first formalized in \cite{jain2024systematic} (\autoref{subsec:resolution_lwe}). In $\lweg k n m \chi$ (Def.~\ref{def:k_sparse_lwe}) one observes $m$ i.i.d. samples
\begin{equation}\label{eq:entry_Fq_discussion}
    \Y = a_{1}x_{i_{1}}+a_{2}x_{i_{2}}+\dots + a_{k}x_{i_{k}} + e \,,
\end{equation}
where $x\in \Z_q^n$ is the secret, $\alpha = \set{i_{1},\dots,i_{k}}$ is a uniformly random index set, $e \sim \chi \in \Z_q$ is i.i.d. noise, and the coefficients $a_{s} \sim_{\text{i.i.d.}} \unif\paren{\Z_q\setminus 0 }$. We prove an average-case reduction mapping 
$$
    \text{from}\quad \lweg k n m {\chi_\theta}\quad \text{to}\quad \lweg {k'} {n'} {m'} {\chi_{\theta'}}
$$
for a wide range of parameters and explicit maps $f$, such that $(k', m', \theta') = f(k, m, \theta)$ (see Lemma~\ref{lem:2path_red_lwe} and informal Thm.~\ref{thm:lwe_informal}). Our results apply as long as the noise family $\set{\chi_\theta}_{\theta\in \Theta}$ is \emph{resolution-stable} (Def.~\ref{def:resolution_stable}), which includes the following common $\lwe$ noise models:
\begin{itemize}
    \item \textbf{Uniform $\mathbb{Z}_q$ Noise:} $e \sim \chi_{\delta}$ if $e = \begin{cases}
        0, &\text{w. prob. }\delta\\
        \unif\paren{\Z_q }, &\text{w. prob. }1-\delta\,.
    \end{cases}$
    \item \textbf{Discrete Gaussian Noise:} 
    $
    \chi_s(e) \propto \sum_{t\in \Z} e^{-\pi (e + tq)^2/s^2}\,.
    $
    \item \textbf{Bounded $\mathbb{Z}_q$ Noise:} $e \sim \chi_{l}$,
    where $\chi_l$ is a given measure on $\set{-l, l}$ for a parameter $l \leq q$.
\end{itemize}
For these noise models we present a simple detection algorithm for $\klwe$ based on equation collision (\autoref{subsec:collision_detection}) that, combined with our reductions, yields new detection algorithms for $\klwe$.

First, exhaustive exploration of resolution for $\klwe$ is outside the scope of this paper and we believe the reductions in Thm.~\ref{thm:lwe_informal} can be extended to a wider parameter range. Second, \cite{jain2024systematic} relate hardness of $\klwe$ in dimension $n$ to that of (dense) $\lwe$ in dimension $k$, and, more recently, \cite{bangachev2024near} give a polynomial-time reduction from dense $\lwe$ to $\klwe$ that yields $n^{\Omega(k/\poly\log(k,n))}$-runtime lower bounds for solving $\klwe$ under standard assumptions for dense $\lwe$. However, the overall complexity profile of $\klwe$ model remains underexplored, especially in its dependence on the order $k$, the number of samples $m$, and the noise distribution.

\begin{openproblem}[$\klwe$]
$ $
\begin{enumerate}
    \item For fixed $k$ and noise $\chi$, characterize the minimal sample complexity $m = m(n)$ for poly-time detection and recovery in $\lweg k n m \chi$, and determine whether there is an computational-statistical gap;
    \item Develop a broader set of average-case reductions between $\klwe$ instances to characterize the complexity profile of $\klwe$.
\end{enumerate}
\end{openproblem}

\paragraph{$\kxor$ against low-noise.}

In the \emph{low-noise} regime of $\xor k n m \delta$, 
$$m = n^{k(1+\eps)/2} \text{ with }\eps<0\qquad \text{and}\qquad\delta = 1 - n^{-\snr}, \snr > 0\,.$$ Recent work \cite{chen2024algorithms} gives efficient detection and recovery algorithms in this regime whenever 
\begin{equation}\label{eq:tradeoff_lownoise}
m \geq \Omega\big(n^{1+(1-\Delta)\cdot \frac{k-1}{2}}\big)\quad\text{and}\quad \delta \geq 1 - n^{-\frac{1+\Delta}{2}}\quad\text{for } \Delta \in \paren{0,1}\,.
\end{equation}
Our reductions in Theorem~\ref{thm:intro_resolution_basic} map
$$
\xor k n {m=n^{k(1+\eps)/2}} \delta \quad\text{to}\quad \xor {k'} n {m'=\Thetat(n^{k'(\eps'+1)/2})} {\delta'}
$$
for even $k' \in (2|\eps| k , k(1-\eps)]$, $\eps' = \frac {2k} {k'} \eps$, and $\delta' = \delta^2 = 1-2 (1-o(1))n^{-\snr}$. This particular reduction does \emph{not} trace the algorithmic trade-off \eqref{eq:tradeoff_lownoise} and does not yield new algorithms in this low-noise regime. 
\begin{openproblem}
    Develop average-case reductions between $\kxor$ instances with different parameters $\paren{k, m, \delta}$ that (1) work in the low-noise regime above, and (2) trace (or improve) the $m\leftrightarrow \delta$ tradeoff \eqref{eq:tradeoff_lownoise}.  
\end{openproblem}

\paragraph{Algorithmic implications of our average-case reductions.}

Our reductions to Tensor PCA (\autoref{sec:to_tensor_pca}) work for any $\eps \in (0,1]$ and map\footnote{The reductions are stated for $\defi \geq 0$, but they succeed for recovery for all $\defi\in\R$.} 
$$
\xorp k n \eps \defi \quad\text{to}\quad \xorp {k'} {n'} {\eps'=1} {2\defi}\,,\qquad k' < \alpha k \text{ is even},
$$
where $n' = n(1-o(1))$, the output secret is a restriction $x' = x_{\sqb{1:n}}$, and $\alpha = \alpha(\eps)$ is an explicit function.
Moreover, these reductions can be adjusted\footnote{ The reductions to Tensor PCA in \autoref{sec:to_tensor_pca} can map to \emph{any} even $k' \leq \alpha k$, so via standard splitting (\autoref{subsubsec:cloning_splitting_overview}), the reduction can map to a collection of $\xorp {k''} {n'} {\eps'=1} {2\defi}$ for $k'' = k', k'-2,\dots$. \autoref{subsubsec:full_asymm_kxor_red} and \autoref{subsubsec:equiv_w_wor_replacement} describe how to then obtain the desired full tensor.} to map to the full Tensor PCA model, in which one observes 
$$
\delta' (x')^{\otimes k'} + \noiseG
$$
for $\delta' = \widetilde\Omega(n^{-k'/4 - \defi})$ and $\noiseG$ is a symmetric tensor of i.i.d. $\N(0,1)$ entries. The Kikuchi-hierarchy algorithms of \cite{wein2019kikuchi} achieve weak recovery in polynomial time for the Tensor PCA model\footnote{For our purposes even the matrix case $k=2$ suffices.} above at $$\delta' \geq \tilde\Theta\paren{1}\cdot n^{-k/4}\,.$$ Putting these together gives:
\begin{proposition}
Let $\eps \in (0,1], k \in \mathbb{Z}^{+}$ and let 
$$
m = n^{k(\eps-1)/2}\quad\text{and}\quad \delta \geq \tilde\Theta\paren{1} \cdot n^{-k\eps/4}\,,
$$
where $\tilde\Theta\paren{1}$ hides the $\poly\log (n)$ factors. For $k$ satisfying 
$$
k > \begin{cases}
    1/\eps, &\eps \in \paren{0,1/3}\\
    2/\paren{1-\eps}, &\eps \in [1/3,1]\,,
\end{cases} 
$$
there exists a $\poly\paren{n}$-time algorithm for weak recovery up to a global sign for $\xor k n m \delta$.
\end{proposition}
The conditions on $k$ are the admissibility conditions for mapping $\xor k n m \delta$ to the full 2-Tensor PCA for $\delta' = \tilde\Theta\paren{1}\cdot n^{-k/4}$ via $\DiscrEqComb$ (see \autoref{sec:to_tensor_pca}).
This leaves open several small-$k$ cases and the removal of the $\tilde\Theta\paren{1}$ factors. While known techniques (e.g., \cite{wein2019kikuchi}) likely extend naturally to the intermediate densities $\eps\in\paren{0,1}$, removing the $\tilde\Theta\paren{1}$ factors for $\eps \in [0,1)$ might be more challenging -- such tight bounds have been obtained for Tensor PCA only recently in \cite{kothari2025smooth}.

\begin{openproblem}[Recovery at Intermediate Sparsities $\eps \in \paren{0,1}$]
    For all $k\geq 2$ and $\eps\in\paren{0,1}$ in the above, obtain detection and recovery algorithms for $\xor k n m \delta$ at the computational threshold \eqref{eq:tradeoff}
    $$
m = n^{k(\eps+1)/2}\quad\text{and}\quad \delta \geq \Theta\paren{1} \cdot n^{-k\eps/4}\,.
$$
\end{openproblem}

\section*{Acknowledgments}

We thank Stefan Tiegel for pointing out a mistake in an earlier draft, leading to a correction in one of the reductions, and Kiril Bangachev, Sam Hopkins, and Tselil Schramm for many comments and suggestions that have improved the paper.
We are also grateful to Kiril Bangachev, Sam Hopkins, Stefan Tiegel, and Vinod Vaikuntanathan for stimulating conversations.

\section{Review of Known Algorithms and Lower Bounds for $\kxor$ and Tensor PCA}\label{sec:lit_review}

Here we overview existing algorithmic and statistical results for the two models.
From our computational equivalence results (Prop.~\ref{prop:gauss_discr_equiv_gen}), the existing algorithms and statistical and computational lower bounds transfer between discrete and Gaussian versions of the problems up to $\poly\log$ factors. 

\subsection{Resolution}\label{subsec:resolution_lit}

The key primitive in our average-case reductions is \emph{resolution}, which multiplies two input $\kxor$ equations (equivalently, adding them over $\mathbb{F}_2$) to form a new $\kxork {k'}$ equation (see \autoref{subsec:resolution_id}). Analyzing resolution has been a major research direction in both the $\kxor$ and broader CSP literature. 

In proof complexity, resolution is studied as a refutation proof system, with foundational lower bounds obtained in \cite{urquhart1987hard,ben1999short,haken1985intractability,chvatal1988many}. Proof systems that reason with linear equations via resolution were later introduced in \cite{raz2008resolution,itsykson2014lower,itsykson2020resolution,efremenko2024lower}. 
Additionally, the analysis of sum-of-squares and Lasserre hierarchies relies on combinatorial properties of $\kxor$ resolution \cite{grigoriev2001linear,schoenebeck2008linear,kothari2017sum}. 
Finally, as an algorithmic primitive, resolution has appeared in \cite{goerdt2001efficient,friedman2005recognizing,guruswami2021algorithms,barak2016noisy}. To the best of our knowledge, however, our work is the first use of resolution as a primitive in a distributionally correct reduction.

\subsection{Information-Theoretic Lower Bounds}\label{subsec:lit_it}

The first information-theoretic lower bounds for Tensor PCA, i.e. proof that the model is indistinguishable from pure noise below the threshold \eqref{eq:it_threshold}, were obtained in \cite{richard2014statistical,  montanari2015limitation} and were generalized to our case of the signal Rademacher prior $x\sim \unif\paren{\set{\pm1}^n}$ in \cite{perry2016statistical}. For sharp analysis of information-theoretically optimal estimation precision, see \cite{jagannath2020statistical}. We remark that the second-moment method of \cite{perry2016statistical} adapts to the case of Bernoulli mask, i.e., to $\ktens$ with intermediate densities $\eps \in [0, 1]$.

\subsection{Algorithms}\label{subsec:lit_algo}

State-of-the art algorithms for the canonical $\kxor$ succeed above the computational threshold \eqref{eq:tradeoff} for detection/refutation \cite{barak2016noisy,allen2015refute,raghavendra2017strongly,guruswami2021algorithms,hsieh2023simple,abascal2021strongly} and recovery \cite{basu2025solving,feldman2015subsampled, guruswami2023efficient}, with several works (starting with \cite{raghavendra2017strongly}) obtaining subexponential algorithms in the conjectured hard regime. 

For Tensor PCA, known algorithms include tensor unfolding  \cite{richard2014statistical}, gradient descent (or Langevin dynamics) and AMP (all suboptimal relative to other efficient algorithms)~\cite{arous2018algorithmic,lesieur2017statistical}, spectral and sum-of-square methods \cite{hopkins2015tensor,hopkins2017power,  hopkins2016fast,bhattiprolu2017sum,moitra2019spectral}, and Kikuchi hierarchy \cite{wein2019kikuchi,kothari2025smooth}. 
The detection algorithm of \cite{wein2019kikuchi} naturally extends to the case of $\eps \in (0,1)$. For recovery in this intermediate regime, we obtain the first known recovery algorithms directly from our reductions to Tensor PCA that succeed for sufficiently large (but constant) tensor order $k$.

\subsection{Computational Lower Bounds}\label{subsec:lit_comp_bounds}

A growing line of work gives evidence towards computational hardness in the conjectured hard regimes by proving that broad restricted classes of algorithms fail to solve the problem. For the canonical $\kxor$, \cite{kothari2017sum,schoenebeck2008linear,grigoriev2001linear} show sum-of-squares (SOS)/Lasserre and \cite{feldman2015complexity} statistical query lower bounds; for Tensor PCA, there exist lower bounds for SOS algorithms \cite{hopkins2015tensor,hopkins2017power,NEURIPS2022_e85454a1}, statistical query model \cite{dudeja2021statistical}, memory-bounded algorithms \cite{dudeja2024statistical}, and low-degree polynomials \cite{kunisky2019notes,kunisky2024tensor}; any method relying on spectral information is ruled out by \cite{montanari2015limitation}.

Additionally, reduction-based hardness from variants of planted clique conjecture (secret leakage and hypergraph PC) was obtained for Tensor PCA in \cite{brennan2020reducibility,luo2022tensor}.

\section{Model Variants and Simple Reductions}

Here we describe preliminary reductions and equivalences used as subroutines in our main results. These include (1) equivalences between variants of $\kxor$ models (e.g., discrete/Gaussian models, different equation sampling procedures, etc), summarized in \autoref{subsec:kxor_variants}, and (2) some simple parameter-changing reductions (\autoref{subsec:overview_simple_red}). The baseline $\kxor$ and $\ktens$ models introduced in \autoref{sec:intro}, \autoref{sec:ideas} are formally defined below. 

\begin{definition}[Family of $\kxor, \ktens$ Models]\label{def:kxor_family}
Let $k,n, m\in \mathbb{N}$ and $\delta \in \sqb{0,1}$ be parameters. Draw the signal vector $x \sim \unif\set{\pm 1}^n$. Let $\binomset n k$ be the set of $k$-sets of distinct indices in $\sqb{n}$ and for a set $\alpha = \set{i_1,\dots,i_k} \in \binomset n k$ write 
$
x_{\alpha} = x_{i_1}x_{i_2} \dots x_{i_k}\,.
$

In $\xor k n m \delta$, we generate a collection of pairs $\calY = \set{\paren{\alpha_j, \Y_{j}}}_{j\in \sqb{m}}$ as follows. For each $j \in \sqb{m}$, sample $\alpha_j \sim \unif\sqb{\binomset n k}$ (with replacement) and set
\begin{equation}\label{eq:def_kxor}
\Y_{j} = 
    x_{\alpha_j}\cdot \noise_{j}\quad \in \set{\pm1}\,,\tag{$\kxor$}
\end{equation}
where $\noise_{j} \sim_\mathrm{i.i.d.} \rad(\delta)$\footnote{We define $\rad(\delta)$ as the $\pm 1$ random variable such that $\E\sqb{\rad\paren{\delta}} = \delta$.} is \emph{entrywise noise}.

In $\gtens k n m \delta$, we replace the multiplicative Rademacher noise $\noise_{j}$ in \eqref{eq:def_kxor} with additive Gaussian $\noiseG_{j} \sim \N(0,1)$: we observe $\calY = \set{\paren{\alpha_j, \Y_{j}}}_{j\in \sqb{m}}$ for $\set{\alpha_j}_{j\in\sqb{m}}$ sampled as above and
\begin{equation}\label{eq:def_ktens}
\Y_{j} = \delta x_{\alpha_j} + \noiseG_{j}\quad\in\R\,,\tag{$\ktens$}
\end{equation}  
where $\noiseG_{j} \sim_\mathrm{i.i.d.} \N(0,1)$.

\end{definition}

\subsection{Variants of $\kxor$ Model}\label{subsec:kxor_variants}

A $\kxor$ instance (Def.~\ref{def:kxor_family}) consists of $m$ independent samples $\paren{\alpha, \Y}$, where $\alpha \in \binomset n k$ and $\Y$ is a noisy observation of $x_{\alpha}$. 
Prop.~\ref{prop:gauss_discr_equiv_gen} is a representative example of the basic equivalences we use as subroutines (here: discrete vs.\ Gaussian observation models).
In this subsection we record several additional variants of $\kxor$ and relate them by simple reductions.

We parametrize a $\kxor$-like model by four choices. By default (i.e., model with no superscript) we use the first option in choices (2)-(4); superscripts indicate deviations from these defaults:

\begin{figure}[h]
  \centering
    
\renewcommand{\arraystretch}{1.15}
\setlength{\tabcolsep}{6pt}
\begin{tabular}{@{}l p{0.72\linewidth}@{}}

\textbf{(1) Noise model} &
\([\checkmark]\) \(\kxor:\ \Y \sim x_\alpha\cdot \rad(\delta)\)\\
&[\textcolor{white}{\checkmark}]\ \(\ktens:\ \Y \sim \delta x_\alpha+\N(0,1)\)\\[4pt]

\textbf{(2) Number of samples} &
\([\checkmark]\) \(\kmodel:\ m\text{ samples}\)\\
&[\textcolor{white}{\checkmark}]\ \(\kmodel^{\mathsf{Pois}}:\ \mpoi\sim\poi(m)\text{ samples}\)\\[4pt]

\textbf{(3) Index space for \(\alpha\)} &
\([\checkmark]\) \(\kmodel:\ \alpha\in\binomset n k\) ($k$-sets of distinct indices)\\
& [\textcolor{white}{\checkmark}]\ \(\kmodel^{\mathsf{FULL}}:\ \alpha\in [n]^k\) ($k$-multisets, index repetitions allowed, Def.~\ref{def:multiset})\\
& [\textcolor{white}{\checkmark}]\ \(\kmodel^{\mathsf{asymm}}:\ \alpha\in [n]^{\underline{k}}\) (ordered $k$-tuples of distinct indices)\\[4pt]

\textbf{(4) Sampling of \(\alpha\)} &
\([\checkmark]\) \(\kmodel:\) with replacement\\
&[\textcolor{white}{\checkmark}]\ \(\kmodel^{\mathsf{wor}}:\) without replacement\\
\end{tabular}

    \caption{Choices in $\kxor$ variants: the defaults are checked and correspond to $\xor k n m \delta$.}
    \label{fig:kxor_variants}
\end{figure}

With this notation, the standard $\xor k n m \delta$ (Def.~\ref{def:kxor_family}) uses the default options in (1)-(4), while $\gtens k n m \delta$ differs only in (1) (Gaussian noise) and keeps the defaults in (2)–(4). 

We use the $(\eps, \defi)$ reparametrization for these models analogously to standard $\kxor$ (see \autoref{sec:ideas}). We focus on the case $\eps\geq 0$ (i.e., $m\geq n^{k/2}$), motivated by our applications; most of the results below can be naturally extended to $\eps < 0$. 

\subsubsection{Equivalence of Models with Different Noise Models: Choice (1)} \label{subsubsec:equiv_noise_model}

Switching between the discrete and Gaussian noise models (choice (1)) results in computationally equivalent models, up to $\poly\log n$ losses. We apply standard scalar reductions between Gaussian and Bernoulli variables (as developed in \cite{ma2015computational,brennan2018reducibility})
that preserve signal level $\delta$ up to $\log n$ factors to our models, entrywise, and obtain in \autoref{subsec:comp_equiv}:
\begin{proposition}[Computational Equivalence of $\kxor^\star$ and $\ktens^\star$]\label{prop:discr_gauss_equivalence_general}
Let $\star\subseteq \set{\mathsf{Pois}, \mathsf{FULL} / \mathsf{asymm}, \mathsf{wor}}$ be some fixed model choices (2)-(4) in Fig.~\ref{fig:kxor_variants}. The corresponding discrete $\xorpstar k n \eps \defi$ model and its Gaussian analog $\gtenspstar k n \eps \defi$ are computationally equivalent: there are $\poly(n)$-time average-case reductions ($\Gaussianize/\Discretize$, Alg.~\ref{alg:gaussianization}, \ref{alg:discretization}) for both detection and recovery in both directions. Moreover, the reductions preserve the secret signal vector $x$ precisely.

\end{proposition}

\subsubsection{Equivalence of Models with Fixed and Poisson Number of Samples: Choice (2)}\label{subsubsec:equiv_sampling_pois}

Options for the number of samples (choice (2)) are also interchangeable for our purposes:
for any $\kxor$-like model in Fig.~\ref{fig:kxor_variants}, modifying the sample count rule in choice (2) yields a computationally equivalent model, up to polylogarithmic losses in parameters.

We prove the following equivalence, that allows us to switch between models with $m$ and $\poi(m'), m' = \Theta(m)$ samples, in \autoref{subsec:poi_sampling}.

\begin{proposition}[Computational Equivalence of $\kmodel$ and $\kmodel^{\mathsf{Pois}}$]\label{prop:sampling_equiv}
    Let $\star \subseteq \set{ \mathsf{FULL}/\mathsf{asymm}, \mathsf{wor} }$ and $\kmodel \in \set{\kxor, \ktens}$ be fixed model choices (1,3,4) in Fig.~\ref{fig:kxor_variants}. Then the fixed-sample $\kmodel^\star_{\eps, \defi}(n)$ and its Poissonized variant $\kmodel^{\star, \Pois}_{\eps, \defi}(n)$ are computationally equivalent: there are $\poly(n)$-time average-case reductions ($\ResampleIID/\ResampleM$, Alg.~\ref{alg:resample_iid}, \ref{alg:resample_M}) for both detection and recovery in both directions. Moreover, the reductions preserve the secret signal vector $x$ precisely.
\end{proposition}

\subsubsection{Models that Allow Index Repetitions and Symmetrization: Choice (3)}\label{subsubsec:full_asymm_kxor_red}

We show in \autoref{subsec:repeated_indexes} that, with a natural choice of measure on the index $k$-multisets (Def.~\ref{def:multiset}), the $\kmodel^{\mathsf{FULL}}$ models are equivalent to a collection of independent lower-order instances $\kmodel, \kmodelk {(k-2)}, \kmodelk {(k-4)},\dots$ that share the same signal vector $x$.

\begin{proposition}[Computational Equivalence of $\kmodel^{\star, \mathsf{FULL}}$ and a Collection of $\kmodel^\star$]\label{prop:comp_equiv_full_collection}
Fix $k\in\mathbb{Z}, \eps\in\sqb{0,1}, \defi\in \R_{+}$ and any choices (1,2,4) $\star \subseteq \set{\Pois, \wor}$ and $\kmodel \in \set{\kxor, \ktens}$ in Fig.~\ref{fig:kxor_variants}. Then,
    $$
    \kmodel_{\eps, \defi}^{\star, \mathsf{FULL}}(n) \quad\text{and}\quad \B\{ \kmodelk {(k-2r)}^{\star}_{\eps_r, \defi} (n) \B\}_{0\leq r \leq \frac{k-1}{2}}\qquad \text{with } \eps_r = \min\set{ \frac{k\eps}{k-2r} , 1}
    $$
    are computationally equivalent: there exist poly-time average-case reductions ($\Decompose/\Compose$, Alg.~\ref{alg:decompose_full}, \ref{alg:compose_full}) both ways.\footnote{Here the second model is a collection of independent $\kmodel$ instances.} These reductions preserve the input signal vector $x\in\set{\pm1}^n$ exactly.

\end{proposition}

Additionally, when sampling is \emph{with replacement} (choice (4)), one can switch between the standard index space $\binomset n k$ (first option in choice (3)) and the space of ordered $k$-tuples (third option in choice (3)), i.e., the symmetric and asymmetric $\kxor$ models become equivalent (\autoref{subsec:symm}).\footnote{We note that the models are equivalent for recovery for any option in choice (4) and the limitation is a reduction from $\kmodel$ to $\kmodel^{\mathsf{asymm}}$ for recovery.} 

\begin{proposition}[Computational Equivalence of $\kmodel^\star$ and $\kmodel^{\star,\mathsf{asymm}}$]\label{prop:symm_asymm_equiv}
Let $\star \subseteq \set{\Pois}$ and $\kmodel \in \set{\kxor, \ktens}$ be any fixed choices (1,2) in Fig.~\ref{fig:kxor_variants} and choice (4) set to sampling with replacement. Then, $\kmodel^{\star}_{\eps, \defi}(n)$ and the corresponding ``asymmetric" variant $\kmodel^{\star, \mathsf{asymm}}_{\eps, \defi}(n)$ are computationally equivalent: there are $\poly(n)$-time average-case reductions ($\DeSymm/\Symm$, Alg.~\ref{alg:desymm}, \ref{alg:symm}) for both detection and recovery in both directions. Moreover, the reductions preserve the secret signal vector $x$ precisely.

\end{proposition}

\subsubsection{Sampling $\alpha$ with or without Replacement: Choice (4)}\label{subsubsec:equiv_w_wor_replacement}

All reductions in this paper are stated for sampling with replacement (the default in choice (4)). With small technical modifications in the reduction algorithms, the same results also hold verbatim for the without-replacement variants, provided both the source and target models use the same sampling rule.

The above claim is not immediate: while known techniques such as cloning (\autoref{subsubsec:cloning}) imply that $\kxor$ variants with entry sampling \emph{with} or \emph{without} replacement are computationally equivalent for recovery, their equivalence for detection in our setting is not addressed, to the best of our knowledge.\footnote{The obstacle to obtaining the detection equivalence is the reduction from $\kmodel^{\star,\wor}$ to $\kmodel^{\star}$: for recovery we can use cloning to obtain entry duplicates, but for detection the null models are mapped incorrectly. This is because our null models (Def.~\ref{def:null_models}) contain a uniformly random signal $w_{\alpha} \sim \unif\sqb{\set{\pm 1}}$ in place of structured $x_\alpha$ and cloning replicates $w_\alpha$.} 
We show in \autoref{subsec:w_replacement}:
\begin{proposition}[Computational Equivalence of $\kmodel^\star$ and $\kmodel^{\star,\mathsf{wor}}$ for Recovery]\label{prop:w_wor_equiv_recovery}
Let $\star \subseteq \set{\Pois, \FULL/\asymm}$ and $\kmodel \in \set{\kxor, \ktens}$ be some fixed model choices (1-3) in Fig.~\ref{fig:kxor_variants}. The two models $\kmodel^\star_{\eps, \defi}(n)$ and $\kmodel^{\star, \wor}_{\eps, \defi}(n)$ are computationally equivalent for recovery: there are $\poly(n)$-time average-case reductions ($\ToTPCA/\ToPoiFromTPCA$, Alg.~\ref{alg:to_pca_sampling}, \ref{alg:to_poi_sampling_from_pca}) for recovery in both directions. Moreover, the reductions preserve the secret signal vector $x$ precisely.
\end{proposition}

Additionally, we show the complete computational equivalence of $\kmodel$ and $\kmodel^{\wor}$ (for both detection and recovery) for $\eps = 1$:
\begin{proposition}[Computational Equivalence of $\kmodel^\star$ and $\kmodel^{\star,\mathsf{wor}}$ for $\eps = 1$]\label{prop:sampling_equiv_eps1}
Let $\star \subseteq \set{\Pois, \FULL/\asymm}$ and $\kmodel \in \set{\kxor, \ktens}$ be some fixed model choices (1-3) in Fig.~\ref{fig:kxor_variants} and let $\defi > 0$. The two models $\kmodel^\star_{\eps=1, \defi}(n)$ and $\kmodel^{\star, \wor}_{\eps=1, \defi}(n)$ are computationally equivalent: there are $\poly(n)$-time average-case reductions ($\ToTPCA/\ToPoiFromTPCA$, Alg.~\ref{alg:to_pca_sampling}, \ref{alg:to_poi_sampling_from_pca}) for both detection and recovery in both directions. Moreover, the reductions preserve the secret signal vector $x$ precisely.
\end{proposition}

Case $\eps = 1$ in Prop.~\ref{prop:sampling_equiv_eps1} is of independent interest, since it captures Tensor PCA model. Recall that in the standard Tensor PCA, denoted as $\ktensorpca_\defi (n)$ (equiv. $\tensorpca k n \delta$ with $\delta = n^{-k/4 - \defi/2}$), there is a secret vector $x\in\set{\pm 1}^n$ and one observes a noisy tensor 
\begin{equation}\label{eq:tpca_full_overview}
\delta \cdot x^{\otimes k} + \noiseG\,,\qquad \delta = n^{-k/4 - \defi/2}\,,
\end{equation}
where $\noiseG$ is a symmetric order-$k$ tensor of $\N(0,1)$. As a special case, Prop.~\ref{prop:sampling_equiv_eps1} relates this model to the other variants we consider.\footnote{To obtain the exact number of samples in $\ktensorpca$, the potentially missing $(1- \poly\log(n)^{-1})$ fraction of entries can be filled with pure noise via the $\DenseRed$ reduction (Alg.~\ref{alg:dense_red}), see \autoref{subsec:entry_adjust}}
\begin{corollary}[of Prop.~\ref{prop:sampling_equiv_eps1}]\label{cor:tensor_pca_wr_equiv}
    The $\ktensorpca_\defi (n)$ (equiv. $\tensorpca k n \delta$ with $\delta = n^{-k/4 - \defi/2}$) model with $\defi > 0$, which observes the noisy symmetric tensor in Eq.~\eqref{eq:tpca_full_overview}, is computationally equivalent to $\kmodel^{\star, \FULL}_{\eps=1, \defi}(n)$ for the choices $\star \subseteq \set{\Pois, \asymm, \wor}$ (and, in particular, $\gtenspfull k n {\eps=1} \defi$). 

    Moreover, the model that observes only distinct-index entries of Eq.~\eqref{eq:tpca_full_overview} is computationally equivalent to $\kmodel^{\star}_{\eps=1, \defi}(n)$ for the choices $\star \subseteq \set{\Pois, \asymm, \wor}$ (and, in particular, $\gtensp k n {\eps=1} \defi$). 
\end{corollary}

\subsection{Simple Reductions}\label{subsec:overview_simple_red}

\subsubsection{Cloning and Splitting Reductions}\label{subsubsec:cloning_splitting_overview}

Here we consider the reductions (cloning and splitting) that map one input instance $\calY \sim \kmodel^{\star}$ to two output instances $\calY\up1, \calY\up2 \sim \kmodel^\star$ sharing the same signal $x$ (\autoref{subsec:cloning_splitting}). Cloning keeps the same index sets as the input but ensures the noise independence across the two outputs. Splitting makes the two outputs fully independent instances (independent indices and independent noise), while preserving the same signal $x$.

\paragraph{Cloning.} Cloning uses the standard Gaussian-cloning identity (see, e.g., \cite{brennan2018reducibility}): given $Y\sim\N(\mu,1)$, generate an independent $Z\sim\N(0,1)$ and set $Y\up 1 = \frac {Y-Z} {\sqrt{2}}$, $ Y\up 2 = \frac{Y+Z}{\sqrt{2}}$. Then $Y\up 1, Y\up2$ are independent and each is distributed as $\N(\mu/\sqrt{2},1)$. Applying this primitive entrywise,\footnote{For $\kmodel=\kxor$ we first reduce to the Gaussian analog via \autoref{subsubsec:equiv_noise_model}, apply the cloning primitive entrywise, then (if needed) map back.} we obtain:

\begin{lemma}[Instance Cloning]\label{lem:instance_clone}

Let $\star \subseteq \set{\Pois, \FULL/\asymm, \wor}$ and $\kmodel \in \set{\kxor, \ktens}$ be some fixed model choices (1-4) in Fig.~\ref{fig:kxor_variants}. There exists a poly-time algorithm ($\XorClone/\kTensClone$, Alg.~\ref{alg:xor_clone}) that, given an input $\calY \sim \kmodel^{\star}(n, m, \delta)$, returns two instances $\calY\up1, \calY\up 2 \sim \kmodel^{\star}(n, m, \Thetat(\delta))$ that have independent entrywise noise, the same index sets, and the same signal vector $x\in\set{\pm1}^n$ as the input $\calY$. 
\end{lemma}
\begin{remark}[Reduction up to $\poly\log n$ Factors]\label{rmk:clone_const}
    In the $\eps, \defi$ parametrization, Lemma~\ref{lem:instance_clone} yields an average-case reduction from $\calY \sim \kmodel^{\star}_{\eps, \defi}(n)$ to $\calY\up1, \calY\up2 \sim \kmodel^{\star}_{\eps, \defi}(n)$ (two same-index copies with independent noise) by Prop.~\ref{prop:map_up_to_constant}.
\end{remark}
\begin{remark}[Creating Several Clones]\label{rmk:clone_several}
    To obtain $A > 2$ clones $\calY\up1, \dots, \calY\up {A}$, we can repeat the cloning procedure recursively $\lceil \log A \rceil$ times, incurring a $\Theta(\sqrt{A})$ loss in $\delta$. In $\eps, \defi$ parametrization this becomes an average-case reduction from $\calY \sim \kmodel^{\star}_{\eps, \defi}(n)$ to $\calY\up1, \dots, \calY\up {A} \sim \kmodel^{\star}_{\eps, \defi'}(n)$ for $\defi' = \defi+ \log_n A$ admitting the same guarantees: entrywise noise independence across the copies, the same index sets and signal vector $x \in \set{\pm 1}^n$ as in $\calY$. 
\end{remark}

\paragraph{Splitting.}
Cloning achieves noise independence but reuses the input index sets. Splitting instead produces two fully independent instances by allocating each observed sample to one of the two outputs.

We state splitting for sampling with replacement. For sampling without replacement our techniques suffice for an average-case reduction for recovery;\footnote{With $\poly\log n$ losses in parameters} with minor modifications this suffices as a subroutine in our main reductions.

\begin{lemma}[Instance Splitting]\label{lem:instance_split}

Let $\star \subseteq \set{\Pois, \FULL/\asymm}$ and $\kmodel \in \set{\kxor, \ktens}$ be some fixed model choices (1-3) in Fig.~\ref{fig:kxor_variants} and choice (4) set to sampling with replacement. There exists a poly-time algorithm ($\kTensSplit$) that, given an input $\calY \sim \kmodel^{\star}(n, m, \delta)$, returns two independent instances $\calY\up1, \calY\up 2 \sim \kmodel^{\star}(n, \Theta(m), \delta)$ that share the same signal vector $x\in\set{\pm1}^n$ as the input $\calY$. 
\end{lemma}
Remarks~\ref{rmk:clone_const} and \ref{rmk:clone_several} similarly apply.

\subsection{Reductions that Change the Number of Samples $m$}\label{subsec:entry_adjust}

In \autoref{subsec:entry_adjust_app} we prove simple reductions that modify the number of samples in $\kmodel^\star(n,m,\delta)$ from $m$ to $m'$ by either (1) injecting pure noise observations, or (2) removing some of the samples. 

\begin{lemma}[Adjusting the Number of Samples $m$]\label{lem:m_adjust}
    Let $\star \subseteq \set{\Pois, \FULL/\asymm, \wor}$ and $\kmodel \in \set{\kxor, \ktens}$ be some fixed model choices (1-4) in Fig.~\ref{fig:kxor_variants} and let $k, n \in \Z^{+}$, $\delta \in \R$ and $m,m' \in\Z^{+}$ be parameters. There exists a poly-time average-case reduction ($\DenseRed/\SparseRed$, Alg.~\ref{alg:dense_red}, \ref{alg:sparse_red}) for both detection and recovery $$\text{from}\quad\kmodel^\star(n,m,\delta)\quad\text{to}\quad\kmodel^\star(n,m',\delta')\,, \qquad\text{with }\delta' = \begin{cases}
        \delta\,, &m' \leq m\\
        \Thetat(\delta\cdot m/m')\,, &m' > m\,. \end{cases}$$
        Moreover, the reduction preserves the secret signal vector $x$ precisely.
\end{lemma}
In the $\eps, \defi$ parametrization this yields:
\begin{corollary}\label{cor:m_adjust}
    Let $\star \subseteq \set{\Pois, \FULL/\asymm, \wor}$ and $\kmodel \in \set{\kxor, \ktens}$ be some fixed model choices (1-4) in Fig.~\ref{fig:kxor_variants} and let $k, n \in \Z^{+}$, $\eps, \eps' \in [0,1]$, and $\defi \in \R$ be parameters. There exists a poly-time average-case reduction ($\DenseRed/\SparseRed$, Alg.~\ref{alg:dense_red}, \ref{alg:sparse_red}) for both detection and recovery 
\begin{align*}&\text{from}\quad\kmodel^\star_{\eps,\defi}(n)\quad\text{to}\quad\kmodel^\star_{\eps',\defi'}(n)\,, \quad \text{where } \defi' = \defi + {k|\eps-\eps'|/2}\,.
    \end{align*}
    Moreover, the reduction preserves the secret signal vector $x$ precisely.
\end{corollary}

\subsection{Reduction Decreasing Order $k$ to $k/a$}\label{subsec:decrease_k_prelim}

In \autoref{subsec:decrease_k_prelim_app} we prove a reduction that maps $\kmodel^\star$ to $\kmodelk {(k/a)}^\star$ for some $a|k$ in the case $\FULL \in \star$. The reduction uses a change of variables, modifying the signal $x$ to $x^{\otimes a} \odot y$ for a known $y\sim\unif\sqb{\set{\pm1}^{n^a}}$. We emphasize that the assumption of the full $\kxor$ variant ($\FULL \in \star$), where index repetitions are allowed, is essential for a correct distributional mapping: otherwise, the resulting index sets do not match the desired distribution.

\begin{lemma}[Reducing Order $k$ to $k/a$]\label{lem:prelim_decrease_k}
    Let $\star \subseteq \set{\Pois, \wor}$ and $\kmodel \in \set{\kxor, \ktens}$ be some fixed model choices (1,2,4) in Fig.~\ref{fig:kxor_variants} and let $k,n \in \Z^{+}, \eps \in [0,1], \defi \in \R$ be parameters.
    For every $a \in \Z^{+}, a | k$, there is an average-case reduction for both detection and recovery ($\ReduceK$, Alg.~\ref{alg:prelim_decrease_k}) $$\text{from}\quad\kmodel^{\star,\FULL}_{\eps,\defi}(n)\quad\text{to}\quad \kmodelk {(k/a)}^{\star,\FULL}_{\eps,\defi/a}(n')\,,\quad n' = n^a\,.$$
    Moreover, Alg.~\ref{alg:prelim_decrease_k} maps an instance with a secret signal vector $x\in \set{\pm 1}^n$ to one with a signal vector $x' = x^{\otimes a}\odot y$ for a known $y\sim \unif\sqb{\set{\pm 1}^{n'}}$. 
\end{lemma}
The reduction in Lemma~\ref{lem:prelim_decrease_k}, in combination with reductions in Prop.~\ref{prop:comp_equiv_full_collection} that decompose an instance of $\kmodel^{\star, \FULL}$ into standard (distinct index) instances of orders $k, k-2, k-4, \dots$, yields the following order-reducing reductions. For convenience, we also refer to the reduction in Cor.~\ref{cor:prelim_decrease_k} as $\ReduceK$ (Alg.~\ref{alg:prelim_decrease_k}).
\begin{corollary}[of Lemma~\ref{lem:prelim_decrease_k}]\label{cor:prelim_decrease_k}
    Let $\star \subseteq \set{\Pois, \wor}$ and $\kmodel \in \set{\kxor, \ktens}$ be some fixed model choices (1,2,4) in Fig.~\ref{fig:kxor_variants} and let $k,n \in \Z^{+}, \eps \in [0,1], \defi \in \R$ be parameters.
    For every $a \in \Z^{+}, a | k$ and $k' \in \set{k/a, k/a-2, k/a-4,\dots} \cap \Z_{>0}$, there is an average-case reduction for both detection and recovery $$\text{from}\quad\kmodel^{\star,\FULL}_{\eps,\defi}(n)\quad\text{to}\quad \kmodelk {k'}^{\star}_{\eps,\defi/a}(n')\,,\quad n' = n^a\,.$$
    Moreover, Alg.~\ref{alg:prelim_decrease_k} maps an instance with a secret signal vector $x\in \set{\pm 1}^n$ to one with a signal vector $x' = x^{\otimes a}\odot y$ for a known $y\sim \unif\sqb{\set{\pm 1}^{n'}}$. 
\end{corollary}

\section{Preliminaries}\label{sec:preliminaries}

\subsection{Notation}\label{sec:notation}

\textbf{Distributions.} Denote $\rad\paren{\delta}$ to be a $\pm 1$ random variable such that $\E \sqb{\rad\paren{\delta}} = \delta$. 

For $M\in \R^{a\times b}$, $\N(M, I_{a\times b})$ is a distribution on $\R^{a\times b}$ with i.i.d. Gaussian entries of mean $M$ and entrywise variance $1$. 

A two-component Gaussian mixture is written as $$
p \N(\mu_1, \sigma_1^2) + (1-p) \N(\mu_2, \sigma_2^2)\,,\quad p \in \paren{0,1}\,.
$$
\textbf{Sets and indices.} For a set $S$, $\binom{S}{k}$ denotes all $k$-tuples of elements in $S$ without repetitions. We write $\sqb{n} = \set{1,\dots,n}$. We write $\set{n}^k$ to denote a set of all $k$-multisets of $n$ and $\sqb{n}^{\underline k}$ to denote a set of all distinct-element $k$-tuples of $n$.

We denote $\calS_n$ to be a set of all permutations $\sqb{n} \to \sqb{n}$. 

For a vector $x \in \R^n$ and an index tuple $\alpha = \paren{i_1,\dots,i_k} \in \binom{\sqb{n}}{k}$, denote $x_{\alpha} = x_{\paren{i_1,\dots,i_k}} = x_{i_1}\dots x_{i_k}$.

Denote $x^{\otimest a}$ to be the vector composed of $x_{\alpha}$ for all $\alpha \in \binom{\sqb{n}}{a}$. This is just the vectorized rank-1 tensor of order $a$ where we keep only entries with distinct indices. 

For a vector $x \in \R^{n}$ and indices $a < b \in [n]$, denote $x_{\sqb{a:b}} \in \R^{b-a+1}$ to be $x_{\sqb{a:b}} = \paren{x_a,\dots,x_b}^\top$. 

For $x, y$ vectors, matrices, or tensors of the same dimension, denote $x\odot y$ to be their entrywise product.\\ 
\textbf{Symmetric tensors.} We say a tensor $\Yb\in \paren{\R^{n}}^{\otimes k}$ is symmetric if it is symmetric up to index permutations, i.e. for any permutation $\sigma \in \calS_k$, $\Yb_{i_1,\dots,i_k} = \Yb_{\sigma(i_1),\dots,\sigma(i_k)}$.

Invariance under index permutations allows us to index symmetric tensors by multisets $\alpha \in \set{n}^k$: we write $\Yb_\alpha$ for $\Yb_{i_1,\dots,i_k}$, where $i_1,\dots,i_k$ is any fixed ordering of $\alpha$. 

For a symmetric tensor $\Yb \in \paren{\R^{n}}^{\otimes k}$, denote the flattening $Y = \mat a {k-a} \paren{\Yb} \in \R^{{\binom{n}{a}} \times \binom{n}{k-a}}$. It is indexed by a pair of sets $(\alpha, \beta) \in \binomset n {a} \times \binomset n {k-a}$ and satisfies 
$$
Y_{\alpha, \beta} = \Yb_{\alpha \cup \beta}\,.
$$
In a special case $a=k$, $Y \in \R^{\binom n a}$ is a vector.
For a matrix $Y \in \R^{{\binom{n}{a}} \times \binom{n}{k-a}}$, denote the reverse operation to be $\Yb = \tenso \paren{Y} \in \paren{\R^{n}}^{\otimes k}$ (where missing entries are set to $0$).\\
\textbf{Polylog factors.} We will sometimes use $\Ot, \Thetat, \widetilde \Omega$ to hide $\poly\log n$ factors, e.g., $\Ot(1)$ may be as large as $\poly\log n$.

\subsection{Concentration Inequalities}\label{subsec:conc}

\begin{proposition}[Concentration of Binomial RV]\label{prop:binomial_conc}
        Let $X\sim\bin(n,p)$, where $p > 0$ and $np \geq n^{\eps}$ for some $\eps > 0$. By Bernstein's inequality,
        $$
        \Pr\sqb{|X - pn| \geq t} \leq 2 \exp\paren{-\frac{t^2/2}{np(1-p) + t/3}}\,,
        $$
        and as a consequence, with probability at least $1 - n^{-c}$,
        $$
        X\in \sqb{pn \pm C\sqrt{pn}\log n} =: \sqb{pn \pm \phi(p,n)}\,,
        $$
        where the constant $C= C(c)$ is an explicit function of $c$. 
    \end{proposition}

\begin{proposition}[Concentration of Poisson RV]\label{prop:pois_conc}
        For $X \sim \poi(\lambda)$, the Chernoff bound argument yields for any $t > 0$,
        $$
        \Pr\sqb{|X - \lambda| \geq t} \leq 2e^{-\frac{t^2}{2(t+\lambda)}}\,.
        $$
        In particular, if $\lambda = n^{\eps}$ for some $\eps > 0$, with probability at least $1 - n^{-c}$,
        $$
        X\in \sqb{\lambda \pm C\sqrt{\lambda}\log n}\,,
        $$
        where the constant $C= C(c)$ is an explicit function of $c$. 
    \end{proposition}

\subsection{Properties of Total Variation}\label{subsec:tv_properties} The analysis of our reductions will make use of the following well-known facts and inequalities concerning total variation distance (see, e.g., \cite{Polyanskiy_Wu_2025}). Recall that the total variation distance between distributions $P$ and $Q$ on a measure space $(\mathcal{X}, \mathcal{B})$ is given by $\tv(P,Q) = \sup_{A\in \mathcal{B} } |P(A) - Q(A)|$.

\begin{fact}[TV Properties] \label{tvfacts}
The distance $\TV$ satisfies the following properties:
\begin{enumerate}
\item (Triangle Inequality) For any distributions $P, Q, S$ on a measurable space $(\mathcal{X}, \mathcal{B})$, it holds that 
$$
\tv (P, Q) \leq \tv(P, S) + \tv(S, Q)\,.
$$
\item (Conditioning on an Event) For any distribution $P$ on a measurable space $(\mathcal{X}, \mathcal{B})$ and event $A \in \mathcal{B}$, it holds that
$$\TV\b( P(\,\cdot\, | A), P \b) = 1 - P(A)\,.$$
\item (Conditioning on a Random Variable, Convexity) For any two pairs of random variables $(X, Y)$ and $(X', Y')$ each taking values in a measurable space $(\mathcal{X}, \mathcal{B})$, it holds that
\begin{align*}\TV\big( \law(X), \law(X') \big) &\le \TV\big( \law(Y), \law(Y') \big)
+ \E_{y \sim Y} \big[ \TV\big( \law(X | Y = y), \law(X' | Y' = y) \big)\big]\,,\end{align*}
where we define $\TV\b( \law(X | Y = y), \law(X' | Y' = y) \b) = 1$ for all $y \not \in \textnormal{supp}(Y')$.
\item (Data Processing Inequality) For any two random variables $X, X'$ on a measurable space $(\mathcal{X}, \mathcal{B})$ and any Markov Kernel $K: \mathcal{X} \to \mathcal{Y}$, it holds that 
$$
\tv\b( \law(K(X)), \law(K(X')) \b) \leq \tv\b(\law(X), \law(X')\b)\,.
$$
\end{enumerate}
\end{fact}

\begin{lemma}[KL-Divergence Between Gaussians]
\label{l:KLGauss}
Let $\mu_1,\mu_2\in \R^n$ and $\Sigma_1,\Sigma_2\in \R^{n\times n}$ be positive definite matrices. Then
$$
\kl(\N(\mu_1,\Sigma_1)\| \N(\mu_2,\Sigma_2))
= \frac12\Big[ 
\log\frac{\det \Sigma_2}{\det \Sigma_1} - n + \mathrm{Tr}(\Sigma_2^{-1}\Sigma_1) +(\mu_2 - \mu_1)^\top \Sigma_2^{-1} (\mu_2 - \mu_1)
\Big]\,.
$$
In particular, by Pinsker's inequality,
\begin{equation}\label{eq:tv_gauss_pinsker}
    \tv(\N(\mu_1,I_n), \N(\mu_2,I_n)) \leq \sqrt{\frac12 \kl(\N(\mu_1,I_n)\| \N(\mu_2,I_n))} = \frac12 \|\mu_2-\mu_1\|_2\,. 
\end{equation}
\end{lemma}

Combining the convexity of total variation (Item 3 in Fact~\ref{tvfacts}) and the bound~\eqref{eq:tv_gauss_pinsker} above, we obtain the following lemma. 

\begin{lemma}[TV Between Gaussians with Shifted Means]
\label{l:GaussianShift}
Let $X, A, B\in \R^d$ be random variables with $X\sim \N(0,I_d)$ independent of $A$ and $B$. Then
$$
\tv\b(\law(X + A + B), \law(X + B)\b) \leq \E \|A\|_2\,.
$$
\end{lemma}

\section{Discrete Resolution}\label{sec:discrete_resolution}

\subsection{Reduction Idea and Lemma Statement} 

In $\xor k n m \delta$ with a secret vector $x \in \set{\pm 1}^n$, we observe a collection $\calY = \set{\paren{\alpha_j, \Y_j}}_{j \in \sqb{m}}$ for $\alpha_j \sim_{i.i.d.} \unif\sqb{\binomset n k}$ and 
\begin{equation*}
    \Y_{j} = 
    x_{\alpha_j}\cdot \noise_{j}\,, 
\end{equation*}
where $\noise_{j} \sim \rad(\delta)$ and we recall that $x_{\alpha} = x_{i_1}\dots x_{i_k}$ for an index set $\alpha = \set{i_1,\dots,i_k} \in \binomset n k$. As described in detail in \autoref{subsec:resolution_id}, the product of two $\kxor$ entries $\Y_{s}, \Y_{t}$ for $s, t\in \sqb{m}$ with $\kstar = |\alpha_s \triangle \alpha_t|$ is distributed as a $\kxork \kstar$ entry with signal level $\delta' = \delta^2$ corresponding to an index set $\alpha_s \triangle \alpha_t$:
\begin{equation}\label{eq:resolution_id7}
    \Y_{s} \Y_{t} = 
    x_{\alpha_s} x_{\alpha_t} \cdot \noise_{s}\noise_{t} = x_{\alpha_s \triangle \alpha_t} \cdot {\paren{\noise_{s}\noise_{t}}} \sim x_{\alpha_s \triangle \alpha_t} \cdot \rad\paren{\delta^2}\,.
\end{equation}
We call the operation of multiplying two $\kxor$ entries to obtain another $\kxork \kstar$ entry \emph{resolution}. The observation in Eq.~\eqref{eq:resolution_id7} suggests two potential na\"ive reductions depending on the number of available pairs $(s,t)$ with $|\alpha_s\triangle\alpha_t| = \kstar$:
\begin{enumerate}
    \item \textbf{Sparse case -- few pairs $(s,t)$:} for all $s,t \in \sqb{m}$ satisfying $|\alpha_s\triangle\alpha_t| = \kstar$, include a pair $\paren{\alpha_s\triangle\alpha_t, \Y_s\Y_t}$ into the output instance $\calZ$:
    \begin{equation}\label{eq:sparse_naiive7}
        \calZ = \set{ \paren{\alpha_s\triangle\alpha_t, \Y_s\Y_t}:\, s,t \in \sqb{m} \text{ and } |\alpha_s\triangle\alpha_t| = \kstar}\tag{Sparse $\EqComb$}\,.
    \end{equation}
    \item \textbf{Dense case -- many pairs $(s,t)$:} if for every index set $\gamma \in \binomset n \kstar$ we expect many $s,t\in\sqb{m}$ satisfying $\alpha_s\triangle\alpha_t = \gamma$, we can collect all of these observations into
        \begin{equation}\label{eq:naiive_resolution7} \Zt_{\gamma}^{obs} = \set{\Y_{s} \Y_{t}: s,t \in [m]\text{ and }\alpha_s \triangle \alpha_t = \gamma}
    \end{equation} 
    and aggregate them via a majority vote, outputting a collection of equations 
    \begin{equation}\label{eq:dense_naiive7}
    \calZ = \set{\paren{\gamma, \maj \b(\Zt_{\gamma}^{obs} \b)}:\, \gamma \in \binomset n \kstar \text{ and }|\Zt_{\gamma}^{obs}| \neq 0}\,.\tag{Dense $\EqComb$}
    \end{equation}
\end{enumerate}

As discussed in \autoref{subsec:resolution_id}, in both cases the output $\calZ$ does not match the $\kxork \kstar$ distribution because of the dependencies in the index sets of observed entries and in the noise distribution. In Lemma~\ref{lem:2path_red}, we build a combination-based reduction for both sparse and dense cases that avoids index set and entry dependence by design. 
Essentially, Algorithm~\ref{alg:two_path} ensures that each input entry $\Y_j$ participates in at most one product pair, so the products in Eqs.~\eqref{eq:sparse_naiive7} and \eqref{eq:dense_naiive7} do not reuse noise across different $\gamma$ (see \autoref{subsec:sparse_resolution} for more intuition). The resulting reduction is the main building block of our densifying reductions in Thm.~\ref{thm:sparse_dense_red} (\autoref{sec:resolution_red}).

\begin{lemma}[Discrete Equation Combination Reduction]\label{lem:2path_red}
    Fix $k,n \in \Z^{+}, \eps \in (-1,1], \defi \in \R$. For any even integer $\kstar$, s.t. 
    $$\kstar \leq k\paren{1-\eps}\,,$$
    there exists an average-case reduction (Def.~\ref{def:avg_case_points})
    for both detection and recovery (Alg.~\ref{alg:two_path}) for $\eps \geq 0$,
    $$\text{from}\quad \xorp k n \eps \defi\quad\text{to}\quad\xorp {\kstar} {n'} {\eps'}{2\defi} \quad\text{with}\quad\eps' = \min\set{\frac{2k}{\kstar}\eps, 1}\,,
    $$
    and for $\eps < 0$ and $k' > 2k |\eps|$,
    $$
    \text{from}\quad\xor k n {m=n^{k(1+\eps)/2}} \delta\quad\text{to}\quad\xor {\kstar} {n'} {m'=\widetilde\Theta((n')^{k'(1+\eps')/2})} {\delta^2}\quad\text{with}\quad \eps' = \frac{2k} {k'} \eps\,.
    $$
    
    Moreover, Alg.~\ref{alg:two_path} maps an instance with a secret signal vector $x\in \set{\pm 1}^n$ to one with the signal vector $x' = x_{[1:n']}$ for $n' = n(1- o(1))$.\footnote{The algorithm can set $x'$ to have any $n'$-coordinate subset of $x$; the algorithm succeeds for any $n' = n\paren{1 - (\poly\log n)^{-1}}$.} 
\end{lemma}

\begin{remark}[Sparse and Dense $\EqComb$]
    In Lemma~\ref{lem:2path_red}, the case $\eps' = \frac{2k} \kstar \eps \leq 1$ 
    corresponds to the \eqref{eq:sparse_naiive7} and $\frac{2k} \kstar \eps > 1$ with $\eps' = 1$ to the \eqref{eq:dense_naiive7}. 
\end{remark}

\begin{remark}[$\kxor$ Sampling Procedure for Alg.~\ref{alg:two_path} and Gaussian Input]\label{rmk:discrete_resolution_alg_sampling_gaussian}
    While Algorithm~\ref{alg:two_path}, as stated, is a reduction from $\kxor^{\Pois}$ to $\kxork{k'}^{\Pois}$ (defined in \autoref{subsec:kxor_variants}), we sometimes use it as a subroutine for $\kxor$ models with the standard sampling procedure, i.e., as a reduction from $\kxor$ to $\kxork {k'}$. It is implicitly assumed that first a Poissonization (Alg.~\ref{alg:resample_iid}) reduction is applied, then Alg.~\ref{alg:two_path}, and then the reduction back to standard sampling model (Alg.~\ref{alg:resample_M}).
    
    Similarly if Alg.~\ref{alg:two_path} is applied as a Gaussian model reduction for $\ktens$ or $\ktens^{\mathsf{Pois}}$: this implicitly assumes the insertion of the appropriate Gaussianization/Discretization (\autoref{subsec:gaussianize}, \autoref{subsec:discretize}) or/and Resampling (\autoref{subsec:poi_sampling}) reductions.
\end{remark}

\begin{algorithm}\SetAlgoLined\SetAlgoVlined\DontPrintSemicolon
    \KwIn{
    $\calY = \set{ \b( \alpha_j, \Y_{j}\b)}_{j\in\sqb{\mpoi}}$, where $\forall j \in\sqb{\mpoi}$, $\b(\alpha_j, \Y_{j}\b) \in \binomset n k \times \set{\pm1}$, $k'\in \Z$, $\kappa \in (0,1)$.
    }
    \KwOut{$\widetilde\calZ = \set{ \b( \gamma_j, \Zt_{j}\b)}_{j\in\sqb{\mpoi'}}$, where $\forall j \in\sqb{\mpoi'}$, $\b(\gamma_j, \Zt_{j}\b) \in \binomset {\kappa n} {k'} \times \set{\pm1}$.}

    \BlankLine
    \tcp{Prune the input equations; only consider index sets $\alpha_j$ that contain exactly $k-k'/2$ coordinates in $\sqb{\kappa n + 1: n}$, denote them as the "cancellation set" $\Ccal(j)$}
    $\Ss = \set{j:\, j\in\sqb{\mpoi}\text{ and } |\alpha_j \cap \sqb{\kappa n}| = k'/2}$

    For each $ j \in \Ss$, $\Ccal(j) \gets \alpha_j \cap \sqb{\kappa n + 1: n}$

    \BlankLine
    \tcp{Aggregate entries with the same cancellation set $\beta$}
    For each $ \beta \in \binomset {\kappa n+1: n} {k-k'/2}$, let $$\calN(\beta) = \set{j \in \sqb{\Ss}:\, \Ccal(j) = \beta}$$

    \BlankLine
    \tcp{Form disjoint combination pairs and record them}
    $\calZ \gets \emptyset$

    \ForAll{$\beta \text{ s.t. } 
    \b|\calN(\beta)\b| \geq 2$}{
        $s, t \gets \unif\sqb{\binom{\calN(\beta)}2}$

        \If{$|\alpha_s\triangle \alpha_t| = k'$}{
            $\calZ \gets \calZ \cup \set{\paren{\alpha_s\triangle\alpha_t, \Y_s\Y_t}}$
        }
        }

    \BlankLine
    \If{$2k\eps \leq k'$}{\tcp{Sparse case -- few observations in $ \calZ$}

        $\mpoit \gets \poi(B_Z)$, $\mpoi' = \min\big\{\mpoit, |\calZ|\big\}$ \tcp{See pf. of Lem.~\ref{lem:2path_red} for defn of $B_Z$}

        $\widetilde \calZ \gets \text{ random }\mpoi'\text{-subset of }\calZ$ 

        Return $\widetilde\calZ$
        
    }\Else{ \tcp{Dense case -- many observations in $\calZ$}
    
        $\forall \gamma \in \binomset {\kappa n} {k'}$, $\Zt^{obs}_{\gamma} \gets \big\{ \Zt:\, \paren{\gamma, \Zt} \in \calZ \big\}$

        $N \gets B_{obs}$ \ \tcp{See proof of Lemma~\ref{lem:2path_red} for defn of $B_{obs}$}

        $\widetilde \calZ \gets \emptyset$ \ \tcp{Initialize the answer to be empty}
        
        \For{$\gamma \in \binomset {\kappa n} {k'}$}{

            $\Zt^{obs}_{\gamma} \gets \text{ random }\min\set{N, |\Zt^{obs}_{\gamma}|}\text{-subset of }\Zt^{obs}_{\gamma}$ 
        
            $\widetilde \calZ \gets \widetilde \calZ \cup \set{\paren{\gamma, \maj\b( \Zt^{obs}_{\gamma} \b)}}$ 
        }

        Return $\ToPoiFromTPCA\b(\widetilde\calZ\b)$ \ \tcp{Alg.~\ref{alg:to_poi_sampling_from_pca}}
    }

    \caption{$\TwoPath\paren{\calY, k', \kappa = 1 - \log^{-1}n}$}
\label{alg:two_path}
\end{algorithm}

We prove Lemma~\ref{lem:2path_red} in \S\ref{subsec:proof_lemma_2path}. In \S\ref{subsec:extended_2path} we generalize Lemma~\ref{lem:2path_red} to accept as an input two instances $k_1\mathsf{\text{-}XOR}$ and $k_2\mathsf{\text{-}XOR}$ of different order, which will be used as a subroutine to our reductions that decrease the tensor order $k$ (\S\ref{subsec:ideas_decrease_k} and \S\ref{sec:decrease_k}). In \S\ref{subsec:resolution_lwe}, we extend Lemma~\ref{lem:2path_red} to the $k$-sparse LWE model, which considers linear equations over $\mathbb{Z}_q$ as opposed to $\mathbb{F}_2$ as in $\kxor$.

\subsection{Proof of Lemma~\ref{lem:2path_red}}\label{subsec:proof_lemma_2path}

    We prove Lemma~\ref{lem:2path_red} for $\kxor^{\mathsf{Pois}}$ (defined in \autoref{subsec:kxor_variants}). 
    The analogous result for the standard $\kxor$ immediately follows by the computational equivalence of $\kxor$ and $\kxor^{\mathsf{Pois}}$: we show in Prop.~\ref{prop:sampling_equiv} that there exist average-case reductions for both detection and recovery from $\kxor^{\Pois}$ to $\kxor$ and vice versa, preserving the number of samples $m$ up to constant factors and noise level $\delta$ unchanged. 

    Algorithm~\ref{alg:two_path} maps $\kxor$ in dimension $n$ with a signal vector $x \in \set{\pm 1}^n$ to $\kxork {k'}$ in dimension $n' = \kappa n$ with a signal vector $x' = x_{[1:n']}$. We present the proof for the planted case and the same argument applies verbatim to the null case $\delta = 0$ (see Def.~\ref{def:null_models} for the null model definition).
    
    Let $\calY = \set{ \b( \alpha_j, \Y_{j}\b)}_{j\in\sqb{\mpoi}}, \mpoi \sim \poi(m)$ be a (planted) instance of $\xorpoi k n m \delta$ with $m = n^{k(1+\eps)/2}$ and $\delta = \begin{cases}
        \Thetat(1), &\eps < 0\\
        n^{-k\eps/4 - \defi/2}, &\eps \geq 0\,.
    \end{cases}$. In particular, $\alpha_j \sim_{i.i.d.} \unif\sqb{\binomset n k}$ and
    $$
    \Y_{j} = 
        x_{\alpha_j}\cdot \noise_{j}\,,\text{ where }\noise_{j} \sim_{i.i.d.} \rad(\delta)\,.
    $$
    We will show that Alg.~\ref{alg:two_path} outputs $\widetilde\calZ = \TwoPath(\calY,k',\kappa = 1 - \log^{-1}n)$, satisfying 
    \begin{equation}\label{eq:final_tv_resolution}  
        \law (\widetilde\calZ) \approx_{\tv = o(1)} \xorpoi {k'} {n'} {m'} {\delta'}\,,
    \end{equation} 
    where $n' = \kappa n$, $m' = \Thetat\paren{ (n')^{k'(1+\eps')/2}}$. For $\eps < 0$, $\delta' = \delta^2$ and $\eps' = \frac{2k}{k'}\eps$; for $\eps \geq 0$, $\delta' = \Thetat\paren{(n')^{-k'\eps'/4 - \defi}}$ and $\eps' = \min\set{\frac{2k}{\kstar}\eps, 1}$.
    Since it is sufficient to map parameters up to $\poly\log n$ factors by Prop.~\ref{prop:map_up_to_constant}, Eq.~\eqref{eq:final_tv_resolution} shows that in the case of $\eps \geq 0$, Alg.~\ref{alg:two_path} is an average-case reduction to $\xorppoi {k'} {n'} {\eps'} {2\defi}$.

\paragraph{Step 1: Distribution of elements in $\calZ$.}  
    We show that $\calZ = \set{\paren{\gamma_j, \Zt_j}}_{j\in \sqb{|\calZ|}}$, which is the set constructed in Lines 4-8 of Alg.~\ref{alg:two_path}, contains elements $\paren{\gamma_j, \Zt_j} \in \binomset {\kappa n}{k'} \times \set{\pm1}$ for which
    \begin{equation}\label{eq:z_dist}
    \gamma_j \sim \unif\sqb{\binomset {\kappa n} {k'}}\quad\text{and}\quad \Zt_j \sim x_{\gamma_j} \cdot \rad\paren{\delta^2}
    \end{equation}
    are independent across $j \in \sqb{|\calZ|}$.
    
    For every $\gamma \in \binomset {\kappa n} {k'}$ and every recorded observation $\Y_s\Y_t \in \calZ$ with $\alpha_s \triangle \alpha_t = \gamma $ in Lines 4-8
    \begin{equation}\label{eq:two_path}
        \Y_{s} \Y_{t} = x_{\alpha_s} x_{\alpha_t} \noise_{s} \noise_{t} \sim x_{\alpha_s\triangle \alpha_t} \rad\paren{\delta^2}\,.
    \end{equation}
    By construction, every entry $\Y_j, j \in\sqb{\mpoi}$ appears in exactly one $\N(\beta)$ for some $\beta \in \binomset {\kappa n+1:n} {k-k'/2}$, and since every $\N(\beta)$ yields at most one recorded observation into $\calZ$, every entry $\Y_j$ appears in at most one equation \eqref{eq:two_path}. Therefore, the entry-wise noise terms $\rad\paren{\delta^2}$ in Eq.~\eqref{eq:two_path} are independent, so the distribution of values $\Zt_j$ in the elements $\paren{\gamma_j, \Zt_j} \in \calZ$ matches the one in Eq.~\eqref{eq:z_dist}.

    Moreover, for any $\beta$ with $|\calN\paren{\beta}| \geq 2$, if the chosen $s,t \sim \unif\sqb{\binom {\N(\beta)} {2}}$ satisfy $|\alpha_s\triangle\alpha_t| = k'$, the resulting index set $\alpha_s\triangle\alpha_t$ is uniformly random on $\binomset {\kappa n } {k'}$, chosen independently for different $\beta$. Hence, the distribution of index sets $\gamma_j$ in the elements $\paren{\gamma_j, \Zt_j} \in \calZ$ matches the one in Eq.~\eqref{eq:z_dist}. 

\paragraph{Step 2: Bound on $|\calZ|$.} We next show a high probability bound on $|\calZ|$, which is the number of observations recorded in Lines 4-8 of Alg.~\ref{alg:two_path}. We show that, with high probability, 
$$
    |\calZ| \geq C_Z \cdot n^{k\eps+k'/2} \eqcolon 2B_Z\,,
    $$
    where $C_Z = \Thetat (1)$ hides constants depending on $k,a$ and $\poly\log n$ factors and denote the corresponding event 
    $$
    S_{|\calZ|} = \set{ |\calZ| \geq 2B_Z }\,,\qquad \Pr\sqb{S_{|\calZ|}^c} = o(1)\,.
    $$   

    From Poisson splitting, for each $\beta \in \binomset {\kappa n+1: n} {k-k'/2}$, given that $\kappa = 1 - \log^{-1} n$,
    $$
    |\calN\paren{\beta}| \sim \poi\paren{m\Thetat(n^{-(k-k'/2)})}\quad\text{and}\quad m\Thetat(n^{-(k-k'/2)}) = \Thetat\paren{n^{k(\eps-1)/2 +k'/2}} = \Ot(1)\,,
    $$
    where the last equality uses $k'\leq k(1-\eps)$. Consequently, 
    $$\Pr\sqb{|\calN\paren{\beta}| \geq 2} = \Thetat \big(n^{k(\eps-1)+k'}\big)\,.$$ 
    Given $|\calN\paren{\beta}| \geq 2$ and $s,t \sim \unif\sqb{\binom {\N(\beta)} {2}}$, with at least constant probability the remaining indices are disjoint and satisfy $|\alpha_s\triangle\alpha_t| = k'$. Then, 
    $$
    |\calZ| \sim \bin\paren{ \binom{n(1-\kappa)}{k-k'/2}, \Thetat \paren{n^{k(\eps-1)+k'}}}\,,
    $$
    and by Binomial concentration in Prop.~\ref{prop:binomial_conc}, with high probability, 
    $$
    |\calZ| \geq  \Thetat (1) \cdot n^{k-k'/2} \cdot n^{k(\eps-1)+k'} = C_Z \cdot n^{k\eps+k'/2} \eqcolon 2B_Z\,,
    $$
    for $C_Z = \Thetat(1)$.

    To summarize Steps 1 and 2, $\calZ$ contains $|\calZ| \geq 2B_Z$ elements $\paren{\gamma_j, \Zt_j}, j \in\sqb{|\calZ|}$ for $\gamma_j \sim_{i.i.d.} \unif\sqb{\binomset {\kappa n}{k'}}$ and independent $\Zt_j \sim x_{\gamma_j} \cdot \rad(\delta^2)$. 

\paragraph{Step 3: Distribution of the output set $\widetilde\calZ$.} Here we conclude the proof by showing 
 $$
    \tv\paren{\law(\widetilde \calZ), \xorpoi {k'} {n'} {m'} {\delta'}} = o_n(1)\,,
    $$
where $\widetilde \calZ$ is the output of the Alg.~\ref{alg:two_path}. This is handled differently in the sparse and dense regimes. 

    \textbf{Case 1: $2k\eps \leq k'$ (sparse).} In Alg.~\ref{alg:two_path}, the output $\widetilde \calZ$ is set to be a random $\mpoi'$ subset of $\widetilde\calZ$, where $\mpoi' = \min\{\mpoit, |\calZ|\}$ for a generated $\mpoit \sim \poi(B_Z)$. By the Poisson concentration in Prop.~\ref{prop:pois_conc}, with high probability, given $S_{|\calZ|}$, $$\mpoit \leq 2B_M \leq |\calZ|
    \,,$$ in which case $\mpoi' =
    \mpoit \sim \poi(B_Z)$. Denote this event $S_\mathrm{Pois} = \set{ \mpoit \leq |\calZ| }$, $\Pr\sqb{S_{Pois}^c} = o(1)$. Conditioned on $S_{Pois}$ and $S_{|\calZ|}$, $\widetilde\calZ$ contains $\mpoi' \sim \poi(B_Z)$ $\kxork {k'}$ entries with uniformly random index sets and signal level $\delta' = \delta^2$, and therefore, for $n' = \kappa n$ and $m' = B_Z$, 
    $$
    \law(\widetilde \calZ | S_{Pois} \cap S_{|\calZ|}) = \xorpoi {k'} {n'} {m'} {\delta'}\,.
    $$
    Then, conditioning on a high probability event in total variation (Fact~\ref{tvfacts}), 
    \begin{align*}
        \tv\paren{\law(\widetilde\calZ), \xorpoi {k'} {n'} {m'} {\delta'}} = \tv\paren{\law(\widetilde\calZ), \law(\widetilde\calZ| S_{Pois} \cap S_{|\calZ|})} = 1 - \Pr\sqb{S_{Pois} \cap S_{|\calZ|}} = o_n(1)\,.
    \end{align*}
    It remains to express $m' = B_Z$ and $\delta' = \delta^2$ in terms of $\eps' = \frac{2k}{k'}\eps$ and $n' = \kappa n$:
    $$
    m' = C_Z/2 \cdot n^{k\eps + k'/2} = \Thetat(1) \cdot \paren{n'}^{k'(\eps' + 1)/2}
    $$
    for $\eps' = \frac{2k}{k'}\eps$ and, if $\eps \geq 0$, $\delta' = n^{-k \eps / 2 - \defi} = \Thetat(1) \cdot \paren{n'}^{- k' \eps' / 4 - \defi}$. This concludes the proof of Eq.~\eqref{eq:final_tv_resolution}.
    
    \textbf{Case 2: $2k\eps > k'$ (dense).} In this case $|\calZ| \geq 2B_Z \gg n^{k'}$, i.e., the number of recorder observations in $\calZ$ significantly exceeds the number of distinct index sets in $\kxork {k'}$. By binomial concentration (Prop.~\ref{prop:binomial_conc}), for all $\gamma \in \binomset {\kappa n}{k'}$, simultaneously with high probability,
    \begin{equation}\label{eq:obs_size}
        \b| \Zt_{\gamma}^{obs} \b| \geq \Thetat(1) \cdot B_Z n^{-k'} = C_{obs} \cdot n^{k\eps - k'/2} \eqcolon B_{obs}
    \end{equation}
    for $C_{obs} = \Thetat(1)$ hiding constants depending on $k, a$ and $\poly\log n$ factors. Denote the corresponding event $S_{obs} = \set{\forall \gamma, \b| \Zt_{\gamma}^{obs} \b| \geq B_{obs}}$, $\Pr[S_{obs}^c] = o(1)$.
    Given $S_{obs}$, $$\widetilde\calZ = \set{\paren{\gamma, \maj\b( \Zt^{obs}_{\gamma} \b)}}_{\gamma \in \binomset {\kappa n}{k'}}\,,$$
    where each majority is taken over exactly $B_{obs}$ observations. In this case, 
    $$
    \E \,\maj\b( \Zt^{obs}_{\gamma} \b) = x_{\gamma} \cdot \Thetat\big( \delta^2 \cdot n^{\frac12\paren{k\eps - k'/2}}\big) = x_{\gamma} \cdot \Thetat\big( \paren{\kappa n}^{-k'/4 - \defi }\big) \eqcolon x_\gamma \cdot \delta'\,,
    $$
    which shows that (see Fig.~\ref{fig:kxor_variants} for the definition of $\kxor^\wor$) with $M = \binom {\kappa n}{k'}$,
    $$
    \law(\widetilde \calZ | S_{|\calZ|} \cap S_{obs}) = \kxork {k'}^{\wor}(\kappa n, M, \delta')\,.
    $$
    Then, conditioning on a high probability event in total variation (Fact~\ref{tvfacts}),
    $$
    \tv\paren{\law (\widetilde\calZ), \kxork {k'}^{\wor}(\kappa n, M, \delta')} = \tv\paren{\law (\widetilde\calZ), \law(\widetilde \calZ | S_{|\calZ|} \cap S_{obs})} = 1 - \Pr\sqb{S_{|\calZ|} \cap S_{obs}} = o_n(1)\,,
    $$
    and therefore, Alg.~\ref{alg:two_path} is an average-case reduction to $\kxork {k'}^{\wor}_{\eps'=1, 2\defi}(\kappa n)$.
    Finally, by Prop.~\ref{prop:sampling_equiv_eps1}, $\ToPoiFromTPCA$ (Alg.~\ref{alg:to_poi_sampling_from_pca}) is an average-case reduction from $\kxork {k'}^{\wor}_{\eps'=1, 2\defi}(\kappa n)$ to $\kxork {k'}^{\Pois}_{\eps'=1, 2\defi}(\kappa n)$, concluding the proof. 
    
\subsection{Resolution for Two Input Tensors of Different Orders}\label{subsec:extended_2path}

As described in \autoref{subsec:ideas_decrease_k}, the equation combination identity of Eq.~\eqref{eq:resolution_id7} can be extended to accept as an input two instances $\calY\up1 \sim \xorp {k_1} n {\eps_1} {\defi_1}$ and $\calY\up2 \sim \xorp {k_2} n {\eps_2} {\defi_2}$ that share a secret vector $x\in\set{\pm1}^n$. Our use case for such a reduction is the order-reducing reduction (\autoref{sec:decrease_k}), for which it suffices to consider $\eps_1 = \eps_2 = 0$. Lemma~\ref{lem:2path_red_ext_discr} is a generalization of Lemma~\ref{lem:2path_red} for this case.\footnote{Similar result can be established for general $\eps_1,\eps_2\in [0,1)$ using analogous techniques.}
\begin{lemma}[Extended Discrete Resolution Reduction]\label{lem:2path_red_ext_discr}

    Given independent instances $\calY\up1 \sim \xorp {k_1} n {\eps=0} {\defi_1}$ and $\calY\up2 \sim \xorp {k_2} n {\eps=0} {\defi_2}$ that share a secret vector $x\in\set{\pm1}^n$, Alg.~\ref{alg:two_path_two_inputs} outputs 
    $\calZ = \TwoPathTwo\paren{\calY\up1, \calY\up2, \kstar}$, where $$\tv\paren{\calZ, \xor \kstar {n'} {m'} {\delta'}} \to 0\quad \text{as }n\to\infty\,,$$ where $m' = \Thetat(n^{\kstar/2})$ and $\delta' = \Thetat(n^{-\defi_1/2 - \defi_2/2})$, if $(k_1+k_2-\kstar)/2$ is an integer and $\kstar$ satisfies
    $$
    \kstar\in \sqb{|k_1-k_2|, \min\set{k_1,k_2}}\,.
    $$
    Moreover, Alg.~\ref{alg:two_path_two_inputs} maps two instances with the same secret signal vector $x\in \set{\pm 1}^n$ to one with the signal vector $x' = x_{[1:n']}$ for $n' = n(1- o(1))$.
\end{lemma}

Remark~\ref{rmk:discrete_resolution_alg_sampling_gaussian} similarly applies.

\begin{algorithm}\SetAlgoLined\SetAlgoVlined\DontPrintSemicolon
    \KwIn{
    $\calY\up a = \set{ \b( \alpha_j\up a, \Y_{j}\up a\b)}_{j\in\sqb{\mpoi\up a}}$ for $a\in\set{1,2}$, where $\forall j \in\sqb{\mpoi\up a}$, $\b(\alpha_j\up a, \Y_{j}\up a\b) \in \binomset n {k_a} \times \set{\pm1}$, $k'\in \Z$, $\kappa \in (0,1)$.
    }
    \KwOut{$\widetilde\calZ = \set{ \b( \gamma_j, \Zt_{j}\b)}_{j\in\sqb{\mpoi'}}$, where $\forall j \in\sqb{\mpoi'}$, $\b(\gamma_j, \Zt_{j}\b) \in \binomset {\kappa n} {k'} \times \set{\pm1}$.}

    \BlankLine
    \tcp{Prune the input equations; only consider index sets $\alpha_j\up a$ that contain exactly $(k_1+k_2-k')/2$ coordinates in $\sqb{\kappa n + 1: n}$, denote them as the "cancellation set" $\Ccal(j)$}
    $\Ss\up a = \set{j:\, j\in\sqb{\mpoi\up a}\text{ and } |\alpha_j\up a \cap \sqb{\kappa n+1:n}| = (k_1+k_2-k')/2}$ for $a \in \set{1,2}$

    For $a\in \set{1,2}$ and each $ j \in \Ss\up a$, $\Ccal\up a(j) \gets \alpha_j\up a \cap \sqb{\kappa n + 1: n}$

    \BlankLine
    \tcp{Aggregate entries with the same cancellation set $\beta$}
    For each $ \beta \in \binomset {\kappa n+1: n} {(k_1+k_2-k')/2}$, let $$\calN\up a(\beta) = \set{j \in \sqb{\Ss\up a}:\, \Ccal\up a(j) = \beta}\text{ for }a\in \set{1,2}$$

    \BlankLine
    \tcp{Form disjoint combination pairs and record them}
    $\calZ \gets \emptyset$

    \ForAll{$\beta \text{ s.t. } 
    \b|\calN\up 1(\beta)\b| \geq 1$ and $\b|\calN\up 2(\beta)\b| \geq 1$}{
        $s \gets \unif\sqb{\calN\up 1(\beta)}$, $t \gets \unif\sqb{\calN\up 2(\beta)}$

        \If{$|\alpha_s\up 1\triangle \alpha_t\up 2| = k'$}{
            $\calZ \gets \calZ \cup \set{\paren{\alpha_s\up 1\triangle\alpha_t\up 2, \Y_s\up 1\Y_t\up 2}}$
        }
        }

    \BlankLine
    $\mpoit \gets \poi(B_Z)$, $\mpoi' = \min\big\{\mpoit, |\calZ|\big\}$ \tcp{See pf. of Lem.~\ref{lem:2path_red_ext_discr} for defn of $B_Z$}

    $\widetilde \calZ \gets \text{ random }\mpoi'\text{-subset of }\calZ$ 

    Return $\widetilde\calZ$

    \caption{$\TwoPathTwo\paren{\calY\up1, \calY\up2, k', \kappa = 1 - \log^{-1}n}$}
\label{alg:two_path_two_inputs}
\end{algorithm}

\begin{proof}[Proof of Lemma~\ref{lem:2path_red_ext_discr}.]

The reduction Alg.~\ref{alg:two_path_two_inputs} is a generalization of the discrete resolution reduction (Alg.~\ref{alg:two_path}) and its analysis for both planted and null distributions closely follows that of Alg.~\ref{alg:two_path} in Lemma~\ref{lem:2path_red}. The only differences are:
\begin{enumerate}
    \item If the two input instances had noise levels $\delta_1, \delta_2$, the output will have i.i.d. $\rad\paren{\delta_1\delta_2}$ entry-wise noise -- this yields the $\defi' = \defi_1 + \defi_2$ deficiency in the output instance.
    \item The cancellation vector is now $(k_1 + k_2 - k')/2$ sparse, which is guaranteed to be $\leq \min\set{k_1, k_2}$ by the condition $k' \geq |k_1 - k_2|$.
    \item In Step 2, we similarly prove $|\calZ| \geq 2 B_Z$ with high probability, but now for $B_Z = \Thetat\paren{n^{k'/2}}$. In particular, we have 
    $$
    \b|\N(\beta)\up 1\b|\sim \poi\paren{\Thetat\paren{n^{k_1/2} \cdot n^{-(k_1+k_2-k')/2}} = \Thetat\paren{n^{k'/2 - k_2/2}}}\,, 
    $$
    and an analogous condition for $\N(\beta)\up 2$.
    We have $k' \leq \min\set{k_1, k_2}$, in which case $|\calZ| \sim \bin\paren{\binom {n(1-\kappa)} {(k_1+k_2-k')/2}, \Thetat\paren{n^{k' - k_1/2 - k_2/2}}}$. By binomial concentration, with high probability, 
    $$|\calZ| \geq C_Z \cdot n^{k'/2} = 2B_Z\,.$$
    for $C_Z  = \Thetat(1)$. Note that $B_Z$ is the desired number of output equations up to $\poly\log n$ factors. 
    \item We only consider the sparse case for the regime of interest $\eps = 0$. \qedhere
\end{enumerate}

\end{proof}

\section{Discrete Resolution for $k$-Sparse $\lwe$.}\label{subsec:resolution_lwe}

In $\kxor$, via the natural $\mathbb{F}_2 \leftrightarrow \set{\pm1}$ map, the entries are evidently noisy linear equations over $\mathbb{F}_2$. In this alternative view, in $\kxor$ we have a signal $x \in \mathbb{F}_2^n$ and observe $m$ equations over $\mathbb{F}_2$ of the form
\begin{equation}\label{eq:entry_F2_s7}
    x_{i_{1}}+x_{i_{2}}+\dots + x_{i_{k}} + e = \Y \quad \paren{\mathsf{mod}\ 2}\,,
\end{equation}
where $e$ is i.i.d. noise and $\alpha = \set{i_{1},\dots,i_{k}}$ is a uniformly random index set from $\binomset n k$. Our $\EqComb$ step adds (or, equivalently, subtracts) two equations whose index sets overlap, thereby canceling the variables that appear twice. 

Noisy linear equations over larger rings $\Z_q$ are central in both cryptography and average-case complexity. The $k$-sparse LWE problem (for constant $k$), i.e., the analog of \eqref{eq:entry_F2_s7} over $\Z_q$, was first formalized in \cite{jain2024systematic}: one observes $m$ samples of the form
\begin{equation}\label{eq:entry_Fq_s7}
    a_{1}x_{i_{1}}+a_{2}x_{i_{2}}+\dots + a_{k}x_{i_{k}} + e = \Y\,,
\end{equation}
where $x\in \Z_q^n$ is the secret, $\alpha = \set{i_{1},\dots,i_{k}}$ is a uniformly random index set, $e \sim \chi \in \Z_q$ is i.i.d. noise, and the coefficients $a_{s} \sim_{\text{i.i.d.}} \unif\paren{\Z_q\setminus 0 }$. 

We extend our discrete resolution reduction (Lem.~\ref{lem:2path_red}) to $k$-sparse $\lwe$ in \autoref{subsec:resolution_red_lwe}. We present a general resolution-based reduction for $k$-sparse $\lwe$ in Lemma~\ref{lem:2path_red_lwe} that succeeds whenever the entrywise noise distribution $\chi$ is \emph{resolution-stable} (Def.~\ref{def:resolution_stable}). A noise distribution $\chi$ is cobination-stable if it is closed under subtraction, i.e., subtracting two equations of the form \eqref{eq:entry_Fq_s7} yields an equation of $k'$-sparse $\lwe$ for $k' = 2k - 2a$ for $ a = (\text{\# of canceled indices})$, possibly with a change in the noise parameter.  

In \autoref{subsec:discr_res_lwe_uniform}, \autoref{subsec:discr_res_lwe_gauss}, and \autoref{subsec:discr_res_lwe_bound}, we apply the general resolution-based reduction for three common $\lwe$ noise models:
\begin{itemize}
    \item \textbf{Uniform $\mathbb{Z}_q$ Noise:} $e \sim \chi_{\delta}$ if $$e = \begin{cases}
        0, &\text{w. prob. }\delta\\
        \unif\paren{\Z_q }, &\text{w. prob. }1-\delta\,.
    \end{cases}$$
    \item \textbf{Discrete Gaussian Noise:} 
    $$
    \chi_s(e) \propto \sum_{i\in \Z} e^{-\pi (e + iq)^2/s^2}\,.
    $$
    \item \textbf{Bounded $\mathbb{Z}_q$ Noise:} $e \sim \chi_{l}$,
    where $\chi_l$ is a given measure on $\sqb{-l, l}$ for a parameter $l \leq q$.
\end{itemize}

Finally, for the three noise models above we present a simple detection algorithm for $k$-sparse $\lwe$ based on equation collision in \autoref{subsec:collision_detection}; combined with our reduction this yields new detection algorithms for $k$-sparse $\lwe$. 
We do not attempt an exhaustive exploration of resolution for $k$-sparse $\lwe$. For instance, one can relax the resolution requirement that coefficients cancel exactly on the support overlap; we leave this and other directions for future work. More generally, the computational complexity of $k$-sparse $\lwe$ remains open: so far, no plausible conjectures have been proposed regarding existence of polynomial-time algorithms and how this depends on number of samples. 

We now formally define the $k$-sparse $\lwe$ ($\klwe$) model as a $\Z_q$ analog of $\kxor$. Here we will call a vector $c \in \Z_q^n$ \emph{$k$-sparse} if it has exactly $k$ nonzero coordinates, i.e., $|\supp (c)| = k$. For simplicity we assume $q$ is prime, which allows better parameters tradeoffs in the case of the uniform $\Z_q$ noise in Lemmas~\ref{lem:collision_detection} and \ref{lem:2path_red_lwe_uniform}. For general odd $q$, the uniform noise detection thresholds incur a factor of $\sqrt{q-1}$\footnote{The condition in Lemma~\ref{lem:collision_detection} becomes $m\delta^2 = \widetilde\Omega\paren{n^{k/2}(q-1)^{(k-1)/2}}$ with $n^{k/2} (q-1)^{k/2} \lesssim n \lesssim n^k (q-1)^k$.}, while the resolution reduction for uniform noise no longer improves over a general reduction; all other results (Lemmas~\ref{lem:2path_red_lwe}, \ref{lem:collision_detection}, \ref{lem:2path_red_lwe_uniform}, \ref{lem:2path_red_lwe_gauss}, and Cor.~\ref{cor:lwe_alg_bounded}) hold without modification. 

\begin{definition}[$k$-Sparse $\lwe$ ($\klwe$)]\label{def:k_sparse_lwe}
Let $k, n, m \in \Z^{+}$, $q$ prime, and $\chi$ be a distribution on $\Z_q$. In the (planted) $\lweg k n m \chi$, there is a signal vector $x \gets \unif\paren{\Z_q^n}$ and we observe a collection of equations $\calY = \set{\paren{c_j, \Y_j}}_{j \in \sqb{m}}$ over $\Z_q$, where
\begin{equation}\label{eq:def_lwe}
\Y_{j} = \langle c_j, x\rangle + e_{j}\,,\tag{$\klwe$}
\end{equation}
where $e_{j} \sim_{i.i.d.} \chi$ is \emph{entrywise noise} and the coefficient vector $c_j \in \Z_q^n$ has uniformly random $k$-sparse support $$\alpha_j \eqcolon \supp(c_j) = \set{i \in \sqb{n}:\, \paren{c_j}_i \neq 0} \sim_{i.i.d.} \unif\sqb{\binomset n k}$$
and satisfies $\paren{c_{j}}_i = \begin{cases}
    \unif\sqb{(\Z_q\setminus 0)}, & i \in \alpha_j\\
    0, &\text{otherwise}\,.
\end{cases}$

In the null distribution of $k$-sparse $\lwe$ ($\lwegnull k n m$), for all $j\in\sqb{m}$, $c_j$ is generated as above, and $$\Y_{j} \sim_{i.i.d.} \unif\paren{\Z_q}\,.$$

\end{definition}

\begin{remark}
    In $\klwe$ we think of $n$ as the main dimensionality parameter, $k$ and $\log_n m$ as constants, and $q \leq n^{O(1)}$. 
\end{remark}

\begin{remark}[Equation Sampling Procedures]\label{rmk:lwe_sampling} Similarly to $\kxor$ and $\ktens$ (see \autoref{subsec:kxor_variants} for $\kxor$ variants overview), the model above with exactly $m$ equations is computationally equivalent to the one with $\mpoi \sim \poi(m)$ equations, which we denote $\lwegpoi k n m \chi$. In particular, there exist average-case reductions in both directions that preserve $m$ up to constant factors (see Alg.~\ref{alg:resample_iid}, \ref{alg:resample_M}). There also exists a trivial reduction from the model above to the one that samples $k$-sparse coefficient vectors $c_j$ independently \emph{without replacement}, but we are not aware of a reverse reduction that would succeed in all parameter regimes (see discussion in \autoref{subsubsec:equiv_w_wor_replacement}). 

Similarly to $\kxor$, our reduction results for $k$-sparse $\lwe$ (Def.~\ref{def:k_sparse_lwe}) hold without modification for the sampling-without-replacement model (where both the input and the output instance have the same sampling procedure).\footnote{The algorithms require small technical modifications to match the distribution in the different sampling model.} Our computational threshold heuristic, however, only exists for the sampling-with-replacement model.
\end{remark}

\subsection{General Resolution Reduction for $k$-sparse $\lwe$}\label{subsec:resolution_red_lwe}

\begin{definition}[Resolution-Stable Noise Models for $\lwe$]\label{def:resolution_stable}
    Let $\chi_\theta$ be a noise distribution over $\Z_q$ parametrized by $\theta \in \Theta$. We say $\chi_\theta$ is $(f,g)$\emph{-resolution-stable} for deterministic functions $f : \Theta \to \Theta$ and $g : \Theta \times \Z^{+} \to \Theta$ if 
    \begin{enumerate}
        \item For $e_1, e_2 \sim \chi_\theta$ independent, 
        $$e = e_1 - e_2 \mod q \sim \chi_{\theta'}$$ for $\theta' = f(\theta)$.
        \item Given a secret $s \in \Z_q$ and $N$ samples $s_1,\dots,s_N \in \Z_q$:
        $$
        s_i = s + e_i\,, \quad e_i \sim \chi_\theta\,, 
        $$
        there is a poly-time algorithm $\mathcal{A}_N: \Z_q^N \to \Z_q$ that for every $N$, outputs $\hat{s} = \mathcal{A}_N(s_1,\dots,s_N)$ such that 
        $$
        \tv\paren{\hat s, s + \chi_{\theta'}} = n^{-C}\,,
        $$
        for an arbitrarily large constant $C$ and $\theta' = g(\theta, N)$.
    \end{enumerate}
\end{definition}
\begin{remark}
    Note that the second condition in Def.~\ref{def:resolution_stable} can be trivially satisfied with $\calA_N\paren{s_1,\dots,s_N} = s_1$ and $g\paren{\theta, N} = \theta$. We aim for $g$ that ``decrease noise" with larger $N$.
\end{remark}

\begin{lemma}[General Discrete Resolution for $\lwe$]\label{lem:2path_red_lwe}
    Fix $k,n,m \in \Z^{+}, \theta\in\Theta$ and let $\chi_\theta$ be a $(f,g)$-resolution-stable distribution over $\Z_q$ (Def.~\ref{def:resolution_stable}). For any even integer $\kstar\leq 2k$, there exists an average-case reduction\footnote{Average-case reductions for $\kxor$ are formally defined in Def.~\ref{def:avg_case_points}; the definition for $\klwe$ is analogous: reduction algorithm outputs an instance $\calZ$ such that $\tv\paren{\law\paren{\calZ}, \lweg {\kstar} {n'} {m'} {\chi_{\theta'}}} = o_n(1)$, which has algorithmic and hardness implications analogous to Lemma~\ref{lem:red_implications_gen}.} for both detection and recovery (Alg.~\ref{alg:two_path_lwe}) $$\text{from}\quad\lweg k n m {\chi_\theta}\quad\text{to}\quad\lweg {\kstar} {n'} {m'} {\chi_{\theta'}}\,,$$ where $n' = n(1-o(1))$ and parameters $m', \theta'$ depend on $\theta$ and the input number of samples $m$ as follows. Denote $\neff = n(q-1)$ and assume $\eps = \log_{\neff} m$ is a fixed constant;
    \begin{itemize}
        \item If $\eps \leq \min\set{k-\frac{k'}2, \frac k2 + \frac{k'}4}$, 
        $$ m' = \Thetat\paren{m^2 \neff^{-(k-k'/2)}} \quad\text{and}\quad \theta' = f(\theta)\,;$$
        \item If $ \eps \in \Big( \frac k 2 + \frac{k'}4, k-\frac{k'}2\Big]$,
        $$
        m' = \Thetat\paren{\neff^{k'}} \quad\text{and}\quad \theta' = g\paren{f(\theta), \Thetat\paren{m^2 \neff^{-k - k'/2}}}\,;
        $$
        \item If $\eps \in \Big(k-\frac{k'}2,k'\Big]$,
        $$
        m' = \Thetat\paren{m}\quad\text{and}\quad \theta' = f(\theta)\,;
        $$
        \item If $\eps > \max\set{k-\frac{k'}2, k'}$,
        $$
        m' = \Thetat\paren{\neff^{k'}} \quad\text{and}\quad \theta' = g\paren{f(\theta), \Thetat\paren{m \neff^{-k'}}}\,.
        $$
    \end{itemize}
    Moreover, Alg.~\ref{alg:two_path_lwe} maps an instance with a secret signal vector $x\in \Z_q^n$ to one with the signal vector $x' = x_{\sqb{1:n'}}$.
\end{lemma}

Similarly to the $\DiscrEqComb$ in Lemma~\ref{lem:2path_red}, Lemma~\ref{lem:2path_red_lwe} splits into regimes depending on the number of available resolution pairs. In the regimes where each $\klwek {k'}$ equation is observed only $O(1)$ times (bullets 1 and 3), the reduction outputs these observations yielding $\theta' = f(\theta)$. In the regimes where for each $\klwek {k'}$ equation many copies are observed (bullets 2 and 4), the reduction additionally applies the ``aggregation" procedure $\calA_N$ from Def.~\ref{def:resolution_stable}, yielding $\theta' = g(f(\theta), N)$ for an appropriate $N$.

\begin{algorithm}\SetAlgoLined\SetAlgoVlined\DontPrintSemicolon
    \KwIn{
    $\calY = \set{ \b(c_j, \Y_{j}\b)}_{j\in\sqb{\mpoi}}$, where $\forall j \in\sqb{\mpoi}$, $\b(c_j, \Y_{j}\b) \in \Z_q^n \times \Z_q$, $k'\in \Z^+$, $\calA = \set{\calA_N: \Z_q^N \to \Z_q}$, $\kappa \in (0,1)$.
    }
    \KwOut{$\widetilde\calZ = \b\{ \b(c'_j, \Zt_{j}\b)\b\}_{j\in\sqb{\mpoi'}}$, where $\forall j \in\sqb{\mpoi'}$, $\b(c'_j, \Zt_{j}\b) \in \Z_q^{\kappa n} \times \Z_q$.}

    \BlankLine
    \tcp{Prune the input equations; only consider those with exactly $k'/2$ nonzero coefficients in $\sqb{\kappa n}$}
    $\Ss = \set{j:\, j\in\sqb{\mpoi}\text{ and } |\supp(c_j) \cap \sqb{\kappa n}| = k'/2}$

    \tcp{Denote the restriction of $c_j$ on $\sqb{\kappa n + 1: n}$ as a cancellation vector $\Ccal(j)$}
    For each $ j \in \Ss$, 
    $$\Ccal(j) \gets \paren{c_j}_{\sqb{\kappa n + 1: n}}$$

    \BlankLine
    \tcp{Aggregate entries with same $(k-k'/2)$-sparse cancellation vector $\beta$}
    For each $(k-k'/2)$-sparse $ \beta \in \Z_q^{(1-\kappa)n}$, let $\calN(\beta) = \set{j \in \sqb{\Ss}:\, \Ccal(j) = \beta}$

    \BlankLine
    \tcp{Form disjoint resolution pairs and record them}
    $\calZ \gets \emptyset$

    \ForAll{$(k-k'/2)$-sparse $\beta \in \Z_q^{(1-\kappa)n}$}{
        \While{$\b|\calN(\beta)\b| \geq 2$}{
            $s, t \gets \unif\sqb{\binom{\calN(\beta)}2}$
    
            \If{$|\supp \paren{c_s - c_t}| = k'$}{
                $\calZ \gets \calZ \cup \set{\paren{(c_s -c_t)_{[1:\kappa n]}, \Y_s - \Y_t}}$
            }

            $\calN(\beta) \gets \calN(\beta)\setminus \set{s, t}$
        }
        }

    \BlankLine
    \If{ $\log_{n(q-1)} B_Z \leq k'$ }{\tcp{Sparse case -- few observations in $ \calZ$}

        $\mpoit \gets \poi(B_Z/2)$, $\mpoi' = \min\big\{\mpoit, |\calZ|\big\}$ \ \tcp{See pf. of Lem.~\ref{lem:2path_red_lwe} for defn of $B_Z$}

        $\widetilde \calZ \gets \text{ random }\mpoi'\text{-subset of }\calZ$ 
        
    }\Else{ \tcp{Dense case -- many observations in $\calZ$}

        For each $k'$-sparse $c' \in \Z_q^{\kappa n}$,
        $$\Zt^{obs}_{c'} \gets \big\{ \Zt:\, \paren{c', \Zt} \in \calZ \big\}$$

        $N \gets B_{obs}$, $\widetilde \calZ \gets \emptyset$ \ \tcp{See proof of Lemma~\ref{lem:2path_red_lwe} for defn of $B_{obs}$}

        
        \ForAll{$k'$-sparse $c' \in \Z_q^{\kappa n}$}{

            $m_{c'} \sim \poi\paren{1}$ 

            \For{$i \in \sqb{m_{c'}}$}{
                $\Zt^{obs}_{c'}(i) \gets \text{ random }\min\set{N, |\Zt^{obs}_{c'}|}\text{-subset of }\Zt^{obs}_{c'}$ 
        
                $\widetilde \calZ \gets \widetilde \calZ \cup \set{\paren{c', \calA\b( \Zt^{obs}_{c'}(i) \b)}}$ 

                $\Zt^{obs}_{c'} \gets \Zt^{obs}_{c'} \setminus \Zt^{obs}_{c'}(i)$
                
            }
             
        }

    }

    Return $\widetilde \calZ$.

    \caption{$\DiscrResolutionLWE\paren{\calY, k', \mathcal{A}, \kappa = 1 - \log^{-1}n}$}
\label{alg:two_path_lwe}
\end{algorithm}

\begin{proof}[Proof of Lemma~\ref{lem:2path_red_lwe}]
    The reduction is Alg.~\ref{alg:two_path_lwe}, which is a direct analogue of the $\DiscrEqComb$ reduction in Alg.~\ref{alg:two_path} for $\kxor$. Here we highlight the main proof steps and differences from the proof of Lemma~\ref{lem:2path_red}.

    As mentioned in Remark~\ref{rmk:lwe_sampling}, Alg.~\ref{alg:resample_iid} and \ref{alg:resample_M} are average-case reductions from $\lweg k n m {\chi_{\theta}}$ to $\lwegpoi k n {m'} {\chi_{\theta}}$ and vice versa, preserving the number of samples $m$ up to constants. Therefore, it is sufficient to show that Alg.~\ref{alg:two_path_lwe} is an average-case reduction for both detection and recovery 
    $$\text{from}\quad\lwegpoi k n m {\chi_{\theta}}\quad\text{to}\quad\lwegpoi {\kstar} {n'} {m'} {\chi_{\theta'}}\,,$$ 
    where $n'=\kappa n, m', \theta'$ satisfy the conditions of Lemma~\ref{lem:2path_red_lwe}.
    We present the proof for the planted case and the same argument applies to the null case. We show that for an input $\calY = \set{ \b(c_j, \Y_{j}\b)}_{j\in\sqb{\mpoi}} \sim \lwegpoi k n m {\chi_{\theta}}, \mpoi \sim \poi(m)$, the output of Alg.~\ref{alg:two_path_lwe} $\widetilde\calZ = \DiscrResolutionLWE(\calY,k',\kappa = 1 - \log^{-1}n, \calA)$ satisfies 
    \begin{equation}\label{eq:final_tv_resolution_lwe}  
        \law (\widetilde\calZ) \approx_{\tv = o(1)} \lwegpoi {\kstar} {n'} {m'} {\chi_{\theta'}}\,.
    \end{equation} 
    Here, $\calA$ is the ``aggregation" algorithm for $\chi_{\theta}$ from the definition of resolution-stable distribution Def.~\ref{def:resolution_stable}. In this proof, we denote $\neff = n(q-1)$: in $\klwe$ we have $\Theta\paren{\neff^k}$ number of possible coefficient vectors.

\paragraph{Step 1: Distribution of elements in $\calZ$.} 
    We show that $\calZ = \big\{(c'_j, \Zt_j)\big\}_{j\in \sqb{|\calZ|}}$, which is the set constructed in Lines 4-8 of Alg.~\ref{alg:two_path_lwe}, contains elements $(c'_j, \Zt_j) \in \Z_q^{\kappa n} \times \Z_q$ distributed as i.i.d. entries of $\klwek {k'}$ with noise distribution $\chi_{f(\theta)}$, where recall that we assumed $\chi$ is $(f,g)$-resolution-stable. 
    
    The proof is analogous to Step 1 in Lemma~\ref{lem:2path_red}. The recorded value $\Y_s - \Y_t$ has noise distribution $\chi_{f(\theta)}$, which follows from the definition of resolution-stable noise distribution (Def.~\ref{def:resolution_stable}). Moreover, since the coefficient vectors $c_s, c_t$ are uniformly random $k$-sparse and satisfy $\supp\paren{c_s - c_t} = \supp \paren{\paren{c_s - c_t}_{[1:\kappa n]}} = k'$, the resulting coefficient vectors $c'_j = \paren{c_s - c_t}_{[1:\kappa n]} \in\Z_q^{\kappa n}$ are uniformly random $k'$-sparse: $$(c'_j)_i = \begin{cases}
        \unif\sqb{(\Z_q\setminus 0)}, &i\in \supp(c_j')\\
        0, &\text{otherwise}\,.
    \end{cases}$$
    Similarly to Lemma~\ref{lem:2path_red}, since the input equations are not reused in the resolution pairs, the resulting entries $(c'_j, \Zt_j)$ are jointly independent. 

\paragraph{Step 2: Bound on $|\calZ|$.} We prove that with high probability the number of recorded equations $|\calZ|$ satisfies the lower bound
$$
    |\calZ| \geq B_Z = \begin{cases}
        C_Z \cdot m^2 \neff^{-(k-k'/2)}, & \log_{\neff}  m \leq k - k'/2 \\
        C_Z \cdot m, &\text{otherwise}\,,
    \end{cases}
    $$
    where $C_Z = \Thetat (1)$ hides constants depending on $k$ and $\poly\log n$ factors; denote the corresponding event 
    $$
    S_{|\calZ|} = \set{ |\calZ| \geq B_Z }\,,\qquad \Pr\sqb{S_{|\calZ|}^c} = o(1)\,.
    $$   
    In particular, from Poisson splitting, for each $(k-k'/2)$-sparse $ \beta \in \Z_q^{(1-\kappa)n}$ (the cancellation vector), we have 
    $$
    |\calN\paren{\beta}| \sim \poi\paren{\Theta\paren{m \neff^{-(k-k'/2)}}}\,.
    $$
    In the case $\log_{\neff}  m \leq k - k'/2$, each cancellation vector $\beta$ has $O(1)$ equations in its set $\N(\beta)$. Then, each set produces a resolution pair with probability $\Thetat\paren{m^2 \cdot \neff^{-2(k-k'/2)}}$, so
    $$
    |\calZ| \sim \bin\paren{ \binom{n(1-\kappa)}{k-k'/2} \cdot (q-1)^{k-k'/2}, \Thetat\paren{m^2 \cdot \neff^{-2(k-k'/2)}}}\,,
    $$
    and by Binomial concentration in Prop.~\ref{prop:binomial_conc}, with high probability, 
    $$
    |\calZ| \geq  C_Z \cdot m^2 \neff^{-(k-k'/2)} \eqcolon B_Z\,,
    $$
    for $C_Z = \Thetat(1)$.

    Otherwise, by Poisson concentration in Prop.~\ref{prop:pois_conc}, with high probability,  simultaneously for all $\beta$, $|\N(\beta)| \geq \Thetat\paren{\E |\N(\beta)|} = \Thetat\paren{m \neff^{-(k-k'/2)}}$. In this case, each $\N(\beta)$ yields $\Theta\paren{|\N(\beta)|}$ resolution pairs, so
    $$
    |\calZ| \geq C_Z \cdot m \eqcolon B_Z\,,
    $$
    for $C_Z = \Thetat(1)$.

\paragraph{Step 3: Distribution of the output set $\widetilde\calZ$.} Here we conclude the proof by showing 
 $$
    \tv\paren{\law(\widetilde \calZ), \lwegpoi {k'} {n'} {m'} {\chi_{\theta'}}} = o_n(1)\,,
    $$
where $\widetilde \calZ$ is the output of the Alg.~\ref{alg:two_path}. This is handled differently in the \emph{sparse regime}, when the number of observed $\klwek {k'}$ equations $B_Z$ is smaller than the number of all  $k'$-sparse coefficient vectors $\Thetat\paren{\neff^{k'}}$, and the \emph{dense regime}, when $B_Z$ is much larger than $\Thetat\paren{\neff^{k'}}$ and the aggregation step is required.

    \textbf{Case 1: $\log_{\neff} B_Z \leq {k'}$ (sparse).} The remaining step of Alg.~\ref{alg:two_path_lwe} in Lines 11-13 chooses a random $\widetilde\mpoi \sim \poi(B_Z/2)$ subset of $\calZ$ if $|\calZ| \geq \widetilde\mpoi$. Given $S_{|\calZ|}$, i.e., $|\calZ| \geq B_Z$, this happens with high probability by Poisson concentration (Prop.~\ref{prop:pois_conc}). Denoting this event $S_{\poi}$, we observe that 
    $$
    \law( \widetilde\calZ | S_{\poi} \cap S_{|\calZ|} ) = \lwegpoi {k'} {n'} {m'} {\chi_{\theta'}}\,,
    $$
    for $\theta' = f(\theta)$ and $m' = B_Z / 2$. By conditioning on the high probability event in total variation (Fact~\ref{tvfacts}), we obtain 
    $$\tv\paren{ \widetilde\calZ, \lwegpoi {k'} {n'} {m'} {\chi_{\theta'}} } \leq 1 - \Pr\sqb{S_{\poi}} = o(1)\,.$$
    The condition $\log_{\neff} B_Z \leq {k'}$ corresponds to bullets 1 and 3 in the Lemma~\ref{lem:2path_red_lwe} statement. 
    
    \textbf{Case 2: $\log_{\neff} B_Z > {k'}$ (dense).} We introduce a slight modification to the argument in Lemma~\ref{lem:2path_red} by skipping the intermediate reduction to $k'\text{-}\mathsf{tPCA}$ and map directly to $\lwegpoi {k'} {n'} {m'} {\chi_{\theta'}}$. Let $C_{\poi} = \Thetat(1)$ be such that 
    $$
    \Pr[\poi(1) \leq C_{\poi}] \geq 1 - n^{-K}\,,
    $$
    for some large fixed constant $K$. The remaining steps of Alg.~\ref{alg:two_path_lwe} in Lines 17-22 produce exactly $m_{c'}$ output equations with $k'$-sparse coefficient vector $c'$, and by above, simultaneously with high probability for all $c'$, $m_{c'} \leq C_\poi$. Moreover, since $\calZ$ contains uniformly random $k'$-sparse equations, given $S_{|\calZ|}$, $|\calZ| \geq B_Z$ and for all $k'$-sparse $c'\in\Z_q^{\kappa n}$, simultaneously with high probability,
    \begin{equation}\label{eq:obs_size_lwe}
        \b| \Zt_{c'}^{obs} \b| \geq \Thetat\paren{B_Z \neff^{-k'}} = C_{Pois}\cdot C_{obs} \cdot B_Z \neff^{-k'} \eqcolon C_{Pois} \cdot B_{obs}
    \end{equation}
    for $C_{obs} = \Thetat(1)$ hiding constants depending on $k$ and $\poly\log n$ factors. Denote the event above as $$S_{obs} = \set{\forall k'\text{-sparse }c'\in\Z_q^{\kappa n}\,,  \b| \Zt_{c'}^{obs} \b| \geq C_{Pois} \cdot B_{obs}}\,.$$
    Given $S_{|\calZ|}, S_{obs}$, the output $\widetilde\calZ$ contains $m_{c'} \sim \poi(1)$ equations for every $k'$-sparse coefficient vector $c'$ with a noise distribution $\chi_{\theta'}$, where 
    $$
    \theta' = g \paren{f(\theta), B_{obs}}\,.
    $$
    Then,
    $$
    \law( \widetilde\calZ | S_{|\calZ|} \cap S_{obs}) = \lwegpoi {k'} {n'} {m'} {\chi_{\theta'}}\,,
    $$
    where $n' = \kappa n$ and $m' = \binom {\kappa n} {k'} (q-1)^{k'} = \Thetat\paren{\neff^{k'}}$. By conditioning on the high probability event in total variation (Fact~\ref{tvfacts}), we obtain 
    $$\tv\paren{ \widetilde\calZ, \lwegpoi {k'} {n'} {m'} {\chi_{\theta'}} } \leq 1 - \Pr\sqb{S_{|\calZ|} \cap S_{obs}} = o(1)\,.$$
    The condition $\log_{\neff} B_Z > {k'}$ corresponds to bullets 2 and 4 in Lemma~\ref{lem:2path_red_lwe}.
\end{proof}

\subsection{Collision-Based Detection Algorithm for $k$-sparse $\lwe$.}\label{subsec:collision_detection}

\begin{algorithm}\SetAlgoLined\SetAlgoVlined\DontPrintSemicolon
    \KwIn{
    $\calY = \set{ \b(c_j, \Y_{j}\b)}_{j\in\sqb{m}}$, where $\forall j \in\sqb{m}$, $\b(c_j, \Y_{j}\b) \in \Z_q^n \times \Z_q$, $T\in \R$.
    }
    \KwOut{$\mathsf{ANS} \in \set{\mathsf{PLANTED}, \mathsf{NULL}}$.}

    \BlankLine
    \tcp{Aggregate entries with the same index set and coefficient vector}
    For all $k$-sparse $c \in \Z_q^n$, $$ \N\paren{c} \gets \set{ j \in \sqb{m}:\, c = c_j } $$

    \BlankLine
    \tcp{Examine the Collisions}
    $C_{1}, C_{0} \gets 0$

    \ForAll{$k$-sparse $c \in \Z_q^n$}{
        \While{$|\N\paren{c}| \geq 2$}{
            select $s\neq t$ from $\N\paren{c}$
            
            $a \gets \mathds{1}\set{\Y_{s} = \Y_{t}}$ 
            
            $C_{a} \gets C_a + 1$

            $\N\paren{c} \gets \N\paren{c} \setminus \set{s, t}$
        }
    }

    \BlankLine
    Return $\mathsf{PLANTED}$ \textbf{if} $\frac {C_1}{C_0 + C_1} \geq T$ \textbf{else} $\mathsf{NULL}$
    
    \caption{$\CollisionDetection\paren{\calY, T}$}
\label{alg:collision}
\end{algorithm}

The collision-based detection algorithm (Alg.~\ref{alg:collision}) accepts $m$ equations $\calY = \set{\paren{c_j, \Y_j}}_{j\in\sqb{m}}$ and finds disjoint ``collisions", i.e., pairs $s,t \in\sqb{m}$ satisfying $$c_s = c_t\,.$$ 
The algorithm then computes the portion of such pairs $s,t$ for which $\Y_s = \Y_t$ (as it would have been in a noiseless setting) and returns $\mathsf{PLANTED}$ if the portion is sufficiently large and $\mathsf{NULL}$ otherwise. 

\begin{lemma}[$\CollisionDetection$ Guarantees]\label{lem:collision_detection}
    Let $\calY$ be drawn from one of the two distributions:
    \begin{align*}
        &H_1: \quad\calY \sim \lweg k n m \chi\\
        &H_0: \quad\calY \sim \lwegnull k n m\,.
    \end{align*}
    Then there exists a threshold $T = T(k,n,m,\chi)$ such that $\mathsf{ANS} = \CollisionDetection\paren{\calY, T}$ (Alg.~\ref{alg:collision}) can be computed in $O(\poly (n,q))$ time and satisfies 
    $$
    \Pr_{\calY\sim H_0}\sqb{\mathsf{ANS} = \mathsf{PLANTED}} + \Pr_{\calY\sim H_1}\sqb{\mathsf{ANS} = \mathsf{NULL}} \to 0\quad \text{as }n\to \infty
    $$
    if the following conditions are met:
    \begin{itemize}
        \item \textbf{Uniform $\Z_q$ Noise $\chi_\delta$ for $\delta = O(1/q)$:} $n^{k/2}(q-1)^{k/2} \lesssim m \lesssim n^{k}(q-1)^{k-1}$ and
        $$
            m \delta^2 = \widetilde \Omega\paren{n^{k/2} (q-1)^{k/2 - 1}}\,;
            $$
        
        \item \textbf{Discrete Gaussian Noise $\chi_s$ for $s \leq q n^{-\vareps}, \vareps>0$:}
            $$
            n^{k/2}(q-1)^{k/2} \cdot \max\set{\sqrt{s}, 1} \lesssim m \lesssim n^k (q-1)^k \,;
            $$
        \item \textbf{Bounded $\Z_q$ Noise:} for $\chi_l$ supported on $\sqb{-l,l}$, $n^{k/2} (q-1)^{k/2}\paren{\|\chi_l\|^2_2 - 1/q}^{-1/2} \lesssim  m \lesssim n^k (q-1)^k$, and
        $$
        m \paren{\|\chi_l\|^2_2 - 1/q} = \widetilde\Omega\paren{n^{k/2}(q-1)^{(k-1)/2}}\,,
        $$
        where $\|\chi_l\|^2_2 = \sum_{e \in \sqb{-l, l}} \chi_l(e)^2$.
    \end{itemize}
    In the above $\widetilde\Omega$ hides constants in $k$ and $\poly\log n$ factors.
\end{lemma}

We now provide proof sketches for the correctness of Alg.~\ref{alg:collision} in the three noise settings. 

\paragraph{Uniform $\Z_q$ Noise.}
In the uniform  $\Z_q$ noise $e \sim \chi_{\delta}$ if $$e = \begin{cases}
        0, &\text{w. prob. }\delta\\
        \unif\paren{\Z_q }, &\text{w. prob. }1-\delta\,,
    \end{cases}$$
where $\delta = O(1/q)$.
Denote $\mu_{c}$ to be the law of the observed equation \eqref{eq:def_lwe} with a $k$-sparse coefficient vector $c$. Then, multiplying $\Y_{j} \sim \mu_{c_j}$ by a nonzero $s \in \Z_q \setminus \set{0}$ preserves the noise distribution and yields
\begin{align*}
s \Y_{j} = \la sc_j, x \ra + s e_{j} \sim \mu_{s c_j}\,.
\end{align*}
Then, under the uniform noise model, we can assume that for all $j\in\sqb{m}$ we have $\paren{c_j}_{\min\supp\paren{c_j}} = 1$ (i.e., we can rescale $\Y_{j}$ by $\paren{c_j}_{\min\supp\paren{c_j}}^{-1}$). Then, each of the $m$ equations has a coefficient vector drawn uniformly out of $ \Thetat\paren{n^k (q-1)^{k-1}}$ options. A direct calculation then gives that, if $m = \Ot\paren{n^k (q-1)^{k-1}}$, with high probability,  
$$
C_1 + C_0 = \widetilde \Theta\paren{ \frac{m^2}{n^k (q-1)^{k-1}} } \eqcolon N\,.
$$
For our algorithm to succeed, we require $N/q \gg 1$, translating into $m = \widetilde\Omega\paren{n^{k/2}(q-1)^{k/2}}$.
Under the two hypothesis distributions we have
\begin{align*}
    \mathsf{PLANTED}: \quad& C_1 \sim \bin\paren{N, \delta^2(1-1/q) + 1/q}\\
    \mathsf{NULL}: \quad& C_1 \sim \bin\paren{N, 1/q}\,.
\end{align*}
From the binomial concentration (Prop.~\ref{prop:binomial_conc}), if we set the threshold $T = 1/q + \widetilde\Theta\paren{\sqrt{1/(Nq)}}$, Alg.~\ref{alg:collision} succeeds at distinguishing $\mathsf{PLANTED}$ vs $\mathsf{NULL}$ with high probability whenever $N \delta^2 = \widetilde\Omega\paren{ \sqrt{N/q}}$, or, equivalently, 
\begin{equation}\label{eq:thr_unif}
m \delta^2 = \widetilde \Omega\paren{n^{k/2} (q-1)^{k/2 - 1}}\,.
\end{equation}
For $q=2$, which is the case of $\kxor$, and any $q = O\paren{\poly\log n}$, Eq.~\eqref{eq:thr_unif} matches the $\kxor$ computational threshold \eqref{eq:tradeoff} up to $\poly\log n$ factors.

\paragraph{Discrete Gaussian Noise.}
In the discrete Gaussian noise, 
$$
    \chi_s(e) \propto \sum_{i\in\Z} e^{-\pi (e + iq)^2/s^2}\,, \quad e\in\Z_q\,,
    $$
where $1 \leq s \leq qn^{-\vareps}$ for some $\vareps > 0$. 
Since there are $\binom n k$ potential values for index sets $\alpha_j = \supp\paren{c_j}$ and $(q-1)^k$ values for the nonzero entries of $c_j$, if $m = \Ot\paren{n^{k}(q-1)^k}$,  with high probability, 
$$
C_1 + C_0 = \widetilde \Theta\paren{ \frac{m^2}{n^k (q-1)^{k}} } \eqcolon N\,.
$$
A direct calculation gives that under the two hypothesis distributions we have
\begin{align*}
    \mathsf{PLANTED}: \quad& C_1 \sim \bin\paren{N, \Theta(1/s)}\\
    \mathsf{NULL}: \quad& C_1 \sim \bin\paren{N, 1/q}\,.
\end{align*}
From the binomial concentration (Prop.~\ref{prop:binomial_conc}), if we set the threshold $T = 1/q + \widetilde\Theta\paren{\sqrt{1/(Nq)}}$, Alg.~\ref{alg:collision} succeeds at distinguishing $\mathsf{PLANTED}$ vs $\mathsf{NULL}$ with high probability whenever $N = \widetilde\Omega\paren{ \max\set{1, s}}$, where $\widetilde\Omega(1)$ hides $\poly\log n$ factors and constants, since $s\ll q$. Equivalently, 
$$
n^{k/2}(q-1)^{k/2} \cdot \max\set{\sqrt{s}, 1} \lesssim m \lesssim n^k (q-1)^k\,.
$$

\paragraph{Bounded $\Z_q$ Noise.} The bounded noise $\chi_l$ is supported on $\sqb{-l, l}$, so, denoting $\|\chi_l\|^2_2 = \sum_{e \in \sqb{-l, l}} \chi_l(e)^2$, we have the probability that two samples from $\chi_l$ collide equal to $\|\chi_l\|^2_2$. By the argument analogous to above, in Alg.~\ref{alg:collision}, if $m = \Ot\paren{n^{k}(q-1)^k}$, with high probability, 
$$
C_1 + C_0 = \widetilde \Theta\paren{ \frac{m^2}{n^k (q-1)^{k}} } \eqcolon N\,,
$$
and therefore,
\begin{align*}
    \mathsf{PLANTED}: \quad& C_1 \sim \bin\paren{N, \|\chi_l\|^2_2}\\
    \mathsf{NULL}: \quad& C_1 \sim \bin\paren{N, 1/q}\,.
\end{align*}
From the binomial concentration (Prop.~\ref{prop:binomial_conc}), if we set the threshold $T = 1/q + \widetilde\Theta\paren{\sqrt{1/(Nq)}}$, Alg.~\ref{alg:collision} succeeds at distinguishing $\mathsf{PLANTED}$ vs $\mathsf{NULL}$ with high probability whenever $N\paren{\|\chi_l\|^2_2 - 1/q} = \widetilde\Omega\paren{ \max\set{1, \sqrt{N/q}}}$, or, equivalently, 
$$
 m \paren{\|\chi_l\|^2_2 - 1/q} = \widetilde\Omega\paren{n^{k/2}(q-1)^{(k-1)/2}} \quad\text{and}\quad n^{k/2} (q-1)^{k/2}\paren{\|\chi_l\|^2_2 - 1/q}^{-1/2} \lesssim  m \lesssim n^k (q-1)^k\,.
$$

\subsection{Resolution Reduction for Uniform $\lwe$ Noise}\label{subsec:discr_res_lwe_uniform}

In this section we give an improved resolution reduction for $\klwe$ with uniform noise in the case of prime $q$ (Lemma~\ref{lem:2path_red_lwe_uniform}). Compared to the general result in Lemma~\ref{lem:2path_red_lwe}, it gains $(q-1)$ factors for the number of equations in the output instance. This is achieved by observing that, similarly to the detection algorithm in \autoref{subsec:collision_detection}, we can assume that in the input instance $\calY = \set{\paren{c_j, \Y_j}}_{j \in \sqb{m}}$, the coefficient vectors $c_j$ are defined up to a multiplicative factor in $\Z_q\setminus 0$.

\begin{proposition}[Uniform $\lwe$ Noise is resolution-stable]
    Let $\chi_\delta$ be the uniform $\lwe$ noise, i.e., $e\sim \chi_{\delta}$ if for $\delta = O(1/q)$,
    $$e = \begin{cases}
        0, &\text{w. prob. }\delta\\
        \unif\paren{\Z_q}, &\text{w. prob. }1-\delta\,.
    \end{cases}$$
    Then, $\chi_\delta$ is $(f,g)$-resolution-stable, where 
    $$
    f(\delta) = \delta^2\quad\text{and}\quad g(\delta, N) = \begin{cases}
        \Thetat\paren{\sqrt{Nq} \delta}\,, & \widetilde\Omega(1) \leq Nq \leq \delta^{-2}\\
        \delta\,, &\text{otherwise},
    \end{cases}
    $$
    where $\Thetat$ hides constants and $\poly\log\paren{Nq}$ factors. 
\end{proposition}
\begin{proof}
    $f(\delta) = \delta^2$ can be verified via a direct calculation; $g(\delta, N)$ is achieved by a majority vote whenever the number of samples is sufficiently large. 
\end{proof}

\begin{lemma}[Discrete Resolution for $\lwe$ with Uniform Noise]\label{lem:2path_red_lwe_uniform}
    Fix $k,n,m,q \in \Z^{+}, \delta = O(1/q) \in (0,1)$ and let $\chi_\delta$ be uniform $\Z_q$ noise. 
    
    For any even integer $\kstar< 2k$, there exists an average-case reduction for both detection and recovery (Alg.~\ref{alg:two_path_lwe}) $$\text{from}\quad\lweg k n m {\chi_\delta}\quad\text{to}\quad\lweg {\kstar} {n'} {m'} {\chi_{\delta'}}\,,$$ where $n' = n(1-o(1))$ and parameters $m', \theta'$ depend on $\theta$ and the input number of samples $m$ as follows. Denote $\qq = \log_n (q-1)$ and assume $\eps = \log_n m$ is a fixed constant;
    \begin{itemize}
        \item If $\eps \leq \min\set{k - \frac{k'}2 + \paren{k-\frac{k'}2-1} \qq, \frac k2 + \frac{k'}4 + \paren{\frac k 2 + \frac{k'}4 - 1}\qq}$ ,
        $$ m' = \Thetat\paren{m^2 n^{-(k-k'/2)} (q-1)^{-(k-k'/2-1)}} \quad\text{and}\quad \delta' = \delta^2\,;$$
        
        \item If $\eps \in \Big(\frac k2 + \frac{k'}4 + \paren{\frac k 2 + \frac{k'}4 - 1}\qq, k - \frac{k'}2 + \paren{k-\frac{k'}2-1} \qq\Big]$,
        $$
        m' = \Thetat\paren{n^{k'}(q-1)^{k'-1}} \quad\text{and}\quad \delta' = \Thetat\paren{\delta^2 mn^{-(k/2+k'/4)}(q-1)^{-(k/2+k'/4-3/2)}}\,;
        $$
        \item If $\eps \in \Big(k - \frac{k'}2 + \paren{k-\frac{k'}2-1} \qq, k' + (k'-1)\qq\Big]$,
        $$
        m' = \Thetat\paren{m}\quad\text{and}\quad \delta' = \delta^2\,;
        $$
        \item If $\eps > \max\set{k - \frac{k'}2 + \paren{k-\frac{k'}2-1} \qq, k' + (k'-1)\qq}$,
        $$
        m' = \Thetat\paren{n^{k'}(q-1)^{k'-1}} \quad\text{and}\quad \delta' = \Thetat\paren{\delta^2 m^{1/2} n^{-k'/2} (q-1)^{-(k'/2-1)}}\,.
        $$
    \end{itemize}
    Moreover, Alg.~\ref{alg:two_path_lwe} maps an instance with a secret signal vector $x\in \Z_q^n$ to one with the signal vector $x' = x_{\sqb{1:n'}}$. 
\end{lemma}
Combined with results in Lemma~\ref{lem:collision_detection}, our average-case reductions (specifically, second bullet point) yield new algorithms for $k$-sparse $\lwe$ with uniform noise that improve upon Lemma~\ref{lem:collision_detection} in a specific parameter regime by a factor of $\sqrt{q-1}$. 
\begin{corollary}[Detection Algorithms for $k$-sparse $\lwe$ with Uniform Noise] There exists a $O(\poly(n,q))$-time algorithm solving the detection task for $\klwe$ with uniform noise $\chi_\delta$ if $k \geq 6$ and for some even integer $k'\geq 2$, $ n^{k/2 + k'/4}(q-1)^{k/2+k'/4 - 1} \lesssim m \lesssim n^{k-k'/2}(q-1)^{k-k'/2-1}$ and 
$$
m \delta^2 = \widetilde \Omega\paren{n^{k/2} (q-1)^{k/2 -3/2}}\,.
$$

\end{corollary}

\begin{proof}[Proof of Lemma~\ref{lem:2path_red_lwe_uniform}]
    The reduction is Alg.~\ref{alg:two_path_lwe} with one modification. For each $j \in \Ss$, we first modify $c_j \gets c_j \cdot\paren{ c_j }_{\min \set{\supp c_j \cap \sqb{\kappa n +1:n}}}^{-1}$, i.e., ensure that the smallest nonzero coordinate in the cancellation vector of $c_j$ is equal to $1$. This decreases the space of all cancellation vectors to have size $\binom {(1-\kappa)n}{k-k'/2} (q-1)^{k-k'/2-1}$ compared to $\binom {(1-\kappa)n}{k-k'/2} (q-1)^{k-k'/2}$ in Lemma~\ref{lem:2path_red_lwe}, which in turn increases the number of resolution pairs in the reduction. The proof follows the proof of Lemma~\ref{lem:2path_red_lwe} with one key modification: the bound on $|\calZ| \geq B_Z$. An argument analogous to that in Lemma~\ref{lem:2path_red_lwe} shows that, with high probability, 
        $$
            |\calZ| \geq B_Z = \begin{cases}
            C_Z \cdot m^2 n^{-(k-k'/2)} (q-1)^{-(k-k'/2-1)}, &\text{if } \log_n  m \leq k - k'/2 + (k-k'/2-1) \qq \\
            C_Z \cdot m, &\text{otherwise}\,,
        \end{cases}
        $$
    where $C_Z = \Thetat(1)$. The sparse and dense cases of Lemma~\ref{lem:2path_red_lwe} correspond to $\log_n B_Z \leq k' + (k'-1)\qq$ and $\log_n B_Z > k' + (k'-1)\qq$. With the rest of the proof unchanged, this yields the bounds in Lemma~\ref{lem:2path_red_lwe_uniform}.
\end{proof}

\subsection{Resolution Reduction for Discrete Gaussian Noise}\label{subsec:discr_res_lwe_gauss}

\begin{proposition}[Discrete Gaussian is resolution-stable]\label{prop:gauss_res_stable}
    Let $\chi_s$ be discrete Gaussian distribution in $\Z_q$ for odd $q>2$, i.e., $\forall e \in \Z_q$,
    $$
    \chi_s(e) \propto \sum_{i\in \Z} e^{-\pi (e + iq)^2/s^2}\,,
    $$
    where $s \leq q n^{-\vareps}$ for $\vareps > 0$.
    Then, $\chi_s$ is $(f,g)$-resolution-stable, where 
    $$
    f(s) = \sqrt{2} s\quad\text{and}\quad g(s, N) = \Thetat\paren{s/\sqrt{N}}\,. 
    $$
\end{proposition}
\begin{proof}
    $f(s) = \sqrt{2}s$ can be verified via a direct calculation; $g(s, N) = \Thetat\paren{s/\sqrt{N}}$ can be achieved as follows. The key identity is that given $s_1, s_2 \sim \chi_s$ independent,
    \begin{equation}\label{eq:discr_gauss_avg}
    \law\paren{\frac {s_1+s_2}2\, \B|\, s_1\equiv s_2 \mod{2}} = \chi_{s/\sqrt{2}}\,,
    \end{equation}
    proved in \cite{aggarwal2015solving,stephens2017gaussian}. Iterative application of \eqref{eq:discr_gauss_avg} achieves the desired $g(s, N)$.
\end{proof}

Plugging the parameters in Prop.~\ref{prop:gauss_res_stable} into the general Lemma~\ref{lem:2path_red_lwe}, we obtain: 
\begin{lemma}[Discrete Resolution for $\lwe$ with Discrete Gaussian Noise]\label{lem:2path_red_lwe_gauss}
    Fix $k,n,m \in \Z^{+}, 0<s<qn^{-\vareps}$ and let $\chi_s$ be be discrete Gaussian noise. For any even integer $\kstar\leq 2k$, there exists an average-case reduction for both detection and recovery (Alg.~\ref{alg:two_path_lwe}) $$\text{from}\quad\lweg k n m {\chi_s}\quad\text{to}\quad\lweg {\kstar} {n'} {m'} {\chi_{s'}}\,,$$ where $n' = n(1-o(1))$ and parameters $m', s'$ depend on $s$ and the input number of samples $m$ as follows. Denote $\neff = n(q-1)$ and assume $\eps = \log_{\neff} m$ is a fixed constant;
    \begin{itemize}
        \item If $\eps \leq \min\set{k-\frac{k'}2, \frac k2 + \frac{k'}4}$, 
        $$ m' = \Thetat\paren{m^2 \neff^{-(k-k'/2)}} \quad\text{and}\quad s' = \sqrt{2}s\,;$$
        \item If $ \eps \in \Big( \frac k 2 + \frac{k'}4, k-\frac{k'}2\Big]$,
        $$
        m' = \Thetat\paren{\neff^{k'}} \quad\text{and}\quad s' = \Thetat(s\cdot m^{-1} \neff^{k/2 + k'/4})\,;
        $$
        \item If $\eps \in \Big(k-\frac{k'}2,k'\Big]$,
        $$
        m' = \Thetat\paren{m}\quad\text{and}\quad s' = \sqrt{2}s\,;
        $$
        \item If $\eps > \max\set{k-\frac{k'}2, k'}$,
        $$
        m' = \Thetat\paren{\neff^{k'}} \quad\text{and}\quad s' = \Thetat(s \cdot m^{-1/2} \neff^{k'/2})\,.
        $$
    \end{itemize}
    Moreover, Alg.~\ref{alg:two_path_lwe} maps an instance with a secret signal vector $x\in \Z_q^n$ to one with the signal vector $x' = x_{\sqb{1:n'}}$. 
\end{lemma}

Combined with results in Lemma~\ref{lem:collision_detection}, our average-case reductions (specifically, third bullet point) yield new algorithms for $k$-sparse $\lwe$ with discrete Gaussian noise. 
\begin{corollary}[Detection Algorithms for $k$-sparse $\lwe$ with Discrete Gaussian Noise] There exists a $O(\poly(n,q))$-time algorithm solving the detection task for $\klwe$ with discrete Gaussian noise $\chi_s$ in the following parameter regime. Denoting $k'\leq k$ be the smallest even integer satisfying $k' \geq \log_{\neff} m$ and $k' > 2(k-\log_{\neff} m)$, it suffices to have 
$$
\neff^{k'/2} \cdot \max\set{\sqrt{s}, 1} \lesssim m\,.
$$

\end{corollary}

\subsection{Resolution Reduction for Bounded $\Z_q$ Noise}\label{subsec:discr_res_lwe_bound}

While the most general version of resolution reduction for bounded $\lwe$ noise is stated in Lemma~\ref{lem:2path_red_lwe}, we illustrate a particular example. Let $\chi_1$ be \emph{binary $\lwe$ noise}, i.e., $\chi_1$ is only supported on $\set{0, \pm 1}$ and assume $q\geq 5$. Denote $\chi_2$, supported on $\set{0,\pm1, \pm2}$ to be the distribution of the difference of two independent samples from $\chi_1$:
$$
\chi_2(e) = \sum_{e_1,e_2 \in \set{0,\pm1}: e_1-e_2 =e} \chi_1(e_1) \chi_1(e_2)\,.
$$
In this case, our resolution reduction (Alg.~\ref{alg:two_path_lwe}) can map an instance with noise distribution $\chi_1$ to noise distribution $\chi_2$. We state the results for regimes of Lemma~\ref{lem:2path_red_lwe} where the output noise is exactly $\chi_2$; other regimes (bullet points two and four) also apply, but require explicit aggregation procedure (e.g., a majority vote) for the bounded noise.

\begin{corollary}[of Lemma~\ref{lem:2path_red_lwe}]\label{cor:lwe_alg_bounded}
    Fix $k,n,m \in \Z^{+}$.
    For any even integer $\kstar\leq 2k$, there exists an average-case reduction for both detection and recovery (Alg.~\ref{alg:two_path_lwe}) $$\text{from}\quad\lweg k n m {\chi_1}\quad\text{to}\quad\lweg {\kstar} {n'} {m'} {\chi_2}\,,$$ where $n' = n(1-o(1))$ and $m'$ depends on the input number of samples $m$ as follows. Denote $\neff = n(q-1)$ and assume $\eps = \log_{\neff} m$ is a fixed constant;
    \begin{itemize}
        \item If $\eps \leq \min\set{k-\frac{k'}2, \frac k2 + \frac{k'}4}$, then 
        $ m' = \Thetat\paren{m^2 \neff^{-(k-k'/2)}}\,;$
        \item If $\eps \in \Big(k-\frac{k'}2,k'\Big]$, then
        $
        m' = \Thetat\paren{m}\,.
        $
    \end{itemize}
    Moreover, Alg.~\ref{alg:two_path_lwe} maps an instance with a secret signal vector $x\in \Z_q^n$ to one with the signal vector $x' = x_{\sqb{1:n'}}$. 
\end{corollary}

\section{Gaussian Resolution}\label{sec:gauss_resolution}

The discrete resolution reduction ($\DiscrEqComb$) in \autoref{sec:discrete_resolution} combines pairs of entries in the input and achieves entrywise noise independence by simply not reusing the input equations -- this works best for sparser instances of $\kxor$, i.e. small density $\eps$. We introduce a new reduction in Lemma~\ref{lem:gauss_resolution}, \emph{Gaussian resolution} ($\GaussEqComb$), that improves upon parameter tradeoffs $k' \leq \min\set{k(1-\eps), 2k|\eps|}$ for target $\eps' = 1$ in \autoref{sec:discrete_resolution} for some inputs with larger density $\eps \geq 1/3$. See \autoref{sec:to_tensor_pca} for the comparison of the two reductions. 

\subsection{Gaussian Resolution Reduction Idea and Lemma Statement} 

We prove a reduction for the Gaussian $\ktens$ models in Lemma~\ref{lem:gauss_resolution} and the analogous result for discrete $\kxor$ follows by the discrete-Gaussian equivalence in Prop.~\ref{prop:discr_gauss_equivalence_general}. In this section we work with densities $\eps \geq 0$.

In $\gtensp k n \eps \defi$ (equivalently $\gtens k n m \delta$ with $m = n^{k(1+\eps)/2}$ and $\delta = n^{-k\eps/4 - \defi/2}$) with a secret vector $x \in \set{\pm 1}^n$, we observe a collection $\calY = \set{\paren{\alpha_j, \Y_j}}_{j \in \sqb{m}}$ for $\alpha_j \sim_{i.i.d.} \unif\sqb{\binomset n k}$ and 
\begin{equation*}
    \Y_{j} = 
    \delta \cdot x_{\alpha_j} + \noiseG_{j}\,, 
\end{equation*}
where $\noiseG_{j} \sim \N(0,1)$ and we recall that $x_{\alpha} = x_{i_1}\dots x_{i_k}$ for an index set $\alpha = \set{i_1,\dots,i_k} \in \binomset n k$. The product of two $\ktens$ entries $\Y_{s}, \Y_{t}$ for distinct $s, t\in \sqb{m}$ with $\kstar = |\alpha_s \triangle \alpha_t|$ has a mean matching the $\kxork \kstar$ entry with signal level $\delta' = \delta^2$ corresponding to an index set $\alpha_s \triangle \alpha_t$:
\begin{equation}\label{eq:resolution_id9}
    \E \sqb{\Y_{s} \Y_{t}} = 
    \delta^2 \cdot x_{\alpha_s} x_{\alpha_t} = \delta^2 \cdot x_{\alpha_s \triangle \alpha_t}\,.
\end{equation}
Contrary to $\DiscrEqComb$ in \autoref{sec:discrete_resolution}, where the \emph{distribution} of $\Y_s \Y_t$ matched that of $\kxork \kstar$, 
\begin{equation}\label{eq:resolution_id9_2}
\Y_{s} \Y_{t} = 
    \delta^2 \cdot x_{\alpha_s \triangle \alpha_t} + {\delta x_{\alpha_s} \noiseG_t + \delta x_{\alpha_t} \noiseG_s} + {\noiseG_s \noiseG_t} \ \ \not\approx_{\tv}\ \ \delta^2 \cdot x_{\alpha_s\triangle \alpha_t} + \N(0,1)\,,
\end{equation}
i.e., multiplying two $\ktens$ entries does not give a $\ktensk \kstar$ entry. Our insight and the main technical contribution of this section is that in some parameter regimes, aggregating many pairs $s,t \in \sqb{m}$ with the same $\gamma = \alpha_s\triangle \alpha_t \in \binomset n \kstar$ \emph{does} yield the correct $\ktensk \kstar$ distribution, even when entries are reused in Eq.~\eqref{eq:resolution_id9_2}.

For the reduction below, we consider $\eps \in (1/3,1], \defi \in \R$ and assume
$$-k\eps/2 < \defi \leq (k-2)/2\,,$$ 
which in particular ensures the entrywise signal $\delta = n^{-k\eps/4 -\defi/2}$ is $o(1)$ and that the input instance is in the conjectured hard regime (see \autoref{subsec:complexity_profiles}). 

\begin{lemma}[Gaussian Resolution: Reduction to Tensor PCA]\label{lem:gauss_resolution}
    Let $\eps \in (1/3,1], \defi \in \R, k \in \mathbb{Z}^{+}$ be the $\kxor$ parameters satisfying $-k\eps/2 < \defi \leq (k-2)/2$. If $k'$ is even and satisfies
    \begin{equation}\label{eq:tensor_red_cond}
        k' < \min\set{ k\eps + 2\defi, \frac{k\paren{3\eps -1}}{2}, \frac{k\paren{1+\eps}}{4}}\,,
    \end{equation}
    then there exists a poly-time average-case reduction\footnote{As defined in Def.~\ref{def:avg_case_points}} for both detection and recovery (Alg.~\ref{alg:gauss_res})
    $$
    \text{from}\quad\gtensp k n \eps \defi \quad\text{to}\quad \gtensp {k'} {n'} {\eps'=1} {2\defi}\,,
    $$
    for $n' = n\paren{1- o(1)}$. Moreover, Alg.~\ref{alg:gauss_res} maps an instance with a secret signal vector $x \in\set{\pm1}^n$ to one with a signal vector $x' = x_{\sqb{1:n'}} \in \set{\pm 1}^{n'}$.\footnote{Here $x'$ can be any subset of $n'$ coordinates of $x$ by permutation invariance of the model. See Sec.~\ref{subsec:recovery_assumptions} for why this is sufficient to transfer weak/strong/exact recovery.}
\end{lemma}
\begin{remark}
    In the case of $\defi \to 0$, i.e. at the \emph{computational threshold}, the parameter regime above can be rewritten as follows: all even $k'$ satisfying
    \begin{enumerate}
        \item $\eps \in (1/3, 3/5)$: $k' < \frac{k\paren{3\eps - 1}}{2}$;
        \item $\eps \in (3/5, 1]$: $k' < \frac {k(1+\eps)} 4$.
    \end{enumerate}
\end{remark}

\begin{remark}[$\ktens$ Sampling Procedure for Alg.~\ref{alg:gauss_res} and Discrete Input]\label{rmk:gaussian_resolution_alg_sampling_discr}
    While Algorithm~\ref{alg:gauss_res}, as stated, is a reduction from $\gtensppoi k n \eps \defi$ to $\gtensppoi {\kstar} {n'} {\eps'=1} {2\defi}$, we sometimes use it as a subroutine for $\ktens$ models with the standard sampling procedure, i.e., as a reduction from $\gtensp k n \eps \defi$ to $\gtensp {\kstar} {n'} {\eps'=1} {2\defi}$. It is implicitly assumed that first a Poissonization (Alg.~\ref{alg:resample_iid}) reduction is applied, then Alg.~\ref{alg:gauss_res}, and then the reduction back to standard sampling model (Alg.~\ref{alg:resample_M}).
    
    Similarly if Alg.~\ref{alg:gauss_res} is applied as a discrete model reduction for $\kxor$ or $\kxor^{\mathsf{Pois}}$: this implicitly assumes the insertion of the appropriate Gaussianization/Discretization (\autoref{subsec:gaussianize}, \autoref{subsec:discretize}) and/or Resampling (\autoref{subsec:poi_sampling}) reductions.
\end{remark}

We present a brief proof sketch here and provide the full proof of Lemma~\ref{lem:gauss_resolution} in \autoref{subsec:gauss_res_proof}. For this proof we adopt a \emph{tensor view} of the input instance $\calY = \set{\paren{\alpha_j, \Y_j}}_{j \in \sqb{m}}$. Let $\Yb \in \paren{\R^n}^{\otimes k}$ be a symmetric tensor satisfying for all $\alpha\in\binomset n k$
\begin{equation*}
    \Yb_{\alpha} = \begin{cases}
        \Y_{j^\star}, &\text{ if }\ j^\star=\min\set{j \in \sqb{m}: \alpha_j = \alpha}\\
        0, &\text{ if }\set{j \in \sqb{m}: \alpha_j = \alpha} = \emptyset\,.
    \end{cases}
\end{equation*}
Here, since $\Yb$ is symmetric up to index permutations, we index its entries by the multisets of indices in $\sqb{n}$. The idea behind the $\GaussEqComb$ is to consider a tensor flattening $Y = \mat {k-k'/2} {k'/2} \Yb \in \R^{\binom n {k-k'/2}\times \binom n {k'/2}}$ -- a matrix whose rows are indexed by elements of $\binomset n {k-k'/2}$ and columns by $\binomset n {k'/2}$.\footnote{See \autoref{sec:notation} for notation.} We consider the tensorization (i.e., the inverse procedure to $\mat {k'/2} {k'/2} \paren{\cdot}$) of a Gram matrix
$$
\Zb = \tenso{\paren{Y^\top Y}}\quad \in \paren{\R^{n}}^{\otimes k'}\,.
$$
Note that for all $\gamma\in \binomset n \kstar$, the resulting entry $\Zb_{\gamma}$ is a sum of products in Eq.~\eqref{eq:resolution_id9_2} and we have 
$$
\E\sqb{\Zb_{\gamma }} \propto x_{\gamma}\,.
$$
To ensure \emph{independence of the entry-wise noise} and correct output distribution, our reduction works with a Gram-like matrix similar to $Y^\top Y$. Our main technical contribution is a distributional result for such  Wishart-like matrices in Lemma~\ref{lem:wishart}, which is a generalization of known results for Wishart matrices from \cite{brennan2021finetti}. Lemma~\ref{lem:wishart} is proved in \autoref{subsec:wishart}.

\begin{lemma}[Generalization of Theorem 2.6 of \cite{brennan2021finetti}]\label{lem:wishart}
Let $\Mr\in[0,1]^{d\times n}, \Ml\in[0,1]^{d\times m}$ have i.i.d. $\ber(\pr)$ and $\ber(\pl)$ entries respectively. Let $\deltaR,\deltaL \in \R$ be parameters such that $\delta_R, \delta_L \ll 1$ and let $\Mur^\full \in \set{\pm 1}^{d \times n}, \Mul^\full \in \set{\pm 1}^{d \times m}$ be any fixed matrices. Let $\paren{\Xgr}^\full \sim \N\paren{0,I_{d \times n}}$ and $\paren{\Xgl}^\full \sim \N\paren{0,I_{d \times m}}$. 

Define for $a \in \set{R, L}$:
$$
\Mu_a = \delta_a \Mu_a^\full \odot \M_a\quad\,, \Xg_a = \paren{\Xg_a}^\full \odot \M_a\,, \quad \text{and} \quad X_a = \Mu_a + \Xg_a\,.
$$

If the following assumptions on parameters $\pr, \pl, \deltaR, \deltaL, d, n, m$ are satisfied: 
\begin{align}
    \pr, \pl &\in \set{1} \cup (0, 1-c) \text{ for some }c>0\tag{a}\label{eq:cond_a}\\
    \deltaR &\ll \min\set{1, \paren{\pr n}^{-1/2}}\text{ and } \deltaL\ll \min\set{1, \paren{\pl m}^{-1/2}}\tag{b}\label{eq:cond_b}\\
    n &\ll \pr^2 \pl d\text{ and }m\ll \pl^2 \pr d\tag{c}\label{eq:cond_c}\\
    mn &\ll \min\set{\pr\pl d, \deltaR^{-4}, \deltaL^{-4}}\tag{d}\label{eq:cond_d}\\
    n^2m &\ll \min\set{\pl d, \pr^{-2} \deltaR^{-4}},\text{ and }nm^2 \ll \min\set{\pr d, \pl^{-2} \deltaL^{-4}}\tag{e}\label{eq:cond_e}\,,
\end{align}
then 
$$
\tv\paren{N\paren{\paren{\pr\pl d}^{-1/2}\Mur^T \Mul, I_{n\times m}}, \law\paren{\paren{\pr \pl d}^{-1/2} \Xr\tr \Xl}} \to 0\,, \quad d\to 
\infty\,.
$$

\end{lemma}

\subsection{Proof of Lemma~\ref{lem:gauss_resolution}}\label{subsec:gauss_res_proof}

\begin{algorithm}\SetAlgoLined\SetAlgoVlined\DontPrintSemicolon
    \KwIn{
    $\calY = \set{ \b( \alpha_j, \Y_{j}\b)}_{j\in\sqb{\mpoi}}$, where $\forall j \in\sqb{\mpoi}$, $\b(\alpha_j, \Y_{j}\b) \in \binomset n k \times \R$, $k'\in \Z$, $\eps\in\sqb{0,1}$, $\kappa \in (0,1)$.
    }
    \KwOut{$\widetilde\calZ = \set{ \b( \gamma_j, \Zt_{j}\b)}_{j\in\sqb{\mpoi'}}$, where $\forall j \in\sqb{\mpoi'}$, $\b(\gamma_j, \Zt_{j}\b) \in \binomset {\kappa n} {k'} \times \R$.}

    \BlankLine
    $\calY\up1, \calY\up2 \gets \XorSplit\paren{\calY}$

    Denote $\calY\up A = \set{ \b( \alpha_j\up A, \Y_{j} \up A\b)}_{j\in\sqb{\mpoi \up A}}$ for $A \in\set{1,2}$
    
    \BlankLine
    \tcp{Construct the tensor view of the input}
    $$
    \Yb_{\alpha}\up A = \begin{cases}
        \Y_{j^\star}\up A, &\text{ if }\ j^\star=\min\set{j \in \sqb{\mpoi\up A}: \alpha_j\up A = \alpha}\\
        0, &\text{ if }\set{j \in \sqb{\mpoi \up A}: \alpha_j\up A = \alpha} = \emptyset
    \end{cases} \quad \text{for }A\in\set{1,2}\text{ and }\alpha \in \binomset n k
    $$
    
    \BlankLine
    \tcp{Flatten the tensors}
    $\hat Y\up A \gets \mat {k - k'/2} {k'/2} \paren{ \Yb\up A} \in \R^{\binomset n {k-k'/2} \times \binomset n {k'/2}}$ for $A\in\set{1,2}$

    \BlankLine
    \tcp{Only keep part of the entries}
    For $A\in\set{1,2}$ let $Y\up A \in \R^{\binom {(1-\kappa)n}{k-k'/2} \times \binom {\kappa n}{k'/2}}$:
    $$
        \forall \alpha \in \binom{\sqb{\kappa n+1,\dots,n}}{k-k'/2}, \beta \in  \binom{\sqb{\kappa n}}{k'/2}\,, \qquad Y_{\alpha,\beta}\up A \gets \hat Y_{\alpha,\beta}\up A
    $$

    \BlankLine
    \tcp{Compute Revealed Entry Probability}
    $p \gets C\cdot n^{k(\eps-1)/2}$, where $C$ is defined in the proof of Lem.~\ref{lem:gauss_resolution} and satisfies $\Pr[\Y\up A \neq 0] = p$

    \BlankLine 
    \tcp{Construct the Gram matrix}
    $$ Z \gets p^{-1}\binom {(1-\kappa)n}{k-k'/2}^{-1/2} \paren{Y\up 1}^\top Y\up 2 $$

    \BlankLine
    \tcp{Arrange the entries of $\tenso\paren{Z}$ into an instance of $\ktensk {k'}^\wor$}
    $\Zb = \tenso\paren{Z}$, $\widetilde\calZ = \set{\paren{\gamma, \Zb_\gamma}}_{\gamma \in \binomset {\kappa n} {k'}}$
    
    \BlankLine
    Return $\ToPoiFromTPCA\b(\widetilde\calZ\b)$. \ \tcp{Alg.~\ref{alg:to_poi_sampling_from_pca}}
    
    \caption{$\GaussResolution\paren{\calY, k', \eps, \kappa = 1 - \log^{-1}n}$}
\label{alg:gauss_res}
\end{algorithm}

We prove that $\GaussResolution\paren{\cdot}$ (Alg.~\ref{alg:gauss_res}) is an average-case reduction for both detection and recovery, as defined in Def.~\ref{def:avg_case_points}, 
$$
    \text{from}\quad\gtensppoi k n \eps \defi \quad\text{to}\quad \gtensppoi {k'} {n'} {\eps'=1} {2\defi}\,,
$$
where parameters $k, k', \eps, \defi, n'$ satisfy the conditions of Lemma~\ref{lem:gauss_resolution}. This immediately yields an average-case reduction from $\gtensp k n \eps \defi$ to $\gtensppoi {k'} {n'} {\eps'=1} {2\defi}$, since by Prop.~\ref{prop:sampling_equiv} there exist average-case reductions between $\gtensp k n \eps \defi$ and $\gtensppoi k n \eps \defi$ both ways.

To prove the reduction, since it is sufficient to map parameters up to $\poly\log n$ factors by Prop.~\ref{prop:map_up_to_constant}, we show that in the planted case $\calY \sim \gtensppoi k n \eps \defi$, the output $\calZ = \GaussResolution\paren{\calY, k', \eps}$ satisfies 
$$
\tv\paren{\law\paren{\calZ}, \gtenspoi {k'} {n'} {m'} {\delta'}} = o(1)\,, \quad \text{where } m' = \Thetat\paren{n^{k'}} \text{ and } \delta' = \Thetat\paren{ n^{-k'/4 - \defi} }\,,
$$
and in the null case $\calY \sim \gtensppoinull k n \eps$, the output $\calZ = \GaussResolution\paren{\calY, k', \eps}$ satisfies 
$$
\tv\paren{\law\paren{\calZ}, \gtenspoinull {k'} {n'} {m'}} = o(1)\,, \quad \text{where } m' = \Thetat\paren{n^{k'}}\,.
$$

Denote $m = n^{k(1+\eps)/2}, \delta = n^{-k\eps/4 - \defi/2}$ and assume the input instance contains $\mpoi \sim \poi\paren{m}$ equations.
\paragraph{Step 1: tensor view of $\ktens$.} In Lines 1-6 of Alg.~\ref{alg:gauss_res}, we reshape the input instance into two matrices $Y\up 1, Y\up 2$, which we hereby describe. We show the calculation for the planted case and the argument in the null case is analogous.
From the guarantees of the splitting procedure (Lemma~\ref{lem:instance_split}),  $\calY\up1, \calY\up2 \sim \gtenspoi k n {\tilde m} {\delta}$, $\tilde m = \Theta(m)$, are independent and share the input signal vector $x$. Then, we can express for all $\alpha\in\binomset n k$
$$\Yb_{\alpha}\up A = \paren{\delta \cdot x_{\alpha} + \noiseG_{\alpha}\up A}\cdot \M_{\alpha}\up A\,, \quad \text{for }A\in \set{1,2}\,,$$
where $\noiseG_{\alpha}\up A \sim \N(0,1)$ is the entry-wise noise and $\M_{\alpha}\up A \sim \ber\paren{p}$ for
$$
p = \Pr\sqb{ \exists j \in \sqb{\mpoi\up A}:\, \alpha_j = \alpha } = \Pr\sqb{\poi\paren{\tilde m \binom n k ^{-1} }} = \Thetat\paren{ n^{k(\eps-1)/2} } \eqcolon C \cdot n^{k(\eps-1)/2} \,,
$$
for a known $C = \Thetat(1)$, where we used $\mpoi\up A \sim \poi\paren{\tilde m}$ and Poisson splitting.
The flattenings $\mat {k-k'/2} {k'/2} \Yb\up A$ are matrices of dimensions $\binom n {k-k'/2} \times \binom n {k'/2}$ with rows indexed by elements in $\binomset n {k-k'/2}$ and columns by $\binomset n {k'/2}$. Denote the pruned sets as $N = \binomset {\kappa n}{k'/2}$ and $M = \binom{(1-\kappa)n}{k-k'/2}$. Then in Line 5 of Alg.~\ref{alg:gauss_res} we set for all $(\alpha, \beta) \in M \times N$ $$Y_{\alpha, \beta}\up A = \paren{\delta \cdot x_{\alpha \cup \beta} + \noiseG_{\alpha \cup \beta}\up A}\cdot \M_{\alpha\cup\beta}\up A\,, \quad \text{for }A\in \set{1,2}\,.$$

Let $x\up {k-k'/2} = \veco\paren{\set{x_{\alpha}}_{\alpha \in M}}$, $x\up {k'/2} = \veco\paren{\set{x_{\beta}}_{\beta \in N}}$, so we can rewrite
$$Y\up A = \paren{\delta x\up {k-k'/2} \paren{x \up {k'/2}}^\top + \noiseG \up A}\odot \M \up A\,, \quad \text{for }A\in \set{1,2}\,,$$
where $\noiseG\up A \sim \N(0, I_{|M|\times |N|})$ are independent noise matrices and $\odot$ denotes the entry-wise multiplication.

\paragraph{Step 2: Gram matrix: a reduction to $\ktensk {k'}^\wor$ (see Fig.~\ref{fig:kxor_variants} and \autoref{subsubsec:equiv_w_wor_replacement}).} In Line 7 Alg.~\ref{alg:gauss_res} computes the Gram matrix 
\begin{equation}\label{eq:gaus_res_Z_def}
Z = p^{-1} |M|^{-1/2} \paren{Y\up 1}^\top Y\up 2\,.
\end{equation}
The goal of this step is to show that in Lines 1-8, Alg.~\ref{alg:gauss_res} makes an average-case reduction for both detection and recovery to $\ktensk {k'}^\wor_{\eps'=1, 2\defi}(n'=\kappa n)$ (defines in \autoref{subsubsec:equiv_w_wor_replacement}). For this it is sufficient to demonstrate that in the planted case 
\begin{equation}\label{eq:g_res_final_tv}
    \tv\paren{\law\paren{Z}, \delta' x \up {k'/2}\paren{x \up {k'/2}}^\top + \N\paren{0, I_{|N|, |N|}}} = o(1)\,,
\end{equation}
where $\delta' = \Thetat\paren{n^{-k'/4 - \defi}}$ and in the null case
\begin{equation}\label{eq:g_res_final_tv_null}
    \tv\paren{\law\paren{Z}, \N\paren{0, I_{|N|, |N|}}} = o(1)\,.
\end{equation}
Note that, since $\Y\up 1, \Y\up 1$ are independent,
\begin{align*}
    \E Z &= p^{-1} |M|^{-1/2} \E \sqb{Y\up 1}^\top \E \sqb{Y\up 2}\\
    &= p^{-1} |M|^{-1/2} \paren{\delta x\up {k'/2} \paren{x \up {k-k'/2}}^\top }\odot \E\sqb{\M \up 1}^\top \paren{\delta x\up {k-k'/2} \paren{x \up {k'/2}}^\top }\odot \E\sqb{\M \up 2}\\
    &= p^{-1} |M|^{-1/2} \delta^2 p^2 |M| x\up {k'/2} \paren{x\up {k'/2}}^\top = \underbrace{\delta^2 p |M|^{1/2}}_{\delta'} x\up {k'/2} \paren{x\up {k'/2}}^\top\,,
\end{align*}
where we have $\delta' = \delta^2 p |M|^{1/2} = \Theta\paren{n^{-k\eps/2 - \defi + k(\eps-1)/2 + (k-k'/2)/2}} = \Theta \paren{n^{-k'/4 - \defi}}$. It remains to show that the entries of $Z$ are indeed independent Gaussian entries with unit variance. We achieve this by applying Lemma~\ref{lem:wishart} proved in \autoref{subsec:wishart}.

Matrix $Z$ in Eq.~\eqref{eq:gaus_res_Z_def} fits the setting of Lemma~\ref{lem:wishart}. Indeed, matrices $Y\up 1, Y \up2 \in \R^{|M| \times |N|}$ satisfy the Lemma~\ref{lem:wishart} conditions with $\tilde \pr = \tilde \pl = p$, $\tilde\deltaR = \tilde\deltaL = \delta$, and $\tilde d = |M|, \tilde n = \tilde m = |N|$, where recall $p = \Thetat\paren{n^{k(\eps-1)/2}}, \delta = \Thetat\paren{n^{-k\eps/4 - \defi/2}}$ and $|M| = \Thetat\b(n^{k-k'/2}\b), |N| = \Thetat\b(n^{k'/2}\b)$. It is left to verify the asymptotic assumptions on parameters $\tilde\pr, \tilde\pl, \tilde\deltaR, \tilde\deltaL, \tilde d, \tilde n, \tilde m$: condition \eqref{eq:cond_a} is trivially satisfied and 
\begin{align*}
    &\eqref{eq:cond_b}: -k\eps/4 - \defi/2 < 0 \quad\text{and}\quad k/2 - \defi < k-k'/2\\
    &\eqref{eq:cond_c}: k(5-3\eps)/4 < k-k'/2\\
    &\eqref{eq:cond_d}: k(3-\eps)/3 < k-k'/2 \quad\text{and}\quad k(2-\eps)/2 - \defi < k-k'/2\\
    &\eqref{eq:cond_e}: k(7-\eps)/8 < k-k'/2 \quad\text{and}\quad 2k/3 - 2\defi/3 < k-k'/2\,,
\end{align*}
all of which are satisfied in the regime of our reduction:
$$k-k'/2 >\max\set{ k\paren{1-\eps/2} - \defi, k(5-3\eps)/4, k(7-\eps)/8} \text{ and } -k\eps/2 < \defi \leq (k-2)/2\,.$$ 

Denoting $$\mu\up A = \delta x\up {k-k'/2} \paren{x\up {k'/2}}^\top \odot \M\up A\,, \quad \text{for }A\in \set{1,2}$$
to be the means of the matrices $\Y\up 1, \Y\up 2 \in \R^{|M| \times |N|}$, we conclude from Lemma~\ref{lem:wishart},
$$
\tv\paren{\N\paren{p^{-1}|M|^{-1/2}\paren{\mu \up 1}^T \mu\up 2, I_{N\times N}}, p^{-1}|M|^{-1/2} \paren{Y \up 1}^T Y\up 2} \to 0\,, \quad n\to 
\infty\,.
$$
By $\tv$ triangle inequality (Fact~\ref{tvfacts}), to conclude the theorem statement it is now sufficient to show 
$$
\tv\paren{\N\paren{p^{-1}|M|^{-1/2}\paren{\mu \up 1}^T \mu\up 2, I_{|N|\times |N|}}, \N\paren{\delta^2 p |M|^{1/2} x\up {k'/2} \paren{x\up{k'/2}}^\top, I_{|N|\times |N|}}} \to 0\,, \quad n\to 
\infty\,.
$$
By Lemma~\ref{l:KLGauss}, Eq.~\eqref{eq:tv_gauss_pinsker}, the square of the total variation above is upper bounded by 
\begin{align*}
    \tv ^2 &\leq \frac 14 \E_{\mu\up 1, \mu\up 2} \B\| p^{-1}|M|^{-1/2}\paren{\mu \up 1}^T \mu\up 2 -  \delta^2 p |M|^{1/2} x\up {k'/2} \paren{x\up{k'/2}}^\top \B\|^2\\
    &= \frac 14 \E_{\mu\up 1, \mu\up 2} \B\| p^{-1}|M|^{-1/2}\paren{\mu \up 1}^T \mu\up 2 -  \E p^{-1}|M|^{-1/2}\paren{\mu \up 1}^T \mu\up 2\B\|^2\,.
\end{align*}
Now recall that for all $i,j \in [|N|]$, 
$$
\paren{\paren{\mu \up 1}^T \mu\up 2}_{ij} \sim \delta^2 \bin\paren{|M|, p^2}\,,
$$
and therefore, $$ \E \sqb{\paren{\paren{\mu \up 1}^T \mu\up 2}_{ij} - \E \paren{\paren{\mu \up 1}^T \mu\up 2}_{ij} }^2 = \delta^4 |M| p^2 (1-p^2) \leq \delta^4 |M| p^2\,.$$
We conclude that 
\begin{align*}
    \tv^2 \leq \frac 14 p^{-2} |M|^{-1} \cdot |N|^2 \delta^4 |M| p^2 = \frac 14 |N|^2 \delta^4 = \Theta\paren{ n^{k' - k\eps - 2\defi}} = o(1)\,,
\end{align*}
since $k' - k\eps - 2 \defi < 0$, which shows Eq.~\eqref{eq:g_res_final_tv} and concludes the proof for the planted case. 

For the null case (Def.~\ref{def:null_models}), the analysis is analogous and Lemma~\ref{lem:wishart} is similarly applicable. We conclude that the output is close in total variation distance to $N\paren{p^{-1}|M|^{-1/2}\paren{\mu \up 1}^T \mu\up 2, I_{|N|\times |N|}}$, where every entry of $\mu\up 1, \mu\up 2$ is i.i.d. $\sim\begin{cases}
    0, &\text{w.p. }1-p\\
    +\delta, &\text{w.p. } p/2\\
    -\delta, &\text{w.p. } p/2\,.
\end{cases}$. Notice that with high probability on the Bernoulli mask,
\begin{align*}
    p^{-2}|M|^{-1} \E \b\|\paren{\mu \up 1}^T \mu\up 2\b\|_2^2 &\leq p^{-2}|M|^{-1} O\paren{ |N|^2 |M| p^2 \delta^4} = O\paren{\delta^4 N^2} = O\paren{n^{k' - k\eps - 2\defi}} = o(1) \,,
\end{align*}
since $k-k'/2 > k\paren{1-\eps/2} - \defi$ and we only consider the off-diagonal entries.
By properties of total variation (Fact~\ref{tvfacts}) and Lemma~\ref{l:KLGauss}, we conclude that the output is close in total variation to $\N\paren{0, I_{|N|\times |N|}}$, proving Eq.~\eqref{eq:g_res_final_tv_null}.

\paragraph{Step 3: from $\ktensk{k'}^\wor$ to $\kgtensppoik {k'}$.} Since Lines 1-8 are an average-case reduction to $\ktensk {k'}^\wor_{\eps'=1, 2\defi}(n'=\kappa n)$, the final step is an average-case reduction to $\gtensppoi {k'} {n'} {\eps'=1} {2\defi}$ by the guarantees for $\ToPoiFromTPCA$ in Prop.~\ref{prop:sampling_equiv_eps1} (and implicit Poissonization step, \autoref{subsubsec:equiv_sampling_pois}). This concludes the proof of Lemma~\ref{lem:gauss_resolution}.

\subsection{Improved Gaussian Resolution for $\eps=1$ via Sparse PCA}\label{subsec:gauss_res_eps_1}

In Lemma~\ref{lem:gauss_resolution} we proved an average-case reduction from $\gtensp k n \eps \defi$ to $\gtensp {k'} {n'} {\eps'=1} {2\defi}$, where $\eps \in (1/3, 1]$. For $\eps=1$, i.e., when the input instance is also Tensor PCA, the reduction in Lemma~\ref{lem:gauss_resolution} maps from Tensor PCA of order $k$ to Tensor PCA of even order $k' < k/2$. In this section we prove Lemma~\ref{lem:2path_red_ext_gauss}: it extends $\GaussEqComb$ to accept as an input two $\ktens$ instances of different orders $k_1, k_2$ that share the same secret vector $x\in\set{\pm1}^n$:
$$
\gtensp {k_1} n {\eps=1} {\defi_1} \ ``+" \ \gtensp {k_2} n {\eps=1} {\defi_2}\ \to \ \gtensp {k'} {n'} {\eps=1} {\defi_1+\defi_2}
$$
for admissible range of parameters $k' < \frac{k_1+k_2}{3}$ and $n' = n(1-o(1))$. Our main insight is a connection to recent results for Spiked Covariance and Spiked Wigner Models in \cite{bresler2025computational} (\autoref{subsubsec:sparse_pca}). Consequently, we improve the reduction of Lemma~\ref{lem:gauss_resolution} to map from order $k$ to all even orders $k' < 2k/3$. The proof of Lemma~\ref{lem:2path_red_ext_gauss} is in \autoref{subsubsec:proof_two_tens_gen}.

We present the reduction for the Gaussian $\ktens$ and the same result holds for discrete $\kxor$ by the discrete-Gaussian equivalence in Prop.~\ref{prop:discr_gauss_equivalence_general}.

\begin{lemma}[Extended Gaussian Resolution]\label{lem:2path_red_ext_gauss}

    Given independent instances $\calY\up1 \sim \gtensp {k_1} n {\eps=1} {\defi_1}$ and $\calY\up2 \sim \gtensp {k_2} n {\eps=1} {\defi_2}$ that share a secret vector $x\in\set{\pm1}^n$ and have $\defi_1,\defi_2 \geq 0$, Alg.~\ref{alg:two_path_gen_tens} outputs 
    $\calZ = \GaussResolutionTwo\paren{\calY\up1, \calY\up2, \kstar}$, where
    $$\tv\paren{\calZ, \gtens \kstar {n'} {m'} {\delta'}} \to 0\quad \text{as }n\to\infty\,,$$ where $m' = \Thetat(n^{k'})$ and $\delta' = \Thetat(n^{-k'/4 - (\defi_1+\defi_2)/2})$, if $(k_1+k_2-\kstar)/2$ is an integer and $\kstar$ satisfies
    $$
    \kstar\in \sqb{|k_1-k_2|, \min\set{k_1,k_2}}\quad\text{and}\quad k' < \frac{k_1+k_2}{3}\,.
    $$
    Moreover, Alg.~\ref{alg:two_path_gen_tens} maps two instances with the same secret signal vector $x\in \set{\pm 1}^n$ to one with the signal vector $x' = x_{[1:n']}$ for $n' = n(1- o(1))$.
\end{lemma}
Remark~\ref{rmk:gaussian_resolution_alg_sampling_discr} similarly applies.

\begin{algorithm}[h]\SetAlgoLined\SetAlgoVlined\DontPrintSemicolon
    \KwIn{
    $\calY\up A = \set{ \b( \alpha_j\up A, \Y_{j}\up A\b)}_{j\in\sqb{M\up A}}$ for $A\in\set{1,2}$, where $\forall j \in\sqb{M\up A}$, $\b(\alpha_j\up A, \Y_{j}\up A\b) \in \binomset n {k_A} \times \R$, $k'\in \Z$, $\kappa \in (0,1)$.
    }
    \KwOut{$\widetilde\calZ = \set{ \b( \gamma_j, \Zt_{j}\b)}_{j\in\sqb{M'}}$, where $\forall j \in\sqb{M'}$, $\b(\gamma_j, \Zt_{j}\b) \in \binomset {\kappa n} {k'} \times \R$.}

    \BlankLine
    \tcp{Construct tensor views of the two inputs}
    $$
    \Yb_{\alpha}\up A = \Y_j\up A, \text{ where }\ j \in \sqb{M\up A}: \alpha_j\up A = \alpha \quad \text{for }A\in\set{1,2}\text{ and } \alpha \in \binomset n {k_A}
    $$

    \BlankLine
    Denote $a \coloneq (k_1+ k_2 - k')/2$ (assumed integer)
    
    \BlankLine
    \tcp{Flatten the tensors}
    $\hat Y\up A \gets \mat {a} { k_A - a} \paren{ \Yb\up A} \in \R^{\binomset n {a} \times \binomset n {k_A - a}}$ for $A\in\set{1,2}$

    \BlankLine
    \tcp{Only keep part of the entries}
    For $A\in\set{1,2}$ let $Y\up A \in \R^{\binom {(1-\kappa)n}{a} \times \binom {\kappa n}{k_A - a}}$:
    $$
        \forall \alpha \in \binom{\sqb{\kappa n+1,\dots,n}}{a}, \beta \in  \binom{\sqb{\kappa n}}{k_A - a}\,, \qquad Y_{\alpha,\beta}\up A \gets \hat Y_{\alpha,\beta}\up A
    $$

    \BlankLine 
    \tcp{Construct the Gram matrix}
    $$ Z \gets \binom {(1-\kappa)n}{a}^{-1/2} \paren{Y\up 1}^\top Y\up 2 $$

    \BlankLine
    \tcp{Arrange the entries of $\tenso\paren{Z}$ into an instance of $\ktensk {k'}^\wor$}
    $\Zb = \tenso\paren{Z}$, $\widetilde\calZ = \set{\paren{\gamma, \Zb_\gamma}}_{\gamma \in \binomset {\kappa n} {k'}}$
    
    \BlankLine
    Return $\ToPoiFromTPCA\b(\widetilde\calZ\b)$. \ \tcp{Alg.~\ref{alg:to_poi_sampling_from_pca}}
    
    \caption{$\GaussResolutionTwo\paren{\calY\up 1, \calY\up 2, k', \kappa = 1-\log^{-1} n}$}
\label{alg:two_path_gen_tens}
\end{algorithm}

\subsubsection{Connection to Spiked Covariance and Spiked Wigner Models}\label{subsubsec:sparse_pca}

To obtain this generalized $\EqComb$ two-to-one reduction,
we draw connections between $\ktens$ and the two models of high-dimensional statistics: Spiked Covariance Model and Spiked Wigner Model. We will use the average-case reductions between them developed in \cite{bresler2025computational}. While this reduction is similar to the one in \autoref{sec:to_tensor_pca}, we directly reduce to the result in \cite{bresler2025computational} to obtain sharper $k$ tradeoffs for a specific case of $\eps = 1$.

For $\eps=1$ we proved in Prop.~\ref{prop:sampling_equiv_eps1} that $\gtensp k n {\eps=1} \defi$ is computationally equivalent to $\ktens^\wor_{\eps=1, \defi}(n)$, i.e., there exist average-case reductions between the models both ways. In $\ktens^\wor_{\eps=1, \defi}(n)$, we observe $m = \binom n k$ $\ktens$ equations, exactly one for each index set $\alpha \in \binomset n k$. Then, in a tensor view, a $\ktens^\wor_{\eps=1, \defi}(n)$ instance $\calY = \set{\paren{\alpha_j, \Y_j}}_{j \in \sqb{m}}$ can be reshaped into a symmetric tensor $\Yb \in \paren{\R^n}^{\otimes k}$ satisfying for all $\alpha\in\binomset n k$
\begin{equation*}
    \Yb_{\alpha} = \Y_j, \text{ where }\ j \in \sqb{m}: \alpha_j = \alpha\,.
\end{equation*}
Here, since $\Yb$ is symmetric up to index permutations, we index its entries by the sets of indices in $\sqb{n}$. A flattening $Y = \mat {a} {k-a} \Yb \in \R^{\binom n {a}\times \binom n {k-a}}$ for some $a$ is an $M = \binom n a$ by $N = \binom n {k-a}$ matrix whose rows are indexed by elements of $\binomset n {a}$ and columns by $\binomset n {k-a}$. 

\begin{definition}[Spiked Covariance and Spiked Wigner Models]\label{def:spiked_cov_and_wigner}
    Under the \emph{Spiked Covariance Model} of Johnstone and Lu \cite{johnstoneSparse04}, one observes $M$ i.i.d. samples from $\N\paren{0,I_N + \theta u u^\top}$ for a unit vector $u\in \R^N$ and an SNR parameter $\theta$. Arranging the samples in the rows of matrix $Z \in \R^{M\times N}$, the model data has a convenient representation 
    $$Z = X + \sqrt{\theta} g u^\top, \qquad \text{where } X\sim \N(0,I_{M\times N})\text{ and } g\sim \N(0, I_M) \text{ are independent}\,.$$
    Under the \emph{Spiked Wigner Model}, one observes a matrix 
    $$
    Y = \lambda u u^\top + W, \qquad W\sim\GOE(N)\,,
    $$
    where $u$ is the unit signal vector and $\lambda$ is the SNR.
\end{definition}

With these definitions in mind, note that $Y = \mat a {k-a} \paren{\Yb}$ can be represented as 
$$
Y^\top = X + \sqrt{\theta} x_L x_R^\top, \qquad X\sim \N\paren{0, I_{M\times N}}\,,
$$
where $x_L = x^{\otimest a} \in \set{\pm 1}^{M}, x_R = N^{-1/2} x^{\otimest k-a} \in \set{\pm 1/\sqrt{N}}^{N}$ and $\theta = \delta^2 N$. With the exception of the distribution of $x_L$ (which corresponds to $g \sim \N(0, I_{M})$), this resembles the distribution of the Spiked Covariance Model with $M$ samples in dimension $N$ and the signal unit vector $x_R = N^{-1/2} x^{\otimest k-a}$. Moreover, the Wigner Model with the same signal vector, $$Y = \lambda x_R x_R^\top + W$$
resembles the flattening of the $(2k-2a)$-order tensor with the secret vector $x$. This connection is the inspiration for this reduction that closely mirrors the reduction of \cite{bresler2025computational} from Spiked Covariance to Spiked Wigner Model.

\subsubsection{Proof of Lemma~\ref{lem:2path_red_ext_gauss}}\label{subsubsec:proof_two_tens_gen}

We show that Alg.~\ref{alg:two_path_gen_tens} is an average-case reduction from $\calY\up 1 \sim \ktensk {k_1}^\wor_{\eps=1, \defi_1}(n)$ and $\calY \up 2 \sim \ktensk {k_2}^\wor_{\eps=1, \defi_2}(n)$, that share a secret vector $x\in\set{\pm1}^n$, to $\gtensp {k'} n {\eps=1} {\defi_1+\defi_2}$ and the Lemma~\ref{lem:2path_red_ext_gauss} follows from the sampling equivalence in Prop.~\ref{prop:sampling_equiv_eps1}. In particular, under the conditions of the lemma, we prove that the output $\calZ = \GaussResolutionTwo\paren{\calY\up1, \calY\up2, \kstar}$ satisfies
    \begin{equation}\label{eq:tv_two_tens}
        \tv\paren{\calZ, \gtens \kstar {n'} {m'} {\delta'}} \to 0\quad \text{as }n\to\infty\,,
    \end{equation}
    where $m' = \Thetat(n^{k'})$ and $\delta' = \Thetat(n^{-k'/4 - \defi_1/2 - \defi_2/2})$, which shows the desired average-case reduction by Prop.~\ref{prop:map_up_to_constant}. 

The proof of Eq.~\eqref{eq:tv_two_tens} closely follows that of Theorem 8.2 in \cite{bresler2025computational}. 
From the definition of $\ktens^\wor$ in \autoref{subsubsec:equiv_w_wor_replacement}, we can express 
\begin{align*}
    \forall \alpha \in \binom{\sqb{n}}{k_1}: &\qquad\Yb_{\alpha}\up 1 = \delta\up 1 \cdot x_{\alpha} + \noiseG_{\alpha}\up 1\,,\\
    \forall \alpha \in \binom{\sqb{n}}{k_2}:&\qquad \Yb_{\alpha}\up 2 = \delta\up 2 \cdot x_{\alpha} + \noiseG_{\alpha}\up 2\,,
\end{align*}
where $\delta\up 1 = \Theta\paren{n^{-k_1/4-\defi_1/2}}$ and $\delta\up 2 = \Theta \paren{n^{-k_2/4-\defi_2/2}}$. The flattenings $Y\up1, Y\up2$ in Lines 3-4 of Alg.~\ref{alg:two_path_gen_tens} then satisfy for all $(\alpha, \beta) \in \binomset {(1-\kappa)n}{a} \times \binomset {\kappa n}{k_1 - a}$, $$Y_{\alpha, \beta}\up 1 = \delta\up 1 \cdot x_{\alpha, \beta} + \noiseG_{\alpha, \beta}\up 1$$
and for all $(\alpha, \beta) \in \binomset {(1-\kappa)n}{a} \times \binomset {\kappa n}{k_2-a}$,
$$Y_{\alpha, \beta}\up 2 = \delta\up 2 \cdot x_{\alpha, \beta} + \noiseG_{\alpha, \beta}\up 2\,,$$ where we index the rows and columns of these matrices by the corresponding index sets.  Denote $N_1 = \binom {\kappa n}{k_1 - a}, N_2 = \binom {\kappa n}{k_2-a}, M = \binom {(1-\kappa)n}{a}$ and $g = \veco\paren{\set{x_{\alpha}}_{\alpha \in \binomset {(1-\kappa)n}{a}}}$, $u\up 1 = N_1^{-1/2}\veco\paren{\set{x_{\beta}}_{\beta \in \binomset {\kappa n}{k_1 - a}}}$, and $u\up 2 = N_2^{-1/2}\veco\paren{\set{x_{\beta}}_{\beta \in \binomset {\kappa n}{k_2-a}}}$, we can express 
$$Y\up 1 = \sqrt{\theta\up 1} g \paren{u\up 1}^\top + X\up 1, \qquad \text{and} \qquad Y\up 2 = \sqrt{\theta\up 2} g \paren{u\up 2}^\top + X\up 2\,,$$
where $X\up 1,X\up 2$ are independent i.i.d. $\N(0,1)$ matrices and $\theta\up1 = \tilde\Theta\paren{n^{k_1/2 - a-\defi_1}} = \tilde\Theta\paren{\sqrt{N_1/M}n^{-\defi_1}}$, $\theta\up 2 = \tilde\Theta\paren{n^{k_2/2 - a-\defi_2}} = \tilde\Theta\paren{\sqrt{N_2/M}n^{-\defi_2}}$. This brings us into the setting of the Spiked Covariance Model (Def.~\ref{def:spiked_cov_and_wigner}), so we apply the tools developed for Theorem 8.2 of \cite{bresler2025computational} directly here in the following claim. 
\begin{claim}\label{claim:assym_wish}
    In the setup above, if $M \gg N_1 N_2$, for $W \in \R^{N_1 \times N_2}$ with i.i.d. $\N(0,1)$ entries,
    $$
    \tv\paren{M^{-1/2} \paren{Y\up 1}^\top Y\up 2, \lambda u\up 1 \paren{u\up 2}^\top + W} \to 0 \qquad \text{as } N,M \to \infty\,,
    $$
    where $\lambda = \sqrt{\theta\up 1 \theta \up 2 M}$.
\end{claim}
\begin{proof}
    The proof (for both planted and null distributions) follows that of Theorem 8.2 in \cite{bresler2025computational} with the following modifications:
    \begin{enumerate}
        \item We skip the symmetrization step of the original Theorem 8.2.
        \item The key proposition of \cite{brennan2021finetti} that shows Claim~\ref{claim:assym_wish} in the case of $Y\up 1, Y\up 2$ having i.i.d. $\N(0,1)$ entries is stated only for the case of $N_1 = N_2$. As proved in the source paper \cite{brennan2021finetti} (Corollary 4.3), the same result holds for any dimensions $N_1, N_2, M$ satisfying $M \gg N_1 N_2$.
        \item The standard in the field argument for removing random cross-terms will succeed whenever both $\paren{\theta\up 1}^2 N_2$ and $\paren{\theta\up 2}^2 N_1$ are $o(1)$ which again holds for $M \gg N_1 N_2$ and $\theta\up1, \theta\up2$ as in Lemma~\ref{lem:2path_red_ext_gauss}.
        \item Again we allow one of the matrices to be a vector. The proof is not altered.
    \end{enumerate}
\end{proof}

Since $a = (k_1+k_2 - k')/2 > (k_1+k_2)/3$, the condition $M \gg N_1 N_2$ holds, and therefore, Claim~\ref{claim:assym_wish} shows that $Z \sim \lambda u\up 1 \paren{u\up 2}^\top + W$ for $\lambda = \sqrt{\theta\up 1 \theta \up 2 M}$. Therefore $\Zb = \tenso\paren{Z}$ satisfies for all $\alpha \in \binom{\sqb{\kappa n}}{k_1+k_2-2a}$,
$$
\Zt_{\alpha} = \sqrt{\theta\up 1 \theta \up 2 M} N_1^{-1/2}N_2^{-1/2} \cdot x_{\alpha} + W_{\alpha}\,,
$$
where $W_{\alpha}$ are i.i.d. $\N(0,1)$. Finally note that $\sqrt{\theta\up 1 \theta \up 2} M N_1^{-1/2}N_2^{-1/2} = \Theta\paren{n^{-(k_1 + k_2 - 2a)/4}-(\defi_1+\defi_2)/2}$, which is the appropriate signal level for $\gtensp {(k_1+k_2-2a)}{\kappa n} {\eps=1} {\defi_1 + \defi_2}$. This concludes the proof of Lemma~\ref{lem:2path_red_ext_gauss}.

\section{Densifying Reduction: Growing $\eps$}\label{sec:resolution_red}

In this section, we build average-case reductions that increase the $\kxor$ density parameter $\eps \geq 0$. The discrete resolution reduction in \autoref{sec:discrete_resolution} gives a densifying reduction
\begin{equation*}
\text{from}\qquad\kxor\text{ at density }\eps\qquad\text{to}\qquad \kxork {k'}\text{ at density }\eps' = \min\set{1, \frac{2k\eps}{k'}}\,,
\end{equation*}
As discussed in \autoref{subsec:sparse_resolution_discrete}, for a fixed starting density $\eps$, the hardness implications of this reduction are limited by the parameter tradeoff above. In particular, even allowing arbitrary input order $k$ and applying a natural density-adjusting post-processing is not sufficient to obtain Corollary~\ref{cor:sparse_dense_red_conj}, which essentially shows 
$$\conjeps \eps \Rightarrow \conjeps {\eps'} \quad\text{for any}\quad \eps>0\quad\text{and}\quad \eps' \in \b(2\eps/(1-\eps),1\b]\,.$$
The goal of this section is to prove Corollary~\ref{cor:sparse_dense_red_conj}. To do so, we use the resolution reduction above as a subroutine together with a careful integer-approximation step for selecting an appropriate starting order $k$.

\subsection{Result Overview}

Our average-case reductions are in Theorem~\ref{thm:sparse_dense_red} and their immediate implications for hardness conjectures (Conj.~\ref{conj:hardness}) are in Corollary~\ref{cor:sparse_dense_red_conj}.


\begin{theorem}[Densifying Reduction: Growing $\eps$]\label{thm:sparse_dense_red} 
    Let $\kstar \in \Z^{+}, \eps, \eps' \in (0,1]$, and $\defi \geq 0$ be constants, s.t. $$\eps' > 2\eps/(1-\eps)\quad\text{or}\quad\eps<\eps'=1\,.$$ 
    There exists sufficiently large $k=k(\eps, \eps', \defi, \kstar) \in \Z^{+}$ and an average-case reduction\footnote{As defined in Def.~\ref{def:avg_case_points}.} for both detection and recovery (Alg.~\ref{alg:sparse_dense}) 
    $$
    \text{from}\quad \xorp {k} n \eps {\defi}\quad\text{to}\quad\xorp {\kstar} {n'} {\eps'} {\defi'}\,,
    $$ 
    where $\defi' \leq 3\defi$ and $n' = \tilde n^A$ for a known constant $A = A\paren{\eps, \eps', \defi, \kstar} \in \mathbb{Z}^{+}$ and $\tilde n = n(1-o(1))$. Moreover, Alg.~\ref{alg:sparse_dense} maps an instance with a secret signal vector $x\in\set{\pm 1}^n$ to one with a secret vector $x' = \paren{x_{\sqb{1:\tilde n}}}^{\otimes A}\odot y$ for a known $y\sim\unif\paren{\set{\pm1}^{n'}}$\footnote{See Sec.~\ref{subsec:recovery_assumptions} on why this transfers weak/strong/exact recovery algorithms.}. 
\end{theorem}
\begin{corollary}[Densifying Reduction (Growing $\eps$), Corollary of Theorem~\ref{thm:sparse_dense_red}]\label{cor:sparse_dense_red_conj}
    For any $\eps, \eps' \in (0,1]$ satisfying $\eps' > 2\eps/(1-\eps)$ or $\eps < \eps'=1$ 
    and every $k' \geq 2$, there exists some $k \geq k'$ such that 
    $$\conj k \eps \Rightarrow \conj {k'} {\eps'}\,.$$
    In particular, for $\eps,\eps'$ as above, 
    $$\conjeps \eps \Rightarrow \conjeps {\eps'}\,.$$
\end{corollary}

\paragraph{Theorem~\ref{thm:sparse_dense_red} proof overview.}
The key reduction subroutine of Alg.~\ref{alg:sparse_dense} is $\TwoPath$ (Alg.~\ref{alg:two_path}), for which the following guarantee is proved in \autoref{sec:discrete_resolution}.

\begin{mylem}{\ref{lem:2path_red}}\textnormal{(Discrete Resolution Reduction)}
    Fix $k,n \in \Z^{+}, \eps \in [0,1], \defi \in \R$. For any even integer $\kstar$, s.t. 
    $$\kstar \leq k\paren{1-\eps}\,,$$
    there exists an average-case reduction for both detection and recovery (Alg.~\ref{alg:two_path}) $$\text{from}\quad\xorp k n \eps \defi\quad\text{to}\quad\xorp {\kstar} {n'} {\eps'} {2\defi}\,,$$ where $\eps' = \min\set{\frac{2k}{\kstar}\eps, 1}$. Moreover, Alg.~\ref{alg:two_path} maps an instance with a secret signal vector $x\in \set{\pm 1}^n$ to one with the signal vector $x' = x_{[1:n']}$ for $n' = n(1- o(1))$.\footnote{The algorithm can set $x'$ to have any $n'$-coordinate subset of $x$; the algorithm succeeds for any $n' = n\paren{1 - (\poly\log n)^{-1}}$.}
\end{mylem}

For the case of $\kstar \geq 2k\eps$, the reduction of Lemma~\ref{lem:2path_red} shows a map from an instance with parameter tuple $\paren{k, \eps, \defi}$ to an instance with parameter tuple $\paren{\kstar, \eps' = \frac{2k}{\kstar}\eps, 2\defi}$ for an appropriate range of even $\kstar$. In particular, given an instance of a $\xorp k n \eps \defi$, one could map to an instance with 
$$\eps' \in \set{ \frac{2k\eps}{ k_{\min}}, \frac{2k\eps}{ k_{\min}+2},\dots,\frac{2k\eps}{k_{\max} } } \subset \B( \frac{2\eps}{1-\eps}, 1\B]\,,$$
where $$k_{\min} = \min_{\kstar\in \mathbb{Z}} \set{\kstar \text{ even}: \kstar \geq 2k\eps}\quad\text{and}\quad k_{\max} = \max_{\kstar\in \mathbb{Z}} \set{\kstar \text{ even}: \kstar < k(1-\eps)}\,.$$ 
The main technical challenge of Theorem~\ref{thm:sparse_dense_red} is to apply the reduction of Lemma~\ref{lem:2path_red} with an appropriate parameter $k$, and map to \emph{any} $\eps' \in \B( \frac{2\eps}{1-\eps}, 1\B]$.
This requires an intricate integer approximation step presented in \autoref{sec:proof_of_sparse_dense}. The guarantees for subroutines $\DenseRed, \SparseRed$ (Alg.~\ref{alg:dense_red}, \ref{alg:sparse_red}) are in \autoref{subsec:entry_adjust} and for $\ReduceK$ (Alg.~\ref{alg:prelim_decrease_k}) in \autoref{subsec:decrease_k_prelim}.

\begin{algorithm}\SetAlgoLined\SetAlgoVlined\DontPrintSemicolon

    \KwIn{$\calY = \set{ \b( \alpha_j, \Y_{j}\b)}_{j\in\sqb{m}}$, where $\forall j \in\sqb{m}$, $\b(\alpha_j, \Y_{j}\b) \in \set{n}^{k} \times \set{\pm1}$, $a, \kstar \in \Z, \eps, \eps' \in [0,1]$
        }
    \KwOut{$\calZ = \set{ \b( \alpha_j', \Zt_{j}\b)}_{j\in\sqb{m'}}$, where $\forall j \in\sqb{m'}$, $\b(\alpha_j', \Zt_{j}\b) \in \set{n'}^{k'} \times \set{\pm1}$.
        }

    \BlankLine 
    $\hat \eps \gets \min\set{k \eps / (k-a), 1}$

    \BlankLine
    $\calY\up1, \dots,\calY\up {k-a} \gets \XorSplit\paren{\calY, k-a}$ \ \tcp{\autoref{subsubsec:splitting}, see Rmk.~\ref{rmk:clone_several}}
    
    \ForAll{$i \in \sqb{k-a}$}{
        $\calY^{\hat \eps}_{2i} \gets \TwoPath\paren{\calY\up i, 2i}$ \ \tcp{Alg.~\ref{alg:two_path}}
    }

    \BlankLine
    \tcp{Combine the $k-a$ instances into one full instance with noise parameter $\hat\eps$}
    $\calY^{\hat\eps} \gets \Compose\paren{\set{\calY^{\hat \eps}_{2i}}_{i\in\sqb{k-a}}, \hat\eps}$ \ \tcp{Alg.~\ref{alg:compose_full}, see Rmk.~\ref{rmk:to_full_discrete_input}}

    \BlankLine
    \tcp{Adjust the density parameter $\eps$}
    \If {$\hat \eps \leq \eps'$}
    {
        $\calY^{\eps'} \gets \DenseRed\paren{\calY^{\hat \eps}, n^{k(\hat\eps+1)/2}, n^{k(\eps'+1)/2}}$ \ \tcp{Alg.~\ref{alg:dense_red}}
    }\Else{
        $\calY^{\eps'} \gets \SparseRed\paren{\calY^{\hat \eps}, n^{k(\hat\eps+1)/2}, n^{k(\eps'+1)/2}}$ \ \tcp{Alg.~\ref{alg:sparse_red}}
    }

    \BlankLine
    \tcp{Map to the desired order $k'$}
    $\calZ \gets \ReduceK\paren{\calY^{\eps'}, (2k-2a)/\kstar}$ \ \tcp{Alg.~\ref{alg:prelim_decrease_k}}

    \BlankLine
    Return $\calZ$.
    
    \caption{$\SparseToDense\paren{\calY, a, \kstar, \eps, \eps'}$}
\label{alg:sparse_dense}
\end{algorithm}

\subsection{Proof of Theorem~\ref{thm:sparse_dense_red}}\label{sec:proof_of_sparse_dense}

For a given target order $\kstar$, densities $\eps, \eps' \in \sqb{0,1}$, and deficiency $\defi > 0$, the reduction algorithm for Theorem~\ref{thm:sparse_dense_red} is Algorithm~\ref{alg:sparse_dense} with carefully chosen parameters $a$ and $k$. Alg.~\ref{alg:sparse_dense} is a sequential composition of several average-case reductions and in this proof we verify that every reduction call is (1) valid and (2) maps the $\kxor$ parameters appropriately.

\subsubsection{Preliminary Lemmas}

In addition to Lemma~\ref{lem:2path_red}, which shows a guarantee on the $\TwoPath$ (Alg.~\ref{alg:two_path}) subroutine, we restate the guarantees for the remaining parts of Alg.~\ref{alg:sparse_dense} here for convenience. 

The guarantee in Cor.~\ref{cor:m_adjust} for $\DenseRed$, $\SparseRed$ and Cor.~\ref{cor:prelim_decrease_k} for $\ReduceK$ are stated for all $\kxor$ variants in \autoref{subsec:entry_adjust} and \autoref{subsec:decrease_k_prelim}; we restate these for $\kxorfull$.
\begin{mycor}{\ref{cor:m_adjust}}\textnormal{(Guarantees for $\DenseRed$ (Alg.~\ref{alg:dense_red}) and $\SparseRed$ (Alg.~\ref{alg:sparse_red}))}
Let $k,n \in \Z^{+}, \eps, \eps' \in [0,1], \defi \in \R^{+}$ be the $\kxor$ parameters. There is an average-case reduction for both detection and recovery ($\DenseRed$ if $\eps \leq \eps'$ and $\SparseRed$ otherwise) 
    \begin{align*}\text{from}\quad\xorpfull k n \eps \defi\quad\text{to}\quad\xorpfull {k} n {\eps'} {\defi'}\,, \quad \text{where } \defi' = \defi + {k|\eps-\eps'|/2}\,.
    \end{align*}
    Moreover, Alg.~\ref{alg:dense_red} and Alg.~\ref{alg:sparse_red} map an instance with a secret signal vector $x\in\set{\pm1}^n$ to one with the same signal vector $x$.
\end{mycor}

\begin{mycor}{\ref{cor:prelim_decrease_k}}\textnormal{(of Lemma~\ref{lem:prelim_decrease_k}: Guarantees of $\ReduceK$)}
    Let $k,n \in \Z^{+}, \eps \in [0,1], \defi \in \R^{+}$ be the $\kxor$ parameters. For every $a \in \Z^{+}, a | k$ there is an average-case reduction for both detection and recovery $$\text{from}\quad\xorpfull k n \eps \defi\quad\text{to}\quad\xorp {k'} {n'} {\eps} {\defi/a}\,,\quad n' = n^a\,,$$
    for all $k' = k/a, k/a-2, k/a-4,\dots$.
    Moreover, the reduction maps an instance with a secret signal vector $x\in \set{\pm 1}^n$ to one with a signal vector $x' = x^{\otimes a}\odot y$ for a known $y\sim \unif\paren{\set{\pm 1}^{n'}}$. From the discrete-Gaussian equivalence (Prop.~\ref{prop:discr_gauss_equivalence_general}), the analogous results hold for $\ktens$ models. 
\end{mycor}

Prop.~\ref{prop:comp_equiv_full_collection} relating $\kxor^\FULL$ and a collection of low-order $\kxor$ is stated for a general model variant in \autoref{subsubsec:full_asymm_kxor_red} and proved in \autoref{subsec:repeated_indexes}; we restate one of the directions for the standard $\kxor$ here. 
\begin{myprop}{\ref{prop:comp_equiv_full_collection}}\textnormal{(Reduction from a Collection of $\kxor$ to $\kxorfull$: Guarantees for $\Compose$)}
    Fix $k\in\mathbb{Z}, \eps\in\sqb{0,1}, \defi\in \R_{+}$. There exists a poly-time average-case reduction\footnote{As defined in def.~\ref{def:avg_case_points}} (Alg.~\ref{alg:compose_full})
$$
\text{from}\quad \B\{ \xorp {(k-2r)} n {\eps_r} \defi \B\}_{0\leq r \leq \frac{k-1}{2}} \quad\text{to}\quad \xorpfull k n {\eps} \defi\qquad \text{with } \eps_r = \min\set{ \frac{k\eps}{k-2r} , 1}\,,
$$
where the input instances share the same signal vector $x\in\set{\pm1}^n$, which is preserved exactly in the reduction.
\end{myprop}

\subsubsection{Special Case: $0 < \eps < \eps' = 1$}

For any $\kstar, \defi$ and $0 < \eps < \eps' = 1$, it is sufficient to choose any $k > \max\set{\frac \kstar \eps, \frac{2\kstar}{1-\eps}}$ and $a = k - \kstar$ as inputs to Alg.~\ref{alg:sparse_dense}. In particular, let $k, a$ be as above and $\calY \sim \xorp k n \eps \defi$ be the input instance to Alg.~\ref{alg:sparse_dense} of order $k$. We show that the map $$\calY \to \calZ = \SparseToDense\paren{\calY, a, \kstar, \eps, \eps'}$$
is an average-case reduction for both detection and recovery to $\xorp {\kstar} {n'} {\eps'} {\defi'}$, where parameters $n', \defi'$ satisfy the requirements of Theorem~\ref{thm:sparse_dense_red}.

By the guarantees for $\XorSplit$ in Lemma~\ref{lem:instance_split}, guarantees for $\TwoPath$ in Lemma~\ref{lem:2path_red}, and guarantees for $\Compose$ in Prop.~\ref{prop:comp_equiv_full_collection}, Lines 1-5 of Alg.~\ref{alg:sparse_dense} make an average-case reduction for both detection and recovery $$\text{from}\quad\xorp k n \eps \defi\quad\text{to}\quad\xorpfull {2\kstar} {\tilde n} {\hat \eps} {2\defi}\,,$$ where $\tilde n = n(1-o(1))$ and $\hat\eps = \min\set{\frac{2k\eps}{2k'}, 1} = 1$. In this case, Lines 6-9 of Alg.~\ref{alg:sparse_dense} keep the instance $\calY^{\hat \eps}$ unchanged. By guarantees for $\ReduceK$ in Lem.~\ref{lem:prelim_decrease_k}, the final step of the algorithm is an average-case reduction 
$$\text{from}\quad\xorpfull {2\kstar} {\tilde n} {\hat \eps=1} {2\defi}\quad\text{to}\quad \xorpfull {\kstar} {\tilde n^2} {\hat \eps=1} {\defi}\,,$$
which concludes the proof.

\subsubsection{Proof for $0 < \eps < \eps' < 1$}

For any fixed parameters $\kstar, \eps, \eps', \defi$, we identify parameters $k, a$ such that Algorithm~\ref{alg:sparse_dense} is an average-case reduction 
$$
    \text{from}\quad \xorp {k} n \eps {\defi}\quad\text{to}\quad\xorp {\kstar} {n'} {\eps'} {\defi'}\,,
$$ 
where $n', \defi'$ satisfy the conditions of Theorem~\ref{thm:sparse_dense_red}. The key technical challenge of the following proof is the appropriate selection of the input parameters $k, a$.

    \paragraph{Step 1: mapping $\eps\to \hat \eps$ in Lines 1-5.} By the guarantees for $\XorSplit$ in Lemma~\ref{lem:instance_split}, guarantees for $\TwoPath$ in Lemma~\ref{lem:2path_red}, and guarantees for $\Compose$ in Prop.~\ref{prop:comp_equiv_full_collection}, Lines 1-5 of Alg.~\ref{alg:sparse_dense} make an average-case reduction for both detection and recovery $$\text{from}\quad\xorp k n \eps \defi\quad\text{to}\quad\xorpfull {2(k-a)} {\tilde n} {\hat \eps} {2\defi}\,,$$
    where $\tilde n = n(1-o(1))$ and $\hat\eps = \min\set{\frac{2k\eps}{2(k-a)}, 1}$ if the conditions of Lemma~\ref{lem:2path_red} are satisfied. Namely, the reduction succeeds for any $k$ and $a \in \b[\frac{k(1+\eps)}{2}, k\b)$ an integer. Set 
    \begin{equation}\label{eq:choice_a}
    a \coloneq \lfloor k(1 - \eps / \eps') \rceil\,.\end{equation}
     where recall that $\lfloor \cdot \rceil$ denotes rounding to the nearest integer and $k$ is a constant to be chosen later.
    Note that, since $\eps' > \frac {2\eps}{1-\eps}$, for sufficiently large constant $k$, $k > a \geq k\paren{1 - \frac{\eps} {\eps'}} - 1\geq \frac{k(1+\eps)}{2}$, i.e., the choice of $a$ in Eq.~\eqref{eq:choice_a} satisfies the conditions of Lemma~\ref{lem:2path_red}. Denote $$\psi \coloneq a - k(1 - \eps / \eps') = \begin{cases}
        \set{k\eps/\eps'}, &\set{k\eps/\eps'} \leq 1/2\\
        \set{k\eps/\eps'}-1, &\set{k\eps/\eps'} > 1/2
    \end{cases}\,,$$
    where $\lfloor \cdot \rfloor, \set{\cdot}$ are the integer and fractional parts of a number. Then, for sufficiently large $k \geq \frac{\eps'}{\eps(1-\eps')}$, $\frac{k\eps}{k-a} \leq 1$, so 
    \begin{align*}
        \hat \eps = \frac{k\eps}{k-a} = \frac{k\eps}{k - k(1 - \eps / \eps')  - \psi} = \frac{k\eps}{k\eps / \eps' - \psi} = \frac{k\eps\eps'}{k\eps - \psi \eps'} = \eps' + \frac{\psi (\eps')^2}{k\eps - \psi \eps'} =: \eps' + \Delta\,.
    \end{align*}
    Our goal is to map to an instance with the sparsity parameter $\eps'$, so we want $\hat \eps \in \sqb{\eps'\pm \xi}$, i.e. $\Delta \in [\pm\xi]$, for a small parameter $\xi>0$ to be chosen later. In order for $|\Delta| \leq \xi$, it is sufficient to take $k$ such that 
    \begin{equation}\label{eq:k_condition}
        \frac k {|\psi|} > \max \set{\frac {\eps'} \eps, \frac {\eps' (\eps' + \xi)}{\eps \xi}} = \frac {\eps' (\eps' + \xi)}{\eps \xi}\,.
    \end{equation}
    The following Claim shows existence of $k$ satisfying \eqref{eq:k_condition} and additional properties to be used later. It is proved in \autoref{subsec:proof_claim_conv}.
    \begin{claim}\label{claim:conv}
        Fix $c, r \in (0,1), \kstar\in \Z^{+}$. There exists a constant $A(c, r, \kstar)$, such that for every $A \geq A(c, r, \kstar)$, there exists $k \in \Z^{+}$ such that $k \leq c A$, $\lfloor k r \rceil$ is a multiple of $\kstar$, and 
        \begin{equation}\label{eq:approx_bound_0}
            \frac{k}{\|kr\|} > A\,,
        \end{equation}
        where $\|x\| = \b| x - \lfloor x \rceil \b|$.

    \end{claim}

    We now use Claim~\ref{claim:conv} to find a good $k$ that satisfies \eqref{eq:k_condition}. Let $c = \paren{2\eps'}^{-1} \defi$ be a constant and apply Claim~\ref{claim:conv} for $c, r = \eps/\eps', \kstar$. This demonstrates that for $\eps' < 1$, there exists a constant $A\paren{c, \eps/\eps', \kstar}$ such that for any $A > A\paren{c, \eps/\eps', \kstar}$ the result of the claim holds.
    In particular, there exists $k \leq c A$ such that $\lfloor k \eps / \eps' \rceil$ is a multiple of $\kstar$ and $$\frac{k}{\|k \eps/\eps'\|} > A = \frac{\eps' (\eps' + \xi)}{\eps \xi}\,,$$
    for $\xi = \xi(A)$ chosen appropriately. Notice that $\|k \eps/\eps'\| = |k \eps/\eps' - \lfloor k \eps/\eps' \rceil| = |\psi|$. As described above, this implies that $\hat{\eps} \in [\eps' \pm \xi(A)]$ for this choice of $k = k(A)$. 
    
    Denote $\xi(c, \eps/\eps', \kstar) = \xi\paren{A(c, \eps/\eps', \kstar)}$, and notice that as $A\to \infty$, $\xi(A) \to 0$. Then, for any constant $c$ and any $\xi \leq \xi(c, \eps/\eps', \kstar)$, we have found $k \in \Z^{+}$ such that:
    \begin{enumerate}
        \item $k \leq c A = c \cdot \frac{\eps' (\eps' + \xi)}{\eps \xi}$;
        \item $\lfloor k \eps / \eps' \rceil = k - \lfloor k (1 - \eps / \eps') \rceil = k-a$ is a multiple of $\kstar$;
        \item lines 1-5 of Alg.~\ref{alg:sparse_dense} make an average-case reduction from $\xorp k n \eps \defi$ to $\xorpfull {2(k-a)} {\tilde n} {\hat \eps} {2\defi}$ for $\hat \eps \in \sqb{\eps' \pm \xi}$ and $a = \lfloor k(1-\eps / \eps')\rceil$.
    \end{enumerate}

    \paragraph{Step 2: mapping $\hat\eps \to \eps'$ in Lines 6-9.} By the guarantees for $\DenseRed, \SparseRed$ in Cor.~\ref{cor:m_adjust}, Lines 6-9 make an average-case reduction 
    $$
    \text{from}\quad \xorpfull {2(k-a)} {\tilde n} {\hat \eps} {2\defi} \quad\text{to}\quad \xorpfull {2(k-a)} {\tilde n} {\eps'} {\defi'}\,,
    $$
    where 
    \begin{align*}
    \defi' = 2\defi + (2k-2a) |\hat \eps - \eps'|/2 &\leq 2\defi + {(k-a)\xi} \\
    &\leq 2\defi+ \paren{k \eps / \eps'+1} \xi \\
    &\leq 2\defi + \xi + c \cdot \frac{\eps'(\eps'+\xi)}{\eps \xi}\cdot \frac{\eps} {\eps'} \cdot \xi \\
    &= 2\defi + {c(\eps' + \xi)+\xi}\\
    &\leq 2\defi + {2c \eps'} \leq 3\defi\,,
    \end{align*}
    for sufficiently small $\xi$ and $c = \paren{2\eps'}^{-1} \defi$.

    \paragraph{Step 3: mapping $(2k-2a) \to \kstar$ in Line 10.} The final step is the $\ReduceK$ reduction that, by guarantees in Lem.~\ref{lem:prelim_decrease_k}, is an average-case reduction 
    $$
    \text{from}\quad \xorpfull {2(k-a)} {\tilde n} {\eps'} {\defi'} \quad\text{to}\quad \xorp {\kstar} {\paren{\tilde n}^{2(k-a)/\kstar}} {\eps'} {\defi''}\,,
    $$
    where $\defi'' = \defi' \kstar / (k-a)/ 2 \leq \defi'$. This is a valid reduction, since Claim~\ref{claim:conv} guarantees that $(k-a)$ is a multiple of $\kstar$. This concludes the proof of Theorem~\ref{thm:sparse_dense_red}.
\qed


\section{Reductions to Tensor PCA}\label{sec:to_tensor_pca}

In this section we summarize the reductions that map a $\kxor$ or $\ktens$ instance with an arbitrary density $\eps\in (0,1]$ to a Tensor PCA instance, i.e. an instance of $\kxor$ or $\ktens$ with $\eps = 1$. We compare two approaches: Discrete Resolution ($\DiscrEqComb$, \autoref{sec:discrete_resolution}) and Gaussian Resolution ($\GaussEqComb$, \autoref{sec:gauss_resolution}). Both methods apply to both discrete $\kxor$ and Gaussian $\ktens$ settings, possibly after a preprocessing step that switches between distributions (\autoref{subsec:gaussianize}, \autoref{subsec:discretize}). For example, to apply $\DiscrEqComb$ to a $\ktens$ instance, one can first discretize the input (Alg.~\ref{alg:discretization}) the input, then apply the $\DiscrEqComb$ reduction, and finally gaussianize the output (Alg.~\ref{alg:gaussianization}) to obtain an instance of $\ktensk {k'}$. For simplicity, we state results for $\ktens$ here; the analogous results for $\kxor$ follow similarly.

\paragraph{Tensor PCA model variants.} In our Tensor PCA model $\gtensp k n {\eps = 1} {\defi}$ (Def.~\ref{def:kxor_family}) we generate $m = n^{k}$ noisy equations of the form $$\Y = \delta x_{\alpha} + \N(0,1)$$ for independent $\alpha\sim\unif\sqb{\binomset n k}$ and $\delta = n^{-k\eps/4 - \defi/2}$. 

A closely related Tensor PCA model is $\ktens^\wor_{\eps=1, \defi} (n)$ (see \autoref{subsec:kxor_variants}), which is computationally equivalent to $\gtensp k n {\eps = 1} {\defi}$ (Prop.~\ref{prop:sampling_equiv_eps1}), and therefore inherits all results below. The difference is purely in sampling procedure: $\ktens^\wor_{\eps=1, \defi} (n)$ reveals exactly $M = \binom n k$ observations -- one for each index subset $\alpha \in \binomset n k$, so it contains all distinct-index entries of 
\begin{equation}\label{eq:tpca_sec11}
\delta\cdot x^{\otimes k} + \noiseG\,,\tag{TensorPCA}
\end{equation}
where $\noiseG$ has i.i.d. $\N(0,1)$ entries. 

Finally, \autoref{sec:intro} defined $\tensorpca k n \delta$ model (equiv. $\tensorpcap k n \defi$), which reveals \emph{all} entries in Eq.~\eqref{eq:tpca_sec11}, including those with repeated indices. In \autoref{subsubsec:full_asymm_kxor_red} we show that $\tensorpcap k n \defi$ is computationally equivalent to a collection of independent $\ktens^\wor_{\eps=1, \defi} (n)$ instances of orders $k, k-2, k-4, \dots$ (Prop.~\ref{prop:comp_equiv_full_collection}), and our reductions can be adapted to output this collection, see \autoref{subsec:full_ideas} for intuition. 

Consequently, any reduction stated below to $\gtensp {k'} {n'} {\eps'=1} {\defi'}$ immediately yields analogous reduction to $\ktensk {k'}^\wor_{\eps'=1, \defi'} (n')$ or $\tensorpcap {k'} {n'} {\defi'}$, without changing the parameters.

\paragraph{Reductions to Tensor PCA.}
The $\DiscrEqComb$ reduction in \autoref{sec:discrete_resolution} as a special case builds a reduction to Tensor PCA in the following regime. (Lemma~\ref{lem:2path_red} in \autoref{sec:discrete_resolution} is stated for the discrete $\kxor$, we restate here the analogous result for $\ktens$, which holds by discrete-Gaussian equivalence in Prop.~\ref{prop:discr_gauss_equivalence_general}).
\begin{corollary}[of Lemma~\ref{lem:2path_red}: Discrete Resolution Reduction]\label{cor:tensor_from_resolution}
    Let $k \in \mathbb{Z}^{+}$, $\eps \in (0, 1]$, and $\defi \in \R$ be $\ktens$ parameters. If $k'$ is even and satisfies 
    \begin{equation}\label{eq:tensor_from_resolution_condition}
        k' \leq \min\set{ k\paren{1-\eps}, 2k\eps}\,,
    \end{equation}
    then there exists a poly-time average-case reduction for both detection and recovery 
    $$
    \text{from}\quad\gtensp k n \eps \defi \quad\text{to}\quad \gtensp {k'} {n} {\eps'=1} {2\defi}\,.
    $$
    Consequently,\footnote{Since the reduction in Lemma~\ref{lem:2path_red} only maps to even $k'$, an additional step of $k' \to k'/2$ reduction in \autoref{subsec:decrease_k_prelim} is required; it is applicable since reduction is Lemma~\ref{lem:2path_red} can be modified to map to a full $\ktensk {k'}$ instance (see \autoref{subsec:kxor_variants}). Formally, this can be achieved by first splitting the input into $k'/2$ independent instances via $\kTensSplit$ and then applying the reduction in Lemma~\ref{lem:2path_red} to every copy, mapping to instances of orders $k', k'-2,k'-4,\dots$. It remains to apply $\Compose$ (Alg.~\ref{alg:compose_full}) to obtain a full $\ktensk {k'}$ instance.} for any $\eps \in \paren{0, 1}$, 
    $$
    \conjeps \eps \Rightarrow \conjeps 1\,.
    $$
\end{corollary}

The $\DiscrEqComb$ reduction combines pairs of entries in the input and achieves noise independence by not reusing them -- this works best for sparser instances of $\ktens$, i.e. small density $\eps$. Our $\GaussEqComb$ reduction in \autoref{sec:gauss_resolution} beats parameter tradeoffs in Cor.~\ref{cor:tensor_from_resolution} in the case of denser input (larger $\eps$). 

\begin{mylem}{\ref{lem:gauss_resolution}}\textnormal{(Gaussian Resolution: Reduction to Tensor PCA)}
    Let $\eps \in (1/3,1], \defi, k \in \mathbb{Z}^{+}$ be the $\ktens$ parameters satisfying $-k\eps/2 < \defi \leq (k-2)/2$. If $k'$ is even and satisfies
    \begin{equation}\label{eq:tensor_red_cond_to_tpca}
        k' < \min\set{ k\eps + 2\defi, \frac{k\paren{3\eps -1}}{2}, \frac{k\paren{1+\eps}}{4}}\,,
    \end{equation}
    then there exists a poly-time average-case reduction\footnote{As defined in Def.~\ref{def:avg_case_points}} for both detection and recovery (Alg.~\ref{alg:gauss_res})
    $$
    \text{from}\quad\gtensp k n \eps \defi \quad\text{to}\quad \gtensp {k'} {n'} {\eps'=1} {2\defi}\,,
    $$
    for $n' = n\paren{1- o(1)}$. Moreover, Alg.~\ref{alg:gauss_res} maps an instance with a secret signal vector $x \in\set{\pm1}^n$ to one with a signal vector $x' = x_{\sqb{1:n'}} \in \set{\pm 1}^{n'}$.\footnote{Here $x'$ can be any subset of $n'$ coordinates of $x$ by permutation invariance of the model. See Sec.~\ref{subsec:recovery_assumptions} for why this is sufficient to transfer weak/strong/exact recovery.}
\end{mylem}
\begin{remark}
    In the case of $\defi \to 0$, i.e. at the \emph{computational threshold}, the parameter regime above can be rewritten as follows: all even $k'$ satisfying: 
    \begin{enumerate}
        \item $\eps \in (1/3, 3/5)$: $k' < \frac{k\paren{3\eps - 1}}{2}$;
        \item $\eps \in (3/5, 1]$: $k' < \frac {k(1+\eps)} 4$.
    \end{enumerate}
\end{remark}

\paragraph{Comparison of Cor.~\ref{cor:tensor_from_resolution} and Lemma~\ref{lem:gauss_resolution}}

Both Cor.~\ref{cor:tensor_from_resolution} ($\DiscrEqComb$) and Lemma~\ref{lem:gauss_resolution} ($\GaussEqComb$) obtain a reduction from 
$$
\gtensp k n \eps \defi \quad\text{to}\quad \gtensp {k'} n {\eps'=1} {2\defi}\,,
$$
for even $k' < k\cdot \tau$, where $\tau = \tau\paren{\eps, \defi}$:
\begin{align*}
    &\text{$\DiscrEqComb$ (Cor.~\ref{cor:tensor_from_resolution})}: \quad \tau = \min\set{1-\eps, 2\eps}\\
    &\text{$\GaussEqComb$ (Lem.~\ref{lem:gauss_resolution})}: \quad \tau = \min\set{\eps + 2\tilde\defi, \frac{3\eps-1}{2}, \frac {1+\eps}4}\,,
\end{align*}
where $-\eps/2<\tilde\defi \coloneq \defi/k < 1/2 - 1/k$. For the instances inside the conjectured computationally hard region ($\tilde\defi > 0$) the following hold:
\begin{enumerate}
    \item $\eps \in (0, 1/3]$: only the result of Cor.~\ref{cor:tensor_from_resolution} ($\DiscrEqComb$) applies;
    \item $\eps \in (1/3, 3/5)$: both reductions are applicable, but Cor.~\ref{cor:tensor_from_resolution} ($\DiscrEqComb$) gives a wider range of admissible $k'$, i.e. $\GaussEqComb$ (Lem.~\ref{lem:gauss_resolution}) does not achieve new results in this regime;
    \item $\eps \in [3/5, 1]$: both reductions are applicable, but Lem.~\ref{lem:gauss_resolution} ($\GaussEqComb$) gives a wider range of admissible $k'$ and shows new reductions not achievable with the $\DiscrEqComb$. In this case we can simplify $\tau = \min\set{\eps + 2\tilde\defi, \frac {1+\eps}4}$.
\end{enumerate}
For $\defi > 0$, we summarize the admissible range of parameters $k' < k \tau$, $k'$ even, as follows:
\begin{equation}\label{eq:admissible}
\tau_{\eps, \defi/k} = \begin{cases}
    2\eps, &\eps \in \paren{0,1/3}\\
    1-\eps, &\eps \in [1/3,3/5)\\
    \frac {1+\eps}4, &\eps \in \sqb{3/5, 1}\,.
\end{cases} 
\end{equation}
Combined with our average-case reductions for decreasing the tensor order $k$ in \autoref{subsec:decrease_k_prelim},\footnote{Both reductions in Cor.~\ref{cor:tensor_from_resolution} and in Lem.~\ref{lem:gauss_resolution} can be modified to map to a full $\kxor$ instance, so the results in Lem.~\ref{lem:prelim_decrease_k} of \autoref{subsec:decrease_k_prelim} are applicable and make a $k'\to k'/2$ average-case reduction.} we obtain the following summary of average-case reductions to Tensor PCA. The equivalent statement for discrete $\kxor$ holds from equivalence in Prop.~\ref{prop:discr_gauss_equivalence_general}.

\begin{corollary}[of Lem.~\ref{lem:gauss_resolution} and Lem.~\ref{lem:2path_red}]\label{cor:tensor_pca_summary}
Let $\eps \in (0,1], \defi, k \in \mathbb{Z}^{+}$ be the $\ktens$ parameters satisfying $0 \leq \defi \leq (k-2)/2$. For $\tau_{\eps,\defi/k}$ defined in Eq.~\eqref{eq:admissible}, if $k'$ satisfies
$$
k' < \begin{cases}
    \tau_{\eps, \defi/k} \cdot k, &k' \text{ even}\\
    \tau_{\eps, \defi/k}/2 \cdot k, &k' \text{ odd}\,,
\end{cases}
$$
there exists a poly-time average-case reduction for both detection and recovery
$$
\text{from}\quad\gtensp k n \eps \defi \quad\text{to}\quad \gtensp {k'} {n'} {\eps'=1} {2\defi}\,,
$$
with $n' = n^{\Theta(1)}$.
\end{corollary}

\section{Reductions that Reduce the Tensor Order $k$}\label{sec:decrease_k}

We devote this section to average-case reductions from $\xorp k n \eps \defi$ to $\xorp {\kstar} {n'} \eps {\defi'}$, i.e. reductions that change the tensor order $k\to \kstar$ at the same density $\eps$, focusing on the canonical cases $\eps\in\set{0,1}$. Fully analogous results for $\ktens$ follow from the discrete-Gaussian equivalence (Prop.~\ref{prop:discr_gauss_equivalence_general}). Our reductions are in Theorem~\ref{thm:decreasek} with the corresponding hardness implications in Corollary~\ref{cor:decreasek}. 

\begin{theorem}[Reductions that Reduce Tensor Order $k$]\label{thm:decreasek} There exists a poly-time average-case reduction for both detection and recovery (Alg.~\ref{alg:decreasek_0}, \ref{alg:decreasek_1}) 
$$
\text{from}\quad \xorp k n \eps \defi\quad\text{to}\quad \xorp {\kstar} {n'} \eps {\defi'}
$$
with $n' = n^{\Theta(1)}, \defi' = O(\defi)$ for the following parameter regimes:
\begin{enumerate}
    \item \textbf{Canonical $\kxor$, $\eps = 0, \defi \geq 0$ (Alg.~\ref{alg:decreasek_0}):} 
    \begin{enumerate}
        \item \textbf{even $\kstar$:} all $k > \kstar$;
        \item \textbf{odd $\kstar$:} all odd $k > \kstar$ and even $k \geq 2\kstar$. 
    \end{enumerate}
    \item \textbf{Tensor PCA, $\eps = 1, \defi \geq 0$ (Alg.~\ref{alg:decreasek_1}):}
    \begin{enumerate}
        \item \textbf{even $\kstar$:} all $k > 3\kstar/2$;
        \item \textbf{odd $\kstar$:} all $k > 3\kstar$. 
    \end{enumerate}
\end{enumerate}  
Moreover, the reductions map an instance of dimension $n$ and signal vector $x\in\set{\pm 1}^n$ to one of dimension $n'$ and signal $x' \in\set{\pm 1}^{n'}$ as follows.\footnote{See \autoref{subsec:recovery_assumptions} for why this transfers weak/strong/exact recovery algorithms.} Denote $\paren{\tilde n, \tilde x} = \paren{n(1-o(1)), x_{\sqb{1: \tilde n}}}$.
\begin{enumerate}
    \item If $\kstar$ is odd and either (i) $\eps=1$ or (ii) $\eps = 0$ with $k$ even:
    $$
        n' = \tilde n^2\quad\text{and}\quad x' = \tilde x^{\otimes 2}\odot y\,,
    $$
    for a known $y \sim \unif\paren{\set{\pm 1}^{n'}}$;
    \item Otherwise: 
    $$
    \paren{n', x'} = \paren{\tilde n, \tilde x}\,.
    $$
\end{enumerate}

\end{theorem}
\begin{corollary}\label{cor:decreasek}
    For parameters $\eps, k$, and $\kstar$ as above, 
    \begin{equation}\label{eq:decreasek_implication}
    \conj k \eps \Rightarrow \conj {\kstar} \eps\,.
    \end{equation}
    In particular, for any $\eps \in \set{0,1}$, the relation Eq.~\eqref{eq:decreasek_implication} holds for any $\kstar$ and sufficiently large $k$.
\end{corollary}

The reductions in Theorem~\ref{thm:decreasek} map from an instance with larger parameter $k$ to an instance with a smaller one. We conjecture (Conj.~\ref{conj:one_way_hardness}) that such a reduction is indeed possible for any pair of parameters $k > \kstar$, i.e. the missing cases of Theorem~\ref{thm:decreasek} are merely a limitation of our technique. 

The results for \emph{even} $\kstar$ are special cases of the corresponding $\EqComb$ reductions: $\DiscrEqComb$ in Lemma~\ref{lem:2path_red} for $\eps = 0$ and $\GaussEqComb$, $\GaussResolutionTwo$ in Lemmas~\ref{lem:gauss_resolution}, \ref{lem:2path_red_ext_gauss} for $\eps = 1$.

To reach \emph{odd} $\kstar$, in the case of even $k$, we stack the discrete (or Gaussian) resolution with the $2\kstar \to \kstar$ reduction of Lemma~\ref{lem:prelim_decrease_k}, obtaining the $k \geq 2\kstar$ range in Thm.~\ref{thm:decreasek}. The reduction in Lemma~\ref{lem:prelim_decrease_k} is applicable, since both results in Lemma~\ref{lem:2path_red} and Lemmas~\ref{lem:gauss_resolution}, \ref{lem:2path_red_ext_gauss} can be modified to map to an instance of the full $\kxork {2k'}$. Formally, this can be achieved by first splitting the input into $k'$ independent instances via $\XorSplit$ and then applying the discrete or Gaussian resolution reduction to every copy, mapping to instances of orders $2k', 2k'-2,2k'-4,\dots$. It remains to apply $\Compose$ (Alg.~\ref{alg:compose_full}) to obtain a full $\kxork {2k'}$ instance.

For $\eps=0$ and $k$ odd, we obtain a wider range by using the generalized discrete resolution (Lemma~\ref{lem:2path_red_ext_discr}) that combines two $\kxor$ inputs that share the same secret $x \in \set{\pm1}^n$.

\subsection{Proof of Theorem~\ref{thm:decreasek}}

\begin{algorithm}\SetAlgoLined\SetAlgoVlined\DontPrintSemicolon
    \KwIn{
    $\calY = \set{ \b( \alpha_j, \Y_{j}\b)}_{j\in\sqb{m}}$, where $\forall j \in\sqb{m}$, $\b(\alpha_j, \Y_{j}\b) \in \binomset n k \times \set{\pm1}$, $k'\in \Z$.
    }
    \KwOut{$\widetilde\calZ = \set{ \b( \gamma_j, \Zt_{j}\b)}_{j\in\sqb{m'}}$, where $\forall j \in\sqb{m'}$, $\b(\gamma_j, \Zt_{j}\b) \in \binomset {n'} {k'} \times \set{\pm1}$.}

    \tcp{Decrease-order reduction for $\eps=0$}

    \BlankLine
    \If{$\kstar$ even}{\tcp{Directly apply discrete resolution}
        $\calZ \gets \TwoPath\paren{\calY, k'}$ \ \tcp{Alg.~\ref{alg:two_path}}
    }

    \BlankLine
    \If{$\kstar$ odd}{
        \If{$k$ even}{\tcp{Map to a collection of instances of even order $2\kstar$ and smaller}
        $\calY\up 1, \dots, \calY\up \kstar \gets \XorSplit\paren{\calY, \kstar}$ \ \tcp{\autoref{subsubsec:splitting}, see Rmk.~\ref{rmk:clone_several}}

        For each $i\in \sqb{k'}$, $\calZ\up i \gets \TwoPath\paren{\calY\up i, 2i}$ \ \tcp{Alg.~\ref{alg:two_path}}

        \tcp{Reduce the order by a factor of $2$}
        $\calZ \gets \ReduceK\paren{\Compose\b(\set{\calZ\up i}_i\b), 2}$ \ \tcp{Alg.~\ref{alg:prelim_decrease_k}, \ref{alg:compose_full}, see Rmk.~\ref{rmk:to_full_discrete_input}}
        }\Else{
        $\calY\up1, \calY\up2 \gets \XorSplit\paren{\calY}$ \ \tcp{\autoref{subsubsec:splitting}}

        \tcp{Create an instance of order $k-1$}
        $\tilde \calY\up{1} \gets \TwoPath\paren{\calY\up1, k-1}$ \ \tcp{Alg.~\ref{alg:two_path}}

        \tcp{Combine tensors of orders $k-1, k$}
        $\calZ \gets \TwoPathTwo\paren{\tilde \calY\up1, \calY\up2, \kstar}$ \ \tcp{Alg.~\ref{alg:two_path_two_inputs}}
        }
    }

    \BlankLine
    Return $\calZ$
    
    \caption{$\ReduceKDiscrete\paren{\calY, \kstar}$}
\label{alg:decreasek_0}
\end{algorithm}
\begin{algorithm}\SetAlgoLined\SetAlgoVlined\DontPrintSemicolon
    \KwIn{
    $\calY = \set{ \b( \alpha_j, \Y_{j}\b)}_{j\in\sqb{m}}$, where $\forall j \in\sqb{m}$, $\b(\alpha_j, \Y_{j}\b) \in \binomset n k \times \R$, $k'\in \Z$.
    }
    \KwOut{$\widetilde\calZ = \set{ \b( \gamma_j, \Zt_{j}\b)}_{j\in\sqb{m'}}$, where $\forall j \in\sqb{m'}$, $\b(\gamma_j, \Zt_{j}\b) \in \binomset {n'} {k'} \times \R$.}

    \tcp{Decrease-order reduction for $\eps=1$}

    \BlankLine
    \If{$\kstar$ even}{\tcp{Directly apply Gaussian resolution}
        $\calY\up 1, \calY\up 2 \gets \kTensSplit(\calY)$ \ \tcp{\autoref{subsubsec:splitting}}
    
        $\calZ \gets \GaussResolutionTwo\paren{\calY\up 1, \calY\up 2, \kstar}$ \ \tcp{Alg.~\ref{alg:two_path_gen_tens}}
    }

    \BlankLine
    \If{$\kstar$ odd}{
        \tcp{Map to a collection of instances of even order $2\kstar$ and smaller}
        $\calY\up 1, \dots, \calY\up \kstar \gets \kTensSplit\paren{\calY, \kstar}$ \ \tcp{\autoref{subsubsec:splitting}, see Rmk.~\ref{rmk:clone_several}}

        \ForAll{$i\in \sqb{k'}$}{
            $\calY\up {i,1}, \calY\up {i,2} \gets \kTensSplit(\calY\up i)$ \ \tcp{\autoref{subsubsec:splitting}}

            $\calZ\up i \gets \GaussResolutionTwo\paren{\calY\up {i,1}, \calY\up{i,2}, 2i}$ \ \tcp{Alg.~\ref{alg:two_path_gen_tens}}
        }

        \tcp{Reduce the order by a factor of $2$}
        $\calZ \gets \ReduceK\paren{\Compose\b(\set{\calZ\up i}_i\b), 2}$ \ \tcp{Alg.~\ref{alg:prelim_decrease_k}, \ref{alg:compose_full}}
    }

    \BlankLine
    Return $\calZ$
    
    \caption{$\ReduceKGauss\paren{\calY, \kstar}$}
\label{alg:decreasek_1}
\end{algorithm}

\paragraph{Order Reduction for $\eps = 0$.}

The reduction algorithm is Algorithm~\ref{alg:decreasek_0}. Let $\calY\sim \xorp k n {\eps=0} \defi$ be the (planted) input. The argument holds for the null case without modification. 

In the case of $\kstar$ even, by Lemma~\ref{lem:2path_red}, $\TwoPath\paren{\calY, \kstar}$ is an average-case reduction for both detection and recovery that maps $\xorp k n {\eps=0} \defi$ to $\xorp \kstar {n'} {\eps=0} {2\defi}$ for $n' = n(1-o(1))$ and the signal vector transformed from $x\in\set{\pm1}^n$ to $x' = x_{\sqb{1:n'}}$. 

In the case of $\kstar$ odd and $k$ even, by Lemma~\ref{lem:instance_split}, $\XorSplit\paren{\calY}$ is an average-case reduction to independent instances of $\xorp k n {\eps=0} \defi$, denoted as $\calY\up1,\dots,\calY\up\kstar$, that share the same signal vector $x$ as the input instance $\calY$. By Lemma~\ref{lem:2path_red}, $\TwoPath\paren{\calY\up i, 2i}$ is an average-case reduction to $\xorp {2i} {\tilde n} {\eps=0} {2\defi}$ for $\tilde n = n(1-o(1))$ and the signal vector $\tilde x = x_{\sqb{1:\tilde n}}$. Finally, by Prop.~\ref{prop:comp_equiv_full_collection} and Lemma~\ref{lem:prelim_decrease_k}, $\calZ \gets \ReduceK\paren{\Compose\b(\set{\calZ\up i}_i\b), 2}$ is an average-case reduction to $\xorp \kstar {n'} {\eps=0} {\defi}$ for $n' = \tilde n^2$ with a signal vector $x' = \tilde x^{\otimes 2} \odot y$ for a known $y \sim \unif\paren{\set{\pm1}^{n'}}$.

In the case of both $k,\kstar$ odd, by Lemma~\ref{lem:instance_split}, $\XorSplit\paren{\calY}$ is an average-case reduction to two independent instances of $\xorp k n {\eps=0} \defi$, denoted as $\calY\up1, \calY\up 2$, that share the same signal vector $x$ as the input instance $\calY$. By Lemma~\ref{lem:2path_red}, $\tilde\calY\up 1 = \TwoPath\paren{\calY\up1, k-1}$ is an average-case reduction to $\xorp {k-1} {\tilde n} {\eps=0} {2\defi}$ for $\tilde n = n(1-o(1))$ and the signal vector $\tilde x = x_{\sqb{1:\tilde n}}$. Then, restricting $\calY\up2$ to the selected $\tilde n$ coordinates, by Lemma~\ref{lem:2path_red_ext_discr}, $\calZ = \TwoPathTwo\paren{\tilde \calY\up1, \calY\up2, \kstar}$ is an average-case reduction to  $\xorp \kstar {n'} {\eps=0} {2\defi}$ for $n' = n(1-o(1))$ and the signal vector $x' = x_{\sqb{1:n'}}$. 

\paragraph{Order Reduction for $\eps = 1$.}

The reduction algorithm is Algorithm~\ref{alg:decreasek_1}. Let $\calY\sim \gtensp k n {\eps=1} \defi$ be the (planted) input. The argument holds for the null case without modification. The analogous result in Thm.~\ref{thm:decreasek} for discrete $\xorp k n {\eps=1} \defi$ follows from discrete-Gaussian equivalence in Prop.~\ref{prop:discr_gauss_equivalence_general}.

In the case of $\kstar$ even, by Lemma~\ref{lem:instance_split} and Lemma~\ref{lem:2path_red_ext_gauss}, $\GaussResolutionTwo\paren{\calY\up 1, \calY\up 2, \kstar}$ is an average-case reduction for both detection and recovery that to $\gtensp \kstar {n'} {\eps=1} {2\defi}$ for $n' = n(1-o(1))$ and the signal vector transformed from $x\in\set{\pm1}^n$ to $x' = x_{\sqb{1:n'}}$. 

In the case of $\kstar$ odd, by Lemma~\ref{lem:instance_split}, $\kTensSplit\paren{\calY}$ is an average-case reduction to independent instances of $\gtensp k n {\eps=1} \defi$, denoted as $\calY\up1,\dots,\calY\up\kstar$, that share the same signal vector $x$ as the input instance $\calY$. By Lemma~\ref{lem:instance_split} and Lemma~\ref{lem:2path_red_ext_gauss}, $\GaussResolutionTwo\paren{\calY\up {i,1}, \calY\up{i,2}, 2i}$ is an average-case reduction to $\gtensp {2i} {\tilde n} {\eps=1} {2\defi}$ for $\tilde n = n(1-o(1))$ and the signal vector $\tilde x = x_{\sqb{1:\tilde n}}$. Finally, by Lemma~\ref{lem:prelim_decrease_k}, $\calZ \gets \ReduceK\paren{\Compose\b(\set{\calZ\up i}_i\b), 2}$ is an average-case reduction to $\gtensp \kstar {n'} {\eps=1} {\defi}$ for $n' = \tilde n^2$ with a signal vector $x' = \tilde x^{\otimes 2} \odot y$ for a known $y \sim \unif\paren{\set{\pm1}^{n'}}$.

\bibliographystyle{alphaurl}
\bibliography{bib}

\appendix

\newpage 
\section{Reductions between $\kxor$ Variants}

For convenience, we restate our notation for $\kxor$ model variants. We parametrize a $\kxor$-like model by four choices. By default (i.e., model with no superscript) we use the first option in choices (2)-(4); superscripts indicate deviations from these defaults:

\begin{figure}[h]
  \centering
    
\renewcommand{\arraystretch}{1.15}
\setlength{\tabcolsep}{6pt}
\begin{tabular}{@{}l p{0.72\linewidth}@{}}

\textbf{(1) Noise model} &
\([\checkmark]\) \(\kxor:\ \Y \sim x_\alpha\cdot \rad(\delta)\)\\
&[\textcolor{white}{\checkmark}]\ \(\ktens:\ \Y \sim \delta x_\alpha+\N(0,1)\)\\[4pt]

\textbf{(2) Number of samples} &
\([\checkmark]\) \(\kmodel:\ m\text{ samples}\)\\
&[\textcolor{white}{\checkmark}]\ \(\kmodel^{\mathsf{Pois}}:\ \mpoi\sim\poi(m)\text{ samples}\)\\[4pt]

\textbf{(3) Index space for \(\alpha\)} &
\([\checkmark]\) \(\kmodel:\ \alpha\in\binomset n k\) ($k$-sets of distinct indices)\\
& [\textcolor{white}{\checkmark}]\ \(\kmodel^{\mathsf{FULL}}:\ \alpha\in [n]^k\) ($k$-multisets, index repetitions allowed, Def.~\ref{def:multiset})\\
& [\textcolor{white}{\checkmark}]\ \(\kmodel^{\mathsf{asymm}}:\ \alpha\in [n]^{\underline{k}}\) (ordered $k$-tuples of distinct indices)\\[4pt]

\textbf{(4) Sampling of \(\alpha\)} &
\([\checkmark]\) \(\kmodel:\) with replacement\\
&[\textcolor{white}{\checkmark}]\ \(\kmodel^{\mathsf{wor}}:\) without replacement\\
\end{tabular}

    \caption{Choices in $\kxor$ variants: the defaults are checked and correspond to $\xor k n m \delta$.}
    \label{fig:kxor_variants_app}
\end{figure}

We use the $(\eps, \defi)$ reparametrization for these models analogously to standard $\kxor$ (see \autoref{sec:ideas}) and focus on the case $\eps \geq 0$.

For example, $\gtenspvar k n \eps \defi {\Pois, \wor}$ (equiv. $\gtensvar k n m \delta {\Pois, \wor}$ with $m = n^{k(1+\eps)/2}, \delta = n^{-k\eps/4 - \defi/2}$) observes $\mpoi \sim \poi(m)$ independent entries $\set{(\alpha_j, \Y_j)}_{j \in \sqb{\mpoi}}$ for $\alpha_j \in \unif\sqb{\binomset n k}$ sampled without replacement and $\Y_j \sim \delta x_{\alpha_j} + \N(0,1)$.

\subsection{Reductions between Discrete ($\kxor$) and Gaussian ($\ktens$) Models}\label{subsec:comp_equiv}

In this section we recall known scalar reductions between Gaussian and Bernoulli variables that preserve signal level $\delta$ up to log factors. Applying these to the $\kxor, \ktens$ models yields the following general equivalence, which will allow us to work with whichever model is more convenient.

\begin{myprop}{\ref{prop:discr_gauss_equivalence_general}}\textnormal{(Computational Equivalence of $\kxor^\star$ and $\ktens^\star$)}
    Let $\star\subseteq \set{\mathsf{Pois}, \mathsf{FULL} / \mathsf{asymm}, \mathsf{wor}}$ be some fixed model choices (2)-(4) in Fig.~\ref{fig:kxor_variants_app}. The corresponding discrete $\xorpstar k n \eps \defi$ model and its Gaussian analog $\gtenspstar k n \eps \defi$ are computationally equivalent: there are $\poly(n)$-time average-case reductions ($\Gaussianize/\Discretize$, Alg.~\ref{alg:gaussianization}, \ref{alg:discretization}) for both detection and recovery in both directions. Moreover, the reductions preserve the secret signal vector $x$ precisely.
\end{myprop}

As a consequence, proving a $\poly(n)$-time average-case reduction from one of $\xorpstar k n \eps \defi$ or $\gtenspstar k n \eps \defi$ to one of $\xorpstar {k'} {n'} {\eps'} {\defi'}$ or $\gtenspstar {k'} {n'} {\eps'} {\defi'}$ gives the remaining three reductions by composition.

\subsubsection{Gaussianization}\label{subsec:gaussianize} First, we describe the guarantees for a subroutine that maps a pair of Bernoulli random variables with different means to a pair of Gaussian random variables with means differing by almost the same amount. Such entry-wise transformations of measure were devised in \cite{ma2015computational}, \cite{hajek2015computational}, and later generalized in \cite{brennan2018reducibility}. The following is a special case of Lemma 5.4 in \cite{brennan2018reducibility}, applied to $\bern {1/2-p}, \bern {1/2+p}$ random variables:

\begin{lemma}[Gaussianization of $\bern {1/2-p}, \bern {1/2+p}$]\label{lem:gaussianization}
   Let $n$ be a parameter and suppose that $p = p(n)$ satisfies $n^{-O(1)} \leq p < c < 1/2$ for some constant $c$. Then there exists a constant $C(c)$ such that for $\mu = \mu(n) \in (0,1)$ defined as
$$
\mu = \frac{Cp}{2 \sqrt{6 \log n + 2 \log (2p)^{-1}}} = \Thetat(p)
$$
the following holds. There exists a map $\mathtt{RK}_{p}$ that can be computed in $O(\poly (n))$-time and satisfies
$$
\tv\big(\mathtt{RK}_p(\bern{1/2+p}), \N(\mu,1)\big) = O_n(n^{-K}), \text{  and  } \tv\big(\mathtt{RK}_p(\bern{1/2-p}), \N(0,1)\big) = O_n(n^{-K})\,,
$$
for any constant $K$.
\end{lemma}

\begin{remark}[Gaussianization of Rademachers]
Using $-1$ in place of $0$ yields the same mapping guarantees for $\rad(\pm p)$ in place of $\bern{1/2\pm p}$. We denote the corresponding map by $\mathtt{RK}_p^{\rad}$. 
\end{remark} 

The following Lemma applies $\mathtt{RK}_p^{\rad}$  entrywise to map from $\kxor^\star$ to $\ktens^\star$ and proves one direction of Prop.~\ref{prop:discr_gauss_equivalence_general}.

\begin{lemma}[Gaussianization of $\kxor^\star$]\label{lem:kxor_gauss}
    For any constant $k, \eps, \defi$, there exists an $O\paren{\poly(n)}$-time average-case reduction (Alg.~\ref{alg:gaussianization}) for both detection and recovery from 
    $$
    \xorpstar k n \eps \defi \quad\text{to}\quad \gtenspstar k n \eps \defi\,.
    $$
    Moreover, given an input instance with a secret signal vector $x \in \set{\pm 1}^n$, Alg.~\ref{alg:gaussianization} outputs an instance with the same signal vector $x$ and the same nonzero entry locations.
\end{lemma}
\begin{algorithm}\SetAlgoLined\SetAlgoVlined\DontPrintSemicolon
    \KwIn{$\calY^{\mathsf{xor}} = \set{ \b( \alpha_j, \Y^{\mathsf{xor}}_{j}\b)}_{j\in\sqb{m}}$; $\delta > 0$.
    }
    \KwOut{$\calY^{\mathsf{tens}} = \set{ \b( \alpha_j, \Y^{\mathsf{tens}}_{j}\b)}_{j\in\sqb{m}}$.}

    \BlankLine
    $\forall j \in \sqb{m}$,
    $$
    \Y^{\mathsf{tens}}_{j} \gets 
        \mathtt{RK}_{\delta}^{\rad} \paren{\Y^{\mathsf{xor}}_{j}}
    $$

    \BlankLine
    
    Return $\calY^{\mathsf{tens}} = \set{ \b( \alpha_j, \Y^{\mathsf{tens}}_{j}\b)}_{j\in\sqb{m}}$.
    
    \caption{$\Gaussianize\paren{\calY^{\mathsf{xor}}, \delta}$}
\label{alg:gaussianization}
\end{algorithm}
\begin{proof}

    Alg.~\ref{alg:gaussianization} applies the Gaussianization kernel $\mathtt{RK}_{\delta}^{\rad}$ on every revealed entry of the input instance. In the planted case, $\calY^{\mathsf{xor}}\sim \xorstar k n m \delta$ (where $m = n^{k(\eps+1)/2}$, $\delta = n^{-k\eps/4-\defi/2}$) and by entry-wise total variation guarantees of Lemma~\ref{lem:gaussianization} with an appropriate choice of constant $K = K(k)$, the output $\calY^{\mathsf{tens}} = \Gaussianize\paren{\calY^{\mathsf{xor}}, \delta}$ satisfies
    $$
    \tv\paren{\law\paren{\calY^{\mathsf{tens}}}, \gtensstar k n m {\delta'}} \leq o_n\paren{1}\,,
    $$
    where $\delta' = \tilde{\Theta}(\delta)$. 
    In the null case, the input $\calY^{\mathsf{xor}} \sim \kxor_{\eps}^{\star, \nulll}(n)$ has i.i.d. $\rad\paren{0}$ entries. By Lemma~\ref{lem:gaussianization}, for any choice of constant $K$,
    $$\tv\paren{ \mathtt{RK}_{\delta}^{\rad}\paren{\rad(0)}, \frac12 \N\paren{\delta', 1}+ \frac12\N\paren{-\delta', 1}} \leq O_n\paren{n^{-K}}\,,$$
    for some $\delta' = \tilde\Theta\paren{\delta}$, and therefore, the output $\calY^{\mathsf{tens}} = \Gaussianize\paren{\calY^{\mathsf{xor}}, \delta}$ satisfies satisfies $$ \tv\paren{ \Gaussianize\paren{\calY^{\mathsf{xor}}, \delta}, \ktens_{\eps, \defi}^{\star, \nulll}(n)} \leq o_n(1) $$
    for sufficiently large $K$ (see Def.~\ref{def:null_models} for the null model definitions).

    As the problem parameters are preserved up to $\poly\log n$ factors, this yields (Prop.~\ref{prop:map_up_to_constant}) the desired average-case reduction from $
    \xorpstar k n \eps \defi$ to $\gtenspstar k n \eps \defi$ for both detection and recovery.
    
\end{proof}

\subsubsection{Discretization}\label{subsec:discretize}

In this section we describe a ``reverse" procedure --\emph{discretization} -- that maps Gaussian random variables to Rademacher (or Bernoulli) random variables preserving the means up to polylog factors. This is achieved by a simple thresholding ($\sign$) operation;  we let $\sign(0) \sim \unif\paren{\set{\pm1}}$.

\begin{lemma}[Discretization of $\N(\mu,1), \N(0,1), \N(-\mu, 1)$]\label{lem:discretization}
   Let $n$ be a parameter and suppose that $\mu = \mu(n)$ satisfies $0 < \mu < 1/2$. Then there exists a sequence $p = p(n) \in (0,1)$, such that $$p = \tilde\Theta(\mu)\,,$$ and a map $\Discr : \R \to \set{\pm 1}$ that can be computed in $O(1)$ time that satisfies 
    $$
    \Discr \Big(\frac12 \N(-\mu, 1) + \frac12 \N(\mu, 1)\Big) \eqdist \rad(0)\,,
    $$
    $$
    \Discr\paren{\N(\mu,1)}\eqdist \rad(p)\,, \quad\text{and} \quad \Discr\paren{\N(-\mu,1)}\eqdist \rad(-p)\,.
    $$
\end{lemma}
\begin{proof}
    Let $\Discr : \R \to \set{\pm 1}$ be simply $\Discr(x) = 
        \sign(x)$, where recall $\sign(0) \sim \unif\paren{\set{\pm1}}$.
    Then, 
    $$
    \Pr_{x \sim \N(\mu,1)} \sqb{\Discr(x) = 1} = \Pr_{x \sim \N(-\mu,1)} \sqb{\Discr(x) = -1} = 1/2 + \Theta(\mu)$$ 
    and
    $$\Pr_{x \sim \frac12 \N(-\mu,1) + \frac12 \N(\mu, 1)} \sqb{\Discr(x) = 1} = 1/2\,,
    $$
    which concludes the proof. 
\end{proof}

The following Lemma maps from $\ktens^\star$ to $\kxor^\star$ and proves the remaining direction of Prop.~\ref{prop:discr_gauss_equivalence_general}.
\begin{lemma}[Discretization of $\ktens^\star$]\label{lem:ktens_discr}
    For any constant $k, \eps, \defi$, there exists an $O\paren{n^k}$-time average-case reduction (Alg.~\ref{alg:discretization}) for both detection and recovery from 
    $$
    \gtenspstar k n \eps \defi \quad\text{to}\quad \xorpstar k n \eps \defi\,.
    $$
    Moreover, given an input instance with a secret signal vector $x \in \set{\pm 1}^n$, Alg.~\ref{alg:discretization} outputs an instance with the same signal vector $x$ and the same nonzero entry locations.
\end{lemma}
\begin{algorithm}\SetAlgoLined\SetAlgoVlined\DontPrintSemicolon
    \KwIn{$\calY^{\mathsf{tens}} = \set{ \b( \alpha_j, \Y^{\mathsf{tens}}_{j}\b)}_{j\in\sqb{m}}$, where $\b(\alpha_j, \Y^{\mathsf{tens}}_{j}\b) \in \binomset n k \times \R\, \forall j \in\sqb{m}$.
    }
    \KwOut{$\calY^{\mathsf{xor}} = \set{ \b( \alpha_j, \Y^{\mathsf{xor}}_{j}\b)}_{j\in\sqb{m}}$, where $\b(\alpha_j, \Y^{\mathsf{xor}}_{j}\b) \in \binomset n k \times \set{\pm1}\, \forall j \in\sqb{m}$.}

    \BlankLine
    $\forall j \in \sqb{m}$,
    $$
    \Y^{\mathsf{xor}}_{j} \gets 
        \Discr \paren{\Y^{\mathsf{tens}}_{j}}
    $$

    \BlankLine
    
    Return $\calY^{\mathsf{xor}} = \set{ \b( \alpha_j, \Y^{\mathsf{xor}}_{j}\b)}_{j\in\sqb{m}}$.
    
    \caption{$\Discretize\paren{\calY^{\mathsf{tens}}}$}
\label{alg:discretization}
\end{algorithm}

\begin{proof}
    Alg.~\ref{alg:discretization} applies the thresholding operation $\Discr$ from Lemma~\ref{lem:discretization} to every revealed entry of the input instance. From LEmma~\ref{lem:discretization}, $\calY^{\mathsf{xor}} = \Discretize\paren{\calY^{\mathsf{tens}}}$ satisfies 
    \begin{align*}
        \law\paren{\calY^{\mathsf{xor}}} = \xorstar k n m {\delta'}\,, &\qquad\text{if } \calY^{\mathsf{tens}} \sim \gtensstar k n m \delta\\
        \law\paren{\calY^{\mathsf{xor}}} = \kxor^{\star, \nulll}(n, m)\,, &\qquad\text{if } \calY^{\mathsf{tens}} \sim \ktens^{\star, \nulll}(n, m, \delta)\,,
    \end{align*}
    where $\delta' = \Theta(\delta)$. As the problem parameters are preserved up to $\poly\log n$ factors, this yields (Prop.~\ref{prop:map_up_to_constant}) an average-case reduction for both detection and recovery from $
    \gtenspstar k n \eps \defi $ to $\xorpstar k n \eps \defi\,.
    $

\end{proof}

\subsection{Poisson Entry Sampling}\label{subsec:poi_sampling}

In this section we formally prove standard reductions between $\kxor$ with exactly $m$ entries and the $\kxorpoi$ in which there are $\mpoi \sim \poi(m)$ entries; this is a widely used equivalence, see, e.g., \cite{bogdanov2025sample}. The following proposition shows the equivalence of the two settings in a variety of models from Fig.~\ref{fig:kxor_variants_app} (see \autoref{subsec:kxor_variants} for overview of $\kxor$ variants).

\begin{myprop}{\ref{prop:sampling_equiv}}\textnormal{(Computational Equivalence of $\kmodel^\star$ and $\kmodel^{\star, \Pois}$)}
    Let $\star \subseteq \set{ \mathsf{FULL}/\mathsf{asymm}, \mathsf{wor} }$ and $\kmodel \in \set{\kxor, \ktens}$ be fixed model choices (1,3,4) in Fig.~\ref{fig:kxor_variants_app}. Then the fixed-sample $\kmodel^\star_{\eps, \defi}(n)$ and its Poissonized variant $\kmodel^{\star, \Pois}_{\eps, \defi}(n)$ are computationally equivalent: there are $\poly(n)$-time average-case reductions ($\ResampleIID/\ResampleM$, Alg.~\ref{alg:resample_iid}, \ref{alg:resample_M}) for both detection and recovery in both directions. Moreover, the reductions preserve the secret signal vector $x$ precisely.
\end{myprop}

\begin{algorithm}\SetAlgoLined\SetAlgoVlined\DontPrintSemicolon
    \KwIn{
    $\calY = \set{ \b( \alpha_j, \Y_{j}\b)}_{j\in\sqb{m}}$.
    }
    \KwOut{$\calZ = \set{ \b( \alpha_j, \Zt_{j}\b)}_{j\in\sqb{\mpoi}}$.}
    
    \BlankLine
    \tcp{Sample the number of observed equations $\mpoi$ in the output instance}
    $\widetilde\mpoi \sim \poi\paren{m/2}$, $\mpoi \gets \min\paren{m, \widetilde\mpoi}$

    \BlankLine
    \tcp{Return exactly $\mpoi$ equations}
    Return $\calZ = \set{ \b( \alpha_j, \Y_{j}\b)}_{j\in\sqb{\mpoi}}$.
    
    \caption{$\ResampleIID\paren{\calY}$}
\label{alg:resample_iid}
\end{algorithm}

\begin{algorithm}\SetAlgoLined\SetAlgoVlined\DontPrintSemicolon
    \KwIn{
    $\calY = \set{ \b( \alpha_j, \Y_{j}\b)}_{j\in\sqb{\mpoi}}$.
    }
    \KwOut{$\calZ = \set{ \b( \alpha_j, \Zt_{j}\b)}_{j\in\sqb{m'}}$.}

    \BlankLine
    \tcp{Compute the number of observed equations $m'$ in the output instance}
    $\widetilde m \gets \lfloor m/2 \rfloor$, $m' \gets \min\paren{\widetilde m, \mpoi}$

    \BlankLine
    \tcp{Return exactly $m'$ equations}
    Return $\calZ = \set{ \b( \alpha_j, \Y_{j}\b)}_{j\in\sqb{m'}}$.
    
    \caption{$\ResampleM\paren{\calY, m}$}
\label{alg:resample_M}
\end{algorithm}

\begin{proof}[Proof of Prop.~\ref{prop:sampling_equiv}.] The analysis for the planted and null cases is the same for both $\ResampleM$ and $\ResampleIID$. 

\textbf{Correctness of $\ResampleIID$ (Alg.~\ref{alg:resample_iid}).} By the concentration of Poisson random variable (Prop.~\ref{prop:pois_conc}), for any $c> 0$, with probability at least $1-n^{-c}$,
$$
 \widetilde\mpoi \leq \frac12 m + C \sqrt{\frac12 m} \log n \leq m\,.
$$
Therefore, with probability at least $1-n^{-c}$, in Alg.~\ref{alg:resample_iid}, $\mpoi = \widetilde\mpoi$, and thus, $\calZ \sim \kmodel^{\star, \Pois}(n, m/2, \delta)$ if $\calY \sim  \kmodel^{\star}(n, m,\delta)$. By the properties of total variation distance (Fact~\ref{tvfacts}), 
$$
\tv\paren{ \law\paren{\calZ}, \kmodel^{\star, \Pois}(n, m/2, \delta)} \leq n^{-c}\,,
$$
yielding an average-case reduction for both detection and recovery from $ \kmodel^{\star}_{\eps, \defi}(n)$ to $\kmodel^{\star, \Pois}_{\eps, \defi}(n)$ (since it is sufficient to map parameters up to constant factors by Prop.~\ref{prop:map_up_to_constant}).

 \textbf{Correctness of $\ResampleM$ (Alg.~\ref{alg:resample_M}).} Since $\mpoi \sim \poi(m)$, by the concentration of Poisson random variable (Prop.~\ref{prop:pois_conc}), for any $c> 0$, with probability at least $1-n^{-c}$,
$$
\mpoi \geq m - C \sqrt{m} \log n \geq \frac 12 m \geq \widetilde m\,.
$$
Therefore, with probability at least $1-n^{-c}$, in Alg.~\ref{alg:resample_M}, $m' = \widetilde m = \lfloor m/2\rfloor$, and thus, $\calZ \sim \kmodel^{\star}(n, {m'}, \delta)$ if $\calY \sim \kmodel^{\star, \Pois}(n, {m}, \delta)$. By the properties of total variation distance (Fact~\ref{tvfacts}), 
$$
\tv\paren{ \law\paren{\calZ}, \kmodel^{\star}(n, {m'}, \delta)} \leq n^{-c}\,,
$$
yielding an average-case reduction for both detection and recovery from $ \kmodel^{\star, \Pois}_{\eps, \defi}(n)$ to $\kmodel^{\star}_{\eps, \defi}(n)$ (since it is sufficient to map parameters up to constant factors by Prop.~\ref{prop:map_up_to_constant}).

\end{proof}

\subsection{Repeated indices in $\kxor$}\label{subsec:repeated_indexes}

While $\kxor$ (Def.~\ref{def:kxor_family}) consists of entries $(\alpha, \Y)$ with $k$-sets of distinct indices $\alpha \in \binomset n k$, the Tensor PCA literature often considers models that allow index repetitions. For example, the standard $\delta x^{\otimes k} + \noiseG$ model, where $\noiseG$ is a symmetric tensor of $\N(0,1)$, includes the entries with repeated indices. We show that this ``full" variant ($\kxorvar \FULL$) that allows index repetitions is computationally equivalent to a collection of independent instances of the models with no index repetitions of orders $k, k-2, k-4, \dots$ that share the same signal vector $x$:

\begin{myprop}{\ref{prop:comp_equiv_full_collection}}\textnormal{(Computational Equivalence of $\kmodel^{\star, \mathsf{FULL}}$ and a Collection of $\kmodel^\star$)}
    Fix $k\in\mathbb{Z}, \eps\in\sqb{0,1}, \defi\in \R_{+}$ and any choices (1,2,4) $\star \subseteq \set{\Pois, \wor}$ and $\kmodel \in \set{\kxor, \ktens}$ in Fig.~\ref{fig:kxor_variants_app}. Then,
    $$
    \kmodel_{\eps, \defi}^{\star, \mathsf{FULL}}(n) \quad\text{and}\quad \B\{ \kmodelk {(k-2r)}^{\star}_{\eps_r, \defi} (n) \B\}_{0\leq r \leq \frac{k-1}{2}}\qquad \text{with } \eps_r = \min\set{ \frac{k\eps}{k-2r} , 1}
    $$
    are computationally equivalent: there exist poly-time average-case reductions ($\Decompose/\Compose$, Alg.~\ref{alg:decompose_full}, \ref{alg:compose_full}) both ways.\footnote{Here the second model is a collection of independent $\kmodel$ instances.} These reductions preserve the input signal vector $x\in\set{\pm1}^n$ exactly.
\end{myprop}

We formally prove Prop.~\ref{prop:comp_equiv_full_collection} for $\ktenspoi$ model (see \autoref{subsec:poi_sampling}). In $\gtensppoi k n \eps \defi$ one observes $\mpoi \sim \poi(m)$ independent samples $(\alpha, \Y)$ with $$\Y \sim \delta x_\alpha + \N(0,1)$$ ($\delta = n^{-k\eps/4 - \defi/2}$) and i.i.d. index sets $\alpha \sim \unif\sqb{\binomset n k}$. The only difference in $\gtenspvar k n \eps \defi {\Pois, \FULL}$ is that $\alpha \in \set{n}^k$ are $k$-multisets sampled from a natural measure $\mu$ from Def.~\ref{def:multiset}.

The analogous results for the $\kmodel$ variants with discrete entries and/or fixed-$m$ sampling follow by the computational equivalence of these models (see \autoref{subsubsec:equiv_noise_model}, \autoref{subsubsec:equiv_sampling_pois} and \autoref{subsec:comp_equiv}, \autoref{subsec:poi_sampling}) and the proof similarly holds for $\kmodel^{\mathsf{wor}}$ models, which sample index sets without replacement.

\begin{definition}[$k$-Multiset of $\sqb{n}$]\label{def:multiset}
    A $k$-multiset $\alpha = \set{i_1,\dots,i_k}$ is an unordered collection of (not necessarily distinct) $k$ elements in $\sqb{n}$. We denote the set of all $k$-multisets of $\sqb{n}$ as $\set{n}^k$ and define a measure $\mu$ on $\set{n}^k$ as:
    $$
    \mu\paren{\set{i_1,\dots,i_k}} \propto \b| \paren{j_1,\dots,j_k} \in \sqb{n}^k:\, \set{j_1,\dots,j_k} = \set{i_1,\dots,i_k} \b| \qquad \text{for all }\set{i_1,\dots,i_k} \in \set{n}^k\,.
    $$
\end{definition}

\subsubsection{Proof Sketch for Prop.~\ref{prop:comp_equiv_full_collection}}
The goal is to show that $\gtenspvar {k} n {\eps} \defi {\Pois, \FULL}$ is equivalent to a collection of $\gtensppoi {(k-2r)} n {\eps_r} \defi$ instances for all $r = 0, 1,\dots, \lfloor \frac {k-1} 2 \rfloor$ sharing the same signal vector $x$: 
\begin{align*}
    \gtenspvar {k} n {\eps} \defi {\Pois, \FULL} \quad\leftrightarrow\quad \B\{ \gtensppoi {(k-2r)} n {\eps_r} \defi \B\}_{0\leq r \leq \frac{k-1}{2}}\,,\qquad \text{where } \eps_r = \min\set{ \frac{k\eps}{k-2r} , 1}\,.
\end{align*}

\paragraph{From $\ktensvar {\Pois, \FULL}$ to a collection of $\ktenspoi$.} For a target order $k' = k-2r$ for some $r \in \sqb{0, \frac{k-1}{2}}$, the reduction finds in the input instance $\calY = \set{\paren{\alpha_j, \Y_j}}_{j\in\sqb{m}} \sim \gtenspvar k n \eps \defi {\Pois, \FULL}$ equations $\paren{\alpha_j, \Y_j}$ such that 
\begin{equation}\label{eq:full_index_multiset}
\alpha_j = \beta_j \cup \set{i_1,i_1,\dots,i_r,i_r} \text{ for } \beta \in \binomset n {k'} \text{ and }i_1,\dots,i_r \in \sqb{n}\,.
\end{equation}
For $\alpha_j$ as above, 
$$
\Y_j = \delta \cdot x_{\beta_j} x_{i_1}x_{i_1} \dots x_{i_r} x_{i_r} + \N(0,1) = \delta\cdot x_{\beta_j} + \N(0,1)\,,
$$
where $\delta = n^{-k\eps/4 - \defi/2}$.
Then the equation $\paren{\beta_j, \Y_j}$ is distributed as an entry of $\ktensk {k'}^{\Pois}$ with signal level $\delta$. The expected number of $\alpha_j$ satisfying Eq.~\eqref{eq:full_index_multiset} is $\approx m \cdot n^{-r} = n^{k(1+\eps)/2 - r}$. The reduction constructs the output $\ktensk {k'}^{\Pois}$ out of such equations $\paren{\beta_j, \Y_j}$ differently, depending on whether there are few ($k(1+\eps)/2 - r \leq k' \Leftrightarrow \frac{k\eps}{k-2r} \leq 1$) or many ($k(1+\eps)/2 - r > k' \Leftrightarrow \frac{k\eps}{k-2r} > 1$) of them available. 
\begin{itemize}
    \item \textbf{Sparse case: few equations ($\frac{k\eps}{k-2r} \leq 1$).} The reduction returns $m'$ equations with signal level $\delta' = \delta$, where 
    $$
    m' \approx n^{k(1+\eps)/2 - r} = n^{k'(1+\eps_r)/2}\quad\text{and}\quad \delta' = \delta = n^{-k\eps/4 - \defi/2} = n^{-k'\eps_r - \defi/2}\,,
    $$
    where $\eps_r = \frac{k\eps}{k-2r}$.
    \item \textbf{Dense case: many equations ($\frac{k\eps}{k-2r} > 1$).} Since there are more than $n^{k'}$ equations available, the reduction aggregates the ones with the same index sets by averaging the corresponding values $\Y_{\alpha_j}$. As we will have $\approx n^{k(1+\eps)/2 - r}n^{-k'} = n^{k(-1+\eps)/2 + r}$ equations for every index support, the resulting parameters are $m', \delta'$:
    $$
    m' \approx n^{k'} = n^{k'(1+\eps_r)/2}\quad\text{and}\quad \delta' \approx \delta \cdot n^{k(-1+\eps)/4 + r/2} = n^{-k'\eps_r/4 - \defi/2}\,,
    $$
    where $\eps_r = 1$.
\end{itemize}

\paragraph{From a collection of $\ktenspoi$ to $\ktensvar {\Pois, \FULL}$.} The reverse reduction combines the entries of the collections of $\ktenspoi$ instances in one $\ktensvar {\Pois, \FULL}$, where the reverse procedure to the aggregation in Case 2 above is the cloning procedure (\autoref{subsec:cloning_splitting}).

\subsubsection{Proof of Prop.~\ref{prop:comp_equiv_full_collection}}

\begin{algorithm}[h]\SetAlgoLined\SetAlgoVlined\DontPrintSemicolon
    \KwIn{$\calY = \set{ \b( \alpha_j, \Y_{j}\b)}_{j\in\sqb{\mpoi}}$, where $\forall j \in\sqb{\mpoi}$, $\b(\alpha_j, \Y_{j}\b) \in \binomset n k \times \R$
        }
    \KwOut{$\calZ\up{k-2r} = \set{ \b( \alpha_j\up{k-2r}, \Zt_{j}\up{k-2r}\b)}_{j\in\sqb{\mpoi\up{k-2r}}}$, where $\forall j \in\sqb{\mpoi\up{k-2r}}$, $\b(\alpha_j\up{k-2r}, \Zt_{j}\up{k-2r}\b) \in [n]^{k-2r} \times \R$ for $0 \leq r \leq \lfloor \frac{k-1}{2}\rfloor$.
        }

    \BlankLine
    \For{ $r \in \sqb{0, \lfloor \frac{k-1}{2}\rfloor} $}{

    \tcp{Create all possible equations of order $k' = k-2r$} 
    
    $
    \calZ = \set{\paren{\beta_j, \Y_j}:\, j \in\sqb{\mpoi},\alpha_j = \beta_j \cup \set{i_1,i_1,\dots,i_r,i_r} , \beta \in \binomset n {k'}, \text{ and }i_1,\dots,i_r \in \sqb{n}\beta }
    $
    
    \If{$\frac{k\eps}{k-2r} \leq 1$}{
    
        $\calZ \up {k-2r} \gets \calZ$
    
    }\Else{

        $\calZ \up {k-2r} \gets \emptyset$
        
        \ForAll{$\beta \in \binomset n {k-2r}$}{

            \tcp{Aggregate equations with the same index set}
            $\N_{\beta} = \set{j \in \sqb{|\calZ|}:\, \beta_j = \beta}$

            \tcp{See pf. of Prop.~\ref{prop:comp_equiv_full_collection} for def. of $N_r$}
            $\N_\beta \gets \text{random }\min\b\{N_r, |\N_\beta|\b\}\text{-subset of } \N_\beta$ 

            \tcp{Denote $\calZ = \set{\paren{\beta_j, \Zt_j}}_{j\in\sqb{|\calZ|}}$}
            $\calZ \up {k-2r} \gets \calZ \up {k-2r} \cup \set{\paren{\beta, \frac{\sum_{j\in\N_\beta} \Zt_j}{ \sqrt{|\N_\beta|}}}}$
        }

        $\calZ \up {k-2r} \gets \ToPoiFromTPCA\paren{\calZ \up {k-2r}}$
    
    }
    }
    
    \BlankLine
    Return $\set{\calZ \up {k-2r}}_r$.
    
    \caption{$\Decompose\paren{\calY, \eps}$}
\label{alg:decompose_full}
\end{algorithm}

\paragraph{ From $\ktensvar {\Pois, \FULL}$ to a collection of $\ktenspoi$: Guarantees for $\Decompose(\calY, \eps)$.}

    We show that Alg.~\ref{alg:decompose_full} is an average-case reduction from $\gtenspvar k n {\eps} \defi {\Pois, \FULL}$, with $\mpoi \sim \poi(m)$ equations to a collection $\B\{ \gtensppoi {(k-2r)} n {\eps_r} \defi \B\}_{0\leq r \leq \frac{k-1}{2}}$ with $\eps_r = \min\set{ \frac{k\eps}{k-2r} , 1}$. We show the proof for the planted case and the same argument holds for the null case. 
    
    Denote the input $\calY = \set{ \b( \alpha_j, \Y_{j}\b)}_{j\in\sqb{\mpoi}} \sim \gtenspvar k n \eps \defi {\Pois, \FULL}$ and let $m = n^{k(1+\eps)/2}, \delta = n^{-k\eps/4-\defi/2}$, so that $\mpoi \sim \poi(m)$. $\Decompose\paren{\calY, \eps}$ returns a collection of instances $\set{\calZ \up {k-2r}}_r$, which are independent, since the entries of $\calY$ are not reused for different $r$. Moreover, for every $r$ in the loop, the collection $\calZ$ contains independent pairs $\paren{\beta_j, \Zt_j}$ for $\beta_j \sim_{i.i.d.} \unif\paren{\binomset n {k-2r}}$ and 
    $$
    \Zt_j = \delta\cdot x_{\beta_j} + \N(0,1)\,.
    $$
    Note that by Poisson splitting, $\calY$ contains $\poi\paren{m \mu(\alpha_j)}$ equations with index multiset $\alpha_j \in \set{n}^k$, where $\mu$ is the measure defined in Def.~\ref{def:multiset}. Then, $\calZ\up{k-2r} = \calZ$ contains $\poi\paren{\Theta\paren{ m n^{-r} n^{-(k-2r)}}}$ entries for each $\beta \in \binomset n {k-2r}$. 

    \paragraph{Case 1: few entries in $\calZ$ ($\frac{k\eps}{k-2r}\leq 1$).} In this case, by the above, 
    $$
    \law\paren{ \calZ\up{k-2r} } = \gtenspoi {(k-2r)} {n} {m'} {\delta'}\,,
    $$
    where 
    $$
    m' = \Theta(mn^{-r}) = \Theta\paren{n^{k'(1+\eps_r)/2}}\quad\text{and}\quad \delta' = \delta = n^{-k\eps/4 - \defi/2} = n^{-(k-2r)\eps_r/4 - \defi/2}\,,
    $$
    and $\eps_r = \frac{k\eps}{k-2r}$. Since it is sufficient to map parameters up to $\poly\log n$ factors (Prop~\ref{prop:map_up_to_constant}), this concludes the proof of the average-case reduction.

    \paragraph{Case 2: many entries in $\calZ$ ($\frac{k\eps}{k-2r}> 1$).} 
    By Poisson concentration in Prop.~\ref{prop:pois_conc}, since $|\N_\beta = \set{j \in \sqb{|\calZ|}:\, \beta_j = \beta }| \sim \poi \paren{\Theta\paren{ m n^{-r} n^{-(k-2r)}}}$, with high probability simultaneously for all $\beta \in \binomset n {k-2r}$, 
    $$
    |\N_\beta| \geq \Thetat(1) \cdot m n^{-r} n^{-(k-2r)} = C \cdot n^{k(-1+\eps)/2 + r} \eqcolon N_r\,,
    $$
    where $C = \Thetat(1)$ hides factors constant in $k, r$ and $\poly\log n$. Conditioned on this event, which we denote, $S_N$, every entry of $\calZ\up{k-2r}$ is an average of $N_r$ observations of the form $\delta\cdot x_{\beta} + \N(0,1)$. Then, 
    $$
    \law\paren{\calZ\up{k-2r} | S_N } = \pcag {(k-2r)} n {\delta'}\,,
    $$
    where $$\delta' = \Thetat\paren{\delta \cdot n^{k(-1+\eps)/4 + r/2}} = \Thetat\paren{n^{-(k-2r)\eps_r/4-\defi/2}}\,,$$
    where $\eps_r = 1$. Therefore, by conditioning in total variation (Fact~\ref{tvfacts}), 
    $$\tv\paren{\law\paren{\calZ\up{k-2r}}, \pcag {(k-2r)} n {\delta'} } \leq \Pr\sqb{S_N^c} = o_n(1)\,.$$
    Since it is sufficient to map parameters up to $\poly\log n$ factors (Prop~\ref{prop:map_up_to_constant}), the average-case reduction follows from the guarantees for the $\ToPoiFromTPCA$ algorithm in Prop.~\ref{prop:sampling_equiv_eps1}.

\paragraph{ From a collection of $\ktenspoi$ to $\ktensvar {\Pois, \FULL}$: Guarantees for $\Compose(\set{\calY \up {k-2r}}_r, \eps)$.}

\begin{algorithm}[h]\SetAlgoLined\SetAlgoVlined\DontPrintSemicolon
    \KwIn{$\calY\up{k-2r} = \set{ \b( \alpha_j\up{k-2r}, \Y_{j}\up{k-2r}\b)}_{j\in\sqb{\mpoi\up{k-2r}}}$, where $\forall j \in\sqb{\mpoi\up{k-2r}}$, $\b(\alpha_j\up{k-2r}, \Y_{j}\up{k-2r}\b) \in [n]^{k-2r} \times \R$ for $0 \leq r \leq \lfloor \frac{k-1}{2}\rfloor$, $\eps \in \sqb{0,1}$
        }
    \KwOut{
    $\calZ = \set{ \b( \beta_j, \Zt_{j}\b)}_{j\in\sqb{\mpoi}}$, where $\forall j \in\sqb{\mpoi}$, $\b(\beta_j, \Zt_{j}\b) \in \binomset n k \times \R$.
        }

    \BlankLine
    \For{ $r \in \sqb{0, \lfloor \frac{k-1}{2}\rfloor} $}{
    
    \If{$\frac{k\eps}{k-2r} > 1$}{

        $\widetilde\calY\up {k-2r} \gets \ToTPCA\paren{\calY\up {k-2r}}$
        
        \ForAll{$j \in \binom n {k-2r}$}{

            \tcp{Here $\GaussClone(\cdot, A)$ is the $\GaussClone(\cdot)$ procedure repeated $\lceil \log A\rceil$ times to obtain $A$ clones.}
            $\Y^1,\dots,\Y^{A_r} \gets \GaussClone(\Y_j\up{k-2r}, A_r)$ \tcp{See pf. of Prop.~\ref{prop:comp_equiv_full_collection} for def. of $A_r$}

            $m_\beta \gets \poi\paren{A_r / 2}, m_\beta \gets \min\set{A_r, m_\beta}$

            $\widetilde\calY\up {k-2r} \gets \widetilde\calY\up {k-2r} \setminus \set{\alpha_j, \Y_j} \cup \set{\paren{\alpha_j, \Y^i}}_{i\in\sqb{m_\beta}}$
        }

        $\calZ \up {k-2r} \gets \ToPoiFromTPCA\paren{\calZ \up {k-2r}}$
    
    }\Else{
        $\widetilde\calY\up {k-2r} \gets \calY\up{k-2r}$
        }
    
    \BlankLine
    \tcp{Create a random output $k$-sparse equation out of each input one} 

    $\calZ \gets \emptyset$

    \ForAll{$j \in \sqb{\mpoi\up{k-2r}}$}{
        $\beta \gets \mu_{\alpha_j\up{k-2r}}$, where $\mu_{\alpha}, \alpha \in \binomset n k$ is law of $\mu$ conditioned on the output equal to $\alpha \cup \set{i_1,i_1,\dots,i_r,i_r}, \text{ where }i_1,\dots,i_r \in \sqb{n}\alpha$

        $\calZ \gets \calZ \cup \set{\paren{\beta, \Y_j\up{k-2r}}}$

    }
    }
    \BlankLine
    Return $\calZ$.
    
    \caption{$\Compose\paren{\set{\calY \up {k-2r}}_r, \eps}$}
\label{alg:compose_full}
\end{algorithm}

    We show that Alg.~\ref{alg:compose_full} is an average-case reduction from a collection $\B\{ \gtensppoi {(k-2r)} n {\eps_r} \defi \B\}_{0\leq r \leq \frac{k-1}{2}}$ with $\eps_r = \min\set{ \frac{k\eps}{k-2r} , 1}$ to $\gtenspvar k n {\eps} \defi {\Pois, \FULL}$. We show the proof for the planted case and the same argument holds for the null case. 
    
    Let the inputs $\calY\up{k-2r} \sim \gtensppoi k n {\eps_r} \defi$ be $\calY\up{k-2r} = \set{ \b( \alpha_j\up{k-2r}, \Y_{j}\up{k-2r}\b)}_{j\in\sqb{\mpoi\up{k-2r}}}$ and denote $m_r = n^{(k-2r)(1+\eps_r)/2}, \delta_r = n^{-(k-2r)\eps_r/4-\defi/2}$. $\Compose\paren{\set{\calY\up {k-2r}}_r, \eps}$ returns an instance $\calZ$ and we will show that
    $$
    \tv\paren{\law\paren{\calZ}, \gtensvar k n {m'} {\delta'} {\Pois, \FULL}} = o_n(1)\,,
    $$
    where $m' = \Thetat\paren{n^{k(1+\eps)/2}}, \delta' = \Thetat\paren{n^{-k\eps/4 - \defi/2}}$ match the parameters of $\gtenspvar k n {\eps} {\defi} {\Pois, \FULL}$ up to $\poly\log n$ factors and constants. By Prop.~\ref{prop:map_up_to_constant}, this proves the desired average-case reduction. 

    Consider any index set $\beta \in \binomset n k$ and assume $\beta =  \alpha \cup \set{i_1,i_1,\dots,i_r,i_r}, \text{ where }i_1,\dots,i_r \in \sqb{n}$ for some $r \in \sqb{0, \frac{k-1}{2}}$, i.e., $\alpha$ is the set of distinct indices that are included an odd number of times in $\beta$. 

    \paragraph{Case 1: $\frac{k\eps}{k-2r} \leq 1$.} In this case, by Poisson splitting, since the input instance $\calY\up {k-2r}$ contained $\poi\paren{m\cdot \binom n {k-2r}^{-1}}$ equations with the index set $\alpha$, the output instance has $\poi\paren{m_r\cdot \binom n {k-2r}^{-1} \cdot \Theta(n^{-r})}$ equations with the index set $\beta$, where 
    $$
    m_r\cdot \binom n {k-2r}^{-1} \cdot \Theta(n^{-r}) = \Theta\paren{ n^{(k-2r)(1+\eps_r)/2 - (k-2r) - r} } = \Theta\paren{n^{k(\eps-1)/2}}\,.
    $$

    \paragraph{Case 2: $\frac{k\eps}{k-2r} > 1$.} In this case we again want $\poi\paren{\Theta\paren{n^{k(\eps-1)/2}}}$ equations with index set $\beta$ in the output index and we set $A_r/2 = \Theta\paren{n^{k(\eps-1)/2}} \cdot \Theta\paren{n^r}$ for the appropriate choice of constant. Then, with high probability, the output will contain $\poi\paren{A_r/2 \cdot \Theta(n^{-r})}$ equations with this index set, each independent of the form $\delta' \cdot x_{\beta} + \N(0,1)$, where 
    $$
    \delta'  = \delta \cdot \Theta\paren{A_r^{-1/2} n^{r/2}} = \Theta\paren{n^{-(k-2r)/4 - \defi/2 - k(\eps-1)/4}} = \Theta\paren{n^{-k\eps/4 - \defi/2}}\,.
    $$
    Denote this high probability event $S_r$ and their union $S = S_1 \cap \dots \cap S_{\lfloor \frac{k-1}{2} \rfloor}$.

    The two cases above show that  
    $$
    \law\paren{\calZ | S} = \gtensvar k n {m'} {\delta'} {\Pois, \FULL}
    $$
    for the appropriate $m', \delta'$ up to $\poly\log n$ factors. By conditioning on event in total variation (Fact~\ref{tvfacts}), we conclude that 
  \begin{equation*}
    \tv\paren{\law\paren{\calZ}, \gtensvar k n {m'} {\delta'} {\Pois, \FULL}} = \Pr\sqb{S^c} = o_n(1)\,.
    \qedhere \end{equation*}

\begin{remark}[Discrete Inputs to Alg.~\ref{alg:decompose_full}, \ref{alg:compose_full}]\label{rmk:to_full_discrete_input}
    While Algorithms \ref{alg:decompose_full} and \ref{alg:compose_full}, as stated, are reductions for the Gaussian $\ktenspoi$, $\ktensvar{\Pois, \FULL}$ models, we sometimes use them as subroutines for Discrete models and/or models with the fixed number of samples. It is implicitly assumed that first a Gaussianization (\autoref{subsec:gaussianize}) and/or Poissonization (\autoref{subsec:poi_sampling}) reduction is applied, then the appropriate algorithm, and then the discretization (\autoref{subsec:discretize}) and/or Fixed-Sample (\autoref{subsec:poi_sampling}) reduction.
\end{remark}

\subsection{Symmetrization and Symmetry Breaking}\label{subsec:symm}

While the standard option for $\kxor$ sampling has index sets $\alpha \in \binomset n k$, the ``asymmetric" variant $\kxor^{\asymm}$ instead indexes equations by \emph{ordered tuples of distinct indices} $\alpha \in \sqb{n}^{\underline k}$. We show that in the case when $\alpha$ are sampled with replacement, the two models are computationally equivalent.

\begin{myprop}{\ref{prop:symm_asymm_equiv}}\textnormal{(Computational Equivalence of $\kmodel^\star$ and $\kmodel^{\star,\mathsf{asymm}}$)}
    Let $\star \subseteq \set{\Pois}$ and $\kmodel \in \set{\kxor, \ktens}$ be any fixed choices (1,2) in Fig.~\ref{fig:kxor_variants_app} and choice (4) set to sampling with replacement. Then, $\kmodel^\star_{\eps, \defi}(n)$ and the corresponding ``asymmetric" variant $\kmodel^{\star, \mathsf{asymm}}_{\eps, \defi}(n)$ are computationally equivalent: there are $\poly(n)$-time average-case reductions ($\DeSymm/\Symm$, Alg.~\ref{alg:desymm}, \ref{alg:symm}) for both detection and recovery in both directions. Moreover, the reductions preserve the secret signal vector $x$ precisely.
\end{myprop}

We prove Prop.~\ref{prop:symm_asymm_equiv} for $\kxorpoi$ model. Recall that in $\xorppoi k n \eps \defi$ (equiv. $\xorpoi k n m \delta$ with $m = n^{k(1+\eps)/2}, \delta = n^{-k\eps/4}$) there are $\mpoi \sim \poi(m)$ independent entries $(\alpha, \Y)$ with $\Y \sim x_\alpha \cdot \rad(\delta)$ and $\alpha\sim_{i.i.d.} \unif\sqb{\binomset n k}$. The only difference in $\xorpvar k n \eps \defi {\Pois, \asymm}$ is that $\alpha$ are sampled i.i.d. from $\sqb{n}^{\underline k}$ -- the set of ordered $k$-tuples of distinct indices. 

The analogous results for the variants with Gaussian entries and/or fixed-$m$ sampling follow by the computational equivalence of these models (see \autoref{subsubsec:equiv_noise_model}, \autoref{subsubsec:equiv_sampling_pois}).

\begin{algorithm}\SetAlgoLined\SetAlgoVlined\DontPrintSemicolon
    \KwIn{$\calY = \set{ \b( \alpha_j, \Y_{j}\b)}_{j\in\sqb{\mpoi}}$, where $\forall j \in\sqb{\mpoi}$, $\b(\alpha_j, \Y_{j}\b) \in \binomset n k \times \R$
        }
    \KwOut{$\calZ = \set{ \b( \alpha_j', \Zt_{j}\b)}_{j\in\sqb{\mpoi}}$, where $\forall j \in\sqb{\mpoi}$, $\b(\alpha_j', \Zt_{j}\b) \in [n]^{\underline{k}} \times \R$.
        }

    \BlankLine
    $\forall j \in \sqb{\mpoi}$, $\sigma_j \gets \unif\sqb{\calS_{k}}$ and set
        $$
        \alpha_j' \gets \sigma_j(\alpha_j)\quad\text{and}\quad \Zt_j \gets \Y_j
        $$

    \BlankLine
    Return $\calZ = \set{ \b( \alpha_j', \Zt_{j}\b)}_{j\in\sqb{\mpoi}}$.
    
    \caption{$\DeSymm\paren{\calY}$}
\label{alg:desymm}
\end{algorithm}
\begin{algorithm}\SetAlgoLined\SetAlgoVlined\DontPrintSemicolon
    \KwIn{$\calY = \set{ \b( \alpha_j, \Y_{j}\b)}_{j\in\sqb{\mpoi}}$, where $\forall j \in\sqb{\mpoi}$, $\b(\alpha_j, \Y_{j}\b) \in [n]^{\underline{k}} \times \R$
        }
    \KwOut{$\calZ = \set{ \b( \alpha_j', \Zt_{j}\b)}_{j\in\sqb{\mpoi}}$, where $\forall j \in\sqb{\mpoi}$, $\b(\alpha_j', \Zt_{j}\b) \in \binomset n k \times \R$.
        }

    \BlankLine
    $\forall j \in \sqb{\mpoi}$, set
        $$
        \alpha_j' \gets \mathsf{set}(\alpha_j)\quad\text{and}\quad \Zt_j \gets \Y_j
        $$

    \BlankLine
    Return $\calZ = \set{ \b( \alpha_j', \Zt_{j}\b)}_{j\in\sqb{\mpoi}}$.
    
    \caption{$\Symm\paren{\calY}$}
\label{alg:symm}
\end{algorithm}
\begin{proof}[Proof of Prop.~\ref{prop:symm_asymm_equiv}]
    Given an input $\calY = \set{ \b( \alpha_j, \Y_{j}\b)}_{j\in\sqb{\mpoi}} \sim \xorppoi k n \eps \defi$, the reduction $\DeSymm$ (Alg.~\ref{alg:desymm}) randomly orders the index sets $\set{\alpha_j}_{j \in \sqb{\mpoi}}$ and preserves the equation values $\Y_j$ intact. Given an input $\calY = \set{ \b( \alpha_j, \Y_{j}\b)}_{j\in\sqb{\mpoi}} \sim \xorpvar k n \eps \defi {\Pois, \asymm}$, the reduction $\Symm$ (Alg.~\ref{alg:symm}) ``removes" the index order from the observations, preserving the equation values $\Y_j$ intact.

    By Poisson splitting, the above are the desired average-case reductions in both planted and null cases. 
\end{proof}

\subsection{Entry Sampling without Replacement}\label{subsec:w_replacement}

In this section we show:

\begin{myprop}{\ref{prop:w_wor_equiv_recovery}}\textnormal{(Computational Equivalence of $\kmodel^\star$ and $\kmodel^{\star,\wor}$ for Recovery)}
    Let $\star \subseteq \set{\Pois, \FULL/\asymm}$ and $\kmodel \in \set{\kxor, \ktens}$ be some fixed model choices (1-3) in Fig.~\ref{fig:kxor_variants_app}. The two models $\kmodel^\star_{\eps, \defi}(n)$ and $\kmodel^{\star, \wor}_{\eps, \defi}(n)$ are computationally equivalent for recovery: there are $\poly(n)$-time average-case reductions ($\ToTPCA/\ToPoiFromTPCA$, Alg.~\ref{alg:to_pca_sampling}, \ref{alg:to_poi_sampling_from_pca}) for recovery in both directions. Moreover, the reductions preserve the secret signal vector $x$ precisely.
\end{myprop}
Additionally, we prove the complete computational equivalence of $\kmodel^\star$ and $\kmodel^{\star, \wor}$ (for both detection and recovery) for $\eps = 1$:
\begin{myprop}{\ref{prop:sampling_equiv_eps1}}\textnormal{(Computational Equivalence of $\kmodel^\star$ and $\kmodel^{\star,\wor}$ for $\eps = 1$)}
    Let $\star \subseteq \set{\Pois, \FULL/\asymm}$ and $\kmodel \in \set{\kxor, \ktens}$ be some fixed model choices (1-3) in Fig.~\ref{fig:kxor_variants_app} and let $\defi > 0$. The two models $\kmodel^\star_{\eps=1, \defi}(n)$ and $\kmodel^{\star, \wor}_{\eps=1, \defi}(n)$ are computationally equivalent: there are $\poly(n)$-time average-case reductions ($\ToTPCA/\ToPoiFromTPCA$, Alg.~\ref{alg:to_pca_sampling}, \ref{alg:to_poi_sampling_from_pca}) for both detection and recovery in both directions. Moreover, the reductions preserve the secret signal vector $x$ precisely.
\end{myprop}

Both propositions are proved via two simple transformations. In the direction “with replacement $\to$ without replacement”, we delete duplicates by keeping one noisy observation per distinct index set (Alg.~\ref{alg:to_pca_sampling}). In the reverse direction, we synthetically introduce duplicates using the cloning procedure from \autoref{subsubsec:cloning} (Alg.~\ref{alg:to_poi_sampling_from_pca}).

\begin{algorithm}[h]\SetAlgoLined\SetAlgoVlined\DontPrintSemicolon

    \KwIn{$\calY = \set{ \b( \alpha_j, \Y_{j}\b)}_{j\in\sqb{m}}$, $\eps \in \sqb{0,1}$;
        }
    \KwOut{$\calZ = \set{ \b( \alpha_j', \Zt_{j}\b)}_{j\in\sqb{m'}}$.
        }

    \BlankLine
    $\mathcal{I} \gets \set{\alpha_j }_{j\in m}$ \ \tcp{$\mathcal{I}$ is the set of distinct index sets appearing in $\calY$}

    $\mathcal{I} \gets \text{random }\min\set{n^{k(1+\eps)/2}/4,|\mathcal{I}|}\text{-subset of }\mathcal{I}$

    \BlankLine
    $\calZ\gets \emptyset$
    
    \tcp{Pick a random input noisy observation for each index set}
    \ForAll{$\alpha \in \mathcal{I}$}{
        $\Zt_{\alpha} \gets \unif\sqb{\set{\Y_j:\, \alpha_j = \alpha}}$

        $\calZ \gets \calZ \cup \set{(\alpha, \Zt_\alpha)}$
    }
    
    \BlankLine
    Return $\calZ$.
    
    \caption{$\ToTPCA\paren{\calY, \eps}$}
\label{alg:to_pca_sampling}
\end{algorithm}

\begin{algorithm}[h]\SetAlgoLined\SetAlgoVlined\DontPrintSemicolon

    \KwIn{$\calY = \set{ \b( \alpha_j, \Y_{j}\b)}_{j\in\sqb{m}}$;
        }
    \KwOut{$\calZ = \set{ \b( \alpha_j', \Zt_{j}\b)}_{j\in\sqb{m'}}$.
        }

    \BlankLine
    Sample $m$ elements $a_1, \dots, a_{m}$ uniformly from $\sqb{|\mathsf{A}|}$ with replacement

    \BlankLine
    Sample $s = |\set{a_i}_{i\in\sqb{m}}|$ elements $\alpha\up 1, \dots, \alpha \up {s}$ from $\set{\alpha_j}_{j\in\sqb{m}}$ at random without replacement 

    \BlankLine
    Let $b_1, \dots, b_s$ be a random enumeration of distinct elements $\set{a_i}_{i\in\sqb{m}}$

    \BlankLine
    \tcp{Compute maximum multiplicity of sampling with replacement}
    For all $i \in \sqb{s}$: $m_i = |\set{ j \in \sqb{m}:\, a_j = b_i }|$

    $m_c \gets \max_{i\in\sqb{s}} m_i$

    \BlankLine
    $\calY\up 1, \dots, \calY \up {m_c} \gets \XorClone\paren{\calY, m_c}$ \ \tcp{$\kTensClone$ if $\kmodel = \ktens$; see Alg.~\ref{alg:xor_clone}, Rmk.~\ref{rmk:clone_several}}
    
    \BlankLine
    $\calZ = \emptyset$
    \ForAll{$ i \in \sqb{s}$}{
        Let $j$ be the unique index s.t. $\alpha_j = \alpha \up i$

        $\calZ \gets \calZ \cup \set{\paren{\alpha_j, \Y_j\up t}}_{t=1}^{m_i}$
    }

    \BlankLine
    Return $\calZ$.
    
    \caption{$\ToPoiFromTPCA\paren{\calY, \mathsf{A} = \binomset n k}$}
\label{alg:to_poi_sampling_from_pca}
\end{algorithm}

\begin{proof}{of Prop.~\ref{prop:w_wor_equiv_recovery} and Prop.~\ref{prop:sampling_equiv_eps1}}
    
    \textbf{Alg.~\ref{alg:to_pca_sampling}: from $\kmodel^{\star}_{\eps, \defi}(n)$ to $\kmodel^{\star, \wor}_{\eps, \defi}(n)$.} Let $\calY = \set{ \b( \alpha_j, \Y_{j}\b)}_{j\in\sqb{m}}$ be drawn from $\kmodel^{\star}_{\eps, \defi}(n)$ (or the corresponding null distribution), where $\alpha_j$ are sampled uniformly from $\mathsf{A}$ with replacement and $\delta = n^{-k\eps/4-\defi/2}$). Alg.~\ref{alg:to_pca_sampling} forms the set $\mathcal I=\set{\alpha_j}_{j\in[m]}$ of distinct index sets appearing in $\calY$, then selects
    $m' = \min\set{n^{k(1+\eps)/2}/4, |\mathcal I|}$
    distinct indices from $\mathcal I$, and for each chosen $\alpha$ outputs one uniformly random associated observation. Hence the output $\calZ = \ToTPCA\paren{\calY, \eps}$ satisfies
    $$
    \law\paren{\calZ} = \kmodel^{\star, \wor}(n, m', \delta)\,.
    $$
    It remains to show that $m' = n^{k(1+\eps)/2}/4$ with high probability. The map $(\alpha_1,\dots,\alpha_m) \to |\mathcal I|$ is 1-Lipschitz, so, combined with the standard bounds on $\E|\mathcal I|$, we have, with high probability, $$
    |\mathcal I| \geq m / 2\,.
    $$
    In the Poissonized case ($\Pois \in \star$) we have
    $m\sim \Pois(n^{k(1+\eps)/2})$, so by Poisson concentration (Prop.~\ref{prop:pois_conc}), $m \geq n^{k(1+\eps)/2}/2$. Therefore, for all models, with high probability, $m' = n^{k(1+\eps)/2}/4$. By conditioning on this high probability event in $\tv$ (Fact~\ref{tvfacts}), we conclude 
    $$
    \tv\paren{\law(\calZ), \kmodel^{\star, \wor}(n, m'=n^{k(1+\eps)/2}/4, \delta)} = o_n(1)\,,
    $$
    which yields an average-case reduction for both detection and recovery, since it is sufficient to map parameters $m, \delta$ up to constant factors (Prop.~\ref{prop:map_up_to_constant}).

    \textbf{Alg.~\ref{alg:to_poi_sampling_from_pca}: from $\kmodel^{\star,\wor}_{\eps, \defi}(n)$ to $\kmodel^{\star}_{\eps, \defi}(n)$.} We first consider the planted case. Let $\calY = \set{ \b( \alpha_j, \Y_{j}\b)}_{j\in\sqb{m}}$ be drawn from $\kmodel^{\star,\wor}_{\eps, \defi}(n)$, where $\alpha_j$ are sampled from a set $\mathsf{A}$ without replacement and denote $m = n^{k(1+\eps)/2}$, $\delta = n^{-k\eps/4-\defi/2}$). 
    Alg.~\ref{alg:to_poi_sampling_from_pca} first samples i.i.d. labels $a_1,\dots,a_m \in \sqb{|\mathsf{A}|}$ (with replacement), computes their multiplicities $\set{m_i}$ and the maximum $m_c = \max_i m_i$, then produces $m_c$ independent clones of the input instance and uses them to realize the desired multiplicity pattern. By Lemma~\ref{lem:instance_clone}, the output $\calZ = \ToPoiFromTPCA\paren{\calY, \mathsf{A}}$ has independent entries and the correct ``with replacement" index distribution; moreover the entrywise signal level is 
    $$
    \delta' = \Thetat(\delta/\sqrt{m_c})\,.
    $$
    Therefore, in the planted case,
    $$
    \law\paren{\calZ} = \kmodel^{\star}(n, m ,\delta')\,.
    $$
    It remains to control $m_c$:
    in the standard ``$m$ balls into $|\mathsf{A}|\geq m$ bins" model, with high probability, the maximum load $m_c = \max_{i \in \sqb{s}} m_i = \Thetat(1)$, yielding $\delta' = \Thetat(\delta)$. Since it is sufficient to map parameters $m, \delta$ up to polylog factors (Prop.~\ref{prop:map_up_to_constant}), this proves that Alg.~\ref{alg:to_poi_sampling_from_pca} is an average-case reduction for recovery. 

    In the null case, cloning replicates the random $\pm 1$ signal component (see Def.~\ref{def:null_models}), so the output does not exactly match the intended null distribution. When $\eps=1$, however, the null model is $o_n(1)$-close in total variation to the distribution with $\delta = 0$ (Prop.~\ref{prop:tv_null}); in this special case the same argument yields the reduction for detection as well, completing Prop.~\ref{prop:sampling_equiv_eps1}. 
    
\end{proof}

\section{Simple Reductions}\label{sec:simple_reductions}

\subsection{Instance Cloning and Splitting}\label{subsec:cloning_splitting}

As described in \autoref{subsubsec:cloning_splitting_overview}, cloning and splitting reductions map one input instance $\calY \sim \kmodel^{\star}$ to two output instances $\calY\up1, \calY\up2 \sim \kmodel^\star$ sharing the same signal $x$. 

\subsubsection{Cloning Reduction}\label{subsubsec:cloning}

Cloning keeps the same index sets as the input but ensures the noise independence across the two outputs:

\begin{mylem}{\ref{lem:instance_clone}}\textnormal{(Instance Cloning)}
Let $\star \subseteq \set{\Pois, \FULL/\asymm, \wor}$ and $\kmodel \in \set{\kxor, \ktens}$ be some fixed model choices (1-4) in Fig.~\ref{fig:kxor_variants_app}. There exists a poly-time algorithm ($\XorClone/\kTensClone$, Alg.~\ref{alg:xor_clone}) that, given an input $\calY \sim \kmodel^{\star}(n, m, \delta)$, returns two instances $\calY\up1, \calY\up 2 \sim \kmodel^{\star}(n, m, \Thetat(\delta))$ that have independent entrywise noise, the same index sets, and the same signal vector $x\in\set{\pm1}^n$ as the input $\calY$. 
\end{mylem}

\begin{algorithm}[h]\SetAlgoLined\SetAlgoVlined\DontPrintSemicolon

    \KwIn{$\calY = \set{ \b( \alpha_j, \Y_{j}\b)}_{j\in\sqb{m}}$.
        }
    \KwOut{$\calY\up1, \calY\up 2$, where $\calY\up a = \set{ \b( \alpha_j, \Y_{j} \up a \b)}_{j\in\sqb{m}}$.}

    \BlankLine
    For each $j \in \sqb{m}$:
    $$
    \paren{\Y_{j}\up1, \Y_{j}\up2} \gets 
        \GaussClone \paren{\Y_{j}}
    $$
    
    \BlankLine
    Return $\calY\up1, \calY\up 2$, where $\calY\up a = \set{ \b( \alpha_j, \Y_{j} \up a \b)}_{j\in\sqb{m}}$ for $a \in\set{1,2}$.
    
    \caption{$\kTensClone\paren{\calY}$}
\label{alg:xor_clone}
\end{algorithm}

The reduction proving Lemma~\ref{lem:instance_clone} uses the standard Gaussian-cloning identity (e.g. \cite{brennan2018reducibility}): given $Y\sim\N(\mu,1)$, generate an independent $Z\sim\N(0,1)$ and set $Y\up 1 = \frac {Y-Z} {\sqrt{2}}$, $ Y\up 2 = \frac{Y+Z}{\sqrt{2}}$. Then $Y\up 1, Y\up2$ are independent and each is distributed as $\N(\mu/\sqrt{2},1)$. The primitive is applied entrywise for $\ktens$ inputs ($\kTensClone$ Alg.~\ref{alg:xor_clone}). For $\kxor$ inputs, we first reduce to the Gaussian analog via \autoref{subsubsec:equiv_noise_model}, apply the cloning primitive entrywise, then (if needed) map back; we denote this algorithm $\XorClone$.

We sometimes apply the cloning reduction recursively to obtain more than two clones (Rmk.~\ref{rmk:clone_several}): we denote the corresponding maps to $A > 2$ clones $\kTensClone(\calY, A)$ and $\XorClone(\calY, A)$.

\begin{proof}[Proof of Lemma~\ref{lem:instance_clone}]
    
    For the $\GaussClone$ primitive above, we have the following guarantee:
    \begin{proposition}[Gaussian Cloning]\label{prop:entry_clone}
    There exists an $O(1)$-time algorithm $\GaussClone : \R \to \R^2$ such that given $Y \sim \N(\mu, 1)$ as an input, $\GaussClone(Y)$ outputs $Y_1, Y_2 \in \R$ that satisfy 
        $$
        \law(Y_1, Y_2) = \N(\mu/\sqrt{2}, 1) \times \N(\mu / \sqrt{2}, 1)\,,
        $$
        i.e., $Y_1, Y_2$ are sampled independently from $\N(\mu/\sqrt{2}, 1)$. Note that the algorithm works simultaneously for all values $\mu \in \R$.
\end{proposition}
For $\ktens$ inputs, $\kTensClone$ (Alg.~\ref{alg:xor_clone}) applies $\GaussClone$ entrywise, so the Lemma statement follows from Prop.~\ref{prop:entry_clone}.

For $\kxor$ inputs, $\XorClone$ first applies Gaussianization, then $\kTensClone$, and then Discretization. The lemma follows from Prop.~\ref{prop:entry_clone} and Gaussianization/Discretization guarantees in \autoref{subsubsec:equiv_noise_model}.
\end{proof}

\subsubsection{Splitting Reduction}\label{subsubsec:splitting}

Splitting makes the two outputs fully independent instances (independent indices and independent noise), while preserving the same signal $x$.

\begin{mylem}{\ref{lem:instance_split}}\textnormal{(Instance Splitting)} Let $\star \subseteq \set{\Pois, \FULL/\asymm}$ and $\kmodel \in \set{\kxor, \ktens}$ be some fixed model choices (1-3) in Fig.~\ref{fig:kxor_variants_app} and choice (4) set to sampling with replacement. There exists a poly-time algorithm ($\kTensSplit$) that, given an input $\calY \sim \kmodel^{\star}(n, m, \delta)$, returns two independent instances $\calY\up1, \calY\up 2 \sim \kmodel^{\star}(n, \Theta(m), \delta)$ that share the same signal vector $x\in\set{\pm1}^n$ as the input $\calY$. 
\end{mylem}
Remarks~\ref{rmk:clone_const}, \ref{rmk:clone_several} similarly apply. Lemma~\ref{lem:instance_split} follows from a simple entry-allocation argument in two cases:
\begin{itemize}
    \item If $\Pois \not\in \star$, we randomly partition $m$ input samples into two groups of size $m/2$ to form the outputs $\calY\up1, \calY\up 2$. This achieves a reduction to two instances of $\kmodel^{\star}(n, m/2, \delta)$.
    \item If $\Pois \in \star$, we allocate each of $m$ input entries to $\calY\up1 $ or $\calY\up 2$ independently at random. This yields a reduction to two instances of $\kmodel^{\star}(n, m/2, \delta)$ by Poisson splitting. 
\end{itemize}
We denote the reductions above as $\kTensSplit(\calY)$. Similarly to cloning, we write $\kTensSplit(\calY, A)$ to denote a reduction that does splitting several times to output $A$ instances (Rmk.~\ref{rmk:clone_several}).

\subsection{Reductions that Change the Number of Revealed Equations $m$}\label{subsec:entry_adjust_app}

We prove simple reductions that modify the number of samples in $\kmodel^\star(n,m,\delta)$ from $m$ to $m'$ by either (1) injecting pure noise observations ($\DenseRed$), or (2) removing some of the samples ($\SparseRed$). Lemma~\ref{lem:m_adjust} and Corollary~\ref{cor:m_adjust} track the resulting change in signal level $\delta$ (or, in $\eps,\defi$ parametrization, the deficiency $\defi$).

\begin{mylem}{\ref{lem:m_adjust}}\textnormal{(Adjusting the Number of Samples $m$)}
    Let $\star \subseteq \set{\Pois, \FULL/\asymm, \wor}$ and $\kmodel \in \set{\kxor, \ktens}$ be some fixed model choices (1-4) in Fig.~\ref{fig:kxor_variants_app} and let $k, n \in \Z^{+}$, $\delta \in \R$ and $m,m' \in\Z^{+}$ be parameters. There exists a poly-time average-case reduction ($\DenseRed/\SparseRed$, Alg.~\ref{alg:dense_red}, \ref{alg:sparse_red}) for both detection and recovery $$\text{from}\quad\kmodel^\star(n,m,\delta)\quad\text{to}\quad\kmodel^\star(n,m',\delta')\,, \qquad\text{with }\delta' = \begin{cases}
        \delta\,, &m' \leq m\\
        \Thetat(\delta\cdot m/m')\,, &m' > m\,. \end{cases}$$
        Moreover, the reduction preserves the secret signal vector $x$ precisely.
\end{mylem}
In the $\eps, \defi$ parametrization this yields:
\begin{mycor}{\ref{cor:m_adjust}}
    Let $\star \subseteq \set{\Pois, \FULL/\asymm, \wor}$ and $\kmodel \in \set{\kxor, \ktens}$ be some fixed model choices (1-4) in Fig.~\ref{fig:kxor_variants_app} and let $k, n \in \Z^{+}$, $\eps, \eps' \in [0,1]$, and $\defi \in \R$ be parameters. There exists a poly-time average-case reduction ($\DenseRed/\SparseRed$, Alg.~\ref{alg:dense_red}, \ref{alg:sparse_red}) for both detection and recovery 
\begin{align*}&\text{from}\quad\kmodel^\star_{\eps,\defi}(n)\quad\text{to}\quad\kmodel^\star_{\eps',\defi'}(n)\,, \quad \text{where } \defi' = \defi + {k|\eps-\eps'|/2}\,.
    \end{align*}
    Moreover, the reduction preserves the secret signal vector $x$ precisely.
\end{mycor}

As stated, $\DenseRed$ (Alg.~\ref{alg:dense_red}) and $\SparseRed$ (Alg.~\ref{alg:sparse_red}) accept and return discrete Poissonized $\kxor$ variants, i.e., $\kmodel^\star$ with $\Pois\in \star$ and $\kmodel = \kxor$. The existence of analogous reductions for fixed-$m$ sampling and Gaussian $\ktens$ follows from equivalence of options in choices (1,2) (\autoref{subsubsec:equiv_noise_model}, \autoref{subsubsec:equiv_sampling_pois}).

\begin{algorithm}[h]\SetAlgoLined\SetAlgoVlined\DontPrintSemicolon

    \KwIn{$\calY = \set{ \b( \alpha_j, \Y_{j}\b)}_{j\in\sqb{\mpoi}}$, $m < m' \in \Z^{+}$
        }
    \KwOut{$\calZ = \set{ \b( \alpha_j', \Zt_{j}\b)}_{j\in\sqb{\mpoi'}}$.
        }

    \BlankLine
    $\Delta_m \gets \poi(m'-m)$, $\mpoi' \gets \mpoi + \Delta_m$ 
    
    \BlankLine
    \tcp{Create new $\poi(m'-m)$ pure noise entries}
    $\forall j = \mpoi+1,\dots, \mpoi'$,
        $$
        \alpha_j \gets \unif\sqb{\binomset n k}\text{ and } \Y_{j} \gets \rad(0)
        $$

    \BlankLine
    $\sigma \gets \calS_{\mpoi'}$ \tcp{Sample a random $\mpoi'$-permutation}

    \BlankLine
    $\forall j \in \sqb{\mpoi'}$, 
        $$
        \alpha_j' \gets \alpha_{\sigma(j)}\quad\text{and}\quad\Zt_j \gets \Y_{\sigma(j)}
        $$

    \BlankLine
    Return $\calZ = \set{ \b( \alpha_j', \Zt_{j}\b)}_{j\in\sqb{\mpoi'}}$.
    
    \caption{$\DenseRed\paren{\calY, m, m'}$}
\label{alg:dense_red}
\end{algorithm}

\begin{algorithm}[h]\SetAlgoLined\SetAlgoVlined\DontPrintSemicolon
    \KwIn{$\calY = \set{ \b( \alpha_j, \Y_{j}\b)}_{j\in\sqb{\mpoi}}$, $m > m' \in \Z^{+}$
        }
    \KwOut{$\calZ = \set{ \b( \alpha_j', \Zt_{j}\b)}_{j\in\sqb{\mpoi'}}$.
        }

    \BlankLine
    $j' \gets 1$
    
    \BlankLine
    \For{$j \in \sqb{\mpoi}$}{
        $I \gets \ber(m'/m)$
        
        \If{$I = 1$}{$$
        \alpha_{j'}' \gets \alpha_j\quad\text{and}\quad \Zt_{j'} \gets \Y_j
        $$
        $j' \gets j' + 1$
        }
    }

    \BlankLine
    $\mpoi' \gets j' - 1$

    \BlankLine
    Return $\calZ = \set{ \b( \alpha_j', \Zt_{j}\b)}_{j\in\sqb{\mpoi'}}$.
    
    \caption{$\SparseRed\paren{\calY, m, m'}$}
\label{alg:sparse_red}
\end{algorithm}

\begin{proof}[Proof of Lemma~\ref{lem:m_adjust}.]
The reduction algorithm is either $\DenseRed$ (Alg.~\ref{alg:dense_red}) or $\SparseRed$ (Alg.~\ref{alg:sparse_red}) depending on whether $m < m'$ or $m > m'$. We show that for $\Pois\in \star$ and $\kmodel = \kxor$ (i.e. the first option in choice (1) and the second option in choice (2) of Fig.~\ref{fig:kxor_variants_app}) these reductions map $\kmodel^\star(n,m,\delta)$ to $\kmodel^\star(n,m',\delta')$ for $\delta'$ as in Lemma~\ref{lem:m_adjust}; the analogous results for $\Pois \not \in \star$ and $\kmodel = \ktens$ follow from the option equivalence in choices (1,2) (\autoref{subsubsec:equiv_noise_model}, \autoref{subsubsec:equiv_sampling_pois}).

First, we demonstrate that the corresponding algorithm correctly maps the planted distributions. Let $\calY \sim \kxor^{\star}( n, m, \delta)$ for $\Pois \in \star$.
    \begin{itemize}
        \item \textbf{Case 1: $m < m'$.} 
        Let $\calZ = \DenseRed\paren{\calY, m, m'}$. In $\calZ = \set{ \b( \alpha_j', \Zt_{j}\b)}_{j\in\sqb{\mpoi'}}$, index sets $\alpha_j' \sim_{i.i.d.} \unif\sqb{\binomset n k}$. Moreover, the distribution of $\calZ$ is such that $\mpoi' \sim \poi(m')$ and $\forall j\in \sqb{\mpoi'}$, 
        $$
        Z_j \sim \begin{cases}
            x_{\alpha_j'} \cdot \rad(\delta) &\text{w. prob. } m/m'\\
            \rad(0) &\text{w. prob. } 1-m/m'\,.
        \end{cases}
        $$
        We conclude that $\calZ \sim \kxor^\star(n, {m'}, {\delta'})$ for $\delta' = \delta \cdot m/m'$.
        \item \textbf{Case 2: $m > m'$.} Let $\calZ = \SparseRed\paren{\calY, m, m'}$. In $\calZ = \set{ \b( \alpha_j', \Zt_{j}\b)}_{j\in\sqb{\mpoi'}}$, index sets $\alpha_j' \sim_{i.i.d.} \unif\sqb{\binomset n k}$. Moreover, the distribution of $\calZ$ is such that $\mpoi' \sim \poi(m')$ and $\forall j\in \sqb{\mpoi'}$, 
        $$
        Z_j \sim
            x_{\alpha_j'} \cdot \rad(\delta)\,.
        $$
        We conclude that $\calZ \sim \kxor^\star(n, {m'}, {\delta})$.
    \end{itemize}
    In the null case of $\calY \sim \kxor^{\star, \nulll}(n,m)$ (Def.~\ref{def:null_models}), by the same analysis, $\calZ \sim \kxor^{\star, \nulll}(n,m')$, which yields an average-case reduction for detection and recovery. 
\end{proof}

\subsection{Reduction Decreasing Order $k$ to $k/a$}\label{subsec:decrease_k_prelim_app}

We prove a reduction that maps $\kmodel^\star$ to $\kmodelk {(k/a)}^\star$ for some $a|k$ in the case $\FULL \in \star$. The reduction uses a change of variables, modifying the signal $x$ to $x^{\otimes a} \odot y$ for a known $y\sim\unif\sqb{\set{\pm1}^{n^a}}$; the assumption $\FULL \in \star$ is crucial for a correct distributional mapping.

\begin{mylem}{\ref{lem:prelim_decrease_k}}\textnormal{(Reducing Order $k$ to $k/a$)}
    Let $\star \subseteq \set{\Pois, \wor}$ and $\kmodel \in \set{\kxor, \ktens}$ be some fixed model choices (1,2,4) in Fig.~\ref{fig:kxor_variants_app} and let $k,n \in \Z^{+}, \eps \in [0,1], \defi \in \R$ be parameters.
    For every $a \in \Z^{+}, a | k$, there is an average-case reduction for both detection and recovery ($\ReduceK$, Alg.~\ref{alg:prelim_decrease_k}) $$\text{from}\quad\kmodel^{\star,\FULL}_{\eps,\defi}(n)\quad\text{to}\quad \kmodelk {(k/a)}^{\star,\FULL}_{\eps,\defi/a}(n')\,,\quad n' = n^a\,.$$
    Moreover, Alg.~\ref{alg:prelim_decrease_k} maps an instance with a secret signal vector $x\in \set{\pm 1}^n$ to one with a signal vector $x' = x^{\otimes a}\odot y$ for a known $y\sim \unif\sqb{\set{\pm 1}^{n'}}$. 
\end{mylem}
The reduction in Lemma~\ref{lem:prelim_decrease_k}, in combination with reductions in Prop.~\ref{prop:comp_equiv_full_collection} that decompose an instance of $\kmodel^{\star, \FULL}$ into standard (distinct index) instances of orders $k, k-2, k-4, \dots$, yields the following order-reducing reductions. For convenience, we also refer to the reduction in Cor.~\ref{cor:prelim_decrease_k} as $\ReduceK$ (Alg.~\ref{alg:prelim_decrease_k}).
\begin{mycor}{\ref{cor:prelim_decrease_k}}\textnormal{(of Lemma~\ref{lem:prelim_decrease_k})}
    Let $\star \subseteq \set{\Pois, \wor}$ and $\kmodel \in \set{\kxor, \ktens}$ be some fixed model choices (1,2,4) in Fig.~\ref{fig:kxor_variants_app} and let $k,n \in \Z^{+}, \eps \in [0,1], \defi \in \R$ be parameters.
    For every $a \in \Z^{+}, a | k$ and $k' \in \set{k/a, k/a-2, k/a-4,\dots} \cap \Z_{>0}$, there is an average-case reduction for both detection and recovery $$\text{from}\quad\kmodel^{\star,\FULL}_{\eps,\defi}(n)\quad\text{to}\quad \kmodelk {k'}^{\star}_{\eps,\defi/a}(n')\,,\quad n' = n^a\,.$$
    Moreover, Alg.~\ref{alg:prelim_decrease_k} maps an instance with a secret signal vector $x\in \set{\pm 1}^n$ to one with a signal vector $x' = x^{\otimes a}\odot y$ for a known $y\sim \unif\sqb{\set{\pm 1}^{n'}}$. 
\end{mycor}

\begin{algorithm}[h]\SetAlgoLined\SetAlgoVlined\DontPrintSemicolon
    \KwIn{$\calY = \set{ \b( \alpha_j, \Y_{j}\b)}_{j\in\sqb{m}}$, $a\in\Z:\, a|k$
        }
    \KwOut{$\calZ = \set{ \b( \alpha_j', \Zt_{j}\b)}_{j\in\sqb{m}}$.
        }

    \BlankLine
    $y \gets \unif \sqb{\set{\pm 1}^{n^a}}$

    \BlankLine
    \tcp{Enumerate elements of $\sqb{n}^a$}
    Let $N: \sqb{n}^a \to \sqb{n^a}$ be any bijection

    \BlankLine
    $\calZ \gets \emptyset$
    
    \ForAll{$j \in \sqb{m}$}{
        $\alphab_j \gets \text{ random ordering of }\alpha_j$

        \tcp{Group indexes of $\alphab_j$  into $k/a$ groups of size $a$}
        $\alpha_j' \gets\set{N\paren{ \paren{\alphab_j}_{ai+1},\dots, \paren{\alphab_j}_{ai + a} }}_{i = 0}^{k/a-1}$ \ \tcp{Here $\set{\cdot}$ denotes multiset}

        $\Zt_j \gets \Y_j \cdot \prod_{i \in \alpha_j'} y_i$

        $\calZ \gets \calZ \cup \{(\alpha_j', \Zt_j)\}$
    }

    \BlankLine
    Return $\calZ$
    
    \caption{$\ReduceK\paren{\calY, a}$}
\label{alg:prelim_decrease_k}
\end{algorithm}

\begin{proof}[Proof of Lemma~\ref{lem:prelim_decrease_k}.] The reduction is $\ReduceK$ (Alg.~\ref{alg:prelim_decrease_k}); we show the planted case and the null case proof is the same. 

    Let $\calY = \set{ \b( \alpha_j, \Y_{j}\b)}_{j\in\sqb{m}} \sim \kmodel^{\star,\FULL}_{\eps,\defi}(n)$, where $m = n^{k(1+\eps)/2}$, and $\calZ = \ReduceK\paren{\calY, a}$ be the output of Algorithm~\ref{alg:prelim_decrease_k}. For every equation $\paren{\alpha, \Y} \in\calY$, the reduction creates one entry in $\calZ$ as follows. Given an index multiset $\alpha = \set{i_1,\dots,i_k}$, the coefficients are first randomly ordered into a tuple $\paren{i_1,\dots,i_k}$ and then grouped into new indices $(i_1,\dots,i_a), (i_{a+1},\dots,i_{2a}), \dots$. Given that $\alpha \sim \mu\paren{\set{n}^k}$ (Def.~\ref{def:multiset}), the tuple $\paren{i_1,\dots,i_k}$ is uniformly random on $\sqb{n}^k$, so the grouped tuple $\paren{(i_1,\dots,i_a), (i_{a+1},\dots,i_{2a}), \dots}$ is uniformly random on $\sqb{\sqb{n}^a}^{k/a}$. Therefore, the index multisets of $\calZ$ are drawn independently from $\mu\paren{\set{\sqb{n}^a}^{k/a}}$.

    Moreover, for an input equation $\paren{\alpha, \Y}$ and the new index multiset $\alpha'$ in $\calZ$, if $\kmodel = \kxor$, the corresponding
    $$
    \Zt = \Y \cdot \prod_{i \in \alpha'} y_i \sim \prod_{i \in \alpha'} \underbrace{\paren{x^{\otimes a}\odot y}_i}_{x'_i} \cdot \rad\paren{\delta}\,,
    $$
    where $\delta = n^{-k\eps/4 - \defi/2}$ is the noise parameter. If $\kmodel = \ktens$, $
    \Zt = \Y \cdot \prod_{i \in \alpha'} y_i \sim \delta \prod_{i \in \alpha'} {x'_i} +\N(0,1)
    $. Therefore, 
    $$
    \law\paren{\calZ} = \kmodelk {(k/a)}^{\star,\FULL}(n^a, m ,\delta)\,,
    $$
    where 
    $$
    m = n^{k(1+\eps)/2} = \paren{n^a}^{k/a\cdot(1+\eps)/2}\quad\text{and}\quad \delta = n^{-k\eps/4 - \defi/2} = \paren{n^a}^{-k/a\cdot\eps/4 - \defi/a/2}\,.
    $$
    This concludes the proof of average-case reduction to $\kmodelk {(k/a)}^{\star,\FULL}_{\eps,\defi/a}(n^a)$.
\end{proof}

\section{Definitions of Average-Case Reductions}\label{sec:avg_case_defs}

In this section we formally define average-case reductions. The definition varies depending on the task, so we separately define average-case reductions for detection and for recovery. Note that all the reductions we prove in this paper are simultaneously reductions for both tasks.

For a more general overview of average-case reductions (both detection and recovery) and parametrization choices for statistical models see, e.g. \cite{bresler2023algorithmic, bresler2025computational}. Here we present formal definitions specific to the tensor problems that are the focus of this paper: $\kxor$ and $\ktens$. 

\subsection{Planted $\kxor$ and $\ktens$ Models}

Recall that for every set of parameters $\paren{k, \eps, \defi}$, a dimensionality parameter $n$, and a signal vector $x \in \set{\pm 1}^n$, the (planted) distributions for $\kxor$ and $\ktens$ are defined as follows:

\begin{mydef}{\ref{def:kxor_family}}\textnormal{(Planted Distributions for $\kxor$, $\ktens$)}
    Fix $k, \eps, \defi$ and let $m = n^{k(\eps+1)/2}, \delta = n^{-k\eps/4 - \defi/2}$ and $x\sim \unif\paren{\set{\pm 1}^n}$.
    \begin{enumerate}
    \item For $\calY \sim \xorp k n \eps \defi$ we observe $m$ pairs $\set{\paren{\alpha_j, \Y_j}}_{j\in\sqb{m}}$, where $\alpha \sim \unif\sqb{\binomset n k}$ are independent and for each $j\in\sqb{m}$,
        \begin{equation}\label{eq:def_kxor_2}
        \Y_{j} = 
            x_{\alpha}\cdot \noise_{j}\,,\tag{$\kxor$}
        \end{equation}
        where $\noise_{j} \sim \rad(\delta)$.

    \item For $\calY \sim \gtensp k n \eps \defi$ we observe $m$ pairs $\set{\paren{\alpha_j, \Y_j}}_{j\in\sqb{m}}$, where $\alpha \sim \unif\sqb{\binomset n k}$ are independent and for each $j\in\sqb{m}$,
        \begin{equation}\label{eq:def_ktens2}
        \Y_{j} = 
            \delta \cdot x_{\alpha} + \noiseG_{j}\,,\tag{$\ktens$}
        \end{equation}
        where $\noiseG_{j} \sim \N(0,1)$.
\end{enumerate}
\end{mydef}
\begin{remark}
    The $(\eps, \defi)$ parametrization above is most natural for $\eps \geq 0$ and is motivated by the computational threshold in this regime (see \autoref{subsec:complexity_profiles}). For the case $\eps < 0$, one can set by default $\defi = - k \eps /2$ (so that $\delta = \Theta(1)$) and use the definitions below without modification.
\end{remark}

For convenience denote $\Omega_n^k$ to be the space of $\kxor(n), \ktens(n)$ instances.

For every parameter tuple $\paren{k, \eps, \defi}$, $\set{\xorp k n \eps \defi}_n$ and $\set{\gtensp k n \eps \defi}_n$ form a \emph{sequence of planted distributions}, i.e. distributions that depend on a signal vector $x$. Both models describe \textit{families} of such sequences, as parameters $\paren{k, \eps, \defi}$ vary. Each fixed point of $\paren{k, \eps, \defi}$ corresponds to a specific planted problem, whose complexity might differ from the one of a different planted problem with parameters $\paren{k', \eps', \defi'}$. 

\subsection{Statistical Tasks: Detection and Recovery}\label{subsec:tasks}

\paragraph{Detection Task.} 
Fix parameters $\paren{k,\eps,\defi}$. The \emph{detection task} is to decide whether a given sample instance $\calY$ was drawn from the planted distribution $\xorp k n \eps \defi$ (or $\gtensp k n \eps \defi$) or a natural \emph{null (pure noise) distribution} that contains no structured signal $x^{\otimes k}$. 

We define the null distributions for $\xorp k n \eps \defi$ and $\gtensp k n \eps \defi$ by replacing the structured signal $x_{\alpha_j}$ in equation $j$ with a random sign $w_{j} \sim \unif\b(\set{\pm 1}\b)$. A common alternative is to set the SNR parameter $\delta$ to be $0$ in Eq. \eqref{eq:def_kxor_2}, \eqref{eq:def_ktens2}. For all discrete models $\xorp k n \eps \defi$, this coincides exactly with our null; for canonical Tensor PCA ($\gtensp k n {\eps=1} \defi$) the two versions are statistically indistinguishable in the conjectured hard regime $\defi > 0$ (i.e., their total variation distance is $o(1)$, see Prop.~\ref{prop:tv_null}). For $\ktens$ with all other densities $\eps\in[0,1)$, the two choices of null distributions are distinct (i.e., their total variation distance is $\Omega(1)$).

The formal null-model Definition~\ref{def:null_models} is the same as the planted Definition~\ref{def:kxor_family} above with the structured signal $x_{\alpha_j}$ in equation $j$ replaced by a random sign $w \in \set{\pm 1}$:

\begin{definition}[Null Models for $\kxor, \ktens$]\label{def:null_models}
Fix $k, \eps, \defi$ and let $m = n^{k(\eps+1)/2}, \delta = n^{-k\eps/4 - \defi/2}$.
    \begin{enumerate}
    \item For $\calY \sim \xorp k n \eps \defi$ we observe $m$ pairs $\set{\paren{\alpha_j, \Y_j}}_{j\in\sqb{m}}$, where $\alpha \sim \unif\sqb{\binomset n k}$ are independent and for each $j\in\sqb{m}$,
        \begin{equation}\label{eq:def_kxor_2_null}
        \Y_{j} = 
            w_{j}\cdot \noise_{j} \eqdist \rad(0)\,,\tag{$\kxor$}
        \end{equation}
        where $w_j \sim \unif\sqb{\set{\pm1}}$ and $\noise_{j} \sim \rad(\delta)$.

    \item For $\calY \sim \gtensp k n \eps \defi$ we observe $m$ pairs $\set{\paren{\alpha_j, \Y_j}}_{j\in\sqb{m}}$, where $\alpha \sim \unif\sqb{\binomset n k}$ are independent and for each $j\in\sqb{m}$,
        \begin{equation}\label{eq:def_ktens_2_null}
        \Y_{j} = 
            \delta \cdot w_{j} + \noiseG_{j}\,,\tag{$\ktens$}
        \end{equation}
        where $w_j \sim \unif\sqb{\set{\pm1}}$ and $\noiseG_{j} \sim \N(0,1)$.
\end{enumerate}
\end{definition}
As stated before, for Tensor PCA, the above null is close in total variation to the standard ``$\delta=0$" null, hence the corresponding detection tasks are equivalent. 
\begin{proposition}\label{prop:tv_null}
    Let $\calY \sim \gtenspnull k n {\eps=1} \defi$ with $\defi > 0$ and $\tilde\calY$ be such that for each $j\in \sqb{m}$ in the Def.~\ref{def:null_models} above,
    $$
    \Y_j = \noiseG_j\,,
    $$
    where $\noiseG_{j} \sim \N(0,1)$ are i.i.d.. Then for any $\defi > 0$,
    $$
    \tv\paren{\law(\calY), \law(\tilde\calY)}\to 0 \quad \text{as }n\to\infty\,.
    $$
\end{proposition}
\begin{proof}
    By Pinsker's inequality and the tensorization of $\kl$ divergence, 
    \begin{align*}
        2\tv\paren{\law(\calY), \law(\tilde\calY)}^2 \leq \kl\paren{\calY\|\tilde\calY} = n^k\cdot \kl\paren{ \N\paren{0,1}\|\N(\delta w_{j},1)}\,,
    \end{align*}
    where $w_{j} \sim \rad\paren{0}$ and $\delta = n^{-k/4-\defi/2} \ll 1$. By definition of $\kl$,
    $$\kl\paren{ \N\paren{0,1}\|\N(\delta w_{j},1)} = \delta^2/2 - \E_{X\sim \N(0,1)}\log\cosh\paren{\delta X} =  \Theta\paren{\delta^4}\,,$$
    so plugging in $\delta \ll n^{-k/4}$, we obtain $\tv\b(\law(\calY), \law(\tilde\calY)\b) \ll 1$.
\end{proof}

Fix $(k,\eps,\defi)$. For a fixed signal vector $x \in \set{\pm 1}^n$, denote the (planted) distribution of Eq.~\eqref{eq:def_kxor_2} $P_x = \xorp k n \eps \defi$ (resp. $P_x = \gtensp k n \eps \defi$) and let $$\Pcal = \set{P_x}_{x\in \set{\pm1}^n}$$ be a collection of all planted distributions. Denote the null distribution $P^{\nulll} = \xorpnull k n \eps$ (resp. $P^\nulll = \gtenspnull k n \eps \defi$). Then the detection task for $\xorp k n \eps \defi$ (corresp. $\gtensp k n \eps \defi$) is defines as follows.
\begin{definition}[Solving the Detection Task]\label{def:detection_problem}
    Given a sample $\calY$ from an unknown distribution $P \in P^\nulll \cup \Pcal$, the detection problem is the task of distinguishing between events $$H_0 = \set{P = P^\nulll}\quad\text{and}\quad H_1 = \set{P \in \Pcal}\,.$$ In particular, an algorithm $\mathcal{A}: \Omega_n^k \to \{0,1\}$ solves the detection task with error $\delta_n$ if the Type $I+II$ error does not exceed $\delta_n$, i.e.,
$$
\P_{\calY \sim P^\nulll} \b[\mathcal{A}(\calY) = 1\b] + \sup_{x\in\set{\pm1}^n} \P_{\calY \sim P_x} \b[\mathcal{A}(\calY) = 0\b] \leq \delta_n\,.
$$
\end{definition}

Moreover, we say that an algorithm $\mathcal{A}$ solves the corresponding detection task with Type $I+II$ asymptotic error $\delta$ if in the definition above we have 
$$
\limsup_{n\to\infty} \delta_n \leq \delta\,.
$$

\paragraph{Recovery (Search) Task.}

Fix $\paren{k,\eps,\defi}$. Similarly denote the planted distribution with signal vector $x\in\set{\pm1}^n$ as $P_x$ (distribution of Eq.~\eqref{eq:def_kxor_2} or Eq.~\eqref{eq:def_ktens2} with parameters $k,\eps,\defi,n$). Let $\Pcal = \set{P_x}_{x\in \set{\pm1}^n}$.

The recovery task is given a sample from $P_x \in \Pcal$ to estimate the signal vector $x$ (up to a global sign)\footnote{The global sign assumption is necessary for even tensor orders $k$. Since some of our reductions change the parity of the tensor order, we define the recovery up to global sign for all tensor orders. See Sec.~\ref{subsec:recovery_assumptions} for further discussion.} under some loss function $\ell_n: \set{\pm 1}^n\times \set{\pm 1}^n \to [0,+\infty)$. Formally, 

\begin{definition}[Recovery Task]\label{def:recovery_problem}
     Given a sample $\calY\in \Omega_n^k$ from an unknown distribution $P_x \in \Pcal$, the recovery task with a loss function $\ell_n$ and a loss threshold $\ell_n^\star$ is to produce an estimate $\widehat{x} \in \set{\pm 1}^n$, such that the expected loss $\ell_n$ is minimized. In particular, an algorithm $\mathcal{A}: \Omega_n^k \to \set{\pm 1}^n$ solves the recovery task (with expected loss at most $\delta_n$) if
$$ \sup_{x \in \set{\pm 1}^n} \E_{\calY\sim P_x} \min \b\{\ell_n\paren{x, \widehat{x} = \mathcal{A}(\calY)},  \ell_n\paren{-x, \widehat{x} = \mathcal{A}(\calY)}\b\}\leq \delta_n \leq \ell_n^\star\,,$$
for some $\delta_n \geq 0$ and a predetermined threshold $\ell_n^\star$.
\end{definition}
Moreover, we say that an algorithm $\mathcal{A}$ solves the recovery problem for $P$ (with asymptotic loss $\delta$) if in the definition above we have 
$$
\limsup_{n\to\infty} \delta_n \leq \delta \leq \ell^\star = \limsup_{n\to\infty} \ell_n^\star\,.
$$
See Sec.~\ref{subsec:recovery_assumptions} for an overview of common recovery problems for $\kxor$ and Tensor PCA.

\subsection{Definitions of Average-Case Reductions}

\subsubsection{Average-Case Reductions for Detection}
First, we define an average-case reduction for a fixed dimensionality parameter $n$ that maps between parameter tuples:
$$
\paren{k,\eps,\defi}\quad\to\quad \paren{k', \eps', \defi'}\,.
$$
Fix the dimensionality parameters $n$ and $n' = n^c$ (for some constant $c > 0$). Consider the detection problems associated with $\xorp k n \eps \defi$ and $\xorp {k'} {n'} {\eps'} {\defi'}$ and associated null and planted distributions:
\begin{align*}
    P^\nulll = \xorpnull k n \eps \quad &\text{and} \quad P_x = \xorp k n \eps \defi \text{ with a secret }x\in\set{\pm 1}^n\,, \Pcal = \set{P_x}_{x\in\set{\pm1}^n}\,,\\
    Q^\nulll = \xorpnull {k'} {n'} {\eps'} \quad &\text{and} \quad Q_x = \xorp {k'} {n'} {\eps'} {\defi'} \text{ with a secret }x\in\set{\pm 1}^{n'}\,, \Qcal = \set{Q_x}_{x\in\set{\pm1}^n}\,.
\end{align*}
The average-case reduction from $\xorp k n \eps \defi$ to $\xorp {k'} {n'} {\eps'} {\defi'}$ (``from $P$ to $Q$") is defined in the following Def.~\ref{def:avg_case_detect}. Note that the average-case reductions $\kxor\to \ktens, \ktens\to \kxor, \ktens\to\ktens$ are defined fully analogously.
\begin{definition}[Average-Case Reduction for Detection from $P$ to $Q$]\label{def:avg_case_detect}

We say that a $O(n^r)$-time (possibly randomized) algorithm $\mathcal{A}^{\text{red}}: \Omega_n^k\to \Omega_{n'}^{k'}$ is an \emph{average-case reduction for detection with TV-error $\epsilon_n$ from $P$ to $Q$} if for all $x \in \set{\pm 1}^{n}$, there exists some $x' \in \set{\pm 1}^{n'}$ satisfying:
    \begin{align*}
        \text{if } \calY \sim P^\nulll,  \text{ then } & \tv\paren{\mathcal{A}^{\text{red}}(\calY), Q^\nulll}\leq \epsilon_n\,, \\
        \text{if } \calY \sim P_x, \text{ then } & \tv\paren{\mathcal{A}^{\text{red}}(\calY), Q_{x'}}\leq \epsilon_n \,,
    \end{align*}
where $n' = n^c$ (for some known constant $c>0$).\footnote{We allow the dimension to change \emph{polynomially} in the reduction, since our focus is on poly-time algorithms.}
\end{definition}

The following lemma states the complexity implications of such a reduction. The proof uses the properties of total variation distance and can be found, e.g. in \cite{bresler2025computational}.

\begin{lemma}[Implication of Reduction for Detection from $P$ to $Q$]\label{lem:red_detect0}
    Assume there exists an $O(n^r)$-time average-case reduction for detection from $P$ to $Q$ with parameter $\epsilon_n$ (Def.~\ref{def:avg_case_detect}). If there is an $O(n^t)$-time algorithm for detection for $Q$ achieving Type $I+II$ error at most $\delta_n$, there exists a $O(n^{ct} + n^r)$-time algorithm for detection for $P$ achieving Type $I+II$ error at most $2\epsilon_n + \delta_{n^c}$.
\end{lemma}

\subsubsection{Average-Case Reductions for Recovery}

Again, we first define an average-case reduction for recovery for a fixed dimensionality parameter $n$ that maps between parameter tuples: 
$$
\paren{k,\eps,\defi}\quad\to\quad \paren{k', \eps', \defi'}\,.
$$
Fix the dimensionality parameters $n$ and $n' = n^c$ (for some constant $c > 0$). Consider the recovery problems associated with $\xorp k n \eps \defi$ and $\xorp {k'} {n'} {\eps'} {\defi'}$ and denote 
\begin{align*}
    P_x = \xorp k n \eps \defi \text{ with a secret }x\in\set{\pm 1}^n\,, \Pcal = \set{P_x}_{x\in\set{\pm1}^n}\,,\\
    Q_x = \xorp {k'} {n'} {\eps'} {\defi'} \text{ with a secret }x\in\set{\pm 1}^{n'}\,, \Qcal = \set{Q_x}_{x\in\set{\pm1}^n}\,.
\end{align*}
Let $\ell_n: \set{\pm 1}^n\times \set{\pm 1}^n \to [0,+\infty)$ and $\ell'_{n'}: \set{\pm 1}^{n'}\times \set{\pm 1}^{n'} \to [0,+\infty)$ be the two associated loss functions and $\ell_n^\star, \paren{\ell'_{n'}}^\star \geq 0$ the associated thresholds. 

The average-case reduction from $\xorp k n \eps \defi$ ($P$) to $\xorp {k'} {n'} {\eps'}{\defi'}$ ($Q$) is defined in the following Def.~\ref{def:avg_case_recover}. Note that the average-case reductions $\kxor\to \ktens, \ktens\to \kxor, \ktens\to\ktens$ are defined fully analogously.

\begin{definition}[Average-Case Reduction for Recovery from $P$ to $Q$]\label{def:avg_case_recover}

We say that a $O(n^r)$-time (possibly randomized) algorithm $\mathcal{A}^{\text{red}}: \Omega_n^k \to \Omega_{n'}^{k'}$ is an \emph{average-case reduction for recovery with TV-error $\epsilon_n$ from $P$ to $Q$} if:
\begin{enumerate}
    \item (Reduction Algorithm) For all $x \in \set{\pm 1}^{n}$, there exists some $x' \in \set{\pm 1}^{n'}$ satisfying:
    \begin{align*}
        \text{if } \calY \sim P_x, \text{ then } & \tv\paren{\mathcal{A}^{\text{red}}(\calY), Q_{x'}}\leq \epsilon_n \,.
    \end{align*}
    \item (Recovery Algorithm) There exists an $O(n^{r})$-time algorithm $\mathcal{A}^{\text{rec}}: \Omega_n^k \times \set{\pm 1}^{n'} \to \set{\pm 1}^{n}$ such that if $\calY, x, x' = x'(\calY)$ are defined as above and $\ell'_{n'}(x', \widehat{x'}) \leq \delta_{n'} \leq \paren{\ell'_{n'}}^{\star}$, then
        $$
        \sup_{x \in \set{\pm 1}^{n}} \E_{\calY \sim P_x} \ell_n\paren{x, \mathcal{A}^{\text{rec}}\paren{\calY, \widehat{x'}}} \leq \eta_n \leq \ell_n^\star\,,
        $$
        where $\eta_n = f(\delta_{n'})$ and $\paren{\ell'_{n'}}^{\star}$ and $\ell_n^\star$ are the loss thresholds for solving the corresponding recovery problems. 
\end{enumerate}
\end{definition}

The following lemma states the complexity implications of such a reduction. 

\begin{lemma}[Implication of Reduction for Recovery from $P$ to $Q$]\label{lem:red_recover0}
    Let $P, Q$ be two recovery problems defined as above and assume there exists an $O(n^{r})$-time average-case reduction for recovery with TV-error $\eps_n$ (Def.~\ref{def:avg_case_recover}). Then, if there is an $O(n^{t})$-time algorithm for recovery for $Q$, that with probability at least $1 - \psi_n$ solves the recovery problem for $Q$, there exists an $O\paren{n^{ct} + n^r}$-time algorithm for recovery for $P$ that with probability at least $1-\epsilon_n - \psi_{n^c}$ achieves loss at most $\eta_n = f(\delta_{n^c}) \leq \ell_n^\star$ for the loss function $\ell_n$.
\end{lemma}

We emphasize that while the definition above is quite general, any average-case reduction that nontrivially modifies the signal vector, i.e. $x' \neq x$ requires additional assumptions on the loss function $\ell_n$ and, potentially, the admissible recovery algorithm itself. For more details on such assumptions for our reductions, see Sec.~\ref{subsec:recovery_assumptions}.

\subsubsection{Average-Case Reductions as Parameter Maps}

\paragraph{Average-case reductions for sequences of problems. }
Notably, while Def.~\ref{def:avg_case_detect}, \ref{def:avg_case_recover} of average-case reductions are specific to a particular value of the dimensionality parameter $n$, we study the complexity of the \emph{sequences of the tasks} $\set{\xorp k n \eps \defi}_n, \set{\gtensp k n \eps \defi}_n$. To address this, for fixed parameters $\paren{k, \eps, \defi}$ and $\paren{k', \eps', \defi'}$, we define an average-case reduction
$$\text{from}\quad\set{P_n} = \set{\xorp k n \eps \defi}_n\quad \text{to}\quad \set{Q_{n'}} =\set{\xorp {k'} {n'} {\eps'} {\defi'}}_{n'}\,.$$ The definition for any combination of $\kxor$ and $\ktens$ is analogous.

\begin{definition}[Average-Case Reduction from $\set{P_n}$ to $\set{Q_{n'}}$]\label{def:avg_case_points}
We say that an $O(n^{r})$-time algorithm $\mathcal{A} = \{\mathcal{A}_n: \Omega_n^k \to \Omega_{n'}^{k'}\}$ is an \emph{average-case reduction for detection (resp. recovery)} $$\text{from} \quad \set{\xorp k n \eps \defi}_n\quad\text{to}\quad\set{\xorp {k'} {n'} {\eps'} {\defi'}}_{n'}$$ if
    for any sequence $\{\paren{\eps_i, \defi_i}\}_{i=1}^{\infty}$, such that $\lim_{i\to\infty} \paren{\eps_i, \defi_i} = \paren{\eps, \defi}$, there exist sequences $\{\paren{\eps_i', \defi_i'}, n_i\}_{i=1}^{\infty}$\footnote{Note that the definition above does not require the algorithm exist strictly for all $n$, only for a growing sequence of dimensions $n$. This is used to avoid unnecessary technical details while allowing assumptions like ``$n$ is even", which do not affect the complexity phenomena we study. We emphasize that all our reductions can formally be adjusted to work for all values of $n$, e.g. by simply ignoring one of the coordinates to make $n$ even.} and an index $i_0$, such that 
    \begin{enumerate}
        \item $\lim_{i\to\infty} \paren{\eps_i', \defi_i'} = \paren{\eps', \defi'}$ and $\lim_{i\to\infty} n_i = \infty$;
        \item $\forall \ i > i_0$, $\mathcal{A}_{n_i}$ is an average-case reduction for detection (resp. recovery) from $\xorp {k_i} {n_i} {\eps_i} {\defi_i}$ to $\xorp {k'_i} {n'_i} {\eps'_i} {\defi'_i}$ with the TV-error $\epsilon_{n_i}$, where $n'_i = n_i^c$ for some constant $c$; 
        \item $\lim_{i\to\infty} \epsilon_{n_i} = 0$. 
    \end{enumerate}
\end{definition}

This definition inherits the average-case reduction implications from Lemmas~\ref{lem:red_detect0} and \ref{lem:red_recover0} (for proof see analogous Lemmas A.2, B.2 of \cite{bresler2025computational}). Essentially, it states that an average-case reduction from $\set{P_n}$ to $\set{Q_{n'}}$ transfers the (detection and recovery) algorithms from $Q$ to $P$ if they are guaranteed to exist in a small $L_1$ parameter neighborhood of $\eps',\defi'$. 
\begin{lemma}[Implication of Reduction for Sequences]\label{lem:red_implications_gen}
    Consider the detection (resp. recovery) problems associated with $\set{P_n} = \set{\xorp k n \eps \defi}_n, \set{Q_{n'}} = \set{\xorp {k'} {n'} {\eps'} {\defi'}}_{n'}$ and assume there exists an $O(n^r)$-time average-case reduction for detection (resp. recovery) for problem sequences (Def.~\ref{def:avg_case_points}) from $\set{P}_n$ to $\set{Q}_n$.

    Fix $\xi > 0$ and assume there is a $O(n^t)$-time algorithm for detection (resp. recovery) for $\xorp {k'} {n} {\tilde\eps} {\tilde\defi}$ for all $\paren{\tilde \eps, \tilde \defi} \in B_{\xi}\paren{\eps',\defi'}$\footnote{Here $B_{\xi}\paren{\eps',\defi'}$ denotes an $L_1$-ball around a parameter tuple $\paren{ \eps',\defi'}$ of radius $\xi$.}  achieving asymptotic Type $I+II$ error $\delta$ (resp. asymptotic loss $\leq \paren{\ell'}^\star$). Then there exists a $O(n^{ct} + n^r)$-time algorithm for detection (resp. recovery) for $\xorp k n \eps \defi$ achieving asymptotic Type $I+II$ error at most $\delta$ (resp. asymptotic loss $\leq \ell^\star$ per Def.~\ref{def:avg_case_recover}).
\end{lemma}

\paragraph{Average-case reduction as a parameter map.}

At the center of this paper is the following complexity question:
\begin{center}
    \emph{Q: for fixed parameters $\paren{k,\eps,\defi}$ and growing dimension $n$, what is the computational complexity of the detection and recovery tasks for $\xorp k n \eps \defi$ (resp. $\gtensp k n \eps \defi$)?}
\end{center}
To address this question, we establish average-case reductions between parameter tuples $\paren{k, \eps, \defi}$ and study the complexity implications of such reductions. Recall that believed hardness of $\kxor, \ktens$ can be summarized in the following conjectures:
\begin{align*}
    &\conjD k \eps: \text{ There is no poly-time algorithm solving \textbf{detection} for }\xorp k n \eps \defi\text{ with any }\defi > 0.\\
    &\conjR k \eps: \text{ There is no poly-time algorithm solving \textbf{recovery} for }\xorp k n \eps \defi\text{ with any }\defi > 0.
\end{align*}
The following hardness implication for these conjectures is a direct corollary of Lemma~\ref{lem:red_implications_gen}.

\begin{proposition}\label{prop:hardness_implication}
    Fix parameters $k,k'\in \mathbb{Z}^{+}, \eps, \eps' \in [0,1]$. If for every $\defi' > 0$ there exists $\defi > 0$ such that there is an average-case reduction for detection (resp. recovery) (Def.~\ref{def:avg_case_points}) from $\xorp k n \eps \defi$ to $\xorp {k'} {n'} {\eps'} {\defi'}$ for $n' = n^c$ (for some constant $c>0$), then $$\conjD k \eps \Rightarrow \conjD {k'} {\eps'}\quad \paren{\text{resp. }\conjR k \eps \Rightarrow \conjR {k'} {\eps'}}\,.$$
\end{proposition}
Note that while the definitions of average-case reductions admit these strong complexity implications, they allow certain flexibility in building the actual reduction algorithms in any given dimension. For example, the following proposition will be part of almost every average-case reduction we present in this paper. Intuitively, it states that the algorithms between instances are considered average-case reductions as long as they map parameters up to $\poly\log(n)$ factors. The proof follows directly from Def.~\ref{def:avg_case_points}.
\begin{proposition}\label{prop:map_up_to_constant}
    Fix parameters $\paren{k, \eps, \defi}$ and $\paren{k', \eps', \defi'}$ and denote $\delta = n^{-k\eps/4 - \defi/2}, m = n^{k(\eps+1)/2}$ and $\delta' = (n')^{-k'\eps'/4 - \defi'/2}, m' = (n')^{k'(\eps'+1)/2}$, where $n' = n^c$ (for some constant $c$). Assume one builds a sequence of algorithms $\set{\mathcal{A}^{\textnormal{red}}_n}_n$ that for a sequence of dimensionality parameters $\set{n_i}_i, \lim_{i\to\infty} n_i = \infty$ are average-case reductions (Def.~\ref{def:avg_case_detect}, \ref{def:avg_case_recover}) from $$ \xor k n m \delta \quad\to\quad \xor {k'} {n'} {\hat{m'}} {\hat{\delta'}}\,,$$
    where $$\hat{m'} = O\paren{\poly \log n}\cdot m'\quad\text{and}\quad \hat{\delta'} = O\paren{\poly\log n}\cdot \delta'\,.$$ Then, $\set{\mathcal{A}^{\textnormal{red}}_n}_n$ constitute an average-case reduction from $\xorp k n \eps \defi$ to $\xorp {k'} {n'} {\eps'} {\defi'}$ (Def.~\ref{def:avg_case_points}).
\end{proposition}

\subsection{Assumptions on the Recovery Task}\label{subsec:recovery_assumptions}
An average-case reduction for recovery from $$P = \xorp k n \eps \defi \quad \text{to}\quad Q = \xorp {k'} {n'} {\eps'} {\defi'}$$
maps an instance of $P$ with signal $x\in\set{\pm1}^n$ to an instance of $Q$ with signal $x' = F(x) \in \set{\pm 1}^{n'}$.

Some of our reductions (e.g., entry creation/erasure, Lem.~\ref{lem:m_adjust}; cloning/splitting procedures, Lem.~\ref{lem:instance_clone}, \ref{lem:instance_split}; and Gaussianization/Discretization, Lem.~\ref{lem:gaussianization}, \ref{lem:discretization}) preserve the signal exactly: $x' = F(x) = x$. In this case, recovery transfers for \emph{any} loss function $\ell_n : \set{\pm 1}^n \times \set{\pm 1}^{n'} \to \R^{+}$: if there exists a poly-time recovery algorithm for $Q$ that succeeds with probability $1 - \psi_n$ at threshold $\ell_n^{\star}$, there exists a poly-time recovery algorithm for $P$ that succeeds with probability $1 - \psi_n -\eps_n$ for the same $(\ell_n, \ell_n^{\star})$, where $\eps_n \to_{n\to \infty} 0$ is the TV-error of the average-case reduction (see Def.~\ref{def:avg_case_recover}).

Our more technical results modify the signal in a controlled way:
\begin{itemize}
    \item[(1)] \textbf{Tensored signal vector:} In the ``Decreasing $k\to k/c$" reduction (Lem.~\ref{lem:prelim_decrease_k}) and its uses (Thm.~\ref{thm:sparse_dense_red}, Thm.~\ref{thm:decreasek}), for some integer $c$ we have $n' = \binom n c$ and $$x' = F(x) = x^{\otimes c}\odot y\,,$$ where $y \sim \unif\paren{\set{\pm 1}^{n'}}$ is known and independent from $x$.
    \item[(2)] \textbf{Subsampled signal vector:} In the discrete and Gaussian equation combination reductions(Lem.~\ref{lem:2path_red} and Lem~\ref{lem:gauss_resolution}), $n' = n\paren{1-o(1)}$ and $$x' = F(x) = x_{[1:n']}\,,$$ i.e. $x'$ is (any) subset of $n'$ coordinates of $x$.
\end{itemize}
In the cases where $x' \neq x$, we cannot claim recovery transfer for an arbitrary loss metric $\ell$. However, we show that our reductions do transfer the standard $\kxor$ and $\ktens$ recovery notions: \emph{weak, strong,} and \emph{exact recovery}. 

\paragraph{Standard Recovery Tasks for $\kxor$, $\ktens$.}
Recovery task (see Sec.~\ref{subsec:tasks}) is defined \emph{up to a global sign}: algorithms produce $\hat x$ such that $\min \b\{\ell(x, \hat x),\ell(-x, \hat x)\b\}$ is small. This requirement is necessary for the even-order tensors, and as our reductions often change the order parity, we impose it for all $k$.

In \emph{exact recovery}, the algorithm is tasked with recovering the signal $x \in \set{\pm 1}^n$ exactly with sufficiently high probability.
\begin{definition}[Exact Recovery]\label{def:exact_recovery}
    In \emph{exact recovery}, the loss function $\ell_n : \set{\pm1}^n \times \set{\pm1}^n \to \set{0,1}$ is 
    $$
    \ell_n\paren{x, \hat x} = \mathsf{1}_{x \neq \hat x}\,,
    $$
    and an algorithm $\mathcal{A} : \Omega_n^k \to \set{\pm1}^n$ solves exact recovery if, given as an input $\calY \sim P_x$ (for $P_x = \kxor$ or $\ktens$) with signal $x \in\set{\pm 1}^n$, it outputs an estimate $\hat x = \mathcal{A}(\calY)$ satisfying 
    $$
    \Pr_{\calY \sim P_x} \sqb{ \min \b\{  \ell_n(x, \hat x), \ell_n(-x,\hat x) \b\} = 0 } \geq 1 - \psi^\star_n\,.
    $$
\end{definition}
In \emph{weak} and \emph{strong recovery}, the algorithm is tasked with finding an estimate $\hat x \in \set{\pm 1}^n$ that has large coordinate-wise overlap with $x$. Formally: 
\begin{definition}[Weak and Strong Recovery]\label{def:weak_strong_recovery}
    In both \emph{weak} and \emph{strong recovery}, the loss function $\ell_n : \set{\pm1}^n \times \set{\pm1}^n \to \sqb{0,1}$ is 
    $$
    \ell_n \paren{x,\widehat{x}} = 1 - \frac1 n \b|\la x, \hat x\ra \b|\,,
    $$
    and an algorithm $\mathcal{A} : \Omega_n^k \to \set{\pm1}^n$ solves recovery if, given as an input $\calY \sim P_x$ for $P_x = \kxor$ (or $\ktens$) with signal $x \in\set{\pm 1}^n$, it outputs an estimate $\hat x = \mathcal{A}(\calY)$ and 
    $$
    \Pr_{\calY \sim P_x} \sqb{\min \b\{  \ell_n(x, \hat x), \ell_n(-x,\hat x) \b\} \leq \ell_n^\star } \geq 1 - \psi_n^\star\,,
    $$
    for some probability threshold $1 - \psi_n^\star$. 

    In \textbf{weak recovery}, $\ell_n^\star = 1 - \Omega(1)$ (recover a constant fraction $>1/2$ of coordinates). 
    
    In \textbf{strong recovery}, $\ell_n^\star = o_n(1)$ (recover $1-o(1)$ fraction of coordinates correctly).
\end{definition}

We now argue that the two signal modifications above still allow transfer of weak, strong, and exact recovery.  

\paragraph{Tensored signal vector.} We prove that given a recovery algorithm for $Q$ that returns an estimate $\widehat{x'} \in \set{\pm 1}^{n'}$, (where $x' = x^{\otimes c} \odot y$), one can recover the original vector $x$. 

\emph{Exact recovery:} if $\widehat{x'} = x'$, then $\widehat{x'} \odot y = x^{\otimes c}$, from which $x$ is trivially recovered (up to a sign flip). 

\emph{Weak/strong recovery:} we have (when recovery Alg. succeeds) 
\begin{equation}\label{eq:rec_err}
\min \b\{ \ell_n\paren{ x', \widehat{x'} }, \ell_n\paren{-x', \widehat{x'}}  \b\} = \min \b\{ \ell_n\paren{ x^{\otimes c}, \widehat{x'}\odot y }, \ell_n\paren{-x^{\otimes c}, \widehat{x'}\odot y }  \b\} \leq \ell_n^\star\,.
\end{equation}
Denote $\ell_n\paren{ x', \widehat{x'} } = \gamma$, then the set of mismatshed coordinates in $x', \widehat{x'}$ has size $n' \gamma / 2$. Because $y$ is uniform and independent of $x$, the signal $x' = x^{\otimes c} \odot y$ in $Q$ is uniform and independent of $x$ and the source instance. If we additionally apply a uniformly random coordinate permutation before using the $Q$-recovery algorithm, conditioning on the loss $\gamma$, the \emph{$n' \gamma/2$ error locations in the estimate $\widehat{x'}$ are chosen at random from $\binom {\sqb{n'}} {n'\gamma/2}$} and are independent of $x$ and the original tensor.

In the case of odd $c$, if $\ell_n\paren{ x', \widehat{x'} } < \ell_n\paren{-x', \widehat{x'}}$, then $\widehat{x'}\odot y$ is a noisy observation of $x^{\otimes c}$, and otherwise it is a noisy observation of $(-x)^{\otimes c}$. In both cases, recovering $x$ up to a global sign from $\widehat{x'}\odot y$ is a noisy $\kxor$-type ($\ktens$-type) problem with $\binom n c$ observations and at least constant entrywise noise -- solvable efficiently with existing algorithms (Sec.~\ref{subsec:lit_algo}). 

In the case of an even $c$, we can efficiently and with high probability distinguish between the two cases. We can then similarly use the noisy observation of $x^{\otimes c}$ (either $\widehat{x'}\odot y$ or $-\widehat{x'}\odot y$) to recover the original $x$. The distinguishing is possible via a following efficiently-computable test. Given a tensor $\Y$ that is either a noisy observation of $x^{\otimest c}$ or $-x^{\otimest c}$, split all $n$ coordinates into $\Theta(n)$ sets of size $3c/2$. For every such set $i_1,\dots,i_{3c/2}$, compute $\Y_{i_1,\dots,i_{c}} \cdot \Y_{i_{c/2},\dots,\Y_{i_{3c/2}}} \cdot \Y_{i_1,\dots,i_{c/2},i_{c},\dots,i_{3c/2}}$. For the case of $x^{\otimes c}$, the the number of these products equal to $1$ is concentrated around $n\paren{1/2 + \Omega(1)}$ and otherwise around $n\paren{1/2 - \Omega(1)}$. These cases are distinguishable with high probability. Thus, both weak and strong recovery are transferred. 

\paragraph{Subsampled signal vector.} In this case, the recovery algorithm for $Q$ returns $\widehat{x'}$ for $x' = x_{\sqb{1:n'}}$.

\emph{Weak/strong recovery:} the estimate $\widehat{x'}$ is already a sufficiently good estimate for the original signal vector $x$:
$$
\ell \b(x, \b(\widehat{x'}, \underbrace{0,\dots,0}_{n - n' = o(1)}\b)\b) \leq \ell^{\star} + o(1) \leq \begin{cases}
    1 - \Omega(1)\,, &\text{for weak recovery}\,,\\
    o(1)\,, &\text{for strong recovery}\,.
\end{cases}\,.
$$

\emph{Exact recovery:} exact recovery of a coordinate subset $x' = x_{\sqb{1:n'}}$ does not determine all $n$ coordinates, but applying the reduction twice with overlapping coordinate set does (with the success probability squared). 

Formally, let $\calY \sim P_x$ be the input instance with signal $x\in\set{\pm1}^n$ and let $$\calY\up1, \calY\up2 = \XorSplit\paren{\calY}$$
be the two independent instances sampled from $P_x$ (see Lem.~\ref{lem:instance_split}). Moreover, let $$\tilde \calY\up1 = \calY\up1 \odot \paren{y\up1}^{\otimes k}\quad\text{and}\quad\tilde \calY\up2 = \calY\up2 \odot \paren{y\up2}^{\otimes k}\,,$$ where $y\up1, y\up2 \sim \unif{\set{\pm 1}^{n}}$. After such transformation, $\calY\up2$ is an instance of $P_{z\up1 = x \odot y\up1}$ and $\calY\up2$ is an independent instance of $P_{z\up2 = x \odot y\up2}$.

Apply the average-case reduction $\mathcal{A}$ independently:
$\mathcal{A}\paren{\calY\up1}$ has the signal vector $w\up1$ and $\mathcal{A}\paren{\calY\up2}$ has the signal vector $w\up2$ for\footnote{The guarantees for reductions that have this signal modification (Alg.~\ref{alg:two_path} and \ref{alg:gauss_res}) state that any $n'$-subset of indices chosen for the output. This can also be achieved by simply permuting the coordinates of $\calY\up2$.}
$$w\up1 = z_{\sqb{1:n'}} = x_{\sqb{1:n'}} \odot y_{\sqb{1:n'}}\up1\quad\text{and}\quad w\up2 = z_{\sqb{n'-n, n}} = x_{\sqb{n'-n, n}} \odot y_{\sqb{n'-n, n}}\up2\,.$$
If the exact recovery algorithm for $Q$ succeeds with probability $1 - \psi_{n'}$ per instance, then, since $\mathcal{A}\paren{\calY\up1}$ and $\mathcal{A}\paren{\calY\up2}$ are fully independent, we recover both $w\up1$ and $w\up2$ with probability at least $(1 - \psi_{n'})^2$. We thus recover $x$ exactly with probability at least $(1 - \psi_{n'})^2$.

\section{Deferred Proofs}

\subsection{Proof of Claim~\ref{claim:conv}}\label{subsec:proof_claim_conv}

    \begin{myclaim}{\ref{claim:conv}}
        Fix $c, r \in (0,1), \kstar\in \Z^{+}$. There exists a constant $A(c, r, \kstar)$, such that for every $A \geq A(c, r, \kstar)$, there exists $k \in \Z^{+}$ such that $k \leq c A$, $\lfloor k r \rceil$ is a multiple of $\kstar$, and 
        \begin{equation}\label{eq:approx_bound}
            \frac{k}{\|kr\|} > A\,,
        \end{equation}
        where $\|x\| = \b| x - \lfloor x \rceil \b|$.

    \end{myclaim}
    \begin{fact}[Integer Approximation Facts \cite{khinchin1964continued}]\label{fact:conv} 
        Given $r \in (0,1)$, we can express $r$ as a simple continued fraction
        \begin{equation}\label{eq:continued_fraction}
            r = 0 + \frac{1}{b_1 + \frac 1 {b_2 + \dots}}\,,
        \end{equation}
        where $b_i > 0$ are integers. We denote expression in \eqref{eq:continued_fraction} as $[b_0=0;b_1,b_2,\dots]$. The \emph{$i$-th convergent} of $r$ is then defined as $Q_i = [0;b_1,\dots,b_i] = [0;b_1,\dots,b_i,0,0,\dots]$. The following holds:
        \begin{itemize}
            \item (Theorem 1 in \cite{khinchin1964continued}) The $i$-th convergent can be expressed as $Q_i = A_i/B_i$, where $\mathsf{gcd}(A_i, B_i) = 1$, and the sequences $\set{A_i}, \set{B_i}$ satisfy the recursive relation
            $$
            X_i = b_i X_{i-1} + X_{i-2}
            $$
            for any $i \geq 2$. As a consequence, $\set{B_i}_{i=2}^{\infty}$ is a growing sequence of integers. 
            \item (Theorem 9 in \cite{khinchin1964continued}) In the above setting, 
            $$
            |r - Q_i| \leq \frac{1}{B_i B_{i+1}}\,,
            $$
            or, equivalently,
            $$
            |r B_i - A_i| \leq \frac{1}{B_{i+1}}\,.
            $$
        \end{itemize} 
    \end{fact}
    \begin{proof}[Proof of Claim~\ref{claim:conv}] Express $r$ as a simple continued fraction $r = [b_0=0;b_1,b_2,\dots]$, where $b_i > 0$ are integers and consider the convergents $\set{Q_i = \frac {A_i}{B_i}}$ for $r$ defined as in Fact~\ref{fact:conv}. Then, from Fact~\ref{fact:conv}, $|rB_i - A_i| \leq \frac 1 { B_{i+1}}$ and, as a consequence, for every integer $t>0$, 
    \begin{equation}\label{eq:conv_approx}
        |r\cdot B_i t - A_i t| \leq \frac t { B_{i+1}}\,.
    \end{equation}
    
    Recall that $\set{B_i}$ is a growing sequence of integers and define $\tilde{c} = c / (2\kstar+1)$ and $i_c = \min \set{i: B_i \geq 1 / \tilde{c}}$. Let $A(c,r,\kstar) = B_{i_c} / \tilde{c}$ and choose any $A > A(c,r,\kstar)$. Let index $i$ be such that $B_i \leq \tilde{c} A < B_{i+1}$. Then, $B_i \geq B_{i_c} \geq 1/\tilde{c}$. From \eqref{eq:conv_approx} we then obtain
    $$
    \frac {B_i q} {\|B_i q r\|} \geq B_i B_{i+1} \geq 1 / \tilde{c} \cdot \tilde{c} A = A\,.
    $$
    Thus, for any $A > A(c,r,\kstar)$ and $q \leq (2\kstar + 1)$, the condition \eqref{eq:approx_bound} of the Claim holds with $k = B_i q \leq \tilde{c}(2\kstar+1) A = c A$ for $q \leq 2\kstar + 1$. 
    It is now sufficient to choose $q$ such that $\lfloor B_i q r\rceil$ is a multiple of $\kstar$. Let $\alpha = \frac{B_i r} {\kstar}$ and let $q \leq 2\kstar + 1$ be such that $$ \| q \alpha\| \leq \frac{1}{2\kstar+1}\,,$$
    where the existence of such $q$ is guaranteed by the pigeonhole principle. Then, 
    $$
    |q B_i r - \kstar \lfloor q \alpha \rceil| \leq \frac{\kstar}{2\kstar+1} < \frac12\,,
    $$
    so $\lfloor q B_i r\rceil = \kstar \lfloor q \alpha \rceil$ is a multiple of $\kstar$.
    \end{proof}

\subsection{Proof of Lemma~\ref{lem:wishart}}\label{subsec:wishart}

\begin{mylem}{\ref{lem:wishart}}\textnormal{(Generalization of Theorem 2.6 of \cite{brennan2021finetti})}
Let $\Mr\in[0,1]^{d\times n}, \Ml\in[0,1]^{d\times m}$ have i.i.d. $\ber(\pr)$ and $\ber(\pl)$ entries respectively. Let $\deltaR,\deltaL \in \R$ be parameters such that $\delta_R, \delta_L \ll 1$ and let $\Mur^\full \in \set{\pm 1}^{d \times n}, \Mul^\full \in \set{\pm 1}^{d \times m}$ be any fixed matrices. Let $\paren{\Xgr}^\full \sim \N\paren{0,I_{d \times n}}$ and $\paren{\Xgl}^\full \sim \N\paren{0,I_{d \times m}}$. 

Define for $a \in \set{R, L}$:
$$
\Mu_a = \delta_a \Mu_a^\full \odot \M_a\quad\,, \Xg_a = \paren{\Xg_a}^\full \odot \M_a\,, \quad \text{and} \quad X_a = \Mu_a + \Xg_a\,.
$$

If the following assumptions on parameters $\pr, \pl, \deltaR, \deltaL, d, n, m$ are satisfied: 
\begin{align}
    \pr, \pl &\in \set{1} \cup (0, 1-c) \text{ for some }c>0\tag{a}\label{eq:cond_a_app}\\
    \deltaR &\ll \min\set{1, \paren{\pr n}^{-1/2}}\text{ and } \deltaL\ll \min\set{1, \paren{\pl m}^{-1/2}}\tag{b}\label{eq:cond_b_app}\\
    n &\ll \pr^2 \pl d\text{ and }m\ll \pl^2 \pr d\tag{c}\label{eq:cond_c_app}\\
    mn &\ll \min\set{\pr\pl d, \deltaL^{-4}, \deltaR^{-4}}\tag{d}\label{eq:cond_d_app}\\
    n^2m &\ll \min\set{\pl d, \pr^{-2} \deltaR^{-4}},\text{ and }nm^2 \ll \min\set{\pr d, \pl^{-2} \deltaL^{-4}}\tag{e}\label{eq:cond_e_app}\,,
\end{align}
then 
$$
\tv\paren{N\paren{\paren{\pr\pl d}^{-1/2}\Mur^T \Mul, I_{n\times m}}, \law\paren{\paren{\pr \pl d}^{-1/2} \Xr\tr \Xl}} \to 0\,, \quad d\to 
\infty\,.
$$

\end{mylem}
\begin{proof}
Recall that for $a \in \set{R, L}$, 
$$
X_a = \Mu_a + \Xg_a \odot \M_a\,,
$$
and therefore the Wishart matrix of interest can be expressed as 
\begin{align*}
    \Xr^\top \Xl = \paren{\Mur +\Xgr}^\top \paren{\Mul+\Xgl} &= \Mur^\top \Mul + \Mur^\top \Xgl + \paren{\Xgr}^\top \Mul + \paren{\Xgr}^\top \Xgl\\
    &= \text{(mean)} + \text{(cross-term 1)} + \text{(cross-term 2)} + \text{(noise)}\,.
\end{align*}
We will prove the Lemma statement in two steps: intuitively, bound \eqref{eq:tv1} removes the first cross-term and proves that the noise terms is Gaussian; bound \eqref{eq:tv2} removes the second cross-term.
\begin{align}
    \tv\paren{\paren{\pr\pl d}^{-1/2} \Xr^\top \Xl,  \N\paren{\paren{\pr \pl d}^{-1/2}\Xr^\top \Mul, I_{n\times m}}} \to 0\,, \quad d\to \infty\,.\tag{TV1}\label{eq:tv1}\\
    \tv\paren{\N\paren{\paren{\pr \pl d}^{-1/2}\Mur^\top \Mul, I_{n\times m}}, \N\paren{\paren{\pr \pl d}^{-1/2}\Xr^\top \Mul, I_{n\times m}}} \to 0\,, \quad d\to \infty\,.\tag{TV2}\label{eq:tv2}
\end{align}
Then, applying the triangle inequality for total variation distance (Fact~\ref{tvfacts}), the statement of the Lemma follows. 

Note that Lemma~\ref{lem:wishart} is a generalization of a result of \cite{brennan2021finetti} for bipartite masks.

\paragraph{Proof of \eqref{eq:tv1}.}
Denote $\norm = \pr\pl d$ and let $$Y = \norm^{-1/2} \Xr\tr \Xl\,,$$
$$\mu = \law\paren{\norm^{-1/2} \Xr\tr \Xl} = \law\paren{Y},\quad \text{and}\quad \nu = \law\paren{\N\paren{\psi^{-1/2} \Xr\tr \Mul, I_{n\times m}}}\,.$$ Our goal for this section is to show that under the conditions of the Lemma, $\tv\paren{\mu, \nu} \to 0$ as $d\to \infty$.

We will identify a high-probability event $S$ on $X_R$ and $\M_L$ in the process of the proof. Then, denoting $\mu_{S}$ to be $\mu$ conditioned on the event $S$ by the properties of total variation (Fact~\ref{tvfacts}), 
$$ \tv\paren{\mu, \nu} \leq \tv\paren{\mu_{S} , \nu} + \Pr\sqb{S^c}\,,$$
and thus it is sufficient to show that $\tv\paren{\mu_{S} , \nu}\to 0$ as $d \to \infty$.
By Pinsker's inequality,
$$
\tv\paren{\mu_{S}, \nu} \leq \sqrt{\frac 12 \kl \paren{\mu_{S} \|\nu}}\,.
$$

Denote $\mu\paren{\Xr, \M_L}, \nu\paren{\Xr, \M_L}$ to be the measures $\mu, \nu$ after conditioning on the value of $\Xr, \M_L \in S$.

Note that the above are product measures by column. For $i \in [m]$ let
\begin{align*}
    \mu\paren{\Xr, \M_L}_{i} = \law\paren{ \norm^{-1/2} \Xr^\top \paren{\Xl}_i}, \quad &\text{and} \quad \nu\paren{\Xr, \M_L}_{i} = \law \paren{\N\paren{\norm^{-1/2} \Xr^\top \paren{M_L}_i, I_n}}\,.
\end{align*}
Note that even after conditioning on $\Xr$ and $\M_L$, the columns of both matrices are independent and Gaussian with parameters computed below. Then, applying the convexity and tensorization of $\kl$ divergence, 
\begin{align}
\kl \paren{\mu_S \| \nu} &\leq  \E_{\Xr, \M_L \in S} \kl \paren{ \mu\paren{\Xr, \M_L} \b\| \nu\paren{\Xr, \M_L}}\notag\\ &= \E_{\Xr, \M_L \in S} \sum_{i \in [m]} \kl \paren{ \mu\paren{\Xr, \M_L}_{i} \b\| \nu\paren{\Xr, \M_L}_{i}}\,.\label{eq:kl_expand}
\end{align}
Both $\mu\paren{\Xr, \Ml}_{i} $ and $ \nu\paren{\Xr, \Ml}_{i}$ are high-dimensional Gaussians -- we now compute their parameters. Denote $\M_i$ to be the set of nonzero coordinates in the column $\paren{\Xl}_i$ and let $A_i$ be the vector $\paren{\Mul}_i$ restricted to those coordinates. Moreover, denote $X_{|i}$ to be the submatrix of $\Xr$ consisting of the rows in $\M_i$. With this notation,
\begin{align*}
    &\mu\paren{\Xr, \Ml}_{i}: &\norm^{-1/2} \Xr^\top \paren{\Xl}_i \sim \N\paren{ \norm^{-1/2} X_{|i}^\top A_i, \norm^{-1} X_{|i}^\top X_{|i}}\\
    &\nu\paren{\Xr, \Ml}_{i}: &\N\paren{ \norm^{-1/2} X_{|i}^\top A_i, I_{n}}\,.
\end{align*}
Denote $\Sigma_{\M_i, \Xr} = \norm^{-1} X_{|i}^\top X_{|i}$ to be the covariance in $\mu\paren{\Xr}_i$ given $\M_i$ and let 
$$
\Delta_{\M_i, \Xr} = \Sigma_{\M_i, \Xr} - I_n\quad\,.
$$
Fix some $c\in (0,1)$ and define the following event:
$$
S_{op} = \set{\forall i\,,\, \|\Delta_{\M_i, \Xr}\|_{op} \leq c}\,.
$$

\begin{claim}\label{claim:tail_bound}
Event $S_{op}$ is high probability, i.e. 
$$
\Pr\sqb{S_{op}^c} \to 0 \quad \text{as }d\to\infty\,.
$$

\end{claim}

\begin{claim}\label{claim:kl_cmp} $\forall i \in \sqb{m}$,
    \begin{align*}
        &\E_{\M_i, \Xr \in S_{op}} \kl\paren{ \N\paren{ \norm^{-1/2} X_{|i}^\top A_i, \Sigma_{\M_i, \Xr}} \b\| \N\paren{\norm^{-1/2} X_{|i}^\top A_i, I_{n}} } = o\paren{m^{-1}}\,.
    \end{align*}
\end{claim}
By Claim~\ref{claim:tail_bound}, we can set $S = S_{op}$. Then, applying Claim~\ref{claim:kl_cmp} to every term of Eq.~\eqref{eq:kl_expand}, we obtain 
$$
\kl\paren{\mu_S \| \nu} \leq o(1)\,,
$$
which concludes the proof of \eqref{eq:tv1}.

\begin{proof}[Proof of Claim~\ref{claim:tail_bound}]
    We will use the following standard bound on the singular values of random matrices.
    \begin{proposition}[Prop. 4.6.1 in \cite{vershynin2018high}]\label{prop:isotropic_rows}
        Let $A$ be an $m\times n$ matrix, where $m \geq n$, whose rows $A_i$ are independent, mean zero, sub-gaussian isotropic random vectors in $\R^n$. Denote by $s_1(A)$ and $s_n(A)$ the largest and smallest singular values of $A$. Then for every $t \geq 0$, we have 
        $$
        \sqrt{m} - CK^2(\sqrt{n}+t) \leq s_n(A) \leq s_1(A) \leq \sqrt{m} + CK^2(\sqrt{n}+t)
        $$
        with probability at least $1-2\exp(-t^2)$. Here $K = \max_i \|A_i\|_{\psi_2}$.
    \end{proposition}

    Recall that conditioned on $\M_i$, $\Delta_{\M_i, \Xr} = \norm^{-1} X_{|i}^\top X_{|i} - I$, where we can decompose $X_{|i} = \Xg_{|i} + \paren{\Mur}_{|i}$. In this proof we will denote $\delta = \deltaR$. First, $$\|\paren{\Mur}_{|i}\|_{op} \leq \delta \|\paren{\Mr}_{|i}\|_{op}\,,$$
    where recall the entries of $\paren{\Mr}_{|i}$ are i.i.d. $\ber(\pr)$. Denote $d' \leq d$ to be the number of rows in $\paren{\Mr}_{|i}$. In case $\pr = 1$, $\|\paren{\Mr}_{|i}\|_{op} \leq \sqrt{nd'}$. Otherwise, for the case $\pr \in \paren{0, 1-c}$, the rows of $\paren{\paren{\Mr}_{|i} - \pr}/\sqrt{\pr(1-\pr)}$ are $\frac{1}{\sqrt{\pr(1-\pr)}}$-subgaussian isotropic. By Prop.~\ref{prop:isotropic_rows}, for any $t\geq 0$ with probability at least $1- 2\exp\paren{-t^2}$,
    \begin{align}
    \|\paren{\Mr}_{|i}\|_{op} &\leq  \pr\sqrt{n d'} + \sqrt{\pr(1-\pr)}\paren{ \sqrt{d'} + \frac C {\pr (1-\pr)} (\sqrt{n} + t)}\notag\\
    &\leq \pr\sqrt{n d'} + O\paren{\sqrt{\pr d'} + \pr^{-1/2}\sqrt{n}}\,,\label{eq:E_op}
    \end{align}
    where we set $t = O\paren{\log n}$. Note that bound \eqref{eq:E_op} also trivially holds for $\pr = 1$. We conclude that with high probability
    \begin{align}
        \norm^{-1/2} \| \paren{\Mur}_{|i} \|_{op} &\leq \norm^{-1/2} \delta \|\paren{\Mr}_{|i}\|_{op}\notag\\
        &\leq \norm^{-1/2}\delta\pr\sqrt{n d'} + O\paren{\norm^{-1/2}\delta\sqrt{\pr d'} + \norm^{-1/2}\delta\pr^{-1/2}\sqrt{n}}\,.\label{eq:bound_op_mean}
    \end{align}

    Now $\Xg_{|i}$ has entries that are $\N(0,1)$ with probability $\pr$ and $0$ otherwise. Then, $\frac 1 {\sqrt{\pr}} \Xg_{|i}$ has isotropic, $\frac 1 {\sqrt{\pr}}$-subgaussian rows, and therefore, by Prop.~\ref{prop:isotropic_rows}, for any $t\geq 0$ with probability at least $1- 2\exp\paren{-t^2}$,
    \begin{align*}
        \frac 1 {\sqrt{\pr}} s_{i}\paren{\Xg_{|i}} \in \sqb{ \sqrt{d'} \pm C \frac 1 \pr \paren{\sqrt{n} + t} }\,.
    \end{align*}
    Setting $t = O\paren{\log n}$, we conclude that with high probability,
    \begin{align}
        s_{i}\paren{\norm^{-1/2}\Xg_{|i}} \in \sqb{ \norm^{-1/2}\sqrt{\pr d'} \pm O\paren{\norm^{-1/2}\pr^{-1/2}\sqrt{n} }}\,.\label{eq:bound_s_gauss}
    \end{align}
    
    Recall that $d' \sim \bin(n, \pl)$. We will use the following fact about concentration of binomial random variables. 
    \begin{myprop}{\ref{prop:binomial_conc}}\textnormal{(Concentration of Binomial RV)}
        Let $X\sim\bin(n,p)$, where $p > 0$ and $np \geq n^{\eps}$ for some $\eps > 0$. By Berstein's inequality,
        $$
        \Pr\sqb{|X - pn| \geq t} \leq 2 \exp\paren{-\frac{t^2/2}{np(1-p) + t/3}}\,,
        $$
        and as a consequence, with probability at least $1 - n^{-c}$,
        $$
        X\in \sqb{pn \pm C\sqrt{pn}\log n} =: \sqb{pn \pm \phi(p,n)}\,,
        $$
        where $C= C(c)$ is a universal constant. 
    \end{myprop}
    By Prop.~\ref{prop:binomial_conc}, with high probability at, $d' \in \sqb{d\pl \pm \tilde O\paren{ \sqrt{\pl d} \log d}}$. Substituting this result into \eqref{eq:bound_s_gauss}, we obtain
    \begin{align*}
        s_{i}\paren{\norm^{-1/2}\Xg_{|i}} &\in \sqb{ 1 \pm \tilde O\paren{\norm^{-1/2}\sqrt{\pr }} \pm O\paren{\norm^{-1/2}\pr^{-1/2}\sqrt{n} } }\\
        &\in \sqb{ 1 \pm o(1) }\,,
    \end{align*}
    where we used $\norm^{-1/2}\sqrt{\pr } = \paren{\pl d}^{-1/2} \ll n^{-\eps}$ from Cond.~\eqref{eq:cond_c_app}, $\norm^{-1/2}\pr^{-1/2}\sqrt{n} = \sqrt{\frac n {\pr^2 \pl d}} \ll 1$ from Cond.~\eqref{eq:cond_c_app}.
    Substituting the concentration of $d'$ into \eqref{eq:bound_op_mean}, we obtain
    \begin{align*}
        \norm^{-1/2} \| \paren{\Mur}_{|i} \|_{op} &\leq O\paren{\norm^{-1/2}\delta\pr\sqrt{n \pl d} + \norm^{-1/2}\delta\sqrt{\pr \pl d} + \norm^{-1/2}\delta\pr^{-1/2}\sqrt{n}}\\
        &= O\paren{\delta\sqrt{\pr n} + \delta + \norm^{-1/2}\delta\pr^{-1/2}\sqrt{n}}\\
        &= o(1)\,,
    \end{align*}
    where we used  $\delta, \delta\sqrt{\pr n} \ll 1$ from Cond.~\eqref{eq:cond_b_app}, $\norm^{-1/2}\delta\pr^{-1/2}\sqrt{n} \ll \sqrt{\frac n {\pr^2 \pl d}} \ll 1$ from Cond.~\eqref{eq:cond_b_app}, \eqref{eq:cond_c_app}.
    Combining the bounds on the singular values of $\norm^{-1/2}\Xg_{|i}$ and on $\norm^{-1/2} \| \paren{\Mur}_{|i} \|_{op}$, we obtain 
        \begin{align*}
        \|\Delta_{\M_i, \Xr}\|_{op} = \|\norm^{-1} X_{|i}^\top X_{|i} - I\|_{op} \ll 1\,,
        \end{align*}
        which concludes the proof. 
\end{proof}

\begin{proof}[Proof of Claim~\ref{claim:kl_cmp}]
We apply the formula of Lemma~\ref{l:KLGauss} for the $\kl$ divergence between Gaussians and the divergence of Claim~\ref{claim:kl_cmp} can be expressed as
\begin{align*}
    - \log \det \paren{I + \Delta_{\M_i, \Xr}} + \Tr\paren{\Delta_{\M_i, \Xr}} &= \Tr\paren{\Delta_{\M_i, \Xr}} + \sum_{t=1}^{\infty} \frac{(-1)^t}{t} \Tr\paren{\Delta_{\M_i, \Xr}^t} = \sum_{t=2}^{\infty} \frac{(-1)^t}{t} \Tr\paren{\Delta_{\M_i, \Xr}^t}\,.
\end{align*}
Moreover, $\forall t > 2$, $\b| \Tr\paren{\Delta_{\M_i,\Xr}^t} \b| \leq \|\Delta_{\M_i,\Xr}\|^{t-2}_{op} \Tr \paren{\Delta_{\M_i,\Xr}^2} \leq c^{t-2}\Tr \paren{\Delta_{\M_i,\Xr}^2}$, since under $S_{op}$ we have $\|\Delta_{\M_i,\Xr}\|_{op} \leq c$. Then,
\begin{align*}
    \E_{\M_i, \Xr\in S_{op}} \sqb{\paren{- \log \det \paren{I + \Delta_{\M_i,\Xr}} + \Tr\paren{\Delta_{\M_i,\Xr}}}} &= O(1)\cdot\E_{\M_i, \Xr\in S_{op}} \Tr \paren{\Delta_{\M_i,\Xr}^2}\,.
\end{align*}
Note that since $S_{op}$ is a high probability event by Claim~\ref{claim:tail_bound},
$$
\E_{\M_i, \Xr\in S_{op}} \Tr \paren{\Delta_{\M_i,\Xr}^2} \leq \Pr\sqb{S_{op}}^{-1} \E_{\M_i, \Xr} \Tr \paren{\Delta_{\M_i,\Xr}^2} = O(1) \cdot \E_{\M_i, \Xr} \|\Delta_{\M_i,\Xr}\|_2^2\,.
$$
We now explicitly compute $\E_{\M_i, \Xr} \|\Delta_{\M_i,\Xr}\|_2^2$. Recall that $\Delta_{\M_i,\Xr}= \norm^{-1} X_{|i}^\top X_{|i} - I$. In this proof, denote $\delta = \deltaR$. Recall that every row of $\Xr$ is included in $X_{|i}$ independently with probability $\pl$ and the entries $\paren{\Xr}_{ab} \sim \N(\pm \delta, 1)\cdot \paren{\Mr}_{ab}$. 
For an off diagonal entry, i.e. $a \neq b \in [n]$:
\begin{align*}
    \norm^2 \E_{\M_i, \Xr} \paren{\Delta_{\M_i, \Xr}}_{ab}^2 &= \paren{1+\delta^2}^2 d \pr^2 \pl + \delta^4 d(d-1) \pr^4 \pl^2\,.
\end{align*}
For a diagonal entry, 
\begin{align*}
        \E_{\M_i, \Xr} \paren{\Delta_{\M_i, \Xr}}_{aa}^2 &= \norm^{-2}\paren{\delta^4+6\delta^2+3} d \pr \pl + 1 + \norm^{-2} \paren{1+\delta^2}^2 d(d-1) \pr^2\pl^2 - 2 \norm^{-1} \paren{1+\delta^2} d \pr\pl\\
        &\leq O\paren{\delta^2\norm^{-1} + \delta^4} + 3 \norm^{-1} + 1 + 1 - 2\\
        &= O\paren{\delta^2\norm^{-1} + \delta^4 + \norm^{-1}}\,.
\end{align*}   
Combining these results we obtain 
\begin{align*}
    \E_{\M_i, \Xr^{G}}\|\Delta_{\M_i, \Xr^G}\|_2^2 &\leq \paren{n^2 - n} \paren{ \norm^{-2}\paren{1+\delta^2}^2 d \pr^2 \pl + \norm^{-2}\delta^4 d(d-1) \pr^4 \pl^2 } + n O\paren{\delta^2\norm^{-1} + \delta^4 + \norm^{-1}}\\
    &= O\paren{ n^2 \pr \norm^{-1} + \delta^4 n^2 \pr^2 + \delta^2 n \norm^{-1} + n\delta^4 + n\norm^{-1}} \ll m^{-1}\,,
\end{align*}
where we have used $\delta \ll 1$ (Cond.~\eqref{eq:cond_b_app}), $n^2 \pr \norm^{-1} = \frac {n^2}{\pl d} \ll m^{-1}$ and $\delta^4 n^2 \pr^2\ll m^{-1}$ from Cond.~\eqref{eq:cond_e_app}, $\max\set{\delta^2 n \norm^{-1}, n\delta^4, n\norm^{-1}} \ll m^{-1}$ from Cond.~\eqref{eq:cond_d_app}, \eqref{eq:cond_b_app}, which concludes the proof.

\end{proof}
\paragraph{Proof of \eqref{eq:tv2}.}
In this section we show the second step of our argument and prove
$$\tv\paren{\N\paren{\paren{\pr \pl d}^{-1/2}\Mur^\top \Mul, I_{n\times m}}, \N\paren{\paren{\pr \pl d}^{-1/2}\Xr^\top \Mul, I_{n\times m}}} \to 0\,, \quad d\to \infty\,,$$
where both measures are mixtures of Gaussians with corresponding distributions for $\Xr, \Mur, \Mul$.
Again denote $\norm = \pr \pl d$ and let 
$$
\mu = \law\paren{\N\paren{\norm^{-1/2}\Xr^\top \Mul, I_{n\times m}}}\,,\quad\text{and},\quad\nu = \law\paren{\N\paren{\norm^{-1/2}\Mur^\top \Mul, I_{n\times m}}}
$$
We will follow the same proof outline as \eqref{eq:tv1} and identify a high-probability event $S$ on $\Mul$ and $\Mr$ in the process of the proof. Then, denoting $\mu_{S}$ to be $\mu$ conditioned on the event $S$ by the properties of total variation (Fact~\ref{tvfacts}), 
$$ \tv\paren{\mu, \nu} \leq \tv\paren{\mu_{S} , \nu} + \Pr\sqb{S^c}\,,$$
and thus it is sufficient to show that $\tv\paren{\mu_{S} , \nu}\to 0$ as $d \to \infty$.
By Pinsker's inequality,
$$
\tv\paren{\mu_{S}, \nu} \leq \sqrt{\frac 12 \kl \paren{\mu_{S} \|\nu}}\,.
$$ 

Denote $\mu\paren{\Mul, \Mr}, \nu\paren{\Mul, \Mr}$ to be the measures $\mu, \nu$ after conditioning on the value of $\Mul, \Mr \in S$. Note that the above are product measures by row. For $i \in [n]$ let
\begin{align*}
    \mu\paren{\Mul, \Mr}_{i} = \law\paren{ \norm^{-1/2} \paren{\Xr}_i^\top \Mul + \N\paren{0, I_m}^\top}, \quad &\text{and} \quad \nu\paren{\Mul, \Mr}_{i} = \law \paren{\norm^{-1/2} \paren{\Mur}_i^\top M_L + \N\paren{0, I_m}^\top}\,.
\end{align*}
Note that even after conditioning on $\Mul$ and $\M_R$, the rows of both matrices are independent and Gaussian with parameters computed below. Then, applying the convexity and tensorization of $\kl$ divergence, 
\begin{align}
\kl \paren{\mu_S \| \nu} &\leq  \E_{\Mul, \M_R \in S} \kl \paren{ \mu\paren{\Mul, \M_R} \b\| \nu\paren{\Mul, \M_R}}\notag\\ &= \E_{\Mul, \M_R \in S} \sum_{i \in [m]} \kl \paren{ \mu\paren{\Mul, \M_R}_{i} \b\| \nu\paren{\Mul, \M_R}_{i}}\,.\label{eq:kl_expand2}
\end{align}
Both $\mu\paren{\Mul, \Mr}_{i} $ and $ \nu\paren{\Mul, \Mr}_{i}$ are high-dimensional Gaussians -- we now compute their parameters. Denote $\M_i$ to be the set of nonzero coordinates in the row $\paren{\Xr}_i$ and recall that $\paren{\Xr}_i = \paren{\Xgr}_i + \paren{\Mur}_i$. Let $A_i$ be the vector $\paren{\Mur}_i$ restricted to the nonzero coordinates $\M_i$. Moreover, denote $\paren{\Mul}_{|i}$ to be the submatrix of $\Mul$ consisting of the rows in $\M_i$. With this notation,
\begin{align*}
    &\mu\paren{\Mul, \Mr}_{i}: &\norm^{-1/2} \paren{\Xr}_i^\top \Mul + \N\paren{0, I_m}^\top \sim \N\paren{ \norm^{-1/2} A_i^\top \paren{\Mul}_{|i}, I_m + \norm^{-1} \paren{\Mul}_{|i}^\top \paren{\Mul}_{|i}}\\
    &\nu\paren{\Mul, \Mr}_{i}: &\N\paren{ \norm^{-1/2} A_i^\top \paren{\Mul}_{|i}, I_{m}}\,.
\end{align*}
Denote 
$$
\Delta_{\Mul, \Mr} = \norm^{-1} \paren{\Mul}_{|i}^\top \paren{\Mul}_{|i}\,,\quad\text{and}\quad \Sigma_{\Mul, \Mr} = I_m + \Delta_{\Mul, \Mr}\,.
$$
Fix some $c\in (0,1)$ and define the following event:
$$
S_{op} = \set{\forall i\,,\, \|\Delta_{\Mul, \Mr}\|_{op} \leq c}\,.
$$
\begin{claim}\label{claim:tail_bound2}
Event $S_{op}$ is high probability, i.e. 
$$
\Pr\sqb{S_{op}^c} \to 0 \quad \text{as }d\to\infty\,.
$$
\end{claim}

\begin{claim}\label{claim:kl_cmp2} $\forall i \in \sqb{m}$,
    \begin{align*}
        &\E_{\Mul, \Mr \in S_{op}} \kl\paren{ \N\paren{ \norm^{-1/2} A_i^\top \paren{\Mul}_{|i}, \Sigma_{\Mul, \Mr}} \b\| \N\paren{ \norm^{-1/2} A_i^\top \paren{\Mul}_{|i}, I_m}} = o\paren{n^{-1}}\,.
    \end{align*}
\end{claim}
By Claim~\ref{claim:tail_bound2}, we can set $S = S_{op}$. Then, applying Claim~\ref{claim:kl_cmp2} to every term of Eq.~\eqref{eq:kl_expand2}, we obtain 
$$
\kl\paren{\mu_S \| \nu} \leq o(1)\,,
$$
which concludes the proof of \eqref{eq:tv2}.
\begin{proof}[Proof of Claim~\ref{claim:tail_bound2}]
    The claim is equivalent to showing that with high probability ($\ll n^{-1}$), 
    $$
    \|\norm^{-1} \paren{\Mul}_{|i}^\top \paren{\Mul}_{|i}\|_{op} \leq c\,.
    $$
    Fully analogous to the proof in Claim~\ref{claim:tail_bound}, we obtain
    \begin{align*}
        \norm^{-1/2} \| \paren{\Mul}_{|i} \|_{op} &\leq O\paren{\norm^{-1/2}\deltaL\pl\sqrt{m \pr d} + \norm^{-1/2}\deltaL\sqrt{\pl \pr d} + \norm^{-1/2}\deltaL\pl^{-1/2}\sqrt{m}}\\
        &= O\paren{\deltaL\sqrt{\pl m} + \deltaL + \norm^{-1/2}\deltaL\pl^{-1/2}\sqrt{m}}\\
        &= o(1)\,,
    \end{align*}
    where we used  $\deltaL, \deltaL\sqrt{\pl m} \ll 1$ from Cond.~\eqref{eq:cond_b_app}, $\norm^{-1/2}\deltaL\pl^{-1/2}\sqrt{m} \ll \sqrt{\frac m {\pl^2 \pr d}} \ll 1$ from Cond.~\eqref{eq:cond_b_app}, \eqref{eq:cond_c_app}, which concludes the proof of the claim.
\end{proof}

\begin{proof}[Proof of Claim~\ref{claim:kl_cmp2}]
    By the same argument as in Claim~\ref{claim:kl_cmp}, it is sufficient to show that 
    $$
    \E_{\Mul,\Mr} \|\Delta_{\Mul,\Mr}\|_2^2 = \E_{\Mul,\Mr} \|\norm^{-1} \paren{\Mul}_{|i}^\top \paren{\Mul}_{|i}\|_2^2 \ll n^{-1}\,.
    $$
    We explicitly compute the expectation. Denote here $\delta = \deltaL$. $\forall a \in \sqb[m]$,
    \begin{align*}
        \norm^{2}\E_{\Mul,\Mr} \paren{\Delta_{\Mul,\Mr}}_{aa}^2 = \delta^4 \pl \pr d + \delta^4 \pl^2 \pr^2 d(d-1) = \delta^4 O\paren{\pl \pr d + \pl^2 \pr^2 d^2 }\,.
    \end{align*}
    $\forall a\neq b \in \sqb{m}$,
    \begin{align*}
        \norm^{2}\E_{\Mul,\Mr} \paren{\Delta_{\Mul,\Mr}}_{ab}^2 \leq \delta^4 \pl^2 \pr d + \delta^4 \pl^4 \pr^2 d(d-1) = \delta^4 O\paren{\pl^2 \pr d + \pl^4 \pr^2 d^2 }\,.
    \end{align*}
    Combining the two bounds together, we obtain 
    \begin{align*}
        \E_{\Mul,\Mr} \|\Delta_{\Mul,\Mr}\|_2^2 &\leq \delta^4 \norm^{-2} \cdot O\paren{ \pl \pr dm + \pl^2 \pr^2 d^2 m + \pl^2 \pr dm^2 + \pl^4 \pr^2 d^2 m^2 }\\
        & = \delta^4 O\paren{ \norm^{-1} m + m + \pl \norm^{-1} m^2 + \pl^2 m^2 }
        \ll n^{-1}\,,
    \end{align*}
    where we used $\delta^4 \norm^{-1} m, \delta^4 m  \ll n^{-1} $ by Cond.~\eqref{eq:cond_d_app} and $\delta^4 \pl \norm^{-1} m^2 = \frac{\delta^4m^2}{\pr d} \ll n^{-1}$, $\delta^4 \pl^2 m^2 \ll n^{-1}$ by Cond.~\eqref{eq:cond_e_app}.
\end{proof}
\end{proof}

\end{document}